\newif\iffull \fullfalse
\newif\ifexample \examplefalse
\newif\ifdtabs \dtabsfalse
\newif\ifdraft \drafttrue

\draftfalse

\documentclass{lmcs}

\usepackage{multirow,multicol}
\usepackage[all]{xy}
\usepackage[normalem]{ulem}

\usepackage{prop}
\usepackage{amsmath}
\usepackage{amssymb}
\usepackage{amsthm}
\usepackage{color}
\usepackage{environ}
\usepackage{enumitem}
\usepackage{thmtools}
\usepackage{url}

\makeatletter
\DeclareRobustCommand*\cal{\@fontswitch\relax\mathcal}
\makeatother

\newcommand{\TS}[1]{\ifdraft\textcolor{blue}{\textsf{TS:{#1}}}\fi}
\newcommand{\AI}[1]{\ifdraft\textcolor{red}{\textsf{AI:{#1}}}\fi}
\newcommand{\redsout}{\bgroup \markoverwith{\textcolor{red}{\rule[0.5ex]{2pt}{0.7pt}}}\ULon}

\newcommand{\lamh}{\ensuremath{\lambda^{H}}}
\newcommand{\fhfix}{\ensuremath{\mathrm{F}_{\!H}}}

\newcommand{\AND}{\mathop{\text{ and }}}

\newcommand{\defeq}{\stackrel{\text{\tiny def}}{=}}

\newcommand{\Rule}[2]
{\ensuremath{\text{({\sc{{#1}\_{#2}}})}}}
\newcommand{\R}[1]
{\Rule{R}{#1}}
\newcommand{\T}[1]
{\Rule{T}{#1}}
\newcommand{\E}[1]
{\Rule{E}{#1}}

\newcommand{\Sub}[1]
{\Rule{S}{#1}}

\newcommand{\CT}[1]
{\Rule{CT}{#1}}

\newcommand{\sect}[1]{Section~\ref{sec:#1}}
\newcommand{\fig}[1]{Figure~\ref{fig:#1}}
\newcommand{\refthm}[1]{Theorem~\ref{thm:#1}}
\newcommand{\coqname}[1]{[\texttt{#1} in \texttt{coterm.v}]}

\makeatletter

\def\reflemma@prefix#1,#2\relax{%
\ifx#2\@empty Lemma%
\else Lemmas%
\fi}

\def\reflemma@ref#1,#2,#3,#4\relax{%
\def\reflabel{\ref{lemma:#1}}%
\ifx#2\@empty\reflabel%
\else%
 \ifx#3\@empty%
  {\reflabel} and \reflemma@ref#2,\@empty,\@empty,\@empty\relax%
 \else%
  {\reflabel}, \expandafter\reflemma@ref#2,#3,#4\relax%
\fi\fi}

\def\reflemma#1{%
{\expandafter\reflemma@prefix#1,\@empty\relax} %
\expandafter\reflemma@ref#1,\@empty,\@empty,\@empty\relax%
}

\prop@constant@setter{env-theorem}{theorem}[lem]
\prop:label{lem}[Lemmas]{Lemma}
\prop:label{thm}[Theorems]{Theorem}
\prop:label{cor}[Corollaries]{Corollary}

\makeatother

\theoremstyle{plain}

\newtheorem*{theorem*}{Theorem}

\newenvironment{enmrt}{\begin{enumerate}[label=(\arabic*)]}{\end{enumerate}}

\NewEnviron{theorem}[2][]{%
 \begin{thm}[#1] \label{thm:#2}%
  \BODY%
 \end{thm}%
}

\long\def\case #1:{\item[Case]{#1}:}
\long\def\subcase #1:{\item[\textsf{Subcase} #1]}
\long\def\subsubcase #1:{\item[\textsf{Subsubcase} #1]}

\newenvironment{fhfigure*}
{\begingroup \small \begin{figure*}}
{\end{figure*} \endgroup}

\newcommand{\ottdrule}[4][]{{\displaystyle\frac{\begin{array}{l}#2\end{array}}{#3}\quad\ottdrulename{#4}}}

\newcommand{\ottpremise}[1]{ #1 \\}
\newenvironment{ottdefnblock}[3][]{ \framebox{\mbox{#2}} \quad #3 \\[0pt]}{}

\newcommand{\ottnt}[1]{\mathit{#1}}
\newcommand{\ottmv}[1]{\mathit{#1}}

\newcommand{\ottsym}[1]{#1}

\newcommand{\ottdrulename}[1]{\textsc{#1}}

\newcommand{\ottliteralin}
{ \mathsf{in} }

\newcommand{\equal}
{ \,\texttt{\string=}\, }

\definecolor{gray96}{rgb}{0.1,0.1,0.8}
\newcommand{\gb}[2][\relax]{ {%
\color{gray96}#2} }

\newcommand{\ottdruleEXXRed}[1]{\ottdrule[#1]{%
\ottpremise{\ottnt{e_{{\mathrm{1}}}}  \rightsquigarrow  \ottnt{e_{{\mathrm{2}}}}}%
}{
\ottnt{E}  [  \ottnt{e_{{\mathrm{1}}}}  ]  \longrightarrow  \ottnt{E}  [  \ottnt{e_{{\mathrm{2}}}}  ]}{%
{\ottdrulename{E\_Red}}{}%
}}

\newcommand{\ottdruleEXXBlame}[1]{\ottdrule[#1]{%
\ottpremise{\ottnt{E}  \mathrel{\neq}  \left[ \, \right]}%
}{
\ottnt{E}  [   \mathord{\Uparrow}  \ell   ]  \longrightarrow   \mathord{\Uparrow}  \ell }{%
{\ottdrulename{E\_Blame}}{}%
}}

\newcommand{\ottdruleWFXXEmpty}[1]{\ottdrule[#1]{%
}{
 \mathord{ \vdash } ~  \emptyset }{%
{\ottdrulename{WF\_Empty}}{}%
}}

\newcommand{\ottdruleWFXXExtendVar}[1]{\ottdrule[#1]{%
\ottpremise{  \mathord{ \vdash } ~  \Gamma   \quad  \Gamma  \vdash  \ottnt{T} }%
}{
 \mathord{ \vdash } ~   \Gamma  ,  \mathit{x}  \mathord{:}  \ottnt{T}  }{%
{\ottdrulename{WF\_ExtendVar}}{}%
}}

\newcommand{\ottdruleWFXXExtendTVar}[1]{\ottdrule[#1]{%
\ottpremise{ \mathord{ \vdash } ~  \Gamma }%
}{
 \mathord{ \vdash } ~  \Gamma  \ottsym{,}  \alpha }{%
{\ottdrulename{WF\_ExtendTVar}}{}%
}}

\newcommand{\ottdruleWFXXBase}[1]{\ottdrule[#1]{%
\ottpremise{ \mathord{ \vdash } ~  \Gamma }%
}{
\Gamma  \vdash  \ottnt{B}}{%
{\ottdrulename{WF\_Base}}{}%
}}

\newcommand{\ottdruleWFXXTVar}[1]{\ottdrule[#1]{%
\ottpremise{  \mathord{ \vdash } ~  \Gamma   \quad  \alpha \, \in \, \Gamma }%
}{
\Gamma  \vdash  \alpha}{%
{\ottdrulename{WF\_TVar}}{}%
}}

\newcommand{\ottdruleWFXXForall}[1]{\ottdrule[#1]{%
\ottpremise{\Gamma  \ottsym{,}  \alpha  \vdash  \ottnt{T}}%
}{
\Gamma  \vdash   \forall   \alpha  .  \ottnt{T} }{%
{\ottdrulename{WF\_Forall}}{}%
}}

\newcommand{\ottdruleWFXXFun}[1]{\ottdrule[#1]{%
\ottpremise{ \Gamma  \vdash  \ottnt{T_{{\mathrm{1}}}}  \quad   \Gamma  ,  \mathit{x}  \mathord{:}  \ottnt{T_{{\mathrm{1}}}}   \vdash  \ottnt{T_{{\mathrm{2}}}} }%
}{
\Gamma  \vdash  {}  \mathit{x}  \ottsym{:}  \ottnt{T_{{\mathrm{1}}}}  \rightarrow  \ottnt{T_{{\mathrm{2}}}}  {}}{%
{\ottdrulename{WF\_Fun}}{}%
}}

\newcommand{\ottdruleWFXXRefine}[1]{\ottdrule[#1]{%
\ottpremise{ \Gamma  \vdash  \ottnt{T}  \quad   \Gamma  ,  \mathit{x}  \mathord{:}  \ottnt{T}   \vdash  \ottnt{e}  \ottsym{:}   \mathsf{Bool}  }%
}{
\Gamma  \vdash   \{  \mathit{x}  \mathord{:}  \ottnt{T}   \mathop{\mid}   \ottnt{e}  \} }{%
{\ottdrulename{WF\_Refine}}{}%
}}

\newcommand{\ottdruleTXXVar}[1]{\ottdrule[#1]{%
\ottpremise{  \mathord{ \vdash } ~  \Gamma   \quad   \mathit{x}  \mathord{:}  \ottnt{T}   \in   \Gamma  }%
}{
\Gamma  \vdash  \mathit{x}  \ottsym{:}  \ottnt{T}}{%
{\ottdrulename{T\_Var}}{}%
}}

\newcommand{\ottdruleTXXConst}[1]{\ottdrule[#1]{%
\ottpremise{ \mathord{ \vdash } ~  \Gamma }%
}{
\Gamma  \vdash  \ottnt{k}  \ottsym{:}   \mathsf{ty}  (  \ottnt{k}  ) }{%
{\ottdrulename{T\_Const}}{}%
}}

\newcommand{\ottdruleTXXBlame}[1]{\ottdrule[#1]{%
\ottpremise{  \mathord{ \vdash } ~  \Gamma   \quad  \emptyset  \vdash  \ottnt{T} }%
}{
\Gamma  \vdash   \mathord{\Uparrow}  \ell   \ottsym{:}  \ottnt{T}}{%
{\ottdrulename{T\_Blame}}{}%
}}

\newcommand{\ottdruleTXXAbs}[1]{\ottdrule[#1]{%
\ottpremise{ \Gamma  ,  \mathit{x}  \mathord{:}  \ottnt{T_{{\mathrm{1}}}}   \vdash  \ottnt{e}  \ottsym{:}  \ottnt{T_{{\mathrm{2}}}}}%
}{
\Gamma  \vdash    \lambda    \mathit{x}  \mathord{:}  \ottnt{T_{{\mathrm{1}}}}  .  \ottnt{e}   \ottsym{:}  \ottsym{(}   \mathit{x} \mathord{:} \ottnt{T_{{\mathrm{1}}}} \rightarrow \ottnt{T_{{\mathrm{2}}}}   \ottsym{)}}{%
{\ottdrulename{T\_Abs}}{}%
}}

\newcommand{\ottdruleTXXApp}[1]{\ottdrule[#1]{%
\ottpremise{ \Gamma  \vdash  \ottnt{e_{{\mathrm{1}}}}  \ottsym{:}  \ottsym{(}   \mathit{x} \mathord{:} \ottnt{T_{{\mathrm{1}}}} \rightarrow \ottnt{T_{{\mathrm{2}}}}   \ottsym{)}  \quad   \Gamma  \vdash  \ottnt{e_{{\mathrm{2}}}}  \ottsym{:}  \ottnt{T_{{\mathrm{1}}}}  \quad  \Gamma  \vdash  \ottnt{T_{{\mathrm{2}}}} \, [  \ottnt{e_{{\mathrm{2}}}}  \ottsym{/}  \mathit{x}  ]  }%
}{
\Gamma  \vdash  \ottnt{e_{{\mathrm{1}}}} \, \ottnt{e_{{\mathrm{2}}}}  \ottsym{:}  \ottnt{T_{{\mathrm{2}}}} \, [  \ottnt{e_{{\mathrm{2}}}}  \ottsym{/}  \mathit{x}  ]}{%
{\ottdrulename{T\_App}}{}%
}}

\newcommand{\ottdruleTXXTAbs}[1]{\ottdrule[#1]{%
\ottpremise{\Gamma  \ottsym{,}  \alpha  \vdash  \ottnt{e}  \ottsym{:}  \ottnt{T}}%
}{
\Gamma  \vdash   \Lambda\!  \, \alpha  .~  \ottnt{e}  \ottsym{:}   \forall   \alpha  .  \ottnt{T} }{%
{\ottdrulename{T\_TAbs}}{}%
}}

\newcommand{\ottdruleTXXTAbsOne}[1]{\ottdrule[#1]{%
\ottpremise{\Gamma  \ottsym{,}  \alpha  \vdash  \ottnt{v}  \ottsym{:}  \ottnt{T}}%
}{
\Gamma  \vdash   \Lambda\!  \, \alpha  .~  \ottnt{v}  \ottsym{:}   \forall   \alpha  .  \ottnt{T} }{%
{\ottdrulename{T\_TAbs1}}{}%
}}

\newcommand{\ottdruleTXXTAbsTwo}[1]{\ottdrule[#1]{%
\ottpremise{\Gamma  \ottsym{,}  \alpha  \vdash  \langle  \ottnt{T_{{\mathrm{1}}}}  \Rightarrow  \ottnt{T_{{\mathrm{2}}}}  \rangle   ^{ \ell }  \, \ottsym{(}  \ottnt{v} \, \alpha  \ottsym{)}  \ottsym{:}  \ottnt{T}}%
}{
\Gamma  \vdash   \Lambda\!  \, \alpha  .~  \langle  \ottnt{T_{{\mathrm{1}}}}  \Rightarrow  \ottnt{T_{{\mathrm{2}}}}  \rangle   ^{ \ell }  \, \ottsym{(}  \ottnt{v} \, \alpha  \ottsym{)}  \ottsym{:}   \forall   \alpha  .  \ottnt{T} }{%
{\ottdrulename{T\_TAbs2}}{}%
}}

\newcommand{\ottdruleTXXTApp}[1]{\ottdrule[#1]{%
\ottpremise{ \Gamma  \vdash  \ottnt{e}  \ottsym{:}   \forall   \alpha  .  \ottnt{T_{{\mathrm{1}}}}   \quad  \Gamma  \vdash  \ottnt{T_{{\mathrm{2}}}} }%
}{
\Gamma  \vdash  \ottnt{e} \, \ottnt{T_{{\mathrm{2}}}}  \ottsym{:}  \ottnt{T_{{\mathrm{1}}}} \, [  \ottnt{T_{{\mathrm{2}}}}  \ottsym{/}  \alpha  ]}{%
{\ottdrulename{T\_TApp}}{}%
}}

\newcommand{\ottdruleTXXCast}[1]{\ottdrule[#1]{%
\ottpremise{ \Gamma  \vdash  \ottnt{T_{{\mathrm{1}}}}  \quad   \Gamma  \vdash  \ottnt{T_{{\mathrm{2}}}}  \quad  \ottnt{T_{{\mathrm{1}}}}  \mathrel{\parallel}  \ottnt{T_{{\mathrm{2}}}}  }%
}{
\Gamma  \vdash  \langle  \ottnt{T_{{\mathrm{1}}}}  \Rightarrow  \ottnt{T_{{\mathrm{2}}}}  \rangle   ^{ \ell }   \ottsym{:}  \ottnt{T_{{\mathrm{1}}}}  \rightarrow  \ottnt{T_{{\mathrm{2}}}}}{%
{\ottdrulename{T\_Cast}}{}%
}}

\newcommand{\ottdruleTXXWCheck}[1]{\ottdrule[#1]{%
\ottpremise{ \Gamma  \vdash   \{  \mathit{x}  \mathord{:}  \ottnt{T_{{\mathrm{1}}}}   \mathop{\mid}   \ottnt{e_{{\mathrm{1}}}}  \}   \quad  \Gamma  \vdash  \ottnt{e_{{\mathrm{2}}}}  \ottsym{:}  \ottnt{T_{{\mathrm{1}}}} }%
}{
\Gamma  \vdash   \langle\!\langle  \,  \{  \mathit{x}  \mathord{:}  \ottnt{T_{{\mathrm{1}}}}   \mathop{\mid}   \ottnt{e_{{\mathrm{1}}}}  \}   \ottsym{,}  \ottnt{e_{{\mathrm{2}}}} \,  \rangle\!\rangle  \,  ^{ \ell }   \ottsym{:}   \{  \mathit{x}  \mathord{:}  \ottnt{T_{{\mathrm{1}}}}   \mathop{\mid}   \ottnt{e_{{\mathrm{1}}}}  \} }{%
{\ottdrulename{T\_WCheck}}{}%
}}

\newcommand{\ottdruleTXXACheck}[1]{\ottdrule[#1]{%
\ottpremise{  \mathord{ \vdash } ~  \Gamma   \quad   \emptyset  \vdash   \{  \mathit{x}  \mathord{:}  \ottnt{T_{{\mathrm{1}}}}   \mathop{\mid}   \ottnt{e_{{\mathrm{1}}}}  \}   \quad   \emptyset  \vdash  \ottnt{e_{{\mathrm{2}}}}  \ottsym{:}   \mathsf{Bool}   \quad   \emptyset  \vdash  \ottnt{v}  \ottsym{:}  \ottnt{T_{{\mathrm{1}}}}  \quad  \ottnt{e_{{\mathrm{1}}}} \, [  \ottnt{v}  \ottsym{/}  \mathit{x}  ]  \longrightarrow^{\ast}  \ottnt{e_{{\mathrm{2}}}}    }%
}{
\Gamma  \vdash  \langle   \{  \mathit{x}  \mathord{:}  \ottnt{T_{{\mathrm{1}}}}   \mathop{\mid}   \ottnt{e_{{\mathrm{1}}}}  \}   \ottsym{,}  \ottnt{e_{{\mathrm{2}}}}  \ottsym{,}  \ottnt{v}  \rangle   ^{ \ell }   \ottsym{:}   \{  \mathit{x}  \mathord{:}  \ottnt{T_{{\mathrm{1}}}}   \mathop{\mid}   \ottnt{e_{{\mathrm{1}}}}  \} }{%
{\ottdrulename{T\_ACheck}}{}%
}}

\newcommand{\ottdruleTXXConv}[1]{\ottdrule[#1]{%
\ottpremise{  \mathord{ \vdash } ~  \Gamma   \quad   \emptyset  \vdash  \ottnt{e}  \ottsym{:}  \ottnt{T_{{\mathrm{1}}}}  \quad   \ottnt{T_{{\mathrm{1}}}}  \equiv  \ottnt{T_{{\mathrm{2}}}}  \quad  \emptyset  \vdash  \ottnt{T_{{\mathrm{2}}}}   }%
}{
\Gamma  \vdash  \ottnt{e}  \ottsym{:}  \ottnt{T_{{\mathrm{2}}}}}{%
{\ottdrulename{T\_Conv}}{}%
}}

\newcommand{\ottdruleTXXForget}[1]{\ottdrule[#1]{%
\ottpremise{  \mathord{ \vdash } ~  \Gamma   \quad  \emptyset  \vdash  \ottnt{v}  \ottsym{:}   \{  \mathit{x}  \mathord{:}  \ottnt{T}   \mathop{\mid}   \ottnt{e}  \}  }%
}{
\Gamma  \vdash  \ottnt{v}  \ottsym{:}  \ottnt{T}}{%
{\ottdrulename{T\_Forget}}{}%
}}

\newcommand{\ottdruleTXXExact}[1]{\ottdrule[#1]{%
\ottpremise{  \mathord{ \vdash } ~  \Gamma   \quad   \emptyset  \vdash  \ottnt{v}  \ottsym{:}  \ottnt{T}  \quad   \emptyset  \vdash   \{  \mathit{x}  \mathord{:}  \ottnt{T}   \mathop{\mid}   \ottnt{e}  \}   \quad  \ottnt{e} \, [  \ottnt{v}  \ottsym{/}  \mathit{x}  ]  \longrightarrow^{\ast}   \mathsf{true}    }%
}{
\Gamma  \vdash  \ottnt{v}  \ottsym{:}   \{  \mathit{x}  \mathord{:}  \ottnt{T}   \mathop{\mid}   \ottnt{e}  \} }{%
{\ottdrulename{T\_Exact}}{}%
}}

\newcommand{\ottdruleCXXBase}[1]{\ottdrule[#1]{%
}{
\ottnt{B}  \mathrel{\parallel}  \ottnt{B}}{%
{\ottdrulename{C\_Base}}{}%
}}

\newcommand{\ottdruleCXXTVar}[1]{\ottdrule[#1]{%
}{
\alpha  \mathrel{\parallel}  \alpha}{%
{\ottdrulename{C\_TVar}}{}%
}}

\newcommand{\ottdruleCXXRefineL}[1]{\ottdrule[#1]{%
\ottpremise{\ottnt{T_{{\mathrm{1}}}}  \mathrel{\parallel}  \ottnt{T_{{\mathrm{2}}}}}%
}{
 \{  \mathit{x}  \mathord{:}  \ottnt{T_{{\mathrm{1}}}}   \mathop{\mid}   \ottnt{e}  \}   \mathrel{\parallel}  \ottnt{T_{{\mathrm{2}}}}}{%
{\ottdrulename{C\_RefineL}}{}%
}}

\newcommand{\ottdruleCXXRefineR}[1]{\ottdrule[#1]{%
\ottpremise{\ottnt{T_{{\mathrm{1}}}}  \mathrel{\parallel}  \ottnt{T_{{\mathrm{2}}}}}%
}{
\ottnt{T_{{\mathrm{1}}}}  \mathrel{\parallel}   \{  \mathit{x}  \mathord{:}  \ottnt{T_{{\mathrm{2}}}}   \mathop{\mid}   \ottnt{e}  \} }{%
{\ottdrulename{C\_RefineR}}{}%
}}

\newcommand{\ottdruleCXXFun}[1]{\ottdrule[#1]{%
\ottpremise{ \ottnt{T_{{\mathrm{11}}}}  \mathrel{\parallel}  \ottnt{T_{{\mathrm{21}}}}  \quad  \ottnt{T_{{\mathrm{12}}}}  \mathrel{\parallel}  \ottnt{T_{{\mathrm{22}}}} }%
}{
 \mathit{x} \mathord{:} \ottnt{T_{{\mathrm{11}}}} \rightarrow \ottnt{T_{{\mathrm{12}}}}   \mathrel{\parallel}   \mathit{y} \mathord{:} \ottnt{T_{{\mathrm{21}}}} \rightarrow \ottnt{T_{{\mathrm{22}}}} }{%
{\ottdrulename{C\_Fun}}{}%
}}

\newcommand{\ottdruleCXXForall}[1]{\ottdrule[#1]{%
\ottpremise{\ottnt{T_{{\mathrm{1}}}}  \mathrel{\parallel}  \ottnt{T_{{\mathrm{2}}}}}%
}{
 \forall   \alpha  .  \ottnt{T_{{\mathrm{1}}}}   \mathrel{\parallel}   \forall   \alpha  .  \ottnt{T_{{\mathrm{2}}}} }{%
{\ottdrulename{C\_Forall}}{}%
}}

\newcommand{\ottdruleCWXXBase}[1]{\ottdrule[#1]{%
\ottpremise{ \mathord{ \vdash } ~  \Gamma }%
}{
 \Gamma   \vdash   \ottnt{B}   \ottsym{:}    \overline{ \Gamma_{\ottmv{i}}  \vdash  \ottnt{e_{\ottmv{i}}}  \ottsym{:}  \ottnt{T_{\ottmv{i}}} }^{ \ottmv{i} }   \mathrel{\circ\hspace{-.4em}\rightarrow} \ast }{%
{\ottdrulename{CW\_Base}}{}%
}}

\newcommand{\ottdruleCWXXTVar}[1]{\ottdrule[#1]{%
\ottpremise{  \mathord{ \vdash } ~  \Gamma   \quad  \alpha \, \in \, \Gamma }%
}{
 \Gamma   \vdash   \alpha   \ottsym{:}    \overline{ \Gamma_{\ottmv{i}}  \vdash  \ottnt{e_{\ottmv{i}}}  \ottsym{:}  \ottnt{T_{\ottmv{i}}} }^{ \ottmv{i} }   \mathrel{\circ\hspace{-.4em}\rightarrow} \ast }{%
{\ottdrulename{CW\_TVar}}{}%
}}

\newcommand{\ottdruleCWXXForall}[1]{\ottdrule[#1]{%
\ottpremise{ \Gamma  \ottsym{,}  \alpha   \vdash   \ottnt{T}^\ottnt{C}   \ottsym{:}    \overline{ \Gamma_{\ottmv{i}}  \vdash  \ottnt{e_{\ottmv{i}}}  \ottsym{:}  \ottnt{T_{\ottmv{i}}} }^{ \ottmv{i} }   \mathrel{\circ\hspace{-.4em}\rightarrow} \ast }%
}{
 \Gamma   \vdash    \forall   \alpha  .  \ottnt{T}^\ottnt{C}    \ottsym{:}    \overline{ \Gamma_{\ottmv{i}}  \vdash  \ottnt{e_{\ottmv{i}}}  \ottsym{:}  \ottnt{T_{\ottmv{i}}} }^{ \ottmv{i} }   \mathrel{\circ\hspace{-.4em}\rightarrow} \ast }{%
{\ottdrulename{CW\_Forall}}{}%
}}

\newcommand{\ottdruleCWXXFun}[1]{\ottdrule[#1]{%
\ottpremise{  \Gamma   \vdash   \ottnt{T}^\ottnt{C}_{{\mathrm{1}}}   \ottsym{:}    \overline{ \Gamma_{\ottmv{i}}  \vdash  \ottnt{e_{\ottmv{i}}}  \ottsym{:}  \ottnt{T_{\ottmv{i}}} }^{ \ottmv{i} }   \mathrel{\circ\hspace{-.4em}\rightarrow} \ast   \quad    \Gamma  ,  \mathit{x}  \mathord{:}  \ottnt{T}^\ottnt{C}_{{\mathrm{1}}}  [   \overline{ \ottnt{e_{\ottmv{i}}} }^{ \ottmv{i} }   ]    \vdash   \ottnt{T}^\ottnt{C}_{{\mathrm{2}}}   \ottsym{:}    \overline{ \Gamma_{\ottmv{i}}  \vdash  \ottnt{e_{\ottmv{i}}}  \ottsym{:}  \ottnt{T_{\ottmv{i}}} }^{ \ottmv{i} }   \mathrel{\circ\hspace{-.4em}\rightarrow} \ast  }%
}{
 \Gamma   \vdash    \mathit{x} \mathord{:} \ottnt{T}^\ottnt{C}_{{\mathrm{1}}} \rightarrow \ottnt{T}^\ottnt{C}_{{\mathrm{2}}}    \ottsym{:}    \overline{ \Gamma_{\ottmv{i}}  \vdash  \ottnt{e_{\ottmv{i}}}  \ottsym{:}  \ottnt{T_{\ottmv{i}}} }^{ \ottmv{i} }   \mathrel{\circ\hspace{-.4em}\rightarrow} \ast }{%
{\ottdrulename{CW\_Fun}}{}%
}}

\newcommand{\ottdruleCWXXRefine}[1]{\ottdrule[#1]{%
\ottpremise{  \Gamma   \vdash   \ottnt{T}^\ottnt{C}   \ottsym{:}    \overline{ \Gamma_{\ottmv{i}}  \vdash  \ottnt{e_{\ottmv{i}}}  \ottsym{:}  \ottnt{T_{\ottmv{i}}} }^{ \ottmv{i} }   \mathrel{\circ\hspace{-.4em}\rightarrow} \ast   \quad   \Gamma  ,  \mathit{x}  \mathord{:}  \ottnt{T}^\ottnt{C}  [   \overline{ \ottnt{e_{\ottmv{i}}} }^{ \ottmv{i} }   ]   \vdash  \ottnt{C}  \ottsym{:}   \overline{ \Gamma_{\ottmv{i}}  \vdash  \ottnt{e_{\ottmv{i}}}  \ottsym{:}  \ottnt{T_{\ottmv{i}}} }^{ \ottmv{i} }   \mathrel{\circ\hspace{-.4em}\rightarrow}   \mathsf{Bool}  }%
}{
 \Gamma   \vdash    \{  \mathit{x}  \mathord{:}  \ottnt{T}^\ottnt{C}   \mathop{\mid}   \ottnt{C}  \}    \ottsym{:}    \overline{ \Gamma_{\ottmv{i}}  \vdash  \ottnt{e_{\ottmv{i}}}  \ottsym{:}  \ottnt{T_{\ottmv{i}}} }^{ \ottmv{i} }   \mathrel{\circ\hspace{-.4em}\rightarrow} \ast }{%
{\ottdrulename{CW\_Refine}}{}%
}}

\newcommand{\ottdruleCTXXHole}[1]{\ottdrule[#1]{%
}{
\Gamma_{\ottmv{j}}  \vdash   \left[ \, \right] _{ \ottmv{j} }   \ottsym{:}   \overline{ \Gamma_{\ottmv{i}}  \vdash  \ottnt{e_{\ottmv{i}}}  \ottsym{:}  \ottnt{T_{\ottmv{i}}} }^{ \ottmv{i} }   \mathrel{\circ\hspace{-.4em}\rightarrow}  \ottnt{T_{\ottmv{j}}}}{%
{\ottdrulename{CT\_Hole}}{}%
}}

\newcommand{\ottdruleCTXXVar}[1]{\ottdrule[#1]{%
\ottpremise{  \mathord{ \vdash } ~  \Gamma   \quad   \mathit{x}  \mathord{:}  \ottnt{T}   \in   \Gamma  }%
}{
\Gamma  \vdash  \mathit{x}  \ottsym{:}   \overline{ \Gamma_{\ottmv{i}}  \vdash  \ottnt{e_{\ottmv{i}}}  \ottsym{:}  \ottnt{T_{\ottmv{i}}} }^{ \ottmv{i} }   \mathrel{\circ\hspace{-.4em}\rightarrow}  \ottnt{T}}{%
{\ottdrulename{CT\_Var}}{}%
}}

\newcommand{\ottdruleCTXXConst}[1]{\ottdrule[#1]{%
\ottpremise{ \mathord{ \vdash } ~  \Gamma }%
}{
\Gamma  \vdash  \ottnt{k}  \ottsym{:}   \overline{ \Gamma_{\ottmv{i}}  \vdash  \ottnt{e_{\ottmv{i}}}  \ottsym{:}  \ottnt{T_{\ottmv{i}}} }^{ \ottmv{i} }   \mathrel{\circ\hspace{-.4em}\rightarrow}   \mathsf{ty}  (  \ottnt{k}  ) }{%
{\ottdrulename{CT\_Const}}{}%
}}

\newcommand{\ottdruleCTXXAbs}[1]{\ottdrule[#1]{%
\ottpremise{ \Gamma  ,  \mathit{x}  \mathord{:}  \ottnt{T}^\ottnt{C}_{{\mathrm{1}}}  [   \overline{ \ottnt{e_{\ottmv{i}}} }^{ \ottmv{i} }   ]   \vdash  \ottnt{C}  \ottsym{:}   \overline{ \Gamma_{\ottmv{i}}  \vdash  \ottnt{e_{\ottmv{i}}}  \ottsym{:}  \ottnt{T_{\ottmv{i}}} }^{ \ottmv{i} }   \mathrel{\circ\hspace{-.4em}\rightarrow}  \ottnt{T'_{{\mathrm{2}}}}}%
}{
\Gamma  \vdash    \lambda    \mathit{x}  \mathord{:}  \ottnt{T}^\ottnt{C}_{{\mathrm{1}}}  .  \ottnt{C}   \ottsym{:}   \overline{ \Gamma_{\ottmv{i}}  \vdash  \ottnt{e_{\ottmv{i}}}  \ottsym{:}  \ottnt{T_{\ottmv{i}}} }^{ \ottmv{i} }   \mathrel{\circ\hspace{-.4em}\rightarrow}   \mathit{x} \mathord{:} \ottnt{T}^\ottnt{C}_{{\mathrm{1}}}  [   \overline{ \ottnt{e_{\ottmv{i}}} }^{ \ottmv{i} }   ] \rightarrow \ottnt{T'_{{\mathrm{2}}}} }{%
{\ottdrulename{CT\_Abs}}{}%
}}

\newcommand{\ottdruleCTXXApp}[1]{\ottdrule[#1]{%
\ottpremise{ \Gamma  \vdash  \ottnt{C_{{\mathrm{1}}}}  \ottsym{:}   \overline{ \Gamma_{\ottmv{i}}  \vdash  \ottnt{e_{\ottmv{i}}}  \ottsym{:}  \ottnt{T_{\ottmv{i}}} }^{ \ottmv{i} }   \mathrel{\circ\hspace{-.4em}\rightarrow}   \mathit{x} \mathord{:} \ottnt{T'_{{\mathrm{1}}}} \rightarrow \ottnt{T'_{{\mathrm{2}}}}   \quad  \Gamma  \vdash  \ottnt{C_{{\mathrm{2}}}}  \ottsym{:}   \overline{ \Gamma_{\ottmv{i}}  \vdash  \ottnt{e_{\ottmv{i}}}  \ottsym{:}  \ottnt{T_{\ottmv{i}}} }^{ \ottmv{i} }   \mathrel{\circ\hspace{-.4em}\rightarrow}  \ottnt{T'_{{\mathrm{1}}}} }%
}{
\Gamma  \vdash  \ottnt{C_{{\mathrm{1}}}} \, \ottnt{C_{{\mathrm{2}}}}  \ottsym{:}   \overline{ \Gamma_{\ottmv{i}}  \vdash  \ottnt{e_{\ottmv{i}}}  \ottsym{:}  \ottnt{T_{\ottmv{i}}} }^{ \ottmv{i} }   \mathrel{\circ\hspace{-.4em}\rightarrow}  \ottnt{T'_{{\mathrm{2}}}} \, [  \ottnt{C_{{\mathrm{2}}}}  [   \overline{ \ottnt{e_{\ottmv{i}}} }^{ \ottmv{i} }   ]  \ottsym{/}  \mathit{x}  ]}{%
{\ottdrulename{CT\_App}}{}%
}}

\newcommand{\ottdruleCTXXTAbs}[1]{\ottdrule[#1]{%
\ottpremise{\Gamma  \ottsym{,}  \alpha  \vdash  \ottnt{C}  \ottsym{:}   \overline{ \Gamma_{\ottmv{i}}  \vdash  \ottnt{e_{\ottmv{i}}}  \ottsym{:}  \ottnt{T_{\ottmv{i}}} }^{ \ottmv{i} }   \mathrel{\circ\hspace{-.4em}\rightarrow}  \ottnt{T}}%
}{
\Gamma  \vdash   \Lambda\!  \, \alpha  .~  \ottnt{C}  \ottsym{:}   \overline{ \Gamma_{\ottmv{i}}  \vdash  \ottnt{e_{\ottmv{i}}}  \ottsym{:}  \ottnt{T_{\ottmv{i}}} }^{ \ottmv{i} }   \mathrel{\circ\hspace{-.4em}\rightarrow}   \forall   \alpha  .  \ottnt{T} }{%
{\ottdrulename{CT\_TAbs}}{}%
}}

\newcommand{\ottdruleCTXXTAbsOne}[1]{\ottdrule[#1]{%
\ottpremise{\Gamma  \ottsym{,}  \alpha  \vdash  \ottnt{C}  \ottsym{:}   \overline{ \Gamma_{\ottmv{i}}  \vdash  \ottnt{e_{\ottmv{i}}}  \ottsym{:}  \ottnt{T_{\ottmv{i}}} }^{ \ottmv{i} }   \mathrel{\circ\hspace{-.4em}\rightarrow}  \ottnt{T}}%
}{
\Gamma  \vdash   \Lambda\!  \, \alpha  .~  \ottnt{C}  \ottsym{:}   \overline{ \Gamma_{\ottmv{i}}  \vdash  \ottnt{e_{\ottmv{i}}}  \ottsym{:}  \ottnt{T_{\ottmv{i}}} }^{ \ottmv{i} }   \mathrel{\circ\hspace{-.4em}\rightarrow}   \forall   \alpha  .  \ottnt{T} }{%
{\ottdrulename{CT\_TAbs1}}{}%
}}

\newcommand{\ottdruleCTXXTAbsTwo}[1]{\ottdrule[#1]{%
\ottpremise{\Gamma  \ottsym{,}  \alpha  \vdash  \langle  \ottnt{T}^\ottnt{C}_{{\mathrm{1}}}  \Rightarrow  \ottnt{T}^\ottnt{C}_{{\mathrm{2}}}  \rangle   ^{ \ell }  \, \ottsym{(}  \ottnt{C} \, \alpha  \ottsym{)}  \ottsym{:}   \overline{ \Gamma_{\ottmv{i}}  \vdash  \ottnt{e_{\ottmv{i}}}  \ottsym{:}  \ottnt{T_{\ottmv{i}}} }^{ \ottmv{i} }   \mathrel{\circ\hspace{-.4em}\rightarrow}  \ottnt{T}}%
}{
\Gamma  \vdash   \Lambda\!  \, \alpha  .~  \langle  \ottnt{T}^\ottnt{C}_{{\mathrm{1}}}  \Rightarrow  \ottnt{T}^\ottnt{C}_{{\mathrm{2}}}  \rangle   ^{ \ell }  \, \ottsym{(}  \ottnt{C} \, \alpha  \ottsym{)}  \ottsym{:}   \overline{ \Gamma_{\ottmv{i}}  \vdash  \ottnt{e_{\ottmv{i}}}  \ottsym{:}  \ottnt{T_{\ottmv{i}}} }^{ \ottmv{i} }   \mathrel{\circ\hspace{-.4em}\rightarrow}   \forall   \alpha  .  \ottnt{T} }{%
{\ottdrulename{CT\_TAbs2}}{}%
}}

\newcommand{\ottdruleCTXXTApp}[1]{\ottdrule[#1]{%
\ottpremise{ \Gamma  \vdash  \ottnt{C}  \ottsym{:}   \overline{ \Gamma_{\ottmv{i}}  \vdash  \ottnt{e_{\ottmv{i}}}  \ottsym{:}  \ottnt{T_{\ottmv{i}}} }^{ \ottmv{i} }   \mathrel{\circ\hspace{-.4em}\rightarrow}   \forall   \alpha  .  \ottnt{T'}   \quad   \Gamma   \vdash   \ottnt{T}^\ottnt{C}   \ottsym{:}    \overline{ \Gamma_{\ottmv{i}}  \vdash  \ottnt{e_{\ottmv{i}}}  \ottsym{:}  \ottnt{T_{\ottmv{i}}} }^{ \ottmv{i} }   \mathrel{\circ\hspace{-.4em}\rightarrow} \ast  }%
}{
\Gamma  \vdash  \ottnt{C} \, \ottnt{T}^\ottnt{C}  \ottsym{:}   \overline{ \Gamma_{\ottmv{i}}  \vdash  \ottnt{e_{\ottmv{i}}}  \ottsym{:}  \ottnt{T_{\ottmv{i}}} }^{ \ottmv{i} }   \mathrel{\circ\hspace{-.4em}\rightarrow}  \ottnt{T'} \, [  \ottnt{T}^\ottnt{C}  [   \overline{ \ottnt{e_{\ottmv{i}}} }^{ \ottmv{i} }   ]  \ottsym{/}  \alpha  ]}{%
{\ottdrulename{CT\_TApp}}{}%
}}

\newcommand{\ottdruleCTXXCast}[1]{\ottdrule[#1]{%
\ottpremise{  \Gamma   \vdash   \ottnt{T}^\ottnt{C}_{{\mathrm{1}}}   \ottsym{:}    \overline{ \Gamma_{\ottmv{i}}  \vdash  \ottnt{e_{\ottmv{i}}}  \ottsym{:}  \ottnt{T_{\ottmv{i}}} }^{ \ottmv{i} }   \mathrel{\circ\hspace{-.4em}\rightarrow} \ast   \quad    \Gamma   \vdash   \ottnt{T}^\ottnt{C}_{{\mathrm{2}}}   \ottsym{:}    \overline{ \Gamma_{\ottmv{i}}  \vdash  \ottnt{e_{\ottmv{i}}}  \ottsym{:}  \ottnt{T_{\ottmv{i}}} }^{ \ottmv{i} }   \mathrel{\circ\hspace{-.4em}\rightarrow} \ast   \quad  \ottnt{T}^\ottnt{C}_{{\mathrm{1}}}  [   \overline{ \ottnt{e_{\ottmv{i}}} }^{ \ottmv{i} }   ]  \mathrel{\parallel}  \ottnt{T}^\ottnt{C}_{{\mathrm{2}}}  [   \overline{ \ottnt{e_{\ottmv{i}}} }^{ \ottmv{i} }   ]  }%
}{
\Gamma  \vdash  \langle  \ottnt{T}^\ottnt{C}_{{\mathrm{1}}}  \Rightarrow  \ottnt{T}^\ottnt{C}_{{\mathrm{2}}}  \rangle   ^{ \ell }   \ottsym{:}   \overline{ \Gamma_{\ottmv{i}}  \vdash  \ottnt{e_{\ottmv{i}}}  \ottsym{:}  \ottnt{T_{\ottmv{i}}} }^{ \ottmv{i} }   \mathrel{\circ\hspace{-.4em}\rightarrow}  \ottsym{(}  \ottnt{T}^\ottnt{C}_{{\mathrm{1}}}  \rightarrow  \ottnt{T}^\ottnt{C}_{{\mathrm{2}}}  \ottsym{)}  [   \overline{ \ottnt{e_{\ottmv{i}}} }^{ \ottmv{i} }   ]}{%
{\ottdrulename{CT\_Cast}}{}%
}}

\newcommand{\ottdruleCTXXConv}[1]{\ottdrule[#1]{%
\ottpremise{  \mathord{ \vdash } ~  \Gamma   \quad   \emptyset  \vdash  \ottnt{C}  \ottsym{:}   \overline{ \Gamma_{\ottmv{i}}  \vdash  \ottnt{e_{\ottmv{i}}}  \ottsym{:}  \ottnt{T_{\ottmv{i}}} }^{ \ottmv{i} }   \mathrel{\circ\hspace{-.4em}\rightarrow}  \ottnt{T_{{\mathrm{1}}}}  \quad   \ottnt{T_{{\mathrm{1}}}}  \equiv  \ottnt{T_{{\mathrm{2}}}}  \quad  \emptyset  \vdash  \ottnt{T_{{\mathrm{2}}}}   }%
}{
\Gamma  \vdash  \ottnt{C}  \ottsym{:}   \overline{ \Gamma_{\ottmv{i}}  \vdash  \ottnt{e_{\ottmv{i}}}  \ottsym{:}  \ottnt{T_{\ottmv{i}}} }^{ \ottmv{i} }   \mathrel{\circ\hspace{-.4em}\rightarrow}  \ottnt{T_{{\mathrm{2}}}}}{%
{\ottdrulename{CT\_Conv}}{}%
}}

\newcommand{\ottdruleCTXXWCheck}[1]{\ottdrule[#1]{%
\ottpremise{  \Gamma   \vdash    \{  \mathit{x}  \mathord{:}  \ottnt{T}^\ottnt{C}_{{\mathrm{1}}}   \mathop{\mid}   \ottnt{C_{{\mathrm{1}}}}  \}    \ottsym{:}    \overline{ \Gamma_{\ottmv{i}}  \vdash  \ottnt{e_{\ottmv{i}}}  \ottsym{:}  \ottnt{T_{\ottmv{i}}} }^{ \ottmv{i} }   \mathrel{\circ\hspace{-.4em}\rightarrow} \ast   \quad  \Gamma  \vdash  \ottnt{C_{{\mathrm{2}}}}  \ottsym{:}   \overline{ \Gamma_{\ottmv{i}}  \vdash  \ottnt{e_{\ottmv{i}}}  \ottsym{:}  \ottnt{T_{\ottmv{i}}} }^{ \ottmv{i} }   \mathrel{\circ\hspace{-.4em}\rightarrow}  \ottnt{T}^\ottnt{C}_{{\mathrm{1}}}  [   \overline{ \ottnt{e_{\ottmv{i}}} }^{ \ottmv{i} }   ] }%
}{
\Gamma  \vdash   \langle\!\langle  \,  \{  \mathit{x}  \mathord{:}  \ottnt{T}^\ottnt{C}_{{\mathrm{1}}}   \mathop{\mid}   \ottnt{C_{{\mathrm{1}}}}  \}   \ottsym{,}  \ottnt{C_{{\mathrm{2}}}} \,  \rangle\!\rangle  \,  ^{ \ell }   \ottsym{:}   \overline{ \Gamma_{\ottmv{i}}  \vdash  \ottnt{e_{\ottmv{i}}}  \ottsym{:}  \ottnt{T_{\ottmv{i}}} }^{ \ottmv{i} }   \mathrel{\circ\hspace{-.4em}\rightarrow}   \{  \mathit{x}  \mathord{:}  \ottnt{T}^\ottnt{C}_{{\mathrm{1}}}   \mathop{\mid}   \ottnt{C_{{\mathrm{1}}}}  \}   [   \overline{ \ottnt{e_{\ottmv{i}}} }^{ \ottmv{i} }   ]}{%
{\ottdrulename{CT\_WCheck}}{}%
}}

\newcommand{\ottdruleCTXXACheck}[1]{\ottdrule[#1]{%
\ottpremise{  \mathord{ \vdash } ~  \Gamma   \quad    \emptyset   \vdash    \{  \mathit{x}  \mathord{:}  \ottnt{T}^\ottnt{C}_{{\mathrm{1}}}   \mathop{\mid}   \ottnt{C_{{\mathrm{1}}}}  \}    \ottsym{:}    \overline{ \Gamma_{\ottmv{i}}  \vdash  \ottnt{e_{\ottmv{i}}}  \ottsym{:}  \ottnt{T_{\ottmv{i}}} }^{ \ottmv{i} }   \mathrel{\circ\hspace{-.4em}\rightarrow} \ast   \quad  \emptyset  \vdash  \ottnt{C_{{\mathrm{2}}}}  \ottsym{:}   \overline{ \Gamma_{\ottmv{i}}  \vdash  \ottnt{e_{\ottmv{i}}}  \ottsym{:}  \ottnt{T_{\ottmv{i}}} }^{ \ottmv{i} }   \mathrel{\circ\hspace{-.4em}\rightarrow}   \mathsf{Bool}   }%
\ottpremise{ \emptyset  \vdash  \ottnt{V}^\ottnt{C}  \ottsym{:}   \overline{ \Gamma_{\ottmv{i}}  \vdash  \ottnt{e_{\ottmv{i}}}  \ottsym{:}  \ottnt{T_{\ottmv{i}}} }^{ \ottmv{i} }   \mathrel{\circ\hspace{-.4em}\rightarrow}  \ottnt{T}^\ottnt{C}_{{\mathrm{1}}}  [   \overline{ \ottnt{e_{\ottmv{i}}} }^{ \ottmv{i} }   ]  \quad  \ottnt{C_{{\mathrm{1}}}}  [   \overline{ \ottnt{e_{\ottmv{i}}} }^{ \ottmv{i} }   ] \, [  \ottnt{V}^\ottnt{C}  [   \overline{ \ottnt{e_{\ottmv{i}}} }^{ \ottmv{i} }   ]  \ottsym{/}  \mathit{x}  ]  \longrightarrow^{\ast}  \ottnt{C_{{\mathrm{2}}}}  [   \overline{ \ottnt{e_{\ottmv{i}}} }^{ \ottmv{i} }   ] }%
}{
\Gamma  \vdash  \langle   \{  \mathit{x}  \mathord{:}  \ottnt{T}^\ottnt{C}_{{\mathrm{1}}}   \mathop{\mid}   \ottnt{C_{{\mathrm{1}}}}  \}   \ottsym{,}  \ottnt{C_{{\mathrm{2}}}}  \ottsym{,}  \ottnt{V}^\ottnt{C}  \rangle   ^{ \ell }   \ottsym{:}   \overline{ \Gamma_{\ottmv{i}}  \vdash  \ottnt{e_{\ottmv{i}}}  \ottsym{:}  \ottnt{T_{\ottmv{i}}} }^{ \ottmv{i} }   \mathrel{\circ\hspace{-.4em}\rightarrow}   \{  \mathit{x}  \mathord{:}  \ottnt{T}^\ottnt{C}_{{\mathrm{1}}}   \mathop{\mid}   \ottnt{C_{{\mathrm{1}}}}  \}   [   \overline{ \ottnt{e_{\ottmv{i}}} }^{ \ottmv{i} }   ]}{%
{\ottdrulename{CT\_ACheck}}{}%
}}

\newcommand{\ottdruleCTXXForget}[1]{\ottdrule[#1]{%
\ottpremise{  \mathord{ \vdash } ~  \Gamma   \quad  \emptyset  \vdash  \ottnt{V}^\ottnt{C}  \ottsym{:}   \overline{ \Gamma_{\ottmv{i}}  \vdash  \ottnt{e_{\ottmv{i}}}  \ottsym{:}  \ottnt{T_{\ottmv{i}}} }^{ \ottmv{i} }   \mathrel{\circ\hspace{-.4em}\rightarrow}   \{  \mathit{x}  \mathord{:}  \ottnt{T}   \mathop{\mid}   \ottnt{e}  \}  }%
}{
\Gamma  \vdash  \ottnt{V}^\ottnt{C}  \ottsym{:}   \overline{ \Gamma_{\ottmv{i}}  \vdash  \ottnt{e_{\ottmv{i}}}  \ottsym{:}  \ottnt{T_{\ottmv{i}}} }^{ \ottmv{i} }   \mathrel{\circ\hspace{-.4em}\rightarrow}  \ottnt{T}}{%
{\ottdrulename{CT\_Forget}}{}%
}}

\newcommand{\ottdruleCTXXExact}[1]{\ottdrule[#1]{%
\ottpremise{  \mathord{ \vdash } ~  \Gamma   \quad   \emptyset  \vdash  \ottnt{V}^\ottnt{C}  \ottsym{:}   \overline{ \Gamma_{\ottmv{i}}  \vdash  \ottnt{e_{\ottmv{i}}}  \ottsym{:}  \ottnt{T_{\ottmv{i}}} }^{ \ottmv{i} }   \mathrel{\circ\hspace{-.4em}\rightarrow}  \ottnt{T}  \quad   \emptyset  \vdash   \{  \mathit{x}  \mathord{:}  \ottnt{T}   \mathop{\mid}   \ottnt{e}  \}   \quad  \ottnt{e} \, [  \ottnt{V}^\ottnt{C}  [   \overline{ \ottnt{e_{\ottmv{i}}} }^{ \ottmv{i} }   ]  \ottsym{/}  \mathit{x}  ]  \longrightarrow^{\ast}   \mathsf{true}    }%
}{
\Gamma  \vdash  \ottnt{V}^\ottnt{C}  \ottsym{:}   \overline{ \Gamma_{\ottmv{i}}  \vdash  \ottnt{e_{\ottmv{i}}}  \ottsym{:}  \ottnt{T_{\ottmv{i}}} }^{ \ottmv{i} }   \mathrel{\circ\hspace{-.4em}\rightarrow}   \{  \mathit{x}  \mathord{:}  \ottnt{T}   \mathop{\mid}   \ottnt{e}  \} }{%
{\ottdrulename{CT\_Exact}}{}%
}}

\newcommand{\ottdruleCTXXBlame}[1]{\ottdrule[#1]{%
\ottpremise{  \mathord{ \vdash } ~  \Gamma   \quad  \emptyset  \vdash  \ottnt{T} }%
}{
\Gamma  \vdash   \mathord{\Uparrow}  \ell   \ottsym{:}   \overline{ \Gamma_{\ottmv{i}}  \vdash  \ottnt{e_{\ottmv{i}}}  \ottsym{:}  \ottnt{T_{\ottmv{i}}} }^{ \ottmv{i} }   \mathrel{\circ\hspace{-.4em}\rightarrow}  \ottnt{T}}{%
{\ottdrulename{CT\_Blame}}{}%
}}

\newcommand{\ottdruleSXXBase}[1]{\ottdrule[#1]{%
}{
\Gamma  \vdash  \ottnt{B}  \ottsym{<:}  \ottnt{B}}{%
{\ottdrulename{S\_Base}}{}%
}}

\newcommand{\ottdruleSXXTVar}[1]{\ottdrule[#1]{%
}{
\Gamma  \vdash  \alpha  \ottsym{<:}  \alpha}{%
{\ottdrulename{S\_TVar}}{}%
}}

\newcommand{\ottdruleSXXForall}[1]{\ottdrule[#1]{%
\ottpremise{\Gamma  \ottsym{,}  \alpha  \vdash  \ottnt{T_{{\mathrm{1}}}}  \ottsym{<:}  \ottnt{T_{{\mathrm{2}}}}}%
}{
\Gamma  \vdash   \forall   \alpha  .  \ottnt{T_{{\mathrm{1}}}}   \ottsym{<:}   \forall   \alpha  .  \ottnt{T_{{\mathrm{2}}}} }{%
{\ottdrulename{S\_Forall}}{}%
}}

\newcommand{\ottdruleSXXFun}[1]{\ottdrule[#1]{%
\ottpremise{ \Gamma  \vdash  \ottnt{T_{{\mathrm{21}}}}  \ottsym{<:}  \ottnt{T_{{\mathrm{11}}}}  \quad   \Gamma  ,  \mathit{x}  \mathord{:}  \ottnt{T_{{\mathrm{21}}}}   \vdash  \ottnt{T_{{\mathrm{12}}}} \, [  \langle  \ottnt{T_{{\mathrm{21}}}}  \Rightarrow  \ottnt{T_{{\mathrm{11}}}}  \rangle   ^{ \ell }  \, \mathit{x}  \ottsym{/}  \mathit{x}  ]  \ottsym{<:}  \ottnt{T_{{\mathrm{22}}}} }%
}{
\Gamma  \vdash   \mathit{x} \mathord{:} \ottnt{T_{{\mathrm{11}}}} \rightarrow \ottnt{T_{{\mathrm{12}}}}   \ottsym{<:}   \mathit{x} \mathord{:} \ottnt{T_{{\mathrm{21}}}} \rightarrow \ottnt{T_{{\mathrm{22}}}} }{%
{\ottdrulename{S\_Fun}}{}%
}}

\newcommand{\ottdruleSXXRefineL}[1]{\ottdrule[#1]{%
\ottpremise{\Gamma  \vdash  \ottnt{T_{{\mathrm{1}}}}  \ottsym{<:}  \ottnt{T_{{\mathrm{2}}}}}%
}{
\Gamma  \vdash   \{  \mathit{x}  \mathord{:}  \ottnt{T_{{\mathrm{1}}}}   \mathop{\mid}   \ottnt{e_{{\mathrm{1}}}}  \}   \ottsym{<:}  \ottnt{T_{{\mathrm{2}}}}}{%
{\ottdrulename{S\_RefineL}}{}%
}}

\newcommand{\ottdruleSXXRefineR}[1]{\ottdrule[#1]{%
\ottpremise{ \Gamma  \vdash  \ottnt{T_{{\mathrm{1}}}}  \ottsym{<:}  \ottnt{T_{{\mathrm{2}}}}  \quad   \Gamma  ,  \mathit{x}  \mathord{:}  \ottnt{T_{{\mathrm{1}}}}   \models  \ottnt{e_{{\mathrm{2}}}} \, [  \langle  \ottnt{T_{{\mathrm{1}}}}  \Rightarrow  \ottnt{T_{{\mathrm{2}}}}  \rangle   ^{ \ell }  \, \mathit{x}  \ottsym{/}  \mathit{x}  ] }%
}{
\Gamma  \vdash  \ottnt{T_{{\mathrm{1}}}}  \ottsym{<:}   \{  \mathit{x}  \mathord{:}  \ottnt{T_{{\mathrm{2}}}}   \mathop{\mid}   \ottnt{e_{{\mathrm{2}}}}  \} }{%
{\ottdrulename{S\_RefineR}}{}%
}}

\newcommand{\ottdruleSatis}[1]{\ottdrule[#1]{%
\ottpremise{\forall  \sigma  .~  \Gamma  \vdash  \sigma \, \mathbin{ \text{implies} } \,  \sigma  (  \ottnt{e}  )   \longrightarrow^{\ast}   \mathsf{true} }%
}{
\Gamma  \models  \ottnt{e}}{%
{\ottdrulename{Satis}}{}%
}}

\def\ottmv#1{\unskip\mathit{#1}}
\renewcommand{\ottdrule}[4][]{{\displaystyle\frac{\begin{array}{c}#2\end{array}}{#3}\;\ottdrulename{#4}}}

\title{Reasoning About Polymorphic Manifest Contracts}

\author[T.~Sekiyama]{Taro Sekiyama}
\address{National Institute of Informatics, Japan}
\email{sekiyama@nii.ac.jp}
\author[A.~Igarashi]{Atsushi Igarashi}
\address{Graduate School of Informatics, Kyoto University, Japan}
\email{igarashi@kuis.kyoto-u.ac.jp}

\begin{document}

\begin{abstract}
 Manifest contract calculi, which integrate cast-based dynamic
 contract checking and refinement type systems, have been studied as
 foundations for hybrid contract checking.
 In this article, we study techniques to reasoning about a
 polymorphic manifest contract calculus, including a few program
 transformations related to static contract verification.
 We first define a polymorphic manifest contract calculus \fhfix,
 which is much simpler than a previously studied one with delayed
 substitution, and a logical relation for it and prove that the
 logical relation is sound with respect to contextual equivalence.
 Next, we show that the \emph{upcast elimination} property, which has
 been studied as correctness of subtyping-based static cast
 verification, holds for \fhfix.  More specifically, we give a
 subtyping relation (which is not part of the calculus) for \fhfix{}
 types and prove that a term obtained by eliminating upcasts---casts
 from one type to a supertype of it---is logically related and so
 contextually equivalent to the original one.  We also justify two
 other program transformations for casts: selfification and
 static cast decomposition, which help upcast elimination.
 A challenge is that, due to the subsumption-free approach to manifest
 contracts, these program transformations do not always
 preserve well-typedness of terms.  To address it, the logical
 relation and contextual equivalence in this work are defined as
 \emph{semityped} relations: only one side of the relations is
 required to be well typed and the other side may be ill typed.
\end{abstract}

\maketitle


\section{Introduction}
\label{sec:intro}

\subsection{Software contracts}

\emph{Software contracts}~\cite{Meyer_1988_book} are a promising program
verification tool to develop robust, dependable software.
Contracts are agreements between a supplier and a client of software components.
On one hand, contracts are what the supplier guarantees.
On the other hand, they are what the client requires.
Following Eiffel~\cite{Meyer_1988_book}, a pioneer of software contracts,
contracts in this work are described as executable Boolean expressions
written in the same language as the program.
%
%
For example, the specification that both numbers $\mathit{x}$ and $\mathit{y}$ are either
positive or negative is described as Boolean expression ``$  \mathit{x}  \mathrel{*}  \mathit{y}   \mathrel{>} \ottsym{0} $''.
%
%
%
%

Contracts can be verified by two complementary approaches: static and dynamic
verification.
Dynamic verification is possible due to executability of contracts---the
run-time system can confirm that a contract holds by evaluating it.
Since Eiffel advocated ``Design by Contracts''~\cite{Meyer_1988_book}, there has
been extensive work on dynamic contract
verification~\cite{Rosenblum_1995_IEEETSE,Kramer_1998_TOOLS,Findler/Felleisen_2002_ICFP,Findler/Guo/Rogers_2007_IFL,Wadler/Findler_2009_ESOP,Strickland/Felleisen_2010_DLS,Disney/Flanagan/McCarthy_2011_ICFP,Chitil_2012_ICFP,Dimoulas/Tobin-Hochstadt/Felleisen_2012_ESOP,Takikawa/Strickland/Tobin-Hochstadt_2013_ESOP}.
Dynamic verification is easy to use, while it brings possibly significant
run-time overhead~\cite{Findler/Guo/Rogers_2007_IFL} and, perhaps worse, it
cannot check all possible execution paths, which may lead to missing critical errors.
Static
verification~\cite{Rondon/Kawaguchi/Jhala_2008_PLDI,Xu/PeytonJones/Claessen_2009_POPL,Bierman/Gordon/Hrictcu/Langworthy_ICFP_2010,Vazou/Rondon/Jhala_2013_ESOP,Nguyen/Tobin-Hochstadt/Horn_2014_ICFP,Vekris/Cosman/Jhala_2016_PLDI}
is another, complementary approach to program verification with contracts.
It causes no run-time overhead and guarantees that contracts are always
satisfied at run time, while it is difficult to use---it often requires heavy
annotations in programs, gives complicated error messages, and restricts the
expressive power of contracts.

\subsection{Manifest contracts}
\label{sec:intro-manifest}
To take the best of both, hybrid contract verification---where contracts are
verified statically if possible and, otherwise, dynamically---was proposed by
Flanagan~\cite{Flanagan_2006_POPL}, and calculi of \emph{manifest
contracts}~\cite{Flanagan_2006_POPL,Greenberg/Pierce/Weirich_2010_POPL,Knowles/Flanagan_2010_TOPLAS,Belo/Greenberg/Igarashi/Pierce_2011_ESOP,Sekiyama/Nishida/Igarashi_2015_POPL,Sekiyama/Igarashi/Greenberg_2016_TOPLAS,Sekiyama/Igarashi_2017_POPL}
have been studied as its theoretical foundation.
Manifest contracts refer to contract systems where contract information occurs as
part of types.
In particular, contracts are embedded into types by \emph{refinement types}
$ \{  \mathit{x}  \mathord{:}  \ottnt{T}   \mathop{\mid}   \ottnt{e}  \} $,\footnote{Although in the context of static verification the
underlying type $\ottnt{T}$ of a refinement type $ \{  \mathit{x}  \mathord{:}  \ottnt{T}   \mathop{\mid}   \ottnt{e}  \} $ is restricted to be a
base type usually, this work allows it to be arbitrary; this extension
is useful to describe contracts for abstract data types~\cite{Belo/Greenberg/Igarashi/Pierce_2011_ESOP,Sekiyama/Igarashi/Greenberg_2016_TOPLAS}.}
which denote a set of values $\ottnt{v}$ of $\ottnt{T}$ such that $\ottnt{v}$ satisfies
Boolean expression $\ottnt{e}$ (which is called a \emph{contract} or a \emph{refinement}), that
is, $\ottnt{e} \, [  \ottnt{v}  \ottsym{/}  \mathit{x}  ]$ evaluates to $ \mathsf{true} $.
For example, using refinement types, a type of positive numbers is represented
by $ \{  \mathit{x}  \mathord{:}   \mathsf{Int}    \mathop{\mid}    \mathit{x}  \mathrel{>} \ottsym{0}   \} $.
%

%
Dynamic verification in manifest contracts is performed by dynamic type
conversion, called \emph{casts}.
A cast $\langle  \ottnt{T_{{\mathrm{1}}}}  \Rightarrow  \ottnt{T_{{\mathrm{2}}}}  \rangle   ^{ \ell } $ checks that, when applied to value $\ottnt{v}$ of source type
$\ottnt{T_{{\mathrm{1}}}}$, $\ottnt{v}$ can behave as target type $\ottnt{T_{{\mathrm{2}}}}$.
In particular, if $\ottnt{T_{{\mathrm{2}}}}$ is a refinement type, the cast checks that $\ottnt{v}$
satisfies the contract of $\ottnt{T_{{\mathrm{2}}}}$.
If the contract check succeeds, the cast returns $\ottnt{v}$; otherwise, if it
fails, an uncatchable exception, called \emph{blame}, will be raised.
For example, let us consider cast $\langle   \{  \mathit{x}  \mathord{:}   \mathsf{Int}    \mathop{\mid}    \mathsf{prime?}  \, \mathit{x}  \}   \Rightarrow   \{  \mathit{x}  \mathord{:}   \mathsf{Int}    \mathop{\mid}    \mathit{x}  \mathrel{>} \ottsym{2}   \}   \rangle   ^{ \ell } $,
where $ \mathsf{prime?} $ is a Boolean function that decides if a given integer is a
prime number.
If this cast is applied to a prime number other than $\ottsym{2}$, the
check succeeds and the cast application returns the number itself.
Otherwise, if it is applied to $\ottsym{2}$, it fails and blame is raised.
The superscript $\ell$ (called blame label) on a cast is used to
indicate which cast has failed.

Static contract verification is formalized as \emph{subtyping}, which statically
checks that any value of a subtype behaves as a supertype.
In particular, a refinement type $ \{  \mathit{x}  \mathord{:}  \ottnt{T_{{\mathrm{1}}}}   \mathop{\mid}   \ottnt{e_{{\mathrm{1}}}}  \} $ is a subtype of another
$ \{  \mathit{x}  \mathord{:}  \ottnt{T_{{\mathrm{2}}}}   \mathop{\mid}   \ottnt{e_{{\mathrm{2}}}}  \} $ if any value of $\ottnt{T_{{\mathrm{1}}}}$ satisfying $\ottnt{e_{{\mathrm{1}}}}$ behaves as $\ottnt{T_{{\mathrm{2}}}}$
and satisfies $\ottnt{e_{{\mathrm{2}}}}$.
For example, $ \{  \mathit{x}  \mathord{:}   \mathsf{Int}    \mathop{\mid}    \mathsf{prime?}  \, \mathit{x}  \} $ is a subtype of $ \{  \mathit{x}  \mathord{:}   \mathsf{Int}    \mathop{\mid}    \mathit{x}  \mathrel{>} \ottsym{0}   \} $
because all prime numbers should be positive.

Hybrid contract verification integrates these two verification mechanisms of
contracts.
In the hybrid approach, for every program point where a type
$\ottnt{T_{{\mathrm{1}}}}$ is required to be a subtype of $\ottnt{T_{{\mathrm{2}}}}$, a type checker first
tries to solve the instance of the subtyping problem statically.
Unfortunately, since contracts are arbitrary Boolean expressions in a
Turing-complete language, the subtyping problem is undecidable in
general.
Thus, the type checker may not be able to solve the problem instance positively or
negatively.
In such a case, it inserts a cast from $\ottnt{T_{{\mathrm{1}}}}$ to $\ottnt{T_{{\mathrm{2}}}}$ into
the program point in
order to dynamically ensure that run-time values of $\ottnt{T_{{\mathrm{1}}}}$ behave as $\ottnt{T_{{\mathrm{2}}}}$.
For example, let us consider function application $\mathit{f} \, \mathit{x}$ where $\mathit{f}$ and
$\mathit{x}$ are given types $ \{  \mathit{y}  \mathord{:}   \mathsf{Int}    \mathop{\mid}    \mathsf{prime?}  \, \mathit{y}  \}   \rightarrow   \mathsf{Int} $ and
$\ottnt{T} \defeq  \{  \mathit{y}  \mathord{:}   \mathsf{Int}    \mathop{\mid}      \ottsym{2}  \mathrel{<}  \mathit{y}   \mathrel{<}  \ottsym{8}   \mathrel{ \mathsf{and} }   \mathsf{odd?}   \, \mathit{y}  \} $, respectively.
Given this expression, the type checker tries to see if $\ottnt{T}$ is a subtype of
$ \{  \mathit{y}  \mathord{:}   \mathsf{Int}    \mathop{\mid}    \mathsf{prime?}  \, \mathit{y}  \} $.
If the checker is strong enough, it will find out that
values of $\ottnt{T}$ are only three, five, and seven and that the subtyping relation holds
and accept $\mathit{f} \, \mathit{x}$; otherwise,
cast $\langle  \ottnt{T}  \Rightarrow   \{  \mathit{y}  \mathord{:}   \mathsf{Int}    \mathop{\mid}    \mathsf{prime?}  \, \mathit{y}  \}   \rangle   ^{ \ell } $ is inserted to check $\mathit{x}$ satisfies
contract $ \mathsf{prime?} $ at run time and the resulting expression $\mathit{f} \, \ottsym{(}  \langle  \ottnt{T}  \Rightarrow   \{  \mathit{y}  \mathord{:}   \mathsf{Int}    \mathop{\mid}    \mathsf{prime?}  \, \mathit{y}  \}   \rangle   ^{ \ell }  \, \mathit{x}  \ottsym{)}$ will be evaluated.

\subsection{Our work}
\label{sec:intro-work}

In this article, we study program reasoning in manifest contracts.
The first goal of the reasoning is to justify hybrid contract verification.
As described in \sect{intro-manifest}, a cast is inserted if an instance
of the subtyping problem is not solved statically.
Unfortunately, due to undecidability of the subtyping problem, it is possible
that casts from a type to its supertype---which we call \emph{upcasts}---are
inserted, though they are actually unnecessary.
How many upcasts are inserted rests on a prover used in static verification: the
more powerful the prover is, the less upcasts are inserted.
In other words, the behavior of programs could be dependent on the prover due to the
insertion of upcasts, which is not very desirable
because the dependency on provers would make it difficult to expect how
programs behave when the prover is modified.
We show that it is not the case, that is, the presence of upcasts has no influences
on the behavior of programs; this property is called the \emph{upcast
elimination}.

In fact, the upcast elimination has been studied in the prior work on manifest
contracts~\cite{Flanagan_2006_POPL,Knowles/Flanagan_2010_TOPLAS,Belo/Greenberg/Igarashi/Pierce_2011_ESOP},
but it is not satisfactory.
Flanagan~\cite{Flanagan_2006_POPL} and Belo et
al.~\cite{Belo/Greenberg/Igarashi/Pierce_2011_ESOP} studied the upcast
elimination for a simply typed manifest contract calculus and a polymorphic one,
respectively, but it turned out that their calculi are flawed~\cite{Knowles/Flanagan_2010_TOPLAS,Sekiyama/Igarashi/Greenberg_2016_TOPLAS}.
While Knowles and Flanagan~\cite{Knowles/Flanagan_2010_TOPLAS} has resolved the
issue of Flanagan, their upcast elimination deals with only closed upcasts;
while Sekiyama et al.~\cite{Sekiyama/Igarashi/Greenberg_2016_TOPLAS} fixed the
flaw in Belo et al., they did not address the upcast elimination;
we discuss in more detail in \sect{relwork}.
As far as we know, this work is the first to show the upcast elimination for open upcasts.

We introduce a \emph{subsumption-free} polymorphic manifest contract calculus
{\fhfix} and show the upcast elimination for it.
{\fhfix} is subsumption-free in the sense that it lacks a typing rule
of subsumption, that is, to promote the type of an expression to a supertype
(in fact, subtyping is not even part of the calculus)
and casts are necessary everywhere a required type is not
syntactically equivalent to the type of an expression.
In this style, static verification is performed ``post facto'', that is, upcasts
are eliminated \emph{post facto} after typechecking.
A subsumption-free manifest contract calculus is first developed by Belo et
al.~\cite{Belo/Greenberg/Igarashi/Pierce_2011_ESOP} to avoid the circularity
issue of manifest contract calculi with
subsumption~\cite{Knowles/Flanagan_2010_TOPLAS,Belo/Greenberg/Igarashi/Pierce_2011_ESOP}.
However, their metatheory turned out to rest on a wrong
conjecture~\cite{Sekiyama/Igarashi/Greenberg_2016_TOPLAS}.
Sekiyama et al.~\cite{Sekiyama/Igarashi/Greenberg_2016_TOPLAS} revised Belo et
al.'s work and resolved their issues by introducing a polymorphic manifest
contract calculus equipped with \emph{delayed substitution}, which suspends
substitution for variables in casts until their refinements are checked.
While delayed substitution ensures type soundness and parametricity, it makes
the metatheory complicated.
In this work, we adopt usual substitution to keep the metatheory simple.
To ensure type soundness under usual substitution, we---inspired by
Sekiyama et al.~\cite{Sekiyama/Nishida/Igarashi_2015_POPL}---modify
the semantics of casts so that all refinements in the target type of a
cast are checked even though they have been ensured by the source
type, whereas checks of refinements which have been ensured are
skipped in the semantics by Belo et
al.~\cite{Belo/Greenberg/Igarashi/Pierce_2011_ESOP} and Sekiyama et
al.~\cite{Sekiyama/Igarashi/Greenberg_2016_TOPLAS}.
For example, given $\langle   \{  \mathit{x}  \mathord{:}   \mathsf{Int}    \mathop{\mid}    \mathsf{prime?}  \, \mathit{x}  \}   \Rightarrow   \{  \mathit{y}  \mathord{:}   \{  \mathit{x}  \mathord{:}   \mathsf{Int}    \mathop{\mid}    \mathsf{prime?}  \, \mathit{x}  \}    \mathop{\mid}    \mathit{y}  \mathrel{>} \ottsym{2}   \}   \rangle   ^{ \ell } $, our
``fussy'' semantics checks both $ \mathsf{prime?}  \, \mathit{x}$ and $ \mathit{y}  \mathrel{>} \ottsym{2} $, while Belo et
al.'s ``sloppy'' semantics checks only $ \mathit{y}  \mathrel{>} \ottsym{2} $ because $ \mathsf{prime?}  \, \mathit{x}$ is
ensured by the source type.  \AI{We should consult Michael about our choice of
words ``fussy'' and ``sloppy''.}
Our fussy semantics resolves the issue of type soundness in Belo et al. and
is arguably simpler than Sekiyama et al.

In addition to the upcast elimination, we study reasoning about casts to make static contract verification more
effective.
In particular, this work studies two additional reasoning techniques.
The first is \emph{selfification}~\cite{Ou/Tan/Mandelbaum/Walker_2004_TCS},
which embeds information of expressions into their types.
For example, it gives expression $\ottnt{e}$ of integer type $ \mathsf{Int} $ a
more informative refinement type $ \{  \mathit{x}  \mathord{:}   \mathsf{Int}    \mathop{\mid}    \mathit{x}  \mathrel{=}_{  \mathsf{Int}  }  \ottnt{e}   \} $ (where
$ \mathrel{=}_{  \mathsf{Int}  } $ is a Boolean equality operator on integers).
The selfification is easily extensible to higher-order types, and it is
especially useful when given type information is not sufficient to solve
subtyping instances; see \sect{reasoning-self} for an example.
We formalize the selfification by casts: given $\ottnt{e}$ of $\ottnt{T}$, we show that
$\ottnt{e}$ is equivalent to a cast application $\langle  \ottnt{T}  \Rightarrow   \mathit{self}  (  \ottnt{T} ,  \ottnt{e}  )   \rangle   ^{ \ell }  \, \ottnt{e}$, where
$ \mathit{self}  (  \ottnt{T} ,  \ottnt{e}  ) $ is the resulting type of embedding $\ottnt{e}$ into $\ottnt{T}$.
In other words, $\ottnt{e}$ behaves as an expression of $ \mathit{self}  (  \ottnt{T} ,  \ottnt{e}  ) $.
The second is \emph{static cast decomposition}, which leads to elimination of
more upcasts obtained by reducing nonredundant casts.
%

We show correctness of three reasoning techniques about casts---the upcast
elimination, the selfification, and the cast decomposition---based on contextual
equivalence: we prove that (1) an upcast is contextually equivalent to an
identity function, (2) a cast application $\langle  \ottnt{T}  \Rightarrow   \mathit{self}  (  \ottnt{T} ,  \ottnt{e}  )   \rangle   ^{ \ell }  \, \ottnt{e}$ is to
$\ottnt{e}$, and (3) a cast is to its static decomposition.
We have to note that contextual equivalence that relates only terms of the same
type (except for the case of type variables) is useless in this work because we
want to show contextual equivalence between \emph{terms of different types}.
For example, an upcast and an identity function may not be given the same type
in our calculus for the lack of subsumption: a possible
type of an upcast $\langle  \ottnt{T_{{\mathrm{1}}}}  \Rightarrow  \ottnt{T_{{\mathrm{2}}}}  \rangle   ^{ \ell } $ is only $\ottnt{T_{{\mathrm{1}}}}  \rightarrow  \ottnt{T_{{\mathrm{2}}}}$, whereas types of
identity functions take the form $\ottnt{T}  \rightarrow  \ottnt{T}$, which is syntactically different
from $\ottnt{T_{{\mathrm{1}}}}  \rightarrow  \ottnt{T_{{\mathrm{2}}}}$ for any $\ottnt{T}$ if $\ottnt{T_{{\mathrm{1}}}}  \mathrel{\neq}  \ottnt{T_{{\mathrm{2}}}}$.
Instead of such usual contextual equivalence---which we call \emph{typed}
contextual equivalence---we introduce \emph{semityped} contextual equivalence,
where a well-typed term and a possibly ill-typed term can be related, and show
correctness of cast reasoning based on it.

Since, as is well known, it is difficult to prove contextual equivalence of
programs directly, we apply a proof technique based on \emph{logical relations}~\cite{Plotkin_1980_TODO,Reynolds_1983_IFIP}.
We develop a logical relation for manifest contracts and show its soundness with
respect to semityped contextual equivalence.
We also show completeness of our logical relation with respect to well-typed
terms in semityped contextual equivalence, via semityped
CIU-equivalence~\cite{Mason/Talcott_1991_JFP}.
The completeness implies transitivity of semityped contextual equivalence, which
is nontrivial in manifest contracts.\footnote{As we will discuss later, showing transitivity of \emph{typed}
contextual equivalence is not trivial, either.}


\subsection{Organization and proofs}

The rest of this paper is organized as follows.
We define our polymorphic manifest contract calculus {\fhfix} equipped with
fussy cast semantics in \sect{lang}%
{\ifexample{, together with a few motivating examples}\else\fi}.
\sect{ctxeq} introduces semityped contextual equivalence and
\sect{logical_relation} develops a logical relation for {\fhfix}.
We show that the logical relation is sound with respect to semityped contextual
equivalence and complete for well-typed terms in \sect{proving}.
Using the logical relation, we show the upcast elimination, the selfification,
and the cast decomposition in \sect{reasoning}.
After discussing related work in \sect{relwork}, we conclude in
\sect{conclusion}.

Most of our proofs are written in the pencil-and-paper style, but the proof of
cotermination, which is a key, but often flawed, property of manifest contracts,
is given by Coq proof script \texttt{coterm.v} at
\url{https://skymountain.github.io/work/papers/fh/coterm.zip}.

\section{Polymorphic Manifest Contract Calculus {\fhfix}}
\label{sec:lang}

This section{\ifexample starts with showing motivating examples of polymorphic
manifest contracts,\fi} formalizes a polymorphic manifest contract calculus
{\fhfix}{\ifexample{,}\fi} and proves its type soundness.  
As described in \sect{intro-work}, our run-time system checks even refinements
which have been ensured already, which enables us to prove \emph{cotermination},
a key property to show type soundness and parametricity without delayed
substitution.
We compare our fussy cast semantics with the sloppy cast semantics provided by
Belo et al.~\cite{Belo/Greenberg/Igarashi/Pierce_2011_ESOP} in
\sect{lang-semantics}.
Greenberg~\cite{Greenberg_2013_PhD} provides a few motivating examples of
polymorphic manifest contracts such as abstract datatypes for natural numbers
and string transducers; see Section 3.1 in the dissertation for details.

{\ifexample
\subsection{Examples}
\label{sec:lang-example}

Parametric polymorphism is so powerful that type languages with it can describe
the behavior of programs.
For example, using list type $ \alpha ~\mathsf{List} $ abstracting over element types, we
find that all functions of $ \alpha ~\mathsf{List}   \rightarrow  \alpha$ return an element in a given list,
while a function of $  \mathsf{Int}  ~\mathsf{List}   \rightarrow   \mathsf{Int} $ does not necessarily do so because it may return,
say, the length of a given list.
Combining parametric polymorphism with contracts makes it possible to specify
program behavior more precisely.
To see it, let us consider filtering function $ \mathsf{filter} $ over lists, which,
taking a predicate and a list, produces a list of elements that are in the given
list and satisfy the predicate.
This function is usually typed at $  \ottsym{(}  \alpha  \rightarrow   \mathsf{Bool}   \ottsym{)}  \rightarrow  \alpha ~\mathsf{List}   \rightarrow  \alpha ~\mathsf{List} $.
This type, however, does not specify the behavior of $ \mathsf{filter} $
very precisely---e.g., it may return the given list itself!
Polymorphic manifest contracts can give more precise types to it.
Using function $ \mathsf{for\_all} $ that returns whether all elements of a given list
satisfy a given predicate, we can give a more precise type to $ \mathsf{filter} $:
\[
  \mathit{x} \mathord{:}  \ottsym{(}  \alpha  \rightarrow   \mathsf{Bool}   \ottsym{)}  \rightarrow  \alpha ~\mathsf{List}  \rightarrow  \{  \mathit{y}  \mathord{:}   \alpha ~\mathsf{List}    \mathop{\mid}    \mathsf{for\_all}  \, \mathit{x} \, \mathit{y}  \}  
\]
where $ \mathit{x} \mathord{:} \ottnt{T_{{\mathrm{1}}}} \rightarrow \ottnt{T_{{\mathrm{2}}}} $ is a dependent function type and denotes functions that,
taking value $\ottnt{v}$ of $\ottnt{T_{{\mathrm{1}}}}$, returns values of $\ottnt{T_{{\mathrm{2}}}} \, [  \ottnt{v}  \ottsym{/}  \mathit{x}  ]$.
This type is more precise than $  \ottsym{(}  \alpha  \rightarrow   \mathsf{Bool}   \ottsym{)}  \rightarrow  \alpha ~\mathsf{List}   \rightarrow  \alpha ~\mathsf{List} $ above in that it
guarantees that elements of the result list satisfy the predicate.\footnote{Even
this type does not still capture the behavior of $ \mathsf{filter} $ exactly---e.g., a
function of that type may return the empty list even if there are elements
satisfying the predicate.}

Polymorphism can also express modules via encoding to existential types, and
its combination with contracts gives precise interfaces.
For example, consider an existential type for functional stacks:
\[\begin{array}{l@{}l}
     \forall   \alpha  .  \exists \, \beta  .~    &  \, \ottsym{(}   \mathsf{is\_empty}   \ottsym{:}  \beta  \rightarrow   \mathsf{Bool}   \ottsym{)}  \mathrel{\times}  \ottsym{(}   \mathsf{empty}   \ottsym{:}  \beta  \ottsym{)}    \mathrel{\times}  \ottsym{(}   \mathsf{top}   \ottsym{:}  \beta  \rightarrow  \alpha  \ottsym{)}   \mathrel{\times}   \\  \,  &  \, \ottsym{(}   \mathsf{push}   \ottsym{:}  \alpha  \rightarrow  \beta  \rightarrow  \beta  \ottsym{)}   \mathrel{\times}  \ottsym{(}   \mathsf{pop}   \ottsym{:}  \beta  \rightarrow  \beta  \ottsym{)} 
  \end{array}\]
where each element of the tuple is named for clearance; $\alpha$ is an element
type and $\beta$ is an abstract type of stacks; $ \mathsf{is\_empty} $ returns
whether a given stack is empty, $ \mathsf{empty} $ creates the empty stack, $ \mathsf{top} $
returns the topmost element of a stack, $ \mathsf{push} $ adds an element at the top of
a stack, and $ \mathsf{pop} $ returns the stack without the topmost element.
Contracts can refine this interface and describe what its members require and
what they guarantee more precisely:
\[\begin{array}{l@{}l}
     \forall   \alpha  .  \exists \, \beta  .~    &  \, \ottsym{(}   \mathsf{is\_empty}   \ottsym{:}  \beta  \rightarrow   \mathsf{Bool}   \ottsym{)}  \mathrel{\times}  \ottsym{(}   \mathsf{empty}   \ottsym{:}   \{  \mathit{x}  \mathord{:}  \beta   \mathop{\mid}    \mathsf{is\_empty}  \, \mathit{x}  \}   \ottsym{)}    \mathrel{\times}   \\  \,  &  \, \ottsym{(}   \mathsf{top}   \ottsym{:}   \{  \mathit{x}  \mathord{:}  \beta   \mathop{\mid}    \mathsf{not}  \, \ottsym{(}   \mathsf{is\_empty}   \ottsym{)}  \}   \rightarrow  \alpha  \ottsym{)}   \mathrel{\times}  \ottsym{(}   \mathsf{push}   \ottsym{:}  \alpha  \rightarrow  \beta  \rightarrow   \{  \mathit{x}  \mathord{:}  \beta   \mathop{\mid}    \mathsf{not}  \, \ottsym{(}   \mathsf{is\_empty}  \, \mathit{x}  \ottsym{)}  \}   \ottsym{)}   \mathrel{\times}   \\  \,  &  \, \ottsym{(}   \mathsf{pop}   \ottsym{:}   \{  \mathit{x}  \mathord{:}  \beta   \mathop{\mid}    \mathsf{not}  \, \ottsym{(}   \mathsf{is\_empty}   \ottsym{)}  \}   \rightarrow  \beta  \ottsym{)} 
  \end{array}\]
where refinement types refer to preceding member $ \mathsf{is\_empty} $ with help of
dependent sum type $ {}  \mathit{x}  \ottsym{:}  \ottnt{T_{{\mathrm{1}}}}  {}  \mathrel{\times}  \ottnt{T_{{\mathrm{2}}}} $, which is a product type such that
$\ottnt{T_{{\mathrm{2}}}}$ depends on $\mathit{x}$ of $\ottnt{T_{{\mathrm{1}}}}$ (it can be encoded by dependent function
types).
The type of $ \mathsf{empty} $ ensures that it is the empty stack; that of $ \mathsf{push} $
guarantees that it produces a nonempty stack; and $ \mathsf{top} $ and $ \mathsf{pop} $
requires a nonempty stack to be passed because $ \mathsf{top} $ takes the topmost
element and $ \mathsf{pop} $ removes it from the stack, respectively.
This interface detects trivial misuse of stacks statically---e.g., $ \mathsf{pop}  \,  \mathsf{empty} $
is no longer well typed---and does even any complicated misuse at run time.

Interested readers can see more examples in Greenberg~\cite{Greenberg_2013_PhD}.
\fi}

\subsection{Syntax}
\label{sec:lang-syntax}

 \begin{fhfigure*}[t!]
  \hspace{-4ex}
  \begin{minipage}[t]{.47\textwidth}
   $\begin{array}[t]{r@{\;\;}c@{\;\;}l}
    \multicolumn{3}{l}{\textbf{Types}} \\[.5ex]
     \ottnt{B} &::=&  \mathsf{Bool}  \mid ... \\
     \ottnt{T} &::=& \ottnt{B} \mid \alpha \mid  \mathit{x} \mathord{:} \ottnt{T_{{\mathrm{1}}}} \rightarrow \ottnt{T_{{\mathrm{2}}}}  \mid  \forall   \alpha  .  \ottnt{T}  \mid
       \{  \mathit{x}  \mathord{:}  \ottnt{T}   \mathop{\mid}   \ottnt{e}  \}  \\[1ex]
    \multicolumn{3}{l}{\textbf{Typing Contexts}} \\[.5ex]
    \Gamma &::=&  \emptyset  \mid  \Gamma  ,  \mathit{x}  \mathord{:}  \ottnt{T}  \mid \Gamma  \ottsym{,}  \alpha \\[1ex]
   \end{array}$
  \end{minipage}
  \begin{minipage}[t]{.52\hsize}
   $\begin{array}[t]{r@{\;\;}c@{\;\;}l}
    \multicolumn{3}{l}{\textbf{Values and Terms}} \\[.5ex]
    \ottnt{v} &::=&
     \ottnt{k} \mid   \lambda    \mathit{x}  \mathord{:}  \ottnt{T}  .  \ottnt{e}  \mid
     \ifdtabs  \Lambda\!  \, \alpha  .~  \ottnt{v} \mid  \gb{  \Lambda\!  \, \alpha  .~  \langle  \ottnt{T_{{\mathrm{1}}}}  \Rightarrow  \ottnt{T_{{\mathrm{2}}}}  \rangle   ^{ \ell }  \, \ottsym{(}  \ottnt{v} \, \alpha  \ottsym{)} }  \else  \Lambda\!  \, \alpha  .~  \ottnt{e} \fi \mid
     \langle  \ottnt{T_{{\mathrm{1}}}}  \Rightarrow  \ottnt{T_{{\mathrm{2}}}}  \rangle   ^{ \ell }  \\
    \ottnt{e} &::=&
     \ottnt{v} \mid \mathit{x} \mid {\tt op} \, \ottsym{(}  \ottnt{e_{{\mathrm{1}}}}  \ottsym{,} \, .. \, \ottsym{,}  \ottnt{e_{\ottmv{n}}}  \ottsym{)} \mid \ottnt{e_{{\mathrm{1}}}} \, \ottnt{e_{{\mathrm{2}}}} \mid \ottnt{e} \, \ottnt{T} \mid
     \\ &&
      \langle\!\langle  \,  \{  \mathit{x}  \mathord{:}  \ottnt{T_{{\mathrm{1}}}}   \mathop{\mid}   \ottnt{e_{{\mathrm{1}}}}  \}   \ottsym{,}  \ottnt{e_{{\mathrm{2}}}} \,  \rangle\!\rangle  \,  ^{ \ell }  \mid \langle   \{  \mathit{x}  \mathord{:}  \ottnt{T_{{\mathrm{1}}}}   \mathop{\mid}   \ottnt{e_{{\mathrm{1}}}}  \}   \ottsym{,}  \ottnt{e_{{\mathrm{2}}}}  \ottsym{,}  \ottnt{v}  \rangle   ^{ \ell }  \mid  \mathord{\Uparrow}  \ell  \\[1ex]
   \end{array}$
  \end{minipage}
  
  \caption{Syntax.}
  \label{fig:fhfix-syntax}
 \end{fhfigure*}

 \fig{fhfix-syntax} shows the syntax of {\fhfix}, which is based on Belo et
 al.~\cite{Belo/Greenberg/Igarashi/Pierce_2011_ESOP}.
 Types, ranged over by $\ottnt{T}$, are from the standard polymorphic lambda calculus
 except dependent function types and refinement types.
 Base types, denoted by $\ottnt{B}$, are parameterized, but we suppose that they
 include Boolean type $ \mathsf{Bool} $ for refinements.
 We also assume that, for each $\ottnt{B}$, there is a set $ {\cal K}_{ \ottnt{B} } $ of constants
 of $\ottnt{B}$; in particular, $ {\cal K}_{  \mathsf{Bool}  }  = \{ \mathsf{true} ,  \mathsf{false} \}$.
 %
 Refinement types $ \{  \mathit{x}  \mathord{:}  \ottnt{T}   \mathop{\mid}   \ottnt{e}  \} $, where variable $\mathit{x}$ of type $\ottnt{T}$ is
 bound in Boolean expression $\ottnt{e}$, denotes the set of values $\ottnt{v}$ of
 $\ottnt{T}$ such that $\ottnt{e} \, [  \ottnt{v}  \ottsym{/}  \mathit{x}  ]$ evaluates to $ \mathsf{true} $.
 As the prior
 work~\cite{Belo/Greenberg/Igarashi/Pierce_2011_ESOP,Sekiyama/Nishida/Igarashi_2015_POPL,Sekiyama/Igarashi/Greenberg_2016_TOPLAS},
 our refinement types are \emph{general} in the sense that any type
 $\ottnt{T}$ can be refined, while some
 work~\cite{Ou/Tan/Mandelbaum/Walker_2004_TCS,Flanagan_2006_POPL} allows
 only base types to be refined.
 %
 Dependent function types $ \mathit{x} \mathord{:} \ottnt{T_{{\mathrm{1}}}} \rightarrow \ottnt{T_{{\mathrm{2}}}} $ bind variable $\mathit{x}$ of domain type
 $\ottnt{T_{{\mathrm{1}}}}$ in codomain type $\ottnt{T_{{\mathrm{2}}}}$, and universal types $ \forall   \alpha  .  \ottnt{T} $ bind
 type variable $\alpha$ in $\ottnt{T}$.
 %
 Typing contexts $\Gamma$ are a sequence of type variables and bindings of the
 form $\mathit{x}:\ottnt{T}$, and we suppose that term and type variables bound in
 a typing context are distinct.

 Values, ranged over by $\ottnt{v}$, consist of casts and usual constructs from the
 call-by-value polymorphic lambda calculus---constants (denoted by $\ottnt{k}$),
 term abstractions, and type abstractions.
 %
 Term abstractions $  \lambda    \mathit{x}  \mathord{:}  \ottnt{T}  .  \ottnt{e} $ and type abstractions $ \Lambda\!  \, \alpha  .~  \ottnt{e}$ bind $\mathit{x}$
 and $\alpha$ in the body $\ottnt{e}$, respectively.
 %
 %
 %
 Casts $\langle  \ottnt{T_{{\mathrm{1}}}}  \Rightarrow  \ottnt{T_{{\mathrm{2}}}}  \rangle   ^{ \ell } $ from source type $\ottnt{T_{{\mathrm{1}}}}$ to target type $\ottnt{T_{{\mathrm{2}}}}$ check
 that arguments of $\ottnt{T_{{\mathrm{1}}}}$ can behave as $\ottnt{T_{{\mathrm{2}}}}$ at run time.
 %
 %
 Label $\ell$ indicates an abstract location of the cast in source code and it
 is used to identify failure casts; in a typical implementation, it would be a
 pair of the file name and the line number where the cast is given.
 We note that casts in {\fhfix} are not equipped with delayed substitution,
 unlike Sekiyama et al.~\cite{Sekiyama/Igarashi/Greenberg_2016_TOPLAS}.
 We discuss how this change affects the design of the logical relation in
 \sect{relwork}.

 The first line of terms, ranged over by $\ottnt{e}$, are standard---values,
 variables (denoted by $\mathit{x}$, $\mathit{y}$, $\mathit{z}$, etc.), primitive operations
 (denoted by $ {\tt op} $), term applications, and type applications.
 We assume that each base type $\ottnt{B}$ is equipped with an equality operator
 $=_{\ottnt{B}}$ to distinguish different constants.

 The second line presents terms which appear at run time for contract checking.
 Waiting checks $ \langle\!\langle  \,  \{  \mathit{x}  \mathord{:}  \ottnt{T_{{\mathrm{1}}}}   \mathop{\mid}   \ottnt{e_{{\mathrm{1}}}}  \}   \ottsym{,}  \ottnt{e_{{\mathrm{2}}}} \,  \rangle\!\rangle  \,  ^{ \ell } $, introduced for fussy cast semantics
 by Sekiyama et al.~\cite{Sekiyama/Nishida/Igarashi_2015_POPL}, check that the
 value of $\ottnt{e_{{\mathrm{2}}}}$ satisfies the contract $\ottnt{e_{{\mathrm{1}}}}$ by turning themselves to
 active checks.
 An active check $\langle   \{  \mathit{x}  \mathord{:}  \ottnt{T_{{\mathrm{1}}}}   \mathop{\mid}   \ottnt{e_{{\mathrm{1}}}}  \}   \ottsym{,}  \ottnt{e_{{\mathrm{2}}}}  \ottsym{,}  \ottnt{v}  \rangle   ^{ \ell } $ denotes an intermediate state of the
 check that $\ottnt{v}$ of $\ottnt{T_{{\mathrm{1}}}}$ satisfies contract $\ottnt{e_{{\mathrm{1}}}}$; $\ottnt{e_{{\mathrm{2}}}}$ is an
 intermediate term during the evaluation of $\ottnt{e_{{\mathrm{1}}}} \, [  \ottnt{v}  \ottsym{/}  \mathit{x}  ]$.
 If $\ottnt{e_{{\mathrm{2}}}}$ evaluates to $ \mathsf{true} $, the active check returns $\ottnt{v}$; otherwise,
 if $\ottnt{e_{{\mathrm{2}}}}$ evaluates to $ \mathsf{false} $, the check fails and uncatchable exception
 $ \mathord{\Uparrow}  \ell $, called \emph{blame}~\cite{Findler/Felleisen_2002_ICFP}, is raised.



 We introduce usual notation.
 We write $ \mathit{FV}  (  \ottnt{e}  ) $ and $ \mathit{FTV}  (  \ottnt{e}  ) $ for the sets of free term variables and
 free type variables that occur in $\ottnt{e}$, respectively.
 Term $\ottnt{e}$ is closed if $ \mathit{FV}  (  \ottnt{e}  )   \mathrel{\cup}   \mathit{FTV}  (  \ottnt{e}  )   \ottsym{=}  \emptyset$.
 %
 $\ottnt{e} \, [  \ottnt{v}  \ottsym{/}  \mathit{x}  ]$ and $\ottnt{e} \, [  \ottnt{T}  \ottsym{/}  \alpha  ]$ denote terms obtained by substituting $\ottnt{v}$ and
 $\ottnt{T}$ for variables $\mathit{x}$ and $\alpha$ in $\ottnt{e}$ in a capture-avoiding
 manner, respectively.
 %
 %
 These notations are also applied to types, typing contexts, and evaluation
 contexts (introduced in \sect{lang-semantics}).
 We write $ \mathit{dom}  (  \Gamma  ) $ for the set of term and type variables bound in $\Gamma$.
 We also write $\ottnt{T_{{\mathrm{1}}}}  \rightarrow  \ottnt{T_{{\mathrm{2}}}}$ for $ \mathit{x} \mathord{:} \ottnt{T_{{\mathrm{1}}}} \rightarrow \ottnt{T_{{\mathrm{2}}}} $ if $\mathit{x}$ does
 not occur free in $\ottnt{T_{{\mathrm{2}}}}$, $\ottnt{e_{{\mathrm{1}}}} ~  {\tt op}  ~ \ottnt{e_{{\mathrm{2}}}}$ for ${\tt op} \, \ottsym{(}  \ottnt{e_{{\mathrm{1}}}}  \ottsym{,}  \ottnt{e_{{\mathrm{2}}}}  \ottsym{)}$, and
 $ \mathsf{let}  ~  \mathit{x}  \mathord{:}  \ottnt{T}  \equal  \ottnt{e_{{\mathrm{1}}}}  ~ \ottliteralin ~  \ottnt{e_{{\mathrm{2}}}} $ for $\ottsym{(}    \lambda    \mathit{x}  \mathord{:}  \ottnt{T}  .  \ottnt{e_{{\mathrm{2}}}}   \ottsym{)} \, \ottnt{e_{{\mathrm{1}}}}$.

 \subsection{Operational Semantics}
 \label{sec:lang-semantics}

 \begin{fhfigure*}[t!]
  \begin{flushleft}
   \framebox{$\ottnt{e_{{\mathrm{1}}}}  \rightsquigarrow  \ottnt{e_{{\mathrm{2}}}}$} \quad {\bf{Reduction Rules}}
  \end{flushleft}

  \begin{center}
   $\begin{array}{rclr}
    {\tt op} \, \ottsym{(}  \ottnt{k_{{\mathrm{1}}}}  ,\, ... \, ,  \ottnt{k_{\ottmv{n}}}  \ottsym{)} & \rightsquigarrow & [\hspace{-.14em}[ {\tt op} ]\hspace{-.14em}] \, \ottsym{(}  \ottnt{k_{{\mathrm{1}}}}  \ottsym{,} \, ... \, \ottsym{,}  \ottnt{k_{\ottmv{n}}}  \ottsym{)} &
     \ottdrulename{R\_Op} \\[0.5ex]

    \ottsym{(}    \lambda    \mathit{x}  \mathord{:}  \ottnt{T}  .  \ottnt{e}   \ottsym{)} \, \ottnt{v} & \rightsquigarrow & \ottnt{e} \, [  \ottnt{v}  \ottsym{/}  \mathit{x}  ] &
     \ottdrulename{R\_Beta} \\[0.5ex]

    \ifdtabs
    \ottsym{(}   \Lambda\!  \, \alpha  .~  \ottnt{v}  \ottsym{)} \, \ottnt{T} & \rightsquigarrow & \ottnt{v} \, [  \ottnt{T}  \ottsym{/}  \alpha  ] &
     \ottdrulename{R\_TBeta1} \\[0.5ex]
    \gb{\ottsym{(}   \Lambda\!  \, \alpha  .~  \langle  \ottnt{T_{{\mathrm{1}}}}  \Rightarrow  \ottnt{T_{{\mathrm{2}}}}  \rangle   ^{ \ell }  \, \ottsym{(}  \ottnt{v} \, \alpha  \ottsym{)}  \ottsym{)} \, \ottnt{T}} & \gb{ \rightsquigarrow } & \gb{\ottsym{(}  \langle  \ottnt{T_{{\mathrm{1}}}}  \Rightarrow  \ottnt{T_{{\mathrm{2}}}}  \rangle   ^{ \ell }  \, \ottsym{(}  \ottnt{v} \, \alpha  \ottsym{)}  \ottsym{)} \, [  \ottnt{T}  \ottsym{/}  \alpha  ]} &
     \gb{\ottdrulename{R\_TBeta2}} \\[0.5ex]
    \else
    \ottsym{(}   \Lambda\!  \, \alpha  .~  \ottnt{e}  \ottsym{)} \, \ottnt{T} & \rightsquigarrow & \ottnt{e} \, [  \ottnt{T}  \ottsym{/}  \alpha  ] &
     \ottdrulename{R\_TBeta} \\[0.5ex]
    \fi

    \langle  \ottnt{B}  \Rightarrow  \ottnt{B}  \rangle   ^{ \ell }  \, \ottnt{v} & \rightsquigarrow & \ottnt{v} & \ottdrulename{R\_Base} \\[0.5ex]

    \langle   \mathit{x} \mathord{:} \ottnt{T_{{\mathrm{11}}}} \rightarrow \ottnt{T_{{\mathrm{12}}}}   \Rightarrow   \mathit{x} \mathord{:} \ottnt{T_{{\mathrm{21}}}} \rightarrow \ottnt{T_{{\mathrm{22}}}}   \rangle   ^{ \ell }  \, \ottnt{v} & \rightsquigarrow & &
     \ottdrulename{R\_Fun} \\
     \multicolumn{4}{r}{
        \lambda    \mathit{x}  \mathord{:}  \ottnt{T_{{\mathrm{21}}}}  .   \mathsf{let}  ~  \mathit{y}  \mathord{:}  \ottnt{T_{{\mathrm{11}}}}  \equal  \langle  \ottnt{T_{{\mathrm{21}}}}  \Rightarrow  \ottnt{T_{{\mathrm{11}}}}  \rangle   ^{ \ell }  \, \mathit{x}  ~ \ottliteralin ~  \langle  \ottnt{T_{{\mathrm{12}}}} \, [  \mathit{y}  \ottsym{/}  \mathit{x}  ]  \Rightarrow  \ottnt{T_{{\mathrm{22}}}}  \rangle   ^{ \ell }    \, \ottsym{(}  \ottnt{v} \, \mathit{y}  \ottsym{)}
     \qquad} \\
     \multicolumn{4}{r}{\text{where }  \mathit{y} \ \text{is a fresh variable} } \\[0.5ex]

    \langle   \forall   \alpha  .  \ottnt{T_{{\mathrm{1}}}}   \Rightarrow   \forall   \alpha  .  \ottnt{T_{{\mathrm{2}}}}   \rangle   ^{ \ell }  \, \ottnt{v} & \rightsquigarrow &  \Lambda\!  \, \alpha  .~  \langle  \ottnt{T_{{\mathrm{1}}}}  \Rightarrow  \ottnt{T_{{\mathrm{2}}}}  \rangle   ^{ \ell }  \, \ottsym{(}  \ottnt{v} \, \alpha  \ottsym{)} &
     \ottdrulename{R\_Forall} \\[1ex]

    \langle   \{  \mathit{x}  \mathord{:}  \ottnt{T_{{\mathrm{1}}}}   \mathop{\mid}   \ottnt{e_{{\mathrm{1}}}}  \}   \Rightarrow  \ottnt{T_{{\mathrm{2}}}}  \rangle   ^{ \ell }  \, \ottnt{v} & \rightsquigarrow & \langle  \ottnt{T_{{\mathrm{1}}}}  \Rightarrow  \ottnt{T_{{\mathrm{2}}}}  \rangle   ^{ \ell }  \, \ottnt{v} &
     \ottdrulename{R\_Forget} \\[1ex]

    \langle  \ottnt{T_{{\mathrm{1}}}}  \Rightarrow   \{  \mathit{x}  \mathord{:}  \ottnt{T_{{\mathrm{2}}}}   \mathop{\mid}   \ottnt{e_{{\mathrm{2}}}}  \}   \rangle   ^{ \ell }  \, \ottnt{v} & \rightsquigarrow &  \langle\!\langle  \,  \{  \mathit{x}  \mathord{:}  \ottnt{T_{{\mathrm{2}}}}   \mathop{\mid}   \ottnt{e_{{\mathrm{2}}}}  \}   \ottsym{,}  \langle  \ottnt{T_{{\mathrm{1}}}}  \Rightarrow  \ottnt{T_{{\mathrm{2}}}}  \rangle   ^{ \ell }  \, \ottnt{v} \,  \rangle\!\rangle  \,  ^{ \ell }  &
     \ottdrulename{R\_PreCheck} \\
     \multicolumn{4}{r}{(\text{if } \forall    \mathit{y}  ,  \ottnt{T'}   ,  \ottnt{e'}   .~  \ottnt{T_{{\mathrm{1}}}}  \mathrel{\neq}   \{  \mathit{y}  \mathord{:}  \ottnt{T'}   \mathop{\mid}   \ottnt{e'}  \} )} \\[0.5ex]

     \langle\!\langle  \,  \{  \mathit{x}  \mathord{:}  \ottnt{T}   \mathop{\mid}   \ottnt{e}  \}   \ottsym{,}  \ottnt{v} \,  \rangle\!\rangle  \,  ^{ \ell }  & \rightsquigarrow & \langle   \{  \mathit{x}  \mathord{:}  \ottnt{T}   \mathop{\mid}   \ottnt{e}  \}   \ottsym{,}  \ottnt{e} \, [  \ottnt{v}  \ottsym{/}  \mathit{x}  ]  \ottsym{,}  \ottnt{v}  \rangle   ^{ \ell }  &
     \ottdrulename{R\_Check} \\[0.5ex]

   \langle   \{  \mathit{x}  \mathord{:}  \ottnt{T}   \mathop{\mid}   \ottnt{e}  \}   \ottsym{,}   \mathsf{true}   \ottsym{,}  \ottnt{v}  \rangle   ^{ \ell }  & \rightsquigarrow & \ottnt{v} & \ottdrulename{R\_OK} \\[0.5ex]

   \langle   \{  \mathit{x}  \mathord{:}  \ottnt{T}   \mathop{\mid}   \ottnt{e}  \}   \ottsym{,}   \mathsf{false}   \ottsym{,}  \ottnt{v}  \rangle   ^{ \ell }  & \rightsquigarrow &  \mathord{\Uparrow}  \ell  & \ottdrulename{R\_Fail} \\
    \end{array}$
  \end{center}

  \vspace{0.5em}
  \begin{flushleft}
   \framebox{$\ottnt{e_{{\mathrm{1}}}}  \longrightarrow  \ottnt{e_{{\mathrm{2}}}}$} \quad {\bf{Evaluation Rules}}
  \end{flushleft}

  \begin{center}
   $\ottdruleEXXRed{}$ \hfil
   $\ottdruleEXXBlame{}$ \hfil
  \end{center}

  \caption{Operational semantics.}
  \label{fig:fhfix-redeval}
 \end{fhfigure*}

 {\fhfix} has call-by-value operational semantics in the small-step style,
 which is given by reduction $ \rightsquigarrow $ and evaluation $ \longrightarrow $ over closed
 terms.
 We write $ \rightsquigarrow^{\ast} $ and $ \longrightarrow^{\ast} $ for the reflexive transitive closures of
 $ \rightsquigarrow $ and $ \longrightarrow $, respectively.
 Reduction and evaluation rules are shown in \fig{fhfix-redeval}.

 \R{Op} says that reduction of primitive operations depends on function
 $[\hspace{-.14em}[\cdot]\hspace{-.14em}]$, which gives a denotation to each
 primitive operation and maps tuples of constants to constants; for example,
 $[\hspace{-.14em}[ {\tt +} ]\hspace{-.14em}](1,3)$ denotes $4$.
 We will describe requirements to $[\hspace{-.14em}[\cdot]\hspace{-.14em}]$ in
 \sect{lang-type-system}.
 Term and type applications evaluate by the standard $\beta$-reduction
 (\R{Beta} and \R{TBeta}).
 %
 {\ifdtabs
 The rule \R{TBeta2} is introduced for the second form $ \Lambda\!  \, \alpha  .~  \langle  \ottnt{T_{{\mathrm{1}}}}  \Rightarrow  \ottnt{T_{{\mathrm{2}}}}  \rangle   ^{ \ell }  \, \ottsym{(}  \ottnt{v} \, \alpha  \ottsym{)}$
 of type abstraction.
 \fi}

 Cast applications evaluate by combination of cast reduction rules, which are
 from Sekiyama et al.~\cite{Sekiyama/Nishida/Igarashi_2015_POPL} except
 \R{Forall}.
 Casts between the same base type behave as an identity function \R{Base}.
 Casts for function types produce a function wrapper involving casts which are
 contravariant on the domain types and covariant on the codomain types \R{Fun}.
 In taking an argument, the wrapper converts the argument with the
 contravariant cast so that the wrapped function $\ottnt{v}$ can accept it; if the
 contravariant cast succeeds, the wrapper invokes $\ottnt{v}$ with the conversion
 result and applies the covariant cast to the value produced by $\ottnt{v}$.
 \R{Fun} renames $\mathit{x}$ in the codomain type $\ottnt{T_{{\mathrm{12}}}}$ of the source function
 type to $\mathit{y}$ because $\ottnt{T_{{\mathrm{12}}}}$ expects $\mathit{x}$ to be replaced with arguments
 to $\ottnt{v}$ but they are actually denoted by $\mathit{y}$ in the wrapper.
 Casts for universal types behave as in the previous
 work~\cite{Belo/Greenberg/Igarashi/Pierce_2011_ESOP,Sekiyama/Igarashi/Greenberg_2016_TOPLAS};
 it produces a wrapper which, applied to a type, invokes the wrapped type
 abstraction and converts the result \R{Forall}.
 Casts for refinements types first peel off all refinements in the source type
 \R{Forget} and then check refinements in the target type with waiting checks
 \R{PreCheck}.
 After checks of inner refinements finish, the outermost refinement will be checked by
 an active check \R{Check}.
 If the check succeeds, the checked value is returned \R{OK}; otherwise, the
 cast is blamed \R{Fail}.


 Evaluation uses evaluation contexts~\cite{Felleisen/Hieb_1992_TCS}, given as
 follows, to reduce subterms \E{Red} and lift up blame \E{Blame}.
 \[
  \ottnt{E} ::=  \left[ \, \right]  \mid {\tt op} \, \ottsym{(}  \ottnt{v_{{\mathrm{1}}}}  \ottsym{,} \, .. \, \ottsym{,}  \ottnt{v_{\ottmv{n}}}  \ottsym{,}  \ottnt{E}  \ottsym{,}  \ottnt{e_{{\mathrm{1}}}}  \ottsym{,} \, .. \, \ottsym{,}  \ottnt{e_{\ottmv{m}}}  \ottsym{)} \mid
   \ottnt{E} \, \ottnt{e} \mid \ottnt{v} \, \ottnt{E} \mid \ottnt{E} \, \ottnt{T} \mid  \langle\!\langle  \,  \{  \mathit{x}  \mathord{:}  \ottnt{T}   \mathop{\mid}   \ottnt{e}  \}   \ottsym{,}  \ottnt{E} \,  \rangle\!\rangle  \,  ^{ \ell }  \mid
   \langle   \{  \mathit{x}  \mathord{:}  \ottnt{T}   \mathop{\mid}   \ottnt{e}  \}   \ottsym{,}  \ottnt{E}  \ottsym{,}  \ottnt{v}  \rangle   ^{ \ell } 
 \]
 This definition indicates that the semantics is call-by-value
 and arguments evaluate from left to right.




 \paragraph{Fussy versus sloppy}
 Our cast semantics is fussy in that, when $\langle  \ottnt{T_{{\mathrm{1}}}}  \Rightarrow  \ottnt{T_{{\mathrm{2}}}}  \rangle   ^{ \ell } $ is applied, all
 refinements in target type $\ottnt{T_{{\mathrm{2}}}}$ are checked even if they have been ensured
 by source type $\ottnt{T_{{\mathrm{1}}}}$.
 For example, let us consider reflexive cast
 $\langle   \{  \mathit{x}  \mathord{:}   \{  \mathit{y}  \mathord{:}   \mathsf{Int}    \mathop{\mid}    \mathit{y}  \mathrel{>} \ottsym{2}   \}    \mathop{\mid}    \mathsf{prime?}  \, \mathit{x}  \}   \Rightarrow   \{  \mathit{x}  \mathord{:}   \{  \mathit{y}  \mathord{:}   \mathsf{Int}    \mathop{\mid}    \mathit{y}  \mathrel{>} \ottsym{2}   \}    \mathop{\mid}    \mathsf{prime?}  \, \mathit{x}  \}   \rangle   ^{ \ell } $.
 When applied to $\ottnt{v}$, the cast application forgets
 the refinements in the source type of the cast \R{Forget}:
 \[\begin{array}{ll}
  & \langle   \{  \mathit{x}  \mathord{:}   \{  \mathit{y}  \mathord{:}   \mathsf{Int}    \mathop{\mid}    \mathit{y}  \mathrel{>} \ottsym{2}   \}    \mathop{\mid}    \mathsf{prime?}  \, \mathit{x}  \}   \Rightarrow   \{  \mathit{x}  \mathord{:}   \{  \mathit{y}  \mathord{:}   \mathsf{Int}    \mathop{\mid}    \mathit{y}  \mathrel{>} \ottsym{2}   \}    \mathop{\mid}    \mathsf{prime?}  \, \mathit{x}  \}   \rangle   ^{ \ell }  \, \ottnt{v} \\
     \longrightarrow^{\ast}  & \langle   \mathsf{Int}   \Rightarrow   \{  \mathit{x}  \mathord{:}   \{  \mathit{y}  \mathord{:}   \mathsf{Int}    \mathop{\mid}    \mathit{y}  \mathrel{>} \ottsym{2}   \}    \mathop{\mid}    \mathsf{prime?}  \, \mathit{x}  \}   \rangle   ^{ \ell }  \, \ottnt{v}
   \end{array}\]
 and then refinements in the target type are checked
 from the innermost through the outermost by using waiting checks
 \R{PreCheck}:
 \[\begin{array}{ll}
  ...\;  \longrightarrow^{\ast}  &  \langle\!\langle  \,  \{  \mathit{x}  \mathord{:}   \{  \mathit{y}  \mathord{:}   \mathsf{Int}    \mathop{\mid}    \mathit{y}  \mathrel{>} \ottsym{2}   \}    \mathop{\mid}    \mathsf{prime?}  \, \mathit{x}  \}   \ottsym{,}   \langle\!\langle  \,  \{  \mathit{y}  \mathord{:}   \mathsf{Int}    \mathop{\mid}    \mathit{y}  \mathrel{>} \ottsym{2}   \}   \ottsym{,}  \langle   \mathsf{Int}   \Rightarrow   \mathsf{Int}   \rangle   ^{ \ell }  \, \ottnt{v} \,  \rangle\!\rangle  \,  ^{ \ell }  \,  \rangle\!\rangle  \,  ^{ \ell }  \\
   \end{array}\]
 even though $\ottnt{v}$ would be typed at $ \{  \mathit{x}  \mathord{:}   \{  \mathit{y}  \mathord{:}   \mathsf{Int}    \mathop{\mid}    \mathit{y}  \mathrel{>} \ottsym{2}   \}    \mathop{\mid}    \mathsf{prime?}  \, \mathit{x}  \} $ and
 satisfy the refinements.

 In contrast, Belo et al.'s
 semantics~\cite{Belo/Greenberg/Igarashi/Pierce_2011_ESOP} is sloppy in that
 checks of refinements that have been ensured are skipped, which is represented
 by two cast reduction rules:
 \[\begin{array}{lll}
  \langle  \ottnt{T}  \Rightarrow  \ottnt{T}  \rangle   ^{ \ell }  \, \ottnt{v}       &  \rightsquigarrow ^{ \mathsf{s} }  & \ottnt{v} \\
  \langle  \ottnt{T}  \Rightarrow   \{  \mathit{x}  \mathord{:}  \ottnt{T}   \mathop{\mid}   \ottnt{e}  \}   \rangle   ^{ \ell }  \, \ottnt{v} &  \rightsquigarrow ^{ \mathsf{s} }  & \langle   \{  \mathit{x}  \mathord{:}  \ottnt{T}   \mathop{\mid}   \ottnt{e}  \}   \ottsym{,}  \ottnt{e} \, [  \ottnt{v}  \ottsym{/}  \mathit{x}  ]  \ottsym{,}  \ottnt{v}  \rangle   ^{ \ell } 
 \end{array}\]
 where $ \rightsquigarrow ^{ \mathsf{s} } $ is the reduction relation in the sloppy semantics.
 The first rule processes reflexive casts as
 if they are identity functions and the second checks only the outermost
 refinement because others have been ensured by the source type.
 Under the sloppy semantics,
 $\langle   \{  \mathit{x}  \mathord{:}   \{  \mathit{y}  \mathord{:}   \mathsf{Int}    \mathop{\mid}    \mathit{y}  \mathrel{>} \ottsym{2}   \}    \mathop{\mid}    \mathsf{prime?}  \, \mathit{x}  \}   \Rightarrow   \{  \mathit{x}  \mathord{:}   \{  \mathit{y}  \mathord{:}   \mathsf{Int}    \mathop{\mid}    \mathit{y}  \mathrel{>} \ottsym{2}   \}    \mathop{\mid}    \mathsf{prime?}  \, \mathit{x}  \}   \rangle   ^{ \ell }  \, \ottnt{v}$
 reduces to $\ottnt{v}$ in one step.
 The sloppy semantics allows a logical relation to take arbitrary binary
 relations on terms for interpretation of type
 variables~\cite{Belo/Greenberg/Igarashi/Pierce_2011_ESOP}.

 It is found that, however, naive sloppy semantics does not satisfy the
 so-called cotermination (\prop:ref{fh-coterm-true}), a key property to show
 type soundness and parametricity in manifest contracts; Sekiyama et al.\ %
 investigated this problem in
 detail~\cite{Sekiyama/Igarashi/Greenberg_2016_TOPLAS}.
 Briefly speaking, the cotermination requires that reduction of subterms
 preserves evaluation results, but the sloppy semantics does not satisfy it.
 For example, let $\ottnt{T} \defeq  \{  \mathit{x}  \mathord{:}   \mathsf{Int}    \mathop{\mid}    \mathsf{not}  \,  \mathsf{true}   \} $ where $ \mathsf{not} $ is a
 negation function on Booleans.
 Since reflexive cast $\langle  \ottnt{T}  \Rightarrow  \ottnt{T}  \rangle   ^{ \ell } $ behaves as an identity function in the
 sloppy semantics, $\langle  \ottnt{T}  \Rightarrow  \ottnt{T}  \rangle   ^{ \ell }  \, \ottnt{v}$ evaluates to $\ottnt{v}$ for any value $\ottnt{v}$.
 Since $ \mathsf{not}  \,  \mathsf{true}   \longrightarrow   \mathsf{false} $, the cotermination requires that
 $\langle   \{  \mathit{x}  \mathord{:}   \mathsf{Int}    \mathop{\mid}    \mathsf{false}   \}   \Rightarrow  \ottnt{T}  \rangle   ^{ \ell }  \, \ottnt{v}$ also evaluate to $\ottnt{v}$ because reduction of
 subterm $ \mathsf{not}  \,  \mathsf{true} $ to $ \mathsf{false} $ must not change the evaluation result.
 However, $\langle   \{  \mathit{x}  \mathord{:}   \mathsf{Int}    \mathop{\mid}    \mathsf{false}   \}   \Rightarrow  \ottnt{T}  \rangle   ^{ \ell }  \, \ottnt{v}$ checks refinement $ \mathsf{not}  \,  \mathsf{true} $ in
 $\ottnt{T}$, which gives rise to blame; thus, the cotermination is invalidated.

 The problem above does not happen in the fussy semantics.
 Under the fussy semantics, since all refinements in a cast are checked, both
 casts $\langle  \ottnt{T}  \Rightarrow  \ottnt{T}  \rangle   ^{ \ell } $ and $\langle   \{  \mathit{x}  \mathord{:}   \mathsf{Int}    \mathop{\mid}    \mathsf{false}   \}   \Rightarrow  \ottnt{T}  \rangle   ^{ \ell } $ check refinement
 $ \mathsf{not}  \,  \mathsf{true} $ and raise blame.

 \subsection{Type System}
 \label{sec:lang-type-system}

 \begin{fhfigure*}[t!]
  \begin{flushleft}
   \framebox{$ \mathord{ \vdash } ~  \Gamma $} \quad {\bf{Context Well-Formedness}}
  \end{flushleft}

  \begin{center}
   $\ottdruleWFXXEmpty{}$ \hfil
   $\ottdruleWFXXExtendVar{}$ \hfil
   $\ottdruleWFXXExtendTVar{}$
  \end{center}

  \vspace{0.5em}
  \begin{flushleft}
   \framebox{$\Gamma  \vdash  \ottnt{T}$} \quad {\bf{Type Well-Formedness}}
  \end{flushleft}

  \begin{center}
   $\ottdruleWFXXBase{}$ \hfil
   $\ottdruleWFXXTVar{}$ \hfil
   $\ottdruleWFXXForall{}$ \\[1ex]
   $\ottdruleWFXXFun{}$ \hfil
   $\ottdruleWFXXRefine{}$
  \end{center}


  \vspace{0.5em}
  \begin{flushleft}
   \framebox{$\Gamma  \vdash  \ottnt{e}  \ottsym{:}  \ottnt{T}$} \quad {\bf{Typing Rules}}
  \end{flushleft}

  \begin{center}
   $\ottdruleTXXVar{}$ \hfil
   $\ottdruleTXXConst{}$ \\[1ex]
   ${\ottdrule[]{%
    \ottpremise{
      \mathord{ \vdash } ~  \Gamma   \quad  \mathsf{ty}  (  {\tt op}  )   \ottsym{=}  {}  \mathit{x_{{\mathrm{1}}}}  \ottsym{:}  \ottnt{T_{{\mathrm{1}}}}  \rightarrow \, ... \, \rightarrow  \mathit{x_{\ottmv{n}}}  \ottsym{:}  \ottnt{T_{\ottmv{n}}}  {}  \rightarrow  \ottnt{T}
    }%
    \ottpremise{
     \forall  \ottmv{i} \, \in \,  \{  \, \ottsym{1}  ,\, ... \, ,  \mathit{n} \,  \}   .~  \Gamma  \vdash  \ottnt{e_{\ottmv{i}}}  \ottsym{:}  \ottnt{T_{\ottmv{i}}} \, [  \ottnt{e_{{\mathrm{1}}}}  \ottsym{/}  \mathit{x_{{\mathrm{1}}}}  ,\, ... \, ,   { \ottnt{e} }_{ \ottmv{i}  \ottsym{-}  \ottsym{1} }   \ottsym{/}   \mathit{x} _{ \ottmv{i}  \ottsym{-}  \ottsym{1} }   ]
    }%
    }{ \Gamma  \vdash  {\tt op} \, \ottsym{(}  \ottnt{e_{{\mathrm{1}}}}  \ottsym{,} \, ... \, \ottsym{,}  \ottnt{e_{\ottmv{n}}}  \ottsym{)}  \ottsym{:}  \ottnt{T} \, [  \ottnt{e_{{\mathrm{1}}}}  \ottsym{/}  \mathit{x_{{\mathrm{1}}}}  \ottsym{,} \, ... \, \ottsym{,}  \ottnt{e_{\ottmv{n}}}  \ottsym{/}  \mathit{x_{\ottmv{n}}}  ] }
    { {\ottdrulename{T\_Op}}{} }
   }$ \\[1ex]
   $\ottdruleTXXAbs{}$ \hfil
   $\ottdruleTXXCast{}$ \\[1ex]
   $\ottdruleTXXApp{}$ \hfil
   {\ifdtabs
   $\ottdruleTXXTAbsOne{}$ \\[1ex]
   \gb{$\ottdruleTXXTAbsTwo{}$} \hfil
   \else
   $\ottdruleTXXTAbs{}$ \\[1ex]
   \fi}
   $\ottdruleTXXTApp{}$ \hfil
   $\ottdruleTXXWCheck{}$ \\[1ex]
   $\ottdruleTXXACheck{}$ \\[1ex]
   $\ottdruleTXXBlame{}$ \hfil
   $\ottdruleTXXConv{}$ \\[1ex]
   $\ottdruleTXXForget{}$ \hfil
   $\ottdruleTXXExact{}$ \hfil
  \end{center}

  \vspace{0.5em}
  \begin{flushleft}
   \framebox{$\ottnt{T_{{\mathrm{1}}}}  \mathrel{\parallel}  \ottnt{T_{{\mathrm{2}}}}$} \quad {\bf{Type Compatibility}}
  \end{flushleft}

  \begin{center}
   $\ottdruleCXXBase{}$ \hfil
   $\ottdruleCXXTVar{}$ \hfil
   $\ottdruleCXXRefineL{}$ \hfil
   $\ottdruleCXXRefineR{}$ \\[1ex]
   $\ottdruleCXXFun{}$ \hfil
   $\ottdruleCXXForall{}$
  \end{center}

  \caption{Typing rules.}
  \label{fig:typing}
 \end{fhfigure*}

 The type system consists of three judgments:
 typing context well-formedness $ \mathord{ \vdash } ~  \Gamma $,
 type well-formedness $\Gamma  \vdash  \ottnt{T}$, and
 term typing $\Gamma  \vdash  \ottnt{e}  \ottsym{:}  \ottnt{T}$.
 They are derived by rules in \fig{typing}.
 The well-formedness rules are
 standard or easy to understand, and the typing rules are based on previous work~\cite{Belo/Greenberg/Igarashi/Pierce_2011_ESOP,Sekiyama/Igarashi/Greenberg_2016_TOPLAS}.
 We suppose that types of constants and primitive operations are provided by
 function $ \mathsf{ty} (-)$.
 Requirements to their types will be described at the end of this section.
 Casts are well typed when their source and target types are compatible \T{Cast}.
 Types are compatible if they are the same modulo refinements.
 This is formalized by type compatibility $\ottnt{T_{{\mathrm{1}}}}  \mathrel{\parallel}  \ottnt{T_{{\mathrm{2}}}}$, which is derived by
 the rules shown at the bottom of \fig{typing}.
 The type $\ottnt{T_{{\mathrm{2}}}} \, [  \ottnt{e_{{\mathrm{1}}}}  \ottsym{/}  \mathit{x}  ]$ of a term application is required to be well
 formed \T{App}.  As we will see the proof in detail, this condition is
 introduced for showing the parametricity (\prop:ref{fh-lr-param}).
 The typing rule \T{WCheck} of waiting checks $ \langle\!\langle  \,  \{  \mathit{x}  \mathord{:}  \ottnt{T_{{\mathrm{1}}}}   \mathop{\mid}   \ottnt{e_{{\mathrm{1}}}}  \}   \ottsym{,}  \ottnt{e_{{\mathrm{2}}}} \,  \rangle\!\rangle  \,  ^{ \ell } $ requires
 $\ottnt{e_{{\mathrm{2}}}}$ to have $\ottnt{T_{{\mathrm{1}}}}$ because it is checked at run time that the
 evaluation result of $\ottnt{e_{{\mathrm{2}}}}$ satisfies $\ottnt{e_{{\mathrm{1}}}}$ which refers to $\mathit{x}$ of
 $\ottnt{T_{{\mathrm{1}}}}$.
 Although waiting checks are run-time terms, \T{WCheck} does not require
 $ \{  \mathit{x}  \mathord{:}  \ottnt{T_{{\mathrm{1}}}}   \mathop{\mid}   \ottnt{e_{{\mathrm{1}}}}  \} $ and $\ottnt{e_{{\mathrm{2}}}}$ to be closed, unlike other run-time typing rules
 such as \T{ACheck}.
 This relaxation allows type-preserving static decomposition of $\langle  \ottnt{T_{{\mathrm{1}}}}  \Rightarrow   \{  \mathit{x}  \mathord{:}  \ottnt{T_{{\mathrm{2}}}}   \mathop{\mid}   \ottnt{e}  \}   \rangle   ^{ \ell } $
 into a smaller cast $\langle  \ottnt{T_{{\mathrm{1}}}}  \Rightarrow  \ottnt{T_{{\mathrm{2}}}}  \rangle   ^{ \ell } $ and a waiting check for refinement $\ottnt{e}$
(\prop:ref{fh-cc-precheck} in \sect{reasoning-cast-decomp}).
 Active checks $\langle   \{  \mathit{x}  \mathord{:}  \ottnt{T_{{\mathrm{1}}}}   \mathop{\mid}   \ottnt{e_{{\mathrm{1}}}}  \}   \ottsym{,}  \ottnt{e_{{\mathrm{2}}}}  \ottsym{,}  \ottnt{v}  \rangle   ^{ \ell } $ are well typed if $\ottnt{e_{{\mathrm{2}}}}$ is an actual
 intermediate state of evaluation of $\ottnt{e_{{\mathrm{1}}}} \, [  \ottnt{v}  \ottsym{/}  \mathit{x}  ]$ \T{ACheck}.
 \T{Forget} and \T{Exact} are run-time typing rules: the former forgets
 a refinement and the latter adds a refinement that holds.

 \T{Conv} is a run-time typing rule to show subject reduction.
 To motivate it, let us consider application $\ottnt{v} \, \ottnt{e}$ where
 $\ottnt{v}$ and $\ottnt{e}$ are typed at $ \mathit{x} \mathord{:} \ottnt{T_{{\mathrm{1}}}} \rightarrow \ottnt{T_{{\mathrm{2}}}} $ and $\ottnt{T_{{\mathrm{1}}}}$, respectively.
 This application would be typed at $\ottnt{T_{{\mathrm{2}}}} \, [  \ottnt{e}  \ottsym{/}  \mathit{x}  ]$ by \T{App}.
 If $\ottnt{e}$ reduces to $\ottnt{e'}$, $\ottnt{v} \, \ottnt{e'}$ would be at $\ottnt{T_{{\mathrm{2}}}} \, [  \ottnt{e'}  \ottsym{/}  \mathit{x}  ]$, which is
 syntactically different from $\ottnt{T_{{\mathrm{2}}}} \, [  \ottnt{e}  \ottsym{/}  \mathit{x}  ]$ in general.
 Since subject reduction requires evaluation of well-typed terms to be
 type-preserving, we need a device that allows $\ottnt{v} \, \ottnt{e'}$ to be typed at
 $\ottnt{T_{{\mathrm{2}}}} \, [  \ottnt{e}  \ottsym{/}  \mathit{x}  ]$.
 To this end, Belo et al.~\cite{Belo/Greenberg/Igarashi/Pierce_2011_ESOP}
 introduced a type conversion relation which relates $\ottnt{T_{{\mathrm{2}}}} \, [  \ottnt{e}  \ottsym{/}  \mathit{x}  ]$ and
 $\ottnt{T_{{\mathrm{2}}}} \, [  \ottnt{e'}  \ottsym{/}  \mathit{x}  ]$ and added a typing rule that allows terms to be retyped at
 convertible types.
 Their type conversion turns out to be flawed, but it is fixed in the succeeding
 work~\cite{Greenberg_2013_PhD,Sekiyama/Nishida/Igarashi_2015_POPL}.
 Our type conversion $ \equiv $ follows the fixed version.
 \begin{defi}[Type Conversion]
  The binary relation $ \Rrightarrow $ over types is defined as follows: $\ottnt{T_{{\mathrm{1}}}}  \Rrightarrow  \ottnt{T_{{\mathrm{2}}}}$ if
  there exist some $\ottnt{T}$, $\mathit{x}$, $\ottnt{e_{{\mathrm{1}}}}$, and $\ottnt{e_{{\mathrm{2}}}}$ such that
  $\ottnt{T_{{\mathrm{1}}}}  \ottsym{=}  \ottnt{T} \, [  \ottnt{e_{{\mathrm{1}}}}  \ottsym{/}  \mathit{x}  ]$ and $\ottnt{T_{{\mathrm{2}}}}  \ottsym{=}  \ottnt{T} \, [  \ottnt{e_{{\mathrm{2}}}}  \ottsym{/}  \mathit{x}  ]$ and $\ottnt{e_{{\mathrm{1}}}}  \longrightarrow  \ottnt{e_{{\mathrm{2}}}}$.
  The \emph{type conversion} $ \equiv $ is the symmetric transitive closure of
  $ \Rrightarrow $.
 \end{defi}

 Finally, we formalize requirements to constants and primitive operations.
 We first define auxiliary function $ \mathit{unref} $, which strips off refinements that are not under other type constructors:
 \[\begin{array}{lcll}
     \mathit{unref}  (   \{  \mathit{x}  \mathord{:}  \ottnt{T}   \mathop{\mid}   \ottnt{e}  \}   )  &=&  \mathit{unref}  (  \ottnt{T}  )  & \\
     \mathit{unref}  (  \ottnt{T}  )        &=& \ottnt{T} &  (\text{if } \ottnt{T} \text{ is not a refinement type})
 \end{array}\]
 Requirements to constants and primitive operations are as follows:
 \begin{itemize}
  \item For each constant $\ottnt{k} \, \in \,  {\cal K}_{ \ottnt{B} } $,
        (1) $ \mathit{unref}  (   \mathsf{ty}  (  \ottnt{k}  )   )   \ottsym{=}  \ottnt{B}$,
        (2) $\emptyset  \vdash   \mathsf{ty}  (  \ottnt{k}  ) $ is derivable, and
        (3) $\ottnt{k}$ satisfies all refinements in $ \mathsf{ty}  (  \ottnt{k}  ) $, that is,
        $\langle  \ottnt{B}  \Rightarrow   \mathsf{ty}  (  \ottnt{k}  )   \rangle   ^{ \ell }  \, \ottnt{k}  \longrightarrow^{\ast}  \ottnt{k}$.

  \item For each primitive operation $ {\tt op} $, $ \mathsf{ty}  (  {\tt op}  ) $ is a monomorphic
        dependent function type of the form
        ${}  \mathit{x_{{\mathrm{1}}}}  \ottsym{:}  \ottnt{T_{{\mathrm{1}}}}  \rightarrow \, ... \, \rightarrow  \mathit{x_{\ottmv{n}}}  \ottsym{:}  \ottnt{T_{\ottmv{n}}}  {}  \rightarrow  \ottnt{T_{{\mathrm{0}}}}$ where,
        for any $\ottmv{i} \, \in \,  \{  \, \ottsym{0}  ,\, ... \, ,  \mathit{n} \,  \} $, there exists some $\ottnt{B}$ such that
        $ \mathit{unref}  (  \ottnt{T_{\ottmv{i}}}  )   \ottsym{=}  \ottnt{B}$.
        Furthermore, we require that $ {\tt op} $ return a value satisfying the
        refinements in the return type $\ottnt{T_{{\mathrm{0}}}}$ when taking constants
        satisfying the refinements in the argument types, that
        is:
        \[\begin{array}{l}
         \forall\,{}  \ottnt{k_{{\mathrm{1}}}}  \ottsym{,} \, ... \, \ottsym{,}  \ottnt{k_{\ottmv{n}}}  {} . \\
           \quad \forall\,\ottmv{i} \, \in \,  \{  \, \ottsym{1}  ,\, ... \, ,  \mathit{n} \,  \}  .
            \left(\begin{array}{l}
             \ottnt{k_{\ottmv{i}}} \, \in \,  {\cal K}_{  \mathit{unref}  (  \ottnt{T_{\ottmv{i}}}  )  }  ~  \text{ and }  \\ \quad
                   \langle   \mathit{unref}  (  \ottnt{T_{\ottmv{i}}}  )   \Rightarrow  \ottnt{T_{\ottmv{i}}} \, [  \ottnt{k_{{\mathrm{1}}}}  \ottsym{/}  \mathit{x_{{\mathrm{1}}}}  ,\, ... \, ,   { \ottnt{k} }_{ \ottmv{i}  \ottsym{-}  \ottsym{1} }   \ottsym{/}   \mathit{x} _{ \ottmv{i}  \ottsym{-}  \ottsym{1} }   ]  \rangle   ^{ \ell }  \, \ottnt{k_{\ottmv{i}}}  \longrightarrow^{\ast}  \ottnt{k_{\ottmv{i}}}
                    \end{array}\right)  \mathbin{ \text{implies} }  \\[1ex]
            \qquad \exists\,\ottnt{k} \, \in \,  {\cal K}_{  \mathit{unref}  (  \ottnt{T_{{\mathrm{0}}}}  )  }  .
            \left(\begin{array}{l}
             [\hspace{-.14em}[ {\tt op} ]\hspace{-.14em}] \, \ottsym{(}  \ottnt{k_{{\mathrm{1}}}}  \ottsym{,} \, ... \, \ottsym{,}  \ottnt{k_{\ottmv{n}}}  \ottsym{)}  \ottsym{=}  \ottnt{k} ~  \text{ and }  \\ \quad
                   \langle   \mathit{unref}  (  \ottnt{T_{{\mathrm{0}}}}  )   \Rightarrow  \ottnt{T_{{\mathrm{0}}}} \, [  \ottnt{k_{{\mathrm{1}}}}  \ottsym{/}  \mathit{x_{{\mathrm{1}}}}  \ottsym{,} \, ... \, \ottsym{,}  \ottnt{k_{\ottmv{n}}}  \ottsym{/}  \mathit{x_{\ottmv{n}}}  ]  \rangle   ^{ \ell }  \, \ottnt{k}  \longrightarrow^{\ast}  \ottnt{k}
                  \end{array}\right)
          \end{array}\]
        In contrast, we assume that $[\hspace{-.14em}[ {\tt op} ]\hspace{-.14em}] \, \ottsym{(}  \ottnt{k_{{\mathrm{1}}}}  \ottsym{,} \, ... \, \ottsym{,}  \ottnt{k_{\ottmv{n}}}  \ottsym{)}$ is undefined if
        some $\ottnt{k_{\ottmv{i}}}$ does not satisfy refinements in $\ottnt{T_{\ottmv{i}}}$, that is,
        $\langle   \mathit{unref}  (  \ottnt{T_{\ottmv{i}}}  )   \Rightarrow  \ottnt{T_{\ottmv{i}}} \, [  \ottnt{k_{{\mathrm{1}}}}  \ottsym{/}  \mathit{x_{{\mathrm{1}}}}  ,\, ... \, ,   { \ottnt{k} }_{ \ottmv{i}  \ottsym{-}  \ottsym{1} }   \ottsym{/}   \mathit{x} _{ \ottmv{i}  \ottsym{-}  \ottsym{1} }   ]  \rangle   ^{ \ell }  \, \ottnt{k_{\ottmv{i}}}  \longrightarrow^{\ast}  \ottnt{k_{\ottmv{i}}}$ cannot be
        derived.
 \end{itemize}

 \subsection{Properties}

 This section proves type soundness via progress and subject
 reduction~\cite{Wright/Felleisen_1994_IC}.
 Type soundness can be shown as in the previous
 work~\cite{Sekiyama/Nishida/Igarashi_2015_POPL,Sekiyama/Igarashi/Greenberg_2016_TOPLAS}
 and so we omit the most parts of its proof.
 %

 We start with showing the cotermination (\prop:ref{fh-coterm-true}), a key property
 for proving not only type soundness but also parametricity and soundness of our logical
 relation with respect to contextual equivalence.
 It states that, if $\ottnt{e_{{\mathrm{1}}}}  \longrightarrow  \ottnt{e_{{\mathrm{2}}}}$, then $\ottnt{e} \, [  \ottnt{e_{{\mathrm{1}}}}  \ottsym{/}  \mathit{x}  ]$ and $\ottnt{e} \, [  \ottnt{e_{{\mathrm{2}}}}  \ottsym{/}  \mathit{x}  ]$ behave
 equivalently, which means that convertible types have the same denotation.
 Following Sekiyama et al.~\cite{Sekiyama/Igarashi/Greenberg_2016_TOPLAS},
 our proof of the cotermination is based on the observation that
 $\mathcal{R} \stackrel{\tiny{\textrm{def}}}{=} \{ (\ottnt{e} \, [  \ottnt{e_{{\mathrm{1}}}}  \ottsym{/}  \mathit{x}  ], \ottnt{e} \, [  \ottnt{e_{{\mathrm{2}}}}  \ottsym{/}  \mathit{x}  ]) \mid \ottnt{e_{{\mathrm{1}}}}  \longrightarrow  \ottnt{e_{{\mathrm{2}}}} \}$ is a weak bisimulation.
 We also refer to the names of the lemmas in the proof script \texttt{coterm.v}.

 \begin{prop}[name=Unique Decomposition \coqname{lemm\_red\_ectx\_decomp}]{fh-red-decomp}
  If
  $\ottnt{e}  \ottsym{=}  \ottnt{E}  [  \ottnt{e_{{\mathrm{1}}}}  ]$ and $\ottnt{e_{{\mathrm{1}}}}  \rightsquigarrow  \ottnt{e_{{\mathrm{2}}}}$ and
  $\ottnt{e}  \ottsym{=}  \ottnt{E'}  [  \ottnt{e'_{{\mathrm{1}}}}  ]$ and $\ottnt{e'_{{\mathrm{1}}}}  \rightsquigarrow  \ottnt{e'_{{\mathrm{2}}}}$,
  then $\ottnt{E}  \ottsym{=}  \ottnt{E'}$ and $\ottnt{e_{{\mathrm{1}}}}  \ottsym{=}  \ottnt{e'_{{\mathrm{1}}}}$.

  \proof

  By induction on $\ottnt{E}$.
 \end{prop}

 \begin{prop}[name=Determinism \coqname{lemm\_eval\_deterministic}]{fh-eval-determinism}
  If $\ottnt{e}  \longrightarrow  \ottnt{e_{{\mathrm{1}}}}$ and $\ottnt{e}  \longrightarrow  \ottnt{e_{{\mathrm{2}}}}$, then $\ottnt{e_{{\mathrm{1}}}}  \ottsym{=}  \ottnt{e_{{\mathrm{2}}}}$.

  \proof

  The case that $\ottnt{e}  \longrightarrow  \ottnt{e_{{\mathrm{1}}}}$ is derived by \E{Red} is shown by
  \prop:ref{fh-red-decomp} and the determinism of the reduction.
  In the case that it is derived by \E{Blame},  let us suppose that
  $\ottnt{e}  \longrightarrow  \ottnt{e_{{\mathrm{2}}}}$ is derived by \E{Red}.
  It is contradictory because, if $\ottnt{e}  \ottsym{=}  \ottnt{E}  [  \ottnt{e'}  ]$ and
  $\ottnt{e'}  \rightsquigarrow  \ottnt{e''}$, then $\ottnt{e}  \mathrel{\neq}  \ottnt{E_{{\mathrm{2}}}}  [   \mathord{\Uparrow}  \ell   ]$ for any $\ottnt{E_{{\mathrm{2}}}}$ and $\ell$.
 \end{prop}

 \begin{prop}[name={Weak bisimulation, left side} \coqname{lemm\_coterm\_left\_eval}]{fh-coterm-left}
  \label{lem:fh-coterm-left}
  If $\ottnt{e_{{\mathrm{1}}}}  \longrightarrow  \ottnt{e_{{\mathrm{2}}}}$ and $\ottnt{e} \, [  \ottnt{e_{{\mathrm{1}}}}  \ottsym{/}  \mathit{x}  ]  \longrightarrow  \ottnt{e'_{{\mathrm{1}}}}$,
  then there exists some $\ottnt{e'}$ such that $\ottnt{e} \, [  \ottnt{e_{{\mathrm{2}}}}  \ottsym{/}  \mathit{x}  ]  \longrightarrow^{\ast}  \ottnt{e'} \, [  \ottnt{e_{{\mathrm{2}}}}  \ottsym{/}  \mathit{x}  ]$
  and $\ottnt{e'_{{\mathrm{1}}}}  \longrightarrow^{\ast}  \ottnt{e'} \, [  \ottnt{e_{{\mathrm{1}}}}  \ottsym{/}  \mathit{x}  ]$.
(See the commuting diagram on the left in \fig{weak_bisimulation}.)
 \end{prop}

 \begin{prop}[name={Weak bisimulation, right side} \coqname{lemm\_coterm\_right\_eval}]{fh-coterm-right}
  \label{lem:fh-coterm-right}
  If $\ottnt{e_{{\mathrm{1}}}}  \longrightarrow  \ottnt{e_{{\mathrm{2}}}}$ and $\ottnt{e} \, [  \ottnt{e_{{\mathrm{2}}}}  \ottsym{/}  \mathit{x}  ]  \longrightarrow  \ottnt{e'_{{\mathrm{2}}}}$,
  then there exists some $\ottnt{e'}$ such that $\ottnt{e} \, [  \ottnt{e_{{\mathrm{1}}}}  \ottsym{/}  \mathit{x}  ]  \longrightarrow^{\ast}  \ottnt{e'} \, [  \ottnt{e_{{\mathrm{1}}}}  \ottsym{/}  \mathit{x}  ]$
  and $\ottnt{e'_{{\mathrm{2}}}}  \longrightarrow^{\ast}  \ottnt{e'} \, [  \ottnt{e_{{\mathrm{2}}}}  \ottsym{/}  \mathit{x}  ]$. (See the commuting diagram on the right in \fig{weak_bisimulation}.)
 \end{prop}

 \begin{figure}
$$
\xymatrix@C=2pt{
  \ottnt{e} \, [  \ottnt{e_{{\mathrm{1}}}}  \ottsym{/}  \mathit{x}  ] \ar@<-1ex>@(u,u)[rr] \ar[d] & \mathcal{R} & \ottnt{e} \, [  \ottnt{e_{{\mathrm{2}}}}  \ottsym{/}  \mathit{x}  ] \ar@{.>}[dd]^>{*}\\
  \ottnt{e'_{{\mathrm{1}}}} \ar@{.>}[d] \\
  \hspace*{-20pt}\exists \ottnt{e'}. \ottnt{e'} \, [  \ottnt{e_{{\mathrm{1}}}}  \ottsym{/}  \mathit{x}  ] & \mathcal{R} & \ottnt{e'} \, [  \ottnt{e_{{\mathrm{2}}}}  \ottsym{/}  \mathit{x}  ]
}
\qquad\qquad
\xymatrix@C=2pt{
  \ottnt{e} \, [  \ottnt{e_{{\mathrm{1}}}}  \ottsym{/}  \mathit{x}  ] \ar@<-1ex>@(u,u)[rr] \ar@{.>}[dd]^>{*} & \mathcal{R} & \ottnt{e} \, [  \ottnt{e_{{\mathrm{2}}}}  \ottsym{/}  \mathit{x}  ] \ar[d] \\
  & & \ottnt{e'_{{\mathrm{2}}}} \ar@{.>}[d] \\
  \hspace*{-20pt}\exists \ottnt{e'}. \ottnt{e'} \, [  \ottnt{e_{{\mathrm{1}}}}  \ottsym{/}  \mathit{x}  ] & \mathcal{R} & \ottnt{e'} \, [  \ottnt{e_{{\mathrm{2}}}}  \ottsym{/}  \mathit{x}  ]
}
$$
   \caption{Lemmas~\ref{lem:fh-coterm-left} and \ref{lem:fh-coterm-right}.}
   \label{fig:weak_bisimulation}
 \end{figure}

 \begin{prop}[name=Cotermination \coqname{lemm\_coterm\_true}]{fh-coterm-true}
 Suppose that $\ottnt{e_{{\mathrm{1}}}}  \longrightarrow  \ottnt{e_{{\mathrm{2}}}}$.
 \begin{enumerate}
  \item If $\ottnt{e} \, [  \ottnt{e_{{\mathrm{1}}}}  \ottsym{/}  \mathit{x}  ]  \longrightarrow^{\ast}  \ottnt{v_{{\mathrm{1}}}}$, then $\ottnt{e} \, [  \ottnt{e_{{\mathrm{2}}}}  \ottsym{/}  \mathit{x}  ]  \longrightarrow^{\ast}  \ottnt{v_{{\mathrm{2}}}}$.
        In particular, if $\ottnt{v_{{\mathrm{1}}}}  \ottsym{=}   \mathsf{true} $, then $\ottnt{v_{{\mathrm{2}}}}  \ottsym{=}   \mathsf{true} $.

  \item If $\ottnt{e} \, [  \ottnt{e_{{\mathrm{2}}}}  \ottsym{/}  \mathit{x}  ]  \longrightarrow^{\ast}  \ottnt{v_{{\mathrm{2}}}}$, then $\ottnt{e} \, [  \ottnt{e_{{\mathrm{1}}}}  \ottsym{/}  \mathit{x}  ]  \longrightarrow^{\ast}  \ottnt{v_{{\mathrm{1}}}}$.
        In particular, if $\ottnt{v_{{\mathrm{2}}}}  \ottsym{=}   \mathsf{true} $, then $\ottnt{v_{{\mathrm{1}}}}  \ottsym{=}   \mathsf{true} $.
 \end{enumerate}

  \proof

  By weak bisimulation and the fact that $\ottnt{v} \, [  \ottnt{e}  \ottsym{/}  \mathit{x}  ]  \ottsym{=}   \mathsf{true} $ implies
  $\ottnt{v}  \ottsym{=}   \mathsf{true} $; note that variables are not values in {\fhfix} and
  it is not the case that $\ottnt{v} = \mathit{x}$ and $\ottnt{e}  \ottsym{=}   \mathsf{true} $.
 \end{prop}

 The cotermination implies the value inversion, which states that well-typed values
 satisfy refinements of their types.
 \begin{defi}
  We define function $ \mathit{refines} $ from types to sets of lambda abstractions that
  denote refinements:
  \[\begin{array}{lcll}
    \mathit{refines}  (   \{  \mathit{x}  \mathord{:}  \ottnt{T}   \mathop{\mid}   \ottnt{e}  \}   )  &=&  \{  \,   \lambda    \mathit{x}  \mathord{:}  \ottnt{T}  .  \ottnt{e}  \,  \}   \mathrel{\cup}   \mathit{refines}  (  \ottnt{T}  )  & \\
    \mathit{refines}  (  \ottnt{T}  )        &=& \emptyset &
    (\text{if } \ottnt{T} \text{ is not a refinement type})
  \end{array}\]
  We write $ \ottnt{e}  \in [\hspace{-.14em}[  \ottnt{T}  ]\hspace{-.14em}] $ if, for any $\ottnt{v} \, \in \,  \mathit{refines}  (  \ottnt{T}  ) $,
  $\ottnt{v} \, \ottnt{e}  \longrightarrow^{\ast}   \mathsf{true} $.
 \end{defi}
 \begin{prop}{fh-val-satis-c-conv}
  For any closed value $\ottnt{v}$,
  if $\ottnt{T_{{\mathrm{1}}}}  \equiv  \ottnt{T_{{\mathrm{2}}}}$, then $ \ottnt{v}  \in [\hspace{-.14em}[  \ottnt{T_{{\mathrm{1}}}}  ]\hspace{-.14em}] $ iff $ \ottnt{v}  \in [\hspace{-.14em}[  \ottnt{T_{{\mathrm{2}}}}  ]\hspace{-.14em}] $.

  \proof

  Straightforward by induction on the derivation of $\ottnt{T_{{\mathrm{1}}}}  \equiv  \ottnt{T_{{\mathrm{2}}}}$.
  The case for $\ottnt{T_{{\mathrm{1}}}}  \Rrightarrow  \ottnt{T_{{\mathrm{2}}}}$ is shown by the cotermination.
 \end{prop}
 \begin{prop}[name=Value Inversion]{fh-val-satis-c}
  If $\emptyset  \vdash  \ottnt{v}  \ottsym{:}  \ottnt{T}$, then $ \ottnt{v}  \in [\hspace{-.14em}[  \ottnt{T}  ]\hspace{-.14em}] $.

  \proof

  Straightforward by induction on the typing derivation.
  The case for \T{Conv} is shown by \prop:ref{fh-val-satis-c-conv}.
 \end{prop}

 In addition to the value inversion, we need auxiliary, standard lemmas to show
 the progress and the subject reduction.  In what follows, only key lemmas are stated;
 readers interested in other lemmas and their proofs are referred to Greenberg's
 dissertation~\cite{Greenberg_2013_PhD} or Sekiyama et
 al.~\cite{Sekiyama/Igarashi/Greenberg_2016_TOPLAS}.
 \begin{prop}[name=Term Weakening]{fh-weak-term}
  Let $\mathit{x}$ be a fresh variable.
  Suppose that $\Gamma  \vdash  \ottnt{T}$.
  \begin{statements}
   \item(term) If $\Gamma  \ottsym{,}  \Gamma'  \vdash  \ottnt{e}  \ottsym{:}  \ottnt{T'}$, then $ \Gamma  ,  \mathit{x}  \mathord{:}  \ottnt{T}   \ottsym{,}  \Gamma'  \vdash  \ottnt{e}  \ottsym{:}  \ottnt{T'}$.
   \item(type) If $\Gamma  \ottsym{,}  \Gamma'  \vdash  \ottnt{T'}$, then $ \Gamma  ,  \mathit{x}  \mathord{:}  \ottnt{T}   \ottsym{,}  \Gamma'  \vdash  \ottnt{T'}$.
   \item(tctx) If $ \mathord{ \vdash } ~  \Gamma  \ottsym{,}  \Gamma' $, then $ \mathord{ \vdash } ~   \Gamma  ,  \mathit{x}  \mathord{:}  \ottnt{T}   \ottsym{,}  \Gamma' $.
  \end{statements}
 \end{prop}

\begin{prop}[name=Type Weakening]{fh-weak-type}
 Let $\alpha$ be a fresh type variable.
 \begin{statements}
  \item(term) If $\Gamma  \ottsym{,}  \Gamma'  \vdash  \ottnt{e}  \ottsym{:}  \ottnt{T}$, then $\Gamma  \ottsym{,}  \alpha  \ottsym{,}  \Gamma'  \vdash  \ottnt{e}  \ottsym{:}  \ottnt{T}$.
  \item(type) If $\Gamma  \ottsym{,}  \Gamma'  \vdash  \ottnt{T}$, then $\Gamma  \ottsym{,}  \alpha  \ottsym{,}  \Gamma'_{{\mathrm{2}}}  \vdash  \ottnt{T}$.
  \item(tctx) If $ \mathord{ \vdash } ~  \Gamma  \ottsym{,}  \Gamma' $, then $ \mathord{ \vdash } ~  \Gamma  \ottsym{,}  \alpha  \ottsym{,}  \Gamma' $.
 \end{statements}
\end{prop}

\begin{prop}[name=Term Substitution]{fh-subst-term}
 Suppose that $\Gamma  \vdash  \ottnt{e}  \ottsym{:}  \ottnt{T}$.
 \begin{statements}
  \item(term) If $ \Gamma  ,  \mathit{x}  \mathord{:}  \ottnt{T}   \ottsym{,}  \Gamma'  \vdash  \ottnt{e'}  \ottsym{:}  \ottnt{T'}$,
        then $\Gamma  \ottsym{,}  \Gamma'  [  \ottnt{e}  \ottsym{/}  \mathit{x}  ]  \vdash  \ottnt{e'} \, [  \ottnt{e}  \ottsym{/}  \mathit{x}  ]  \ottsym{:}  \ottnt{T'} \, [  \ottnt{e}  \ottsym{/}  \mathit{x}  ]$.
  \item(type) If $ \Gamma  ,  \mathit{x}  \mathord{:}  \ottnt{T}   \ottsym{,}  \Gamma'  \vdash  \ottnt{T'}$,
        then $\Gamma  \ottsym{,}  \Gamma'  [  \ottnt{e}  \ottsym{/}  \mathit{x}  ]  \vdash  \ottnt{T'} \, [  \ottnt{e}  \ottsym{/}  \mathit{x}  ]$.
  \item(tctx) If $ \mathord{ \vdash } ~   \Gamma  ,  \mathit{x}  \mathord{:}  \ottnt{T}   \ottsym{,}  \Gamma' $, then $ \mathord{ \vdash } ~  \Gamma  \ottsym{,}  \Gamma'  [  \ottnt{e}  \ottsym{/}  \mathit{x}  ] $.
 \end{statements}
\end{prop}

\begin{prop}[name=Type Substitution]{fh-subst-type}
 Suppose that $\Gamma  \vdash  \ottnt{T}$.
 \begin{statements}
  \item(term) If $\Gamma  \ottsym{,}  \alpha  \ottsym{,}  \Gamma'  \vdash  \ottnt{e'}  \ottsym{:}  \ottnt{T'}$,
        then $\Gamma  \ottsym{,}  \Gamma'  [  \ottnt{T}  \ottsym{/}  \alpha  ]  \vdash  \ottnt{e'} \, [  \ottnt{T}  \ottsym{/}  \alpha  ]  \ottsym{:}  \ottnt{T'} \, [  \ottnt{T}  \ottsym{/}  \alpha  ]$.
  \item(type) If $\Gamma  \ottsym{,}  \alpha  \ottsym{,}  \Gamma'  \vdash  \ottnt{T'}$,
        then $\Gamma  \ottsym{,}  \Gamma'  [  \ottnt{T}  \ottsym{/}  \alpha  ]  \vdash  \ottnt{T'} \, [  \ottnt{T}  \ottsym{/}  \alpha  ]$.
  \item(tctx) If $ \mathord{ \vdash } ~  \Gamma  \ottsym{,}  \alpha  \ottsym{,}  \Gamma' $, then $ \mathord{ \vdash } ~  \Gamma  \ottsym{,}  \Gamma'  [  \ottnt{T}  \ottsym{/}  \alpha  ] $.
 \end{statements}
\end{prop}

\begin{prop}[name=Canonical Forms]{fh-canonical}
 Suppose that $\emptyset  \vdash  \ottnt{v}  \ottsym{:}  \ottnt{T}$.
 \begin{statements}
  \item(base) If $ \mathit{unref}  (  \ottnt{T}  )   \ottsym{=}  \ottnt{B}$, then $\ottnt{v} \, \in \,  {\cal K}_{ \ottnt{B} } $.
  \item(fun) If $ \mathit{unref}  (  \ottnt{T}  )   \ottsym{=}   \mathit{x} \mathord{:} \ottnt{T_{{\mathrm{1}}}} \rightarrow \ottnt{T_{{\mathrm{2}}}} $, then
   $\ottnt{v}  \ottsym{=}    \lambda    \mathit{x}  \mathord{:}  \ottnt{T'_{{\mathrm{1}}}}  .  \ottnt{e} $ for some $\mathit{x}$, $\ottnt{T'_{{\mathrm{1}}}}$, and $\ottnt{e}$, or
   $\ottnt{v}  \ottsym{=}  \langle  \ottnt{T'_{{\mathrm{1}}}}  \Rightarrow  \ottnt{T'_{{\mathrm{2}}}}  \rangle   ^{ \ell } $ for some $\ottnt{T'_{{\mathrm{1}}}}$, $\ottnt{T'_{{\mathrm{2}}}}$, and $\ell$.
  \item(univ) If $ \mathit{unref}  (  \ottnt{T}  )   \ottsym{=}   \forall   \alpha  .  \ottnt{T'} $, then
        {\ifdtabs
        \begin{enumerate}
         \item $\ottnt{v}  \ottsym{=}   \Lambda\!  \, \alpha  .~  \ottnt{v'}$
               for some $\alpha$ and $\ottnt{v'}$, or
         \item $\ottnt{v}  \ottsym{=}   \Lambda\!  \, \alpha  .~  \langle  \ottnt{T'_{{\mathrm{1}}}}  \Rightarrow  \ottnt{T'_{{\mathrm{2}}}}  \rangle   ^{ \ell }  \, \ottsym{(}  \ottnt{v'} \, \alpha  \ottsym{)}$
               for some $\alpha$, $\ottnt{T'_{{\mathrm{1}}}}$, $\ottnt{T'_{{\mathrm{2}}}}$, $\ell$ and $\ottnt{v'}$.
        \end{enumerate}
        \else
        $\ottnt{v}  \ottsym{=}   \Lambda\!  \, \alpha  .~  \ottnt{e}$ for some $\ottnt{e}$.
        \fi}
 \end{statements}
\end{prop}

\begin{prop}[name=Progress]{fh-progress}
 If $\emptyset  \vdash  \ottnt{e}  \ottsym{:}  \ottnt{T}$, then:
 \begin{itemize}
  \item $\ottnt{e}  \longrightarrow  \ottnt{e'}$ for some $\ottnt{e'}$;
  \item $\ottnt{e}$ is a value; or
  \item $\ottnt{e}  \ottsym{=}   \mathord{\Uparrow}  \ell $ for some $\ell$.
 \end{itemize}
\end{prop}

\begin{prop}[name=Subject Reduction]{fh-subjred}
 If $\emptyset  \vdash  \ottnt{e}  \ottsym{:}  \ottnt{T}$ and $\ottnt{e}  \longrightarrow  \ottnt{e'}$, then $\emptyset  \vdash  \ottnt{e'}  \ottsym{:}  \ottnt{T}$.
\end{prop}

\begin{prop}[type=thm,name=Type Soundness]{fh-type-sound}
 If $\emptyset  \vdash  \ottnt{e}  \ottsym{:}  \ottnt{T}$, then one of the followings holds.
 \begin{itemize}
  \item $\ottnt{e}$ diverges;
  \item $\ottnt{e}  \longrightarrow^{\ast}  \ottnt{v}$ for some $\ottnt{v}$
        such that $\emptyset  \vdash  \ottnt{v}  \ottsym{:}  \ottnt{T}$ and $ \ottnt{v}  \in [\hspace{-.14em}[  \ottnt{T}  ]\hspace{-.14em}] $; or
  \item $\ottnt{e}  \longrightarrow^{\ast}   \mathord{\Uparrow}  \ell $ for some $\ell$.
 \end{itemize}

 \proof

 By the progress, the subject reduction, and the value inversion.
\end{prop}

\section{Semityped Contextual Equivalence}
\label{sec:ctxeq}

We introduce semityped contextual equivalence to formalize the upcast elimination property.
It relates terms $\ottnt{e_{{\mathrm{1}}}}$ and $\ottnt{e_{{\mathrm{2}}}}$ such that (1) they are contextually
equivalent, that is, behave equivalently under any well-typed program context,
and (2) $\ottnt{e_{{\mathrm{1}}}}$ is well-typed.
Semityped contextual equivalence does not enforce any condition on the
type of $\ottnt{e_{{\mathrm{2}}}}$,\footnote{In fact, it does not even require it to be
  well typed.}  so it can even relate terms having different types
such as an upcast and an identity function.

\begin{fhfigure*}[t]
\[\begin{array}{rcl}
  \ottnt{C} &::=&
    \left[ \, \right] _{ \ottmv{i} }  \mid \ottnt{k} \mid   \lambda    \mathit{x}  \mathord{:}  \ottnt{T}^\ottnt{C}  .  \ottnt{C}  \mid \langle  \ottnt{T}^\ottnt{C}_{{\mathrm{1}}}  \Rightarrow  \ottnt{T}^\ottnt{C}_{{\mathrm{2}}}  \rangle   ^{ \ell }  \mid
    \Lambda\!  \, \alpha  .~  \ottnt{C} \mid \mathit{x} \mid {\tt op} \, \ottsym{(}  \ottnt{C_{{\mathrm{1}}}}  \ottsym{,} \, ... \, \ottsym{,}  \ottnt{C_{\ottmv{n}}}  \ottsym{)} \mid \ottnt{C_{{\mathrm{1}}}} \, \ottnt{C_{{\mathrm{2}}}} \mid \ottnt{C_{{\mathrm{1}}}} \, \ottnt{T}^\ottnt{C}_{{\mathrm{2}}} \mid
   \\ &&
    \mathord{\Uparrow}  \ell  \mid  \langle\!\langle  \,  \{  \mathit{x}  \mathord{:}  \ottnt{T}^\ottnt{C}_{{\mathrm{1}}}   \mathop{\mid}   \ottnt{C_{{\mathrm{1}}}}  \}   \ottsym{,}  \ottnt{C_{{\mathrm{2}}}} \,  \rangle\!\rangle  \,  ^{ \ell }  \mid \langle   \{  \mathit{x}  \mathord{:}  \ottnt{T}^\ottnt{C}_{{\mathrm{1}}}   \mathop{\mid}   \ottnt{C_{{\mathrm{1}}}}  \}   \ottsym{,}  \ottnt{C_{{\mathrm{2}}}}  \ottsym{,}  \ottnt{V}^\ottnt{C}  \rangle   ^{ \ell } 
   \\[1ex]
  \ottnt{V}^\ottnt{C} &::=&
   \ottnt{k} \mid   \lambda    \mathit{x}  \mathord{:}  \ottnt{T}^\ottnt{C}  .  \ottnt{C}  \mid \langle  \ottnt{T}^\ottnt{C}_{{\mathrm{1}}}  \Rightarrow  \ottnt{T}^\ottnt{C}_{{\mathrm{2}}}  \rangle   ^{ \ell }  \mid  \Lambda\!  \, \alpha  .~  \ottnt{C} \\
  \ottnt{T}^\ottnt{C} &::=&
   \ottnt{B} \mid \alpha \mid  \mathit{x} \mathord{:} \ottnt{T}^\ottnt{C}_{{\mathrm{1}}} \rightarrow \ottnt{T}^\ottnt{C}_{{\mathrm{2}}}  \mid  \forall   \alpha  .  \ottnt{T}^\ottnt{C}  \mid  \{  \mathit{x}  \mathord{:}  \ottnt{T}^\ottnt{C}   \mathop{\mid}   \ottnt{C}  \} 
   \\[1ex]
\end{array}\]
 \caption{Syntax of contexts and type contexts.}
 \label{fig:ctx}
\end{fhfigure*}
Figure~\ref{fig:ctx} shows the syntax of multi-hole program contexts $\ottnt{C}$, value
contexts $\ottnt{V}^\ottnt{C}$, and type contexts $\ottnt{T}^\ottnt{C}$.
Contexts have zero or more holes $ \left[ \, \right] _{ \ottmv{i} } $ indexed by positive numbers $\ottmv{i}$, and
the same hole $ \left[ \, \right] _{ \ottmv{i} } $ can occur in a context an arbitrary number of times.
Thus, any term, value, and type are contexts without holes.
Replacement of the holes in program contexts, value contexts, and type contexts
with terms produces terms, values, and types, respectively.
For any program context $\ottnt{C}$ where indices of the holes range over
$\ottsym{1}$ through $\mathit{n}$ and any terms $\ottnt{e_{{\mathrm{1}}}},...,\ottnt{e_{\ottmv{n}}}$, we write
$\ottnt{C}  [  \ottnt{e_{{\mathrm{1}}}}  ,\, ... \, ,  \ottnt{e_{\ottmv{n}}}  ]$, or $\ottnt{C}  [   \overline{ \ottnt{e_{\ottmv{i}}} }^{ \ottmv{i} }   ]$ simply if $\mathit{n}$ is clear from the
context or not important, to denote a term obtained by replacing each hole
$ \left[ \, \right] _{ \ottmv{i} } $ with term $\ottnt{e_{\ottmv{i}}}$.
In particular, $\ottnt{e}  [   \overline{ \ottnt{e_{\ottmv{i}}} }^{ \ottmv{i} }   ]  \ottsym{=}  \ottnt{e}$ because there are zero holes in $\ottnt{e}$.
We use similar notation for value and type contexts.
%
%

Contexts having \emph{multiple} holes is crucial in semityped contextual
equivalence.
If we restrict contexts to have a single hole, replacements of terms with
contextually-equivalent ones would be performed one by one.
However, a replacement with an ill-typed term produces an ill-typed program, and
then, since semityped contextual equivalence requires terms on one side to be
well typed, the results of the remaining replacements could not be guaranteed to be
contextually equivalent to the original program.
For example, the replacement of term $\ottnt{e_{{\mathrm{1}}}}$ in $\ottnt{C}  [  \ottnt{e_{{\mathrm{1}}}}  \ottsym{,}  \ottnt{e_{{\mathrm{2}}}}  ]$ with ill-typed
term $\ottnt{e'_{{\mathrm{1}}}}$ produces an ill-typed program $\ottnt{C}  [  \ottnt{e'_{{\mathrm{1}}}}  \ottsym{,}  \ottnt{e_{{\mathrm{2}}}}  ]$.
In this case, even if there is an ill-typed term $\ottnt{e'_{{\mathrm{2}}}}$ contextually
equivalent to $\ottnt{e_{{\mathrm{2}}}}$, semityped contextual equivalence cannot contain
$\ottnt{C}  [  \ottnt{e'_{{\mathrm{1}}}}  \ottsym{,}  \ottnt{e_{{\mathrm{2}}}}  ]$ and $\ottnt{C}  [  \ottnt{e'_{{\mathrm{1}}}}  \ottsym{,}  \ottnt{e'_{{\mathrm{2}}}}  ]$ because both are ill typed.
The same issue arises even if we first replace $\ottnt{e_{{\mathrm{2}}}}$ and then $\ottnt{e_{{\mathrm{1}}}}$.
As a result, we could not show that $\ottnt{C}  [  \ottnt{e_{{\mathrm{1}}}}  \ottsym{,}  \ottnt{e_{{\mathrm{2}}}}  ]$ and $\ottnt{C}  [  \ottnt{e'_{{\mathrm{1}}}}  \ottsym{,}  \ottnt{e'_{{\mathrm{2}}}}  ]$ are
contextually equivalent.
This is problematic also in the upcast elimination, especially when programs
have multiple upcasts.
We address this issue by contexts with multiple holes, which
allow simultaneous replacements.
In the example above, if $\ottnt{e_{{\mathrm{1}}}}$ and $\ottnt{e_{{\mathrm{2}}}}$ are shown to be contextually
equivalent to $\ottnt{e'_{{\mathrm{1}}}}$ and $\ottnt{e'_{{\mathrm{2}}}}$ respectively, we can relate $\ottnt{C}  [  \ottnt{e_{{\mathrm{1}}}}  \ottsym{,}  \ottnt{e_{{\mathrm{2}}}}  ]$
to $\ottnt{C}  [  \ottnt{e'_{{\mathrm{1}}}}  \ottsym{,}  \ottnt{e'_{{\mathrm{2}}}}  ]$ directly, not via $\ottnt{C}  [  \ottnt{e'_{{\mathrm{1}}}}  \ottsym{,}  \ottnt{e_{{\mathrm{2}}}}  ]$ nor $\ottnt{C}  [  \ottnt{e_{{\mathrm{1}}}}  \ottsym{,}  \ottnt{e'_{{\mathrm{2}}}}  ]$.

The semityped contextual equivalence considers three kinds of observable
results, that is, termination, blame, and being stuck---the last has to be
considered because semityped contextual equivalence contains possibly ill-typed
terms.
We write
$ \ottnt{e}  \downarrow $ if $\ottnt{e}  \longrightarrow^{\ast}  \ottnt{v}$ for some $\ottnt{v}$,
$ \ottnt{e}  \, \mathord{\Uparrow}  \ell $ if $\ottnt{e}  \longrightarrow^{\ast}   \mathord{\Uparrow}  \ell $, and
$ \ottnt{e}  \upharpoonleft $ if $\ottnt{e}  \longrightarrow^{\ast}  \ottnt{e'}$ for some $\ottnt{e'}$
such that $\ottnt{e'}$ cannot evaluate and it is neither a value nor blame.
\begin{defi}[Observable Equivalence]
 We write $\ottnt{e_{{\mathrm{1}}}}  \Downarrow  \ottnt{e_{{\mathrm{2}}}}$ if
 (1) $ \ottnt{e_{{\mathrm{1}}}}  \downarrow $ iff $ \ottnt{e_{{\mathrm{2}}}}  \downarrow $,
 (2) $ \ottnt{e_{{\mathrm{1}}}}  \, \mathord{\Uparrow}  \ell $ iff $ \ottnt{e_{{\mathrm{2}}}}  \, \mathord{\Uparrow}  \ell $ for any $\ell$, and
 (3) $ \ottnt{e_{{\mathrm{1}}}}  \upharpoonleft $ iff $ \ottnt{e_{{\mathrm{2}}}}  \upharpoonleft $.
\end{defi}

Now, we could define semityped contextual equivalence as follows.
\begin{quotation}
 Terms $\ottnt{e_{{\mathrm{11}}}},..., { \ottnt{e} }_{  \ottsym{1} \mathit{n}  } $ and $\ottnt{e_{{\mathrm{21}}}},..., { \ottnt{e} }_{  \ottsym{2} \mathit{n}  } $ are contextually
 equivalent at $\ottnt{T_{{\mathrm{1}}}},...,\ottnt{T_{\ottmv{n}}}$ under $\Gamma_{{\mathrm{1}}},...,\Gamma_{\ottmv{n}}$, respectively,
 when
 (1) for any $\ottmv{i}$,
 $\Gamma_{\ottmv{i}}  \vdash   { \ottnt{e} }_{  \ottsym{1} \ottmv{i}  }   \ottsym{:}  \ottnt{T_{\ottmv{i}}}$ and $ \mathit{FV}  (   { \ottnt{e} }_{  \ottsym{2} \ottmv{i}  }   )   \mathrel{\cup}   \mathit{FTV}  (   { \ottnt{e} }_{  \ottsym{2} \ottmv{i}  }   )   \subseteq   \mathit{dom}  (  \Gamma_{\ottmv{i}}  ) $, and
 (2) for any $\ottnt{T}$ and $\ottnt{C}$,
 if $\emptyset  \vdash  \ottnt{C}  [  \ottnt{e_{{\mathrm{11}}}}  ,\, ... \, ,   { \ottnt{e} }_{  \ottsym{1} \mathit{n}  }   ]  \ottsym{:}  \ottnt{T}$,
 then $\ottnt{C}  [  \ottnt{e_{{\mathrm{11}}}}  ,\, ... \, ,   { \ottnt{e} }_{  \ottsym{1} \mathit{n}  }   ]  \Downarrow  \ottnt{C}  [  \ottnt{e_{{\mathrm{21}}}}  ,\, ... \, ,   { \ottnt{e} }_{  \ottsym{2} \mathit{n}  }   ]$.
\end{quotation}
Thanks to program contexts with multiple holes, we can replace two or more
well-typed terms with possibly ill-typed, contextually equivalent terms at the same time.

\begin{fhfigure*}[t]
 \begin{flushleft}
  \noindent  \framebox{$\Gamma  \vdash  \ottnt{C}  \ottsym{:}   \overline{ \Gamma_{\ottmv{i}}  \vdash  \ottnt{e_{\ottmv{i}}}  \ottsym{:}  \ottnt{T_{\ottmv{i}}} }^{ \ottmv{i} }   \mathrel{\circ\hspace{-.4em}\rightarrow}  \ottnt{T}$} \quad
 {\bf{Context Typing Rules}}
 \end{flushleft}
 \begin{center}
  $\ottdruleCTXXHole{}$ \hfil
  $\ottdruleCTXXVar{}$ \\[1.5ex]
  $\ottdruleCTXXConst{}$ \\[1.5ex]
  ${\ottdrule[]{%
    \ottpremise{
      \mathord{ \vdash } ~  \Gamma  \quad
      \mathsf{ty}  (  {\tt op}  )   \ottsym{=}  {}  \mathit{x_{{\mathrm{1}}}}  \ottsym{:}  \ottnt{T'_{{\mathrm{1}}}}  \rightarrow \, ... \, \rightarrow  \mathit{x_{\ottmv{n}}}  \ottsym{:}  \ottnt{T'_{\ottmv{n}}}  {}  \rightarrow  \ottnt{T}
    }
    \ottpremise{
     \forall  \ottmv{j} \, \in \,  \{  \, \ottsym{1}  ,\, ... \, ,  \mathit{n} \,  \}   .~  \Gamma  \vdash  \ottnt{C_{\ottmv{j}}}  \ottsym{:}   \overline{ \Gamma_{\ottmv{i}}  \vdash  \ottnt{e_{\ottmv{i}}}  \ottsym{:}  \ottnt{T_{\ottmv{i}}} }^{ \ottmv{i} }   \mathrel{\circ\hspace{-.4em}\rightarrow}  \ottnt{T'_{\ottmv{j}}} \, [  \ottnt{C_{{\mathrm{1}}}}  [   \overline{ \ottnt{e_{\ottmv{i}}} }^{ \ottmv{i} }   ]  \ottsym{/}  \mathit{x_{{\mathrm{1}}}}  ,\, ... \, ,   { \ottnt{C} }_{ \ottmv{j}  \ottsym{-}  \ottsym{1} }   [   \overline{ \ottnt{e_{\ottmv{i}}} }^{ \ottmv{i} }   ]  \ottsym{/}   \mathit{x} _{ \ottmv{j}  \ottsym{-}  \ottsym{1} }   ]
    }}%
   {\Gamma  \vdash  {\tt op} \, \ottsym{(}  \ottnt{C_{{\mathrm{1}}}}  \ottsym{,} \, ... \, \ottsym{,}  \ottnt{C_{\ottmv{n}}}  \ottsym{)}  \ottsym{:}   \overline{ \Gamma_{\ottmv{i}}  \vdash  \ottnt{e_{\ottmv{i}}}  \ottsym{:}  \ottnt{T_{\ottmv{i}}} }^{ \ottmv{i} }   \mathrel{\circ\hspace{-.4em}\rightarrow}  \ottnt{T} \, [  \ottnt{C_{{\mathrm{1}}}}  [   \overline{ \ottnt{e_{\ottmv{i}}} }^{ \ottmv{i} }   ]  \ottsym{/}  \mathit{x_{{\mathrm{1}}}}  ,\, ... \, ,  \ottnt{C_{\ottmv{n}}}  [   \overline{ \ottnt{e_{\ottmv{i}}} }^{ \ottmv{i} }   ]  \ottsym{/}  \mathit{x_{\ottmv{n}}}  ]}
  {{\ottdrulename{CT\_Op}}{}}}$ \\[1.5ex]
  $\ottdruleCTXXAbs{}$ \\[1.5ex]
  $\ottdruleCTXXCast{}$ \\[1.5ex]
  $\ottdruleCTXXApp{}$ \\[1.5ex]
  {\ifdtabs
  $\ottdruleCTXXTAbsOne{}$ \\[1.5ex]
  $\ottdruleCTXXTAbsTwo{}$ \\[1.5ex]
  \else
  $\ottdruleCTXXTAbs{}$ \\[1.5ex]
  \fi}
  $\ottdruleCTXXTApp{}$ \\[1.5ex]
  $\ottdruleCTXXConv{}$ \\[1.5ex]
  $\ottdruleCTXXWCheck{}$ \\[1.5ex]
  $\ottdruleCTXXACheck{}$ \\[1.5ex]
  $\ottdruleCTXXBlame{}$ \hfil
  $\ottdruleCTXXForget{}$ \\[1.5ex]
  $\ottdruleCTXXExact{}$
 \end{center}

 \caption{Program context well-formedness rules.}
 \label{fig:ctx-rules}
\end{fhfigure*}
\begin{fhfigure*}
 \begin{flushleft}
  \noindent  \framebox{$ \Gamma   \vdash   \ottnt{T}^\ottnt{C}   \ottsym{:}    \overline{ \Gamma_{\ottmv{i}}  \vdash  \ottnt{e_{\ottmv{i}}}  \ottsym{:}  \ottnt{T_{\ottmv{i}}} }^{ \ottmv{i} }   \mathrel{\circ\hspace{-.4em}\rightarrow} \ast $} \quad
  {\bf{Type Context Well-Formedness Rules}}
 \end{flushleft}
 \begin{center}
  $\ottdruleCWXXBase{}$ \hfil
  $\ottdruleCWXXTVar{}$ \\[1.5ex]
  $\ottdruleCWXXFun{}$ \\[1.5ex]
  $\ottdruleCWXXForall{}$ \\[1.5ex]
  $\ottdruleCWXXRefine{}$ \\[1.5ex]
 \end{center}

 \caption{Type context well-formedness rules.}
 \label{fig:ctx-typ-rules}
\end{fhfigure*}
The semityped contextual equivalence defined in this way is well defined as it is
but we find it more convenient to consider contexts as typed objects to discuss
composition of contexts and terms rigorously.
To this end, we introduce judgments for
program context well-formedness $\Gamma  \vdash  \ottnt{C}  \ottsym{:}   \overline{ \Gamma_{\ottmv{i}}  \vdash  \ottnt{e_{\ottmv{i}}}  \ottsym{:}  \ottnt{T_{\ottmv{i}}} }^{ \ottmv{i} }   \mathrel{\circ\hspace{-.4em}\rightarrow}  \ottnt{T}$ and type
context well-formedness $ \Gamma   \vdash   \ottnt{T}^\ottnt{C}   \ottsym{:}    \overline{ \Gamma_{\ottmv{i}}  \vdash  \ottnt{e_{\ottmv{i}}}  \ottsym{:}  \ottnt{T_{\ottmv{i}}} }^{ \ottmv{i} }   \mathrel{\circ\hspace{-.4em}\rightarrow} \ast $, which mean that, if $\ottnt{e_{\ottmv{i}}}$
is typed at $\ottnt{T_{\ottmv{i}}}$ under $\Gamma_{\ottmv{i}}$ for any $\ottmv{i}$, $\ottnt{C}  [   \overline{ \ottnt{e_{\ottmv{i}}} }^{ \ottmv{i} }   ]$ and
$\ottnt{T}^\ottnt{C}  [   \overline{ \ottnt{e_{\ottmv{i}}} }^{ \ottmv{i} }   ]$ are a well-typed term of $\ottnt{T}$ under $\Gamma$ and a
well-formed type under $\Gamma$, respectively.\footnote{Since value contexts are
a subset of program contexts, value context well-formedness is given by program
context well-formedness.}
These well-formedness judgments need information on terms $ \overline{ \Gamma_{\ottmv{i}}  \vdash  \ottnt{e_{\ottmv{i}}}  \ottsym{:}  \ottnt{T_{\ottmv{i}}} }^{ \ottmv{i} } $ with which holes are replaced as well as typing
context and type information because whether composition of a context with
terms produces a well-typed term rests on the composed terms.
For example, let us consider $\ottnt{C}  \ottsym{=}  \langle   \{  \mathit{x}  \mathord{:}   \mathsf{Int}    \mathop{\mid}    \mathit{x}  \mathrel{>} \ottsym{0}   \}   \Rightarrow   \mathsf{Int}   \rangle   ^{ \ell }  \, \ottsym{(}  \mathit{f} \,  \left[ \, \right] _{ \ottsym{1} }   \ottsym{)}$ where $\mathit{f}$
is typed at $ \mathit{y} \mathord{:}  \mathsf{Int}  \rightarrow  \{  \mathit{x}  \mathord{:}   \mathsf{Int}    \mathop{\mid}    \mathit{x}  \mathrel{>} \mathit{y}   \}  $.
$\ottnt{C}  [  {}  \ottsym{0}  {}  ]$ is well typed because the type of $\mathit{f} \, \ottsym{0}$ matches with the
source type of the cast, while $\ottnt{C}  [  {}  \ottsym{2}  {}  ]$ is not because the type of
$\mathit{f} \, \ottsym{2}$ is $ \{  \mathit{x}  \mathord{:}   \mathsf{Int}    \mathop{\mid}    \mathit{x}  \mathrel{>} \ottsym{2}   \} $, which is different from $ \{  \mathit{x}  \mathord{:}   \mathsf{Int}    \mathop{\mid}    \mathit{x}  \mathrel{>} \ottsym{0}   \} $.
If no type information in $\ottnt{C}$ depends on holes, the derivation of
$\Gamma  \vdash  \ottnt{C}  \ottsym{:}   \overline{ \Gamma_{\ottmv{i}}  \vdash  \ottnt{e_{\ottmv{i}}}  \ottsym{:}  \ottnt{T_{\ottmv{i}}} }^{ \ottmv{i} }   \mathrel{\circ\hspace{-.4em}\rightarrow}  \ottnt{T}$ refers only to $ \overline{ \Gamma_{\ottmv{i}} }^{ \ottmv{i} } $ and
$ \overline{ \ottnt{T_{\ottmv{i}}} }^{ \ottmv{i} } $, not any of $ \overline{ \ottnt{e_{\ottmv{i}}} }^{ \ottmv{i} } $.
Inference rules for the judgments are shown in Figures~\ref{fig:ctx-rules} and
\ref{fig:ctx-typ-rules}; they correspond to term typing and type well-formedness
rules given in \sect{lang-type-system}.

We show a few properties of well-typed contexts: (1) composition of a
well-formed context with well-typed terms produces a well-typed term, (2) free
variables and free type variables are preserved by the composition, and (3)
well-typed terms, well-typed values, well-formed types are well-formed program
contexts, value contexts, and type contexts, respectively.
\begin{prop}{fh-ctxeq-typed}
 Suppose $\Gamma_{{\mathrm{1}}}  \vdash  \ottnt{e_{{\mathrm{1}}}}  \ottsym{:}  \ottnt{T_{{\mathrm{1}}}}$, \ldots,
  $\Gamma_{\ottmv{n}}  \vdash  \ottnt{e_{\ottmv{n}}}  \ottsym{:}  \ottnt{T_{\ottmv{n}}}$.
 \begin{statements}
  \item(term) If $\Gamma  \vdash  \ottnt{C}  \ottsym{:}   \overline{ \Gamma_{\ottmv{i}}  \vdash  \ottnt{e_{\ottmv{i}}}  \ottsym{:}  \ottnt{T_{\ottmv{i}}} }^{ \ottmv{i} }   \mathrel{\circ\hspace{-.4em}\rightarrow}  \ottnt{T}$,
   then $\Gamma  \vdash  \ottnt{C}  [   \overline{ \ottnt{e_{\ottmv{i}}} }^{ \ottmv{i} }   ]  \ottsym{:}  \ottnt{T}$.
  \item(type) If $ \Gamma   \vdash   \ottnt{T}^\ottnt{C}   \ottsym{:}    \overline{ \Gamma_{\ottmv{i}}  \vdash  \ottnt{e_{\ottmv{i}}}  \ottsym{:}  \ottnt{T_{\ottmv{i}}} }^{ \ottmv{i} }   \mathrel{\circ\hspace{-.4em}\rightarrow} \ast $,
   then $\Gamma  \vdash  \ottnt{T}^\ottnt{C}  [   \overline{ \ottnt{e_{\ottmv{i}}} }^{ \ottmv{i} }   ]$.
 \end{statements}

 \proof

 By induction on the derivations of
 $\Gamma  \vdash  \ottnt{C}  \ottsym{:}   \overline{ \Gamma_{\ottmv{i}}  \vdash  \ottnt{e_{\ottmv{i}}}  \ottsym{:}  \ottnt{T_{\ottmv{i}}} }^{ \ottmv{i} }   \mathrel{\circ\hspace{-.4em}\rightarrow}  \ottnt{T}$ and
 $ \Gamma   \vdash   \ottnt{T}^\ottnt{C}   \ottsym{:}    \overline{ \Gamma_{\ottmv{i}}  \vdash  \ottnt{e_{\ottmv{i}}}  \ottsym{:}  \ottnt{T_{\ottmv{i}}} }^{ \ottmv{i} }   \mathrel{\circ\hspace{-.4em}\rightarrow} \ast $.
\end{prop}
\begin{prop}{fh-ctxeq-closed}
 For any $\Gamma_{{\mathrm{1}}}, ..., \Gamma_{\ottmv{n}}$, $\ottnt{e_{{\mathrm{1}}}}, ..., \ottnt{e_{\ottmv{n}}}$, and
 $\ottnt{T_{{\mathrm{1}}}}, ..., \ottnt{T_{\ottmv{n}}}$ such that
 $ \mathit{FV}  (  \ottnt{e_{\ottmv{i}}}  )   \mathrel{\cup}   \mathit{FTV}  (  \ottnt{e_{\ottmv{i}}}  )   \subseteq   \mathit{dom}  (  \Gamma_{\ottmv{i}}  ) $ for any $\ottmv{i}$,
 \begin{statements}
  \item(term) if $\Gamma  \vdash  \ottnt{C}  \ottsym{:}   \overline{ \Gamma_{\ottmv{i}}  \vdash  \ottnt{e_{\ottmv{i}}}  \ottsym{:}  \ottnt{T_{\ottmv{i}}} }^{ \ottmv{i} }   \mathrel{\circ\hspace{-.4em}\rightarrow}  \ottnt{T}$,
  then $ \mathit{FV}  (  \ottnt{C}  [   \overline{ \ottnt{e_{\ottmv{i}}} }^{ \ottmv{i} }   ]  )   \mathrel{\cup}   \mathit{FTV}  (  \ottnt{C}  [   \overline{ \ottnt{e_{\ottmv{i}}} }^{ \ottmv{i} }   ]  )   \subseteq   \mathit{dom}  (  \Gamma  ) $, and
  \item(type) if $ \Gamma   \vdash   \ottnt{T}^\ottnt{C}   \ottsym{:}    \overline{ \Gamma_{\ottmv{i}}  \vdash  \ottnt{e_{\ottmv{i}}}  \ottsym{:}  \ottnt{T_{\ottmv{i}}} }^{ \ottmv{i} }   \mathrel{\circ\hspace{-.4em}\rightarrow} \ast $,
  then $ \mathit{FV}  (  \ottnt{T}^\ottnt{C}  [   \overline{ \ottnt{e_{\ottmv{i}}} }^{ \ottmv{i} }   ]  )   \mathrel{\cup}   \mathit{FTV}  (  \ottnt{C}  [   \overline{ \ottnt{e_{\ottmv{i}}} }^{ \ottmv{i} }   ]  )   \subseteq   \mathit{dom}  (  \Gamma  ) $.
 \end{statements}

 \proof

 By induction on the derivations of
 $\Gamma  \vdash  \ottnt{C}  \ottsym{:}   \overline{ \Gamma_{\ottmv{i}}  \vdash  \ottnt{e_{\ottmv{i}}}  \ottsym{:}  \ottnt{T_{\ottmv{i}}} }^{ \ottmv{i} }   \mathrel{\circ\hspace{-.4em}\rightarrow}  \ottnt{T}$ and
 $ \Gamma   \vdash   \ottnt{T}^\ottnt{C}   \ottsym{:}    \overline{ \Gamma_{\ottmv{i}}  \vdash  \ottnt{e_{\ottmv{i}}}  \ottsym{:}  \ottnt{T_{\ottmv{i}}} }^{ \ottmv{i} }   \mathrel{\circ\hspace{-.4em}\rightarrow} \ast $.
\end{prop}

\begin{prop}{fh-ctxeq-refl}
 For any $\Gamma_{{\mathrm{1}}}, ..., \Gamma_{\ottmv{n}}$, $\ottnt{e_{{\mathrm{1}}}}, ..., \ottnt{e_{\ottmv{n}}}$, and
 $\ottnt{T_{{\mathrm{1}}}}, ..., \ottnt{T_{\ottmv{n}}}$,
 \begin{statements}
  \item(term) if $\Gamma  \vdash  \ottnt{e}  \ottsym{:}  \ottnt{T}$, then $\Gamma  \vdash  \ottnt{e}  \ottsym{:}   \overline{ \Gamma_{\ottmv{i}}  \vdash  \ottnt{e_{\ottmv{i}}}  \ottsym{:}  \ottnt{T_{\ottmv{i}}} }^{ \ottmv{i} }   \mathrel{\circ\hspace{-.4em}\rightarrow}  \ottnt{T}$,
  \item(val) if $\Gamma  \vdash  \ottnt{v}  \ottsym{:}  \ottnt{T}$, then $\Gamma  \vdash  \ottnt{v}  \ottsym{:}   \overline{ \Gamma_{\ottmv{i}}  \vdash  \ottnt{e_{\ottmv{i}}}  \ottsym{:}  \ottnt{T_{\ottmv{i}}} }^{ \ottmv{i} }   \mathrel{\circ\hspace{-.4em}\rightarrow}  \ottnt{T}$, and
  \item(type) if $\Gamma  \vdash  \ottnt{T}$, then $ \Gamma   \vdash   \ottnt{T}   \ottsym{:}    \overline{ \Gamma_{\ottmv{i}}  \vdash  \ottnt{e_{\ottmv{i}}}  \ottsym{:}  \ottnt{T_{\ottmv{i}}} }^{ \ottmv{i} }   \mathrel{\circ\hspace{-.4em}\rightarrow} \ast $.
 \end{statements}

 \proof

 By induction on the derivations of
 $\Gamma  \vdash  \ottnt{e}  \ottsym{:}  \ottnt{T}$, $\Gamma  \vdash  \ottnt{v}  \ottsym{:}  \ottnt{T}$, and $\Gamma  \vdash  \ottnt{T}$.
\end{prop}

Finally, we define semityped contextual equivalence by using well-formed contexts.
\begin{defi}[Semityped Contextual Equivalence]
 \label{def:ctxeq} \sloppy{
 Terms $\ottnt{e_{{\mathrm{11}}}},..., { \ottnt{e} }_{  \ottsym{1} \mathit{n}  } $ and $\ottnt{e_{{\mathrm{21}}}},..., { \ottnt{e} }_{  \ottsym{2} \mathit{n}  } $ are \emph{contextually
 equivalent} at $\ottnt{T_{{\mathrm{1}}}},...,\ottnt{T_{\ottmv{n}}}$ under $\Gamma_{{\mathrm{1}}},...,\Gamma_{\ottmv{n}}$, respectively,
 written as $ \overline{ \Gamma_{\ottmv{i}}   \vdash    { \ottnt{e} }_{  \ottsym{1} \ottmv{i}  }    =_\mathsf{ctx}    { \ottnt{e} }_{  \ottsym{2} \ottmv{i}  }    \ottsym{:}   \ottnt{T_{\ottmv{i}}} }^{ \ottmv{i} \, \in \,  \{  \, \ottsym{1}  ,\, ... \, ,  \mathit{n} \,  \}  } $,
 if and only if
 (1) for any $\ottmv{i}$,
 $\Gamma_{\ottmv{i}}  \vdash   { \ottnt{e} }_{  \ottsym{1} \ottmv{i}  }   \ottsym{:}  \ottnt{T_{\ottmv{i}}}$ and $ \mathit{FV}  (   { \ottnt{e} }_{  \ottsym{2} \ottmv{i}  }   )   \mathrel{\cup}   \mathit{FTV}  (   { \ottnt{e} }_{  \ottsym{2} \ottmv{i}  }   )   \subseteq   \mathit{dom}  (  \Gamma_{\ottmv{i}}  ) $, and
 (2) for any $\ottnt{T}$ and $\ottnt{C}$,
 if
 \ifdtabs $\ottnt{C}  [  \ottnt{e_{{\mathrm{11}}}}  ,\, ... \, ,   { \ottnt{e} }_{  \ottsym{1} \mathit{n}  }   ]$ and $\ottnt{C}  [  \ottnt{e_{{\mathrm{21}}}}  ,\, ... \, ,   { \ottnt{e} }_{  \ottsym{2} \mathit{n}  }   ]$ are valid, and \fi
 $\emptyset  \vdash  \ottnt{C}  \ottsym{:}   \overline{ \Gamma_{\ottmv{i}}  \vdash   { \ottnt{e} }_{  \ottsym{1} \ottmv{i}  }   \ottsym{:}  \ottnt{T_{\ottmv{i}}} }^{ \ottmv{i} }   \mathrel{\circ\hspace{-.4em}\rightarrow}  \ottnt{T}$,
 then
 $\ottnt{C}  [  \ottnt{e_{{\mathrm{11}}}}  ,\, ... \, ,   { \ottnt{e} }_{  \ottsym{1} \mathit{n}  }   ]  \Downarrow  \ottnt{C}  [  \ottnt{e_{{\mathrm{21}}}}  ,\, ... \, ,   { \ottnt{e} }_{  \ottsym{2} \mathit{n}  }   ]$.
 For simplification,
 we write
 $ \overline{ \Gamma_{\ottmv{i}}   \vdash    { \ottnt{e} }_{  \ottsym{1} \ottmv{i}  }    =_\mathsf{ctx}    { \ottnt{e} }_{  \ottsym{2} \ottmv{i}  }    \ottsym{:}   \ottnt{T_{\ottmv{i}}} }^{ \ottmv{i} } $ if $\mathit{n}$ is not important
 and
 $\Gamma_{{\mathrm{1}}}  \vdash  \ottnt{e_{{\mathrm{11}}}} \, =_\mathsf{ctx} \, \ottnt{e_{{\mathrm{21}}}}  \ottsym{:}  \ottnt{T_{{\mathrm{1}}}}$ if $\mathit{n}  \ottsym{=}  \ottsym{1}$.
 }
\end{defi}
\sloppy{ We note that we state semityped contextual equivalence
for \emph{pairs} of terms and that equivalennce is preseved by dropping some pairs:
 that is, if $ \overline{ \Gamma_{\ottmv{i}}   \vdash    { \ottnt{e} }_{  \ottsym{1} \ottmv{i}  }    =_\mathsf{ctx}    { \ottnt{e} }_{  \ottsym{2} \ottmv{i}  }    \ottsym{:}   \ottnt{T_{\ottmv{i}}} }^{ \ottmv{i} \, \in \,  \{  \, \ottsym{1}  ,\, ... \, ,  \mathit{n} \,  \}  } $, then
$ \overline{ \Gamma_{\ottmv{i}}   \vdash    { \ottnt{e} }_{  \ottsym{1} \ottmv{i}  }    =_\mathsf{ctx}    { \ottnt{e} }_{  \ottsym{2} \ottmv{i}  }    \ottsym{:}   \ottnt{T_{\ottmv{i}}} }^{ \ottmv{i} \, \in \,  \{  \, \ottsym{1}  ,\, ... \, ,  \ottmv{m} \,  \}  } $ for $m \leq n$.}

Finally, we make a few remarks on semityped contextual equivalence.
Although we call it semityped contextual ``equivalence,'' this relation is not quite an
equivalence relation because symmetry does not hold (ill-typed terms cannot be
on the left-hand side).
More interestingly, even showing its transitivity is not trivial.
For proving the transitivity, we have to show that, if $\Gamma  \vdash  \ottnt{e_{{\mathrm{1}}}} \, =_\mathsf{ctx} \, \ottnt{e_{{\mathrm{2}}}}  \ottsym{:}  \ottnt{T}$ and
$\Gamma  \vdash  \ottnt{e_{{\mathrm{2}}}} \, =_\mathsf{ctx} \, \ottnt{e_{{\mathrm{3}}}}  \ottsym{:}  \ottnt{T}$, then $\ottnt{e_{{\mathrm{1}}}}$ and $\ottnt{e_{{\mathrm{3}}}}$ behave equivalently under
any program context $\ottnt{C}$ which is well formed for $\ottnt{e_{{\mathrm{1}}}}$.
We might expect that $\ottnt{e_{{\mathrm{2}}}}$ and $\ottnt{e_{{\mathrm{3}}}}$ behave in the same way under $\ottnt{C}$,
but it is not clear because $\ottnt{C}$ may not be well formed for
$\ottnt{e_{{\mathrm{2}}}}$.
Fortunately, our logical relation enables us to show (restricted) transitivity
of semityped contextual equivalence via completeness with respect to semityped
contextual equivalence (\prop:ref{fh-lr-ctx-trans}).

In some work~\cite{Lassen_1998_Phd,Pitts_2005_ATTAPL}, contextual equivalence is
defined for A-normal forms, where arguments to functions are restricted to
values and terms are composed by $ \mathsf{let} $-expressions (so, they are not
shorthand of term applications there) to reduce clutter.
In fact, we have adopted that style at an early stage of the study but it turned out that it did
not work quite well, because a term in A-normal form is not closed under term
substitution.
To see the problem, let us consider a typing rule for $ \mathsf{let} $-expression
$ \mathsf{let}  ~  \mathit{x}  \mathord{:}  \ottnt{T_{{\mathrm{1}}}}  \equal  \ottnt{e_{{\mathrm{1}}}}  ~ \ottliteralin ~  \ottnt{e_{{\mathrm{2}}}} $, which could be given as follows:
\[
\frac{\Gamma  \vdash  \ottnt{e_{{\mathrm{1}}}}  \ottsym{:}  \ottnt{T_{{\mathrm{1}}}} \quad  \Gamma  ,  \mathit{x}  \mathord{:}  \ottnt{T_{{\mathrm{1}}}}   \vdash  \ottnt{e_{{\mathrm{2}}}}  \ottsym{:}  \ottnt{T_{{\mathrm{2}}}}}{\Gamma  \vdash   \mathsf{let}  ~  \mathit{x}  \mathord{:}  \ottnt{T_{{\mathrm{1}}}}  \equal  \ottnt{e_{{\mathrm{1}}}}  ~ \ottliteralin ~  \ottnt{e_{{\mathrm{2}}}}   \ottsym{:}  \ottnt{T_{{\mathrm{2}}}} \, [  \ottnt{e_{{\mathrm{1}}}}  \ottsym{/}  \mathit{x}  ]}
\]
The problem is that the index type $\ottnt{T_{{\mathrm{2}}}} \, [  \ottnt{e_{{\mathrm{1}}}}  \ottsym{/}  \mathit{x}  ]$ possibly includes refinements
which are not A-normal forms if $\ottnt{e_{{\mathrm{1}}}}$ is neither a variable nor a value.
For example, $ \mathsf{let}  ~  \mathit{x}  \mathord{:}   \mathsf{Int}   \equal   \ottsym{2}  \mathrel{+}  \ottsym{3}   ~ \ottliteralin ~  \langle   \mathsf{Int}   \Rightarrow   \{  \mathit{y}  \mathord{:}   \mathsf{Int}    \mathop{\mid}    \mathit{x}  \mathrel{>} \ottsym{0}   \}   \rangle   ^{ \ell }   \, \ottsym{0}$ is typed at
$ \{  \mathit{x}  \mathord{:}   \mathsf{Int}    \mathop{\mid}     \ottsym{2}  \mathrel{+}  \ottsym{3}   \mathrel{>} \ottsym{0}   \} $, but the refinement
$  \ottsym{2}  \mathrel{+}  \ottsym{3}   \mathrel{>} \ottsym{0} $ is not in A-normal form.
We might be able to define substitution so that $ \{  \mathit{y}  \mathord{:}   \mathsf{Int}    \mathop{\mid}     \mathsf{let}  ~  \mathit{x}  \mathord{:}   \mathsf{Int}   \equal   \ottsym{2}  \mathrel{+}  \ottsym{3}   ~ \ottliteralin ~  \mathit{x}   \mathrel{>} \ottsym{0}   \} $ would be obtained, but we avoid such ``peculiar'' substitution.

While semitypedness of our contextual equivalence is motivated by the upcast
elimination, perhaps surprisingly, it appears unclear to us how
to define \emph{typed} contextual equivalence.
One naive definition of it is to demand that, for each $ { \ottnt{e} }_{  \ottsym{2} \ottmv{i}  } $ in
Definition~\ref{def:ctxeq}, $ { \ottnt{e} }_{  \ottsym{2} \ottmv{i}  } $ is well typed at $\ottnt{T_{\ottmv{i}}}$ under
$\Gamma_{\ottmv{i}}$.
However, this gives rise to ill-typed terms.
For example, suppose that we want to equate $\ottsym{0}$ and $\ottsym{(}    \lambda    \mathit{y}  \mathord{:}   \mathsf{Int}   .  \mathit{y}   \ottsym{)} \, \ottsym{0}$.
To show their contextual equivalence, we have to evaluate them in any
program context.
Here, a context $\langle   \{  \mathit{x}  \mathord{:}   \mathsf{Int}    \mathop{\mid}    \ottsym{0}  \mathrel{<}  \mathit{x}   \}   \Rightarrow   \mathsf{Int}   \rangle   ^{ \ell }  \, \ottsym{(}  \mathit{f} \,  \left[ \, \right] _{ \ottmv{i} }   \ottsym{)}$ given above is well-formed for
$\ottsym{0}$ but not for $\ottsym{(}    \lambda    \mathit{y}  \mathord{:}   \mathsf{Int}   .  \mathit{y}   \ottsym{)} \, \ottsym{0}$; note that we cannot apply \CT{Conv}
to $\mathit{f} \, \ottsym{(}  \ottsym{(}    \lambda    \mathit{y}  \mathord{:}   \mathsf{Int}   .  \mathit{y}   \ottsym{)} \, \ottsym{0}  \ottsym{)}$ due to the reference to free variable $\mathit{f}$.
A better definition may be to require contexts to be well-formed for both
terms that we want to equate.
This definition could exclude contexts like the above whereas it seems to cause
another issue: are program contexts in such a restricted form enough to test
terms?
We conclude that defining \emph{typed} contextual equivalence for a dependently typed calculus
is still an open problem.

\section{Logical Relation}
\label{sec:logical_relation}

We develop a logical relation for two reasons.
The first is parametricity, which ensures abstraction and enables reasoning for
programs in polymorphic calculi~\cite{Wadler_1989_FPCA}.
Parametricity is usually stated as ``any well typed term is
logically related to itself.''
The second is to show contextual equivalence easily.
It is often difficult to prove that given two terms are contextually equivalent
since it involves quantification over \emph{all} program contexts.
Much work has developed techniques to reason about contextual
equivalence more easily, and many of such reasoning techniques are
based on logical relations.
We will also use the logical relation to reason about casts in 
Section~\ref{sec:reasoning}.

In this section, we first give an informal overview of main ideas in
our logical relation in \sect{logical_relation-overview}.  Then, after
preliminary definitions in \sect{logical_relation-preliminary}, we
formally define the logical relation in \sect{logical_relation-def}
and state its soundness and completeness with respect to semityped
contextual equivalence in \sect{logical_relation-statement}.
The completeness is given in a restricted form---two contextually equivalent,
\emph{well-typed} terms are logically related; completeness without restrictions
is left open.

\subsection{Informal Overview}
\label{sec:logical_relation-overview}

The definition of our logical relation follows Belo et
al.~\cite{Belo/Greenberg/Igarashi/Pierce_2011_ESOP} and Sekiyama et
al.~\cite{Sekiyama/Igarashi/Greenberg_2016_TOPLAS}.  We start with two
type-indexed families of relations $ \ottnt{v_{{\mathrm{1}}}}  \simeq_{\mathtt{v} }  \ottnt{v_{{\mathrm{2}}}}   \ottsym{:}   \ottnt{T} ;  \theta ;  \delta $
for closed values and $ \ottnt{e_{{\mathrm{1}}}}  \simeq_{\mathtt{e} }  \ottnt{e_{{\mathrm{2}}}}   \ottsym{:}   \ottnt{T} ;  \theta ;  \delta $ for closed
terms and a relation $ \ottnt{T_{{\mathrm{1}}}}  \simeq  \ottnt{T_{{\mathrm{2}}}}   \ottsym{:}   \ast ;  \theta ;  \delta $ for (open) types.  The
\emph{type interpretation} $\theta$ assigns value relations to type
variables---which is common to relational semantics for a polymorphic
language---and $\delta$, called \emph{value assignment}, gives
pairs of values to free term variables in $\ottnt{T}$, $\ottnt{T_{{\mathrm{1}}}}$, and
$\ottnt{T_{{\mathrm{2}}}}$.  Value assignments are introduced by Belo et
al.~\cite{Belo/Greenberg/Igarashi/Pierce_2011_ESOP} to handle
dependency of types on terms.  Main differences from the previous work
\cite{Belo/Greenberg/Igarashi/Pierce_2011_ESOP,Sekiyama/Igarashi/Greenberg_2016_TOPLAS}
are that (1) our logical relation is semityped just like our
contextual equivalence (whereas the previous work does not enforce
well-typedness conditions) and that (2) different closure conditions
are assumed for relations assigned to type variables.  (We will
elaborate (2) shortly.)  Then, we extend these relations
to open terms/types and define $\Gamma  \vdash  \ottnt{e_{{\mathrm{1}}}} \,  \mathrel{ \simeq }  \, \ottnt{e_{{\mathrm{2}}}}  \ottsym{:}  \ottnt{T}$ and
$ \Gamma \vdash \ottnt{T_{{\mathrm{1}}}}  \mathrel{ \simeq }  \ottnt{T_{{\mathrm{2}}}}  : \ast $.

Formally, a type interpretation $\theta$ assigns a type variable
$\alpha$ a triple $(\ottnt{r},\ottnt{T_{{\mathrm{1}}}},\ottnt{T_{{\mathrm{2}}}})$ where $\ottnt{r}$ is a binary
relation on closed values $\ottsym{(}  \ottnt{v_{{\mathrm{1}}}}  \ottsym{,}  \ottnt{v_{{\mathrm{2}}}}  \ottsym{)}$, where $\ottnt{v_{{\mathrm{1}}}}$ is of type
$\ottnt{T_{{\mathrm{1}}}}$.  There are two closure conditions on $\ottnt{r}$.

The first condition on $\ottnt{r}$ is that it has to be \emph{closed under
  wrappers produced by reflexive casts}: if $\ottsym{(}  \ottnt{v_{{\mathrm{1}}}}  \ottsym{,}  \ottnt{v_{{\mathrm{2}}}}  \ottsym{)} \, \in \, \ottnt{r}$,
the value of $\langle  \ottnt{T_{\ottmv{i}}}  \Rightarrow  \ottnt{T_{\ottmv{i}}}  \rangle   ^{ \ell }  \, \ottnt{v_{\ottmv{i}}}$ is related to $ { \ottnt{v} }_{ \ottsym{3}  \ottsym{-}  \ottmv{i} } $ (for
$i=1,2$).
This closure condition is needed due to polymorphic casts of the form $\langle  \alpha  \Rightarrow  \alpha  \rangle   ^{ \ell } $.
A polymorphic cast is a function typed at $\alpha  \rightarrow  \alpha$, so it should
produce values related at $\alpha$ when taking arguments related at
$\alpha$.
Since values related at $\alpha$ should be in $\ottnt{r}$, the results of
evaluating $\langle  \alpha  \Rightarrow  \alpha  \rangle   ^{ \ell }  \, \ottnt{v_{{\mathrm{1}}}}$ and $\langle  \alpha  \Rightarrow  \alpha  \rangle   ^{ \ell }  \, \ottnt{v_{{\mathrm{2}}}}$ should be in $\ottnt{r}$
for any $\ottsym{(}  \ottnt{v_{{\mathrm{1}}}}  \ottsym{,}  \ottnt{v_{{\mathrm{2}}}}  \ottsym{)} \in \ottnt{r}$ (if they terminate at values).
Unfortunately, it could not be achieved if $\ottnt{r}$ were arbitrary,
because, if $\alpha$ is instantiated with higher-order types,
$\langle  \alpha  \Rightarrow  \alpha  \rangle   ^{ \ell }  \, \ottnt{v_{\ottmv{i}}}$ produces wrappers (e.g., by \R{Fun}) but they
may not be in $\ottnt{r}$.
Thus, instead of taking arbitrary $\ottnt{r}$, we require $\ottnt{r}$ to
contain also the wrappers.\footnote{The prior
  work~\cite{Belo/Greenberg/Igarashi/Pierce_2011_ESOP,Sekiyama/Igarashi/Greenberg_2016_TOPLAS}
  does not need this requirement because reflexive casts always behave
  like identity functions, regardless of their types.  However, it
  causes different
  problems~\cite{Sekiyama/Igarashi/Greenberg_2016_TOPLAS}.}
%
Actually, an alternative requirement that $\ottnt{r}$ relates the value of
$\langle  \ottnt{T_{\ottmv{i}}}  \Rightarrow  \ottnt{T_{\ottmv{i}}}  \rangle   ^{ \ell }  \, \ottnt{v_{\ottmv{i}}}$ to that of $ { \langle   { \ottnt{T} }_{ \ottsym{3}  \ottsym{-}  \ottmv{i} }   \Rightarrow   { \ottnt{T} }_{ \ottsym{3}  \ottsym{-}  \ottmv{i} }   \rangle   ^{ \ell }  \, \ottnt{v} }_{ \ottsym{3}  \ottsym{-}  \ottmv{i} } $ would be enough
if we are interested only in soundness of the logical relation.
Our closure condition---without $\langle   { \ottnt{T} }_{ \ottsym{3}  \ottsym{-}  \ottmv{i} }   \Rightarrow   { \ottnt{T} }_{ \ottsym{3}  \ottsym{-}  \ottmv{i} }   \rangle   ^{ \ell } $---subsumes
this alternative and, in fact, is a key to proving correctness of the upcast
elimination and the selfification.

The second closure condition is that $\ottnt{r}$ is closed under
(semityped) CIU-equivalence so that the logical relation is complete
with respect to contextual equivalence, following the prior
work~\cite{Ahmed_2006_ESOP}.
CIU-equivalence~\cite{Mason/Talcott_1991_JFP} relates two closed terms if they behave equivalently
under any \emph{evaluation} context (\underline{u}se of the terms),
and it is extended to open terms with closing substitutions
(\underline{c}losed \underline{i}nstantiations).
Actually, this condition subsumes the first but it will turn out so,
\emph{only after we finish proving the upcast elimination property in
  \sect{reasoning-upcast-elim}}.  So, we have to introduce the two conditions
separately.
Interestingly, the closure under CIU-equivalence also enables us to
show transitivity of the logical relation.
We will show that CIU-equivalence, the logical relation, and contextual
equivalence coincide on well-typed terms via a property similar to Pitts'
``equivalence-respecting property''~\cite{Pitts_2005_ATTAPL}.

\subsection{Preliminaries}
\label{sec:logical_relation-preliminary}

Here, we give a few preliminary definitions, including CIU-equivalence
and the closure conditions on $\ottnt{r}$, to define the logical relation.

 \begin{defi} \noindent
  \begin{itemize}
   \item $ \mathsf{Typ} $ is the set $\{ \ottnt{T} \mid \emptyset  \vdash  \ottnt{T} \}$ of all closed,
         well-formed types;
   \item $ \mathsf{UTyp} $ is the set $\{ \ottnt{T} \mid  \mathit{FV}  (  \ottnt{T}  )   \mathrel{\cup}   \mathit{FTV}  (  \ottnt{T}  )   \ottsym{=}  \emptyset \}$
         of all closed types;
   \item For each $\ottnt{T} \, \in \, \mathsf{Typ}$,
         $ \mathsf{Val}  (  \ottnt{T}  ) $ is the set $\{ \ottnt{v} \mid \emptyset  \vdash  \ottnt{v}  \ottsym{:}  \ottnt{T} \}$ of all
         closed values of $\ottnt{T}$; and
   \item $ \mathsf{UVal} $ is the set $\{ \ottnt{v} \mid  \mathit{FV}  (  \ottnt{v}  )   \mathrel{\cup}   \mathit{FTV}  (  \ottnt{v}  )   \ottsym{=}  \emptyset \}$ of
         all closed values.
  \end{itemize}
 \end{defi}

 In what follows, (capture-avoiding) \emph{substitutions}, denoted by
 $\sigma$, are maps from term and type variables to closed terms
 and types, respectively, and they can be extended to maps over terms,
 types, etc.\ straightforwardly.
 We write $\sigma  [  \ottnt{v}  \ottsym{/}  \mathit{x}  ]$ and $\sigma  [  \ottnt{T}  \ottsym{/}  \alpha  ]$ for substitutions that
 map $\mathit{x}$ and $\alpha$ to $\ottnt{v}$ and $\ottnt{T}$, respectively,
 and other term/type variables according to $\sigma$.
 Then, we define the notion of closing substitutions.
 \begin{defi}[Closing Substitutions]
  Substitution $\sigma$ is a \emph{closing substitution that respects $\Gamma$}, written $\Gamma  \vdash  \sigma$, if and only if $ \sigma  (  \mathit{x}  )  \, \in \,  \mathsf{Val}  (   \sigma  (   \Gamma  (  \mathit{x}  )   )   ) $ for any
  $\mathit{x} \, \in \,  \mathit{dom}  (  \Gamma  ) $ and $ \sigma  (  \alpha  )  \, \in \, \mathsf{Typ}$ for any $\alpha \, \in \,  \mathit{dom}  (  \Gamma  ) $.
 \end{defi}

 We define CUI-equivalence below.
  Our CIU-equivalence rests on \emph{static} evaluation contexts $\ottnt{E}^\ottnt{S}$, where
 holes do not occur under run-time term constructors such as active checks.
 \[\begin{array}{lll}
  \ottnt{E}^\ottnt{S} &::=&  \left[ \, \right]  \mid {\tt op} \, \ottsym{(}  \ottnt{v_{{\mathrm{1}}}}  \ottsym{,} \, ... \, \ottsym{,}  \ottnt{v_{\ottmv{n}}}  \ottsym{,}  \ottnt{E}^\ottnt{S}  \ottsym{,}  \ottnt{e_{{\mathrm{1}}}}  \ottsym{,} \, ... \, \ottsym{,}  \ottnt{e_{\ottmv{m}}}  \ottsym{)} \mid
   \ottnt{E}^\ottnt{S} \, \ottnt{e} \mid \ottnt{v} \, \ottnt{E}^\ottnt{S} \mid \ottnt{E}^\ottnt{S} \, \ottnt{T}
   \end{array}\]
 Since a static evaluation context is also a (single-hole) context, we use
 the context well-formedness judgments also for static evaluation contexts and
 write $\Gamma  \vdash  \ottnt{E}^\ottnt{S}  \ottsym{:}  \ottsym{(}  \Gamma_{{\mathrm{1}}}  \vdash  \ottnt{e_{{\mathrm{1}}}}  \ottsym{:}  \ottnt{T_{{\mathrm{1}}}}  \ottsym{)}  \mathrel{\circ\hspace{-.4em}\rightarrow}  \ottnt{T'}$.
 Use of static evaluation contexts, instead of evaluation contexts, is important
 to show the equivalence-respecting property, especially,
 \prop:ref{fh-lr-comp-sectx-hole-red}.
 \begin{defi}[Semityped CIU-Equivalence]
  Terms $\ottnt{e_{{\mathrm{1}}}}$ and $\ottnt{e_{{\mathrm{2}}}}$ are \emph{CIU-equivalent at $\ottnt{T}$ under $\Gamma$},
  written $\Gamma  \vdash  \ottnt{e_{{\mathrm{1}}}} \, =_\mathsf{ciu} \, \ottnt{e_{{\mathrm{2}}}}  \ottsym{:}  \ottnt{T}$, if and only if
  (1) $\Gamma  \vdash  \ottnt{e_{{\mathrm{1}}}}  \ottsym{:}  \ottnt{T}$,
  (2) $ \mathit{FV}  (  \ottnt{e_{{\mathrm{2}}}}  )   \mathrel{\cup}   \mathit{FTV}  (  \ottnt{e_{{\mathrm{2}}}}  )   \subseteq   \mathit{dom}  (  \Gamma  ) $, and
  (3) $\ottnt{E}^\ottnt{S}  [   \sigma  (  \ottnt{e_{{\mathrm{1}}}}  )   ]  \Downarrow  \ottnt{E}^\ottnt{S}  [   \sigma  (  \ottnt{e_{{\mathrm{2}}}}  )   ]$, 
      for any $\sigma$, $\ottnt{E}^\ottnt{S}$, and $\ottnt{T'}$ such that
      $\Gamma  \vdash  \sigma$ and
      $\emptyset  \vdash  \ottnt{E}^\ottnt{S}  \ottsym{:}  \ottsym{(}  \emptyset  \vdash   \sigma  (  \ottnt{e_{{\mathrm{1}}}}  )   \ottsym{:}   \sigma  (  \ottnt{T}  )   \ottsym{)}  \mathrel{\circ\hspace{-.4em}\rightarrow}  \ottnt{T'}$.
 \end{defi}

 Using the semityped CIU-equivalence, we define the universe $ \mathsf{VRel}  (  \ottnt{T_{{\mathrm{1}}}} ,  \ottnt{T_{{\mathrm{2}}}}  ) $ of interpretations used for $\ottnt{r}$.
 \begin{defi}[Universe of Interpretations]
  For $\ottnt{T_{{\mathrm{1}}}} \, \in \, \mathsf{Typ}$ and $\ottnt{T_{{\mathrm{2}}}} \, \in \, \mathsf{UTyp}$,
  \[\begin{array}{l@{\;}l@{\;}l}
    \mathsf{VRel}  (  \ottnt{T_{{\mathrm{1}}}} ,  \ottnt{T_{{\mathrm{2}}}}  )  \defeq
   \{ \; R \mathrel{\subseteq}  \mathsf{Val}  (  \ottnt{T_{{\mathrm{1}}}}  )  \mathop{\times}  \mathsf{UVal}  \mid &
   \multicolumn{2}{l}{\forall \ottsym{(}  \ottnt{v_{{\mathrm{1}}}}  \ottsym{,}  \ottnt{v_{{\mathrm{2}}}}  \ottsym{)} \in R.} \\ & \qquad
     \exists \ottnt{v'_{{\mathrm{1}}}}, \ottnt{v'_{{\mathrm{2}}}}. &
      \langle  \ottnt{T_{{\mathrm{1}}}}  \Rightarrow  \ottnt{T_{{\mathrm{1}}}}  \rangle   ^{ \ell }  \, \ottnt{v_{{\mathrm{1}}}}  \longrightarrow^{\ast}  \ottnt{v'_{{\mathrm{1}}}} \text{ and } \\ &&
      \langle  \ottnt{T_{{\mathrm{2}}}}  \Rightarrow  \ottnt{T_{{\mathrm{2}}}}  \rangle   ^{ \ell }  \, \ottnt{v_{{\mathrm{2}}}}  \longrightarrow^{\ast}  \ottnt{v'_{{\mathrm{2}}}} \text{ and } \\ &&
      \ottsym{(}  \ottnt{v'_{{\mathrm{1}}}}  \ottsym{,}  \ottnt{v_{{\mathrm{2}}}}  \ottsym{)}, \ottsym{(}  \ottnt{v_{{\mathrm{1}}}}  \ottsym{,}  \ottnt{v'_{{\mathrm{2}}}}  \ottsym{)} \in R, \text{ and } \\ &
      \multicolumn{2}{l}{ \quad\ \;
       \forall \ottnt{v}.\, \emptyset  \vdash  \ottnt{v} \, =_\mathsf{ciu} \, \ottnt{v_{{\mathrm{1}}}}  \ottsym{:}  \ottnt{T_{{\mathrm{1}}}} \text{ implies }
        \ottsym{(}  \ottnt{v}  \ottsym{,}  \ottnt{v_{{\mathrm{2}}}}  \ottsym{)} \in \mathit{R} \}
      }.
   \end{array}\]
  We write $\langle  \ottnt{r}  \ottsym{,}  \ottnt{T_{{\mathrm{1}}}}  \ottsym{,}  \ottnt{T_{{\mathrm{2}}}}  \rangle$ if $\ottnt{T_{{\mathrm{1}}}} \, \in \, \mathsf{Typ}$,
  $\ottnt{T_{{\mathrm{2}}}} \, \in \, \mathsf{UTyp}$, and $\ottnt{r} \, \in \,  \mathsf{VRel}  (  \ottnt{T_{{\mathrm{1}}}} ,  \ottnt{T_{{\mathrm{2}}}}  ) $.
 \end{defi}
 The conditions above on $\mathit{R}$ represent the closure conditions discussed in
 \sect{logical_relation-overview}.

 \subsection{Formal Definition of Logical Relation}
 \label{sec:logical_relation-def}

 We formally define our logical relation, after defining type interpretations and value assignments below.

  \begin{defi}[Type Interpretations]
  A \emph{type interpretation} $\theta$ is a finite map from type variables to tuples
  $(\ottnt{r},\ottnt{T_{{\mathrm{1}}}},\ottnt{T_{{\mathrm{2}}}})$ such that $\langle  \ottnt{r}  \ottsym{,}  \ottnt{T_{{\mathrm{1}}}}  \ottsym{,}  \ottnt{T_{{\mathrm{2}}}}  \rangle$.
  We write $\theta \,  \{  \,  \alpha  \mapsto ( \ottnt{r} , \ottnt{T_{{\mathrm{1}}}} , \ottnt{T_{{\mathrm{2}}}} )  \,  \} $ for the same map as $\theta$ except
  that $\alpha$ is mapped to $(\ottnt{r},\ottnt{T_{{\mathrm{1}}}},\ottnt{T_{{\mathrm{2}}}})$.
  We also write $\theta_{\ottmv{i}}$ ($\ottmv{i} \, \in \,  \{   \ottsym{1}  \ottsym{,}  \ottsym{2}   \} $) for a substitution that maps
  type variables $\alpha$ to types $\ottnt{T_{\ottmv{i}}}$ such that $ \theta  (  \alpha  ) = (  \ottnt{r} ,  \ottnt{T_{{\mathrm{1}}}} ,  \ottnt{T_{{\mathrm{2}}}}  ) $.
  $ \mathit{dom}  (  \theta  ) $ denotes the set of type variables mapped by
  $\theta$.
 \end{defi}
 \begin{defi}[Value Assignments]
  A \emph{value assignment} $\delta$ is a finite map from term variables to pairs
  $(\ottnt{v_{{\mathrm{1}}}},\ottnt{v_{{\mathrm{2}}}})$ such that $\ottnt{v_{{\mathrm{1}}}} \, \in \,  \mathsf{Val}  (  \ottnt{T}  ) $ for some type $\ottnt{T}$ and
  $\ottnt{v_{{\mathrm{2}}}} \, \in \, \mathsf{UVal}$.
  We write $ \delta    [  \,  (  \ottnt{v_{{\mathrm{1}}}}  ,  \ottnt{v_{{\mathrm{2}}}}  ) /  \mathit{x}  \,  ]  $ for the same mapping as $\delta$
  except that $\mathit{x}$ is mapped to $(\ottnt{v_{{\mathrm{1}}}},\ottnt{v_{{\mathrm{2}}}})$.
  We also write $\delta_{\ottmv{i}}$ ($\ottmv{i} \, \in \,  \{   \ottsym{1}  \ottsym{,}  \ottsym{2}   \} $) for a substitution that maps
  term variables $\mathit{x}$ to values $\ottnt{v_{\ottmv{i}}}$ such that $ \delta  (  \mathit{x}  ) = (  \ottnt{v_{{\mathrm{1}}}} ,  \ottnt{v_{{\mathrm{2}}}}  ) $.
  $ \mathit{dom}  (  \delta  ) $ denotes the set of term variables mapped by
  $\delta$.
 \end{defi}

  \begin{fhfigure*}[t!]
  \begin{flushleft}
   \noindent
   \framebox{$ \ottnt{v_{{\mathrm{1}}}}  \simeq_{\mathtt{v} }  \ottnt{v_{{\mathrm{2}}}}   \ottsym{:}   \ottnt{T} ;  \theta ;  \delta $} \quad {\bf{Value Relation}}
  \end{flushleft}
  \[\begin{array}{rcl}
    \ottnt{v_{{\mathrm{1}}}}  \simeq_{\mathtt{v} }  \ottnt{v_{{\mathrm{2}}}}   \ottsym{:}   \alpha ;  \theta ;  \delta  & \iff &
    \exists \,   \ottnt{r}  ,  \ottnt{T_{{\mathrm{1}}}}   ,  \ottnt{T_{{\mathrm{2}}}}   .~   \theta  (  \alpha  ) = (  \ottnt{r} ,  \ottnt{T_{{\mathrm{1}}}} ,  \ottnt{T_{{\mathrm{2}}}}  )  \text{ and } \ottsym{(}  \ottnt{v_{{\mathrm{1}}}}  \ottsym{,}  \ottnt{v_{{\mathrm{2}}}}  \ottsym{)} \, \in \, \ottnt{r} \\[1ex]
    \ottnt{v_{{\mathrm{1}}}}  \simeq_{\mathtt{v} }  \ottnt{v_{{\mathrm{2}}}}   \ottsym{:}   \ottnt{B} ;  \theta ;  \delta  & \iff &
    \ottnt{v_{{\mathrm{1}}}}  \ottsym{=}  \ottnt{v_{{\mathrm{2}}}} \text{ and } \ottnt{v_{{\mathrm{1}}}} \, \in \,  {\cal K}_{ \ottnt{B} }  \\[1ex]
    \ottnt{v_{{\mathrm{1}}}}  \simeq_{\mathtt{v} }  \ottnt{v_{{\mathrm{2}}}}   \ottsym{:}    \mathit{x} \mathord{:} \ottnt{T_{{\mathrm{1}}}} \rightarrow \ottnt{T_{{\mathrm{2}}}}  ;  \theta ;  \delta  & \iff &
     \forall \ottnt{v'_{{\mathrm{1}}}},\ottnt{v'_{{\mathrm{2}}}}.\;  \ottnt{v'_{{\mathrm{1}}}}  \simeq_{\mathtt{v} }  \ottnt{v'_{{\mathrm{2}}}}   \ottsym{:}   \ottnt{T_{{\mathrm{1}}}} ;  \theta ;  \delta  \text{ implies }
       \\ && \qquad\qquad
        \ottnt{v_{{\mathrm{1}}}} \, \ottnt{v'_{{\mathrm{1}}}}  \simeq_{\mathtt{e} }  \ottnt{v_{{\mathrm{2}}}} \, \ottnt{v'_{{\mathrm{2}}}}   \ottsym{:}   \ottnt{T_{{\mathrm{2}}}} ;  \theta ;   \delta    [  \,  (  \ottnt{v'_{{\mathrm{1}}}}  ,  \ottnt{v'_{{\mathrm{2}}}}  ) /  \mathit{x}  \,  ]   
   \\[1ex]
    \ottnt{v_{{\mathrm{1}}}}  \simeq_{\mathtt{v} }  \ottnt{v_{{\mathrm{2}}}}   \ottsym{:}    \forall   \alpha  .  \ottnt{T}  ;  \theta ;  \delta  & \iff &
    \forall \ottnt{T_{{\mathrm{1}}}}, \ottnt{T_{{\mathrm{2}}}}, \ottnt{r}.\; \langle  \ottnt{r}  \ottsym{,}  \ottnt{T_{{\mathrm{1}}}}  \ottsym{,}  \ottnt{T_{{\mathrm{2}}}}  \rangle \text{ implies }
     \\ && \qquad\qquad
      \ottnt{v_{{\mathrm{1}}}} \, \ottnt{T_{{\mathrm{1}}}}  \simeq_{\mathtt{e} }  \ottnt{v_{{\mathrm{2}}}} \, \ottnt{T_{{\mathrm{2}}}}   \ottsym{:}   \ottnt{T} ;  \theta \,  \{  \,  \alpha  \mapsto  \ottnt{r} , \ottnt{T_{{\mathrm{1}}}} , \ottnt{T_{{\mathrm{2}}}}  \,  \}  ;  \delta 
   \\[1ex]
    \ottnt{v_{{\mathrm{1}}}}  \simeq_{\mathtt{v} }  \ottnt{v_{{\mathrm{2}}}}   \ottsym{:}    \{  \mathit{x}  \mathord{:}  \ottnt{T}   \mathop{\mid}   \ottnt{e}  \}  ;  \theta ;  \delta  & \iff &
      \ottnt{v_{{\mathrm{1}}}}  \simeq_{\mathtt{e} }  \ottnt{v_{{\mathrm{2}}}}   \ottsym{:}   \ottnt{T} ;  \theta ;  \delta  \text{ and } \\ &&
      \theta_{{\mathrm{1}}}  (   \delta_{{\mathrm{1}}}  (  \ottnt{e} \, [  \ottnt{v_{{\mathrm{1}}}}  \ottsym{/}  \mathit{x}  ]  )   )   \longrightarrow^{\ast}   \mathsf{true}  \text{ and }
      \theta_{{\mathrm{2}}}  (   \delta_{{\mathrm{2}}}  (  \ottnt{e} \, [  \ottnt{v_{{\mathrm{2}}}}  \ottsym{/}  \mathit{x}  ]  )   )   \longrightarrow^{\ast}   \mathsf{true}  \\
  \end{array}\] \\[2ex]

  \begin{flushleft}
   \noindent
   \framebox{$ \ottnt{e_{{\mathrm{1}}}}  \simeq_{\mathtt{e} }  \ottnt{e_{{\mathrm{2}}}}   \ottsym{:}   \ottnt{T} ;  \theta ;  \delta $} \quad {\bf{Term Relation}}
  \end{flushleft}
  \[\begin{array}{rcl}
    \ottnt{e_{{\mathrm{1}}}}  \simeq_{\mathtt{e} }  \ottnt{e_{{\mathrm{2}}}}   \ottsym{:}   \ottnt{T} ;  \theta ;  \delta  & \iff &
    \ottnt{e_{{\mathrm{1}}}}  \longrightarrow^{\ast}   \mathord{\Uparrow}  \ell  \text{ and } \ottnt{e_{{\mathrm{2}}}}  \longrightarrow^{\ast}   \mathord{\Uparrow}  \ell , \text{ or } \\ &&
    \ottnt{e_{{\mathrm{1}}}}  \longrightarrow^{\ast}  \ottnt{v_{{\mathrm{1}}}} \text{ and } \ottnt{e_{{\mathrm{2}}}}  \longrightarrow^{\ast}  \ottnt{v_{{\mathrm{2}}}} \text{ and }
      \ottnt{v_{{\mathrm{1}}}}  \simeq_{\mathtt{v} }  \ottnt{v_{{\mathrm{2}}}}   \ottsym{:}   \ottnt{T} ;  \theta ;  \delta 
  \end{array}\]

  \begin{flushleft}
   \noindent
   \framebox{$ \ottnt{T_{{\mathrm{1}}}}  \simeq  \ottnt{T_{{\mathrm{2}}}}   \ottsym{:}   \ast ;  \theta ;  \delta $} \quad {\bf{Type Relation}}
  \end{flushleft}
  \[\begin{array}{rcl}
     \ottnt{B}  \simeq  \ottnt{B}   \ottsym{:}   \ast ;  \theta ;  \delta  \\[1ex]

     \alpha  \simeq  \alpha   \ottsym{:}   \ast ;  \theta ;  \delta  \\[1ex]

      \mathit{x} \mathord{:} \ottnt{T_{{\mathrm{11}}}} \rightarrow \ottnt{T_{{\mathrm{12}}}}   \simeq   \mathit{x} \mathord{:} \ottnt{T_{{\mathrm{21}}}} \rightarrow \ottnt{T_{{\mathrm{22}}}}    \ottsym{:}   \ast ;  \theta ;  \delta  & \iff &
      \ottnt{T_{{\mathrm{11}}}}  \simeq  \ottnt{T_{{\mathrm{21}}}}   \ottsym{:}   \ast ;  \theta ;  \delta  \AND \\
     & & \quad
      \forall   \ottnt{v_{{\mathrm{1}}}}  ,  \ottnt{v_{{\mathrm{2}}}}   .~   \ottnt{v_{{\mathrm{1}}}}  \simeq_{\mathtt{v} }  \ottnt{v_{{\mathrm{2}}}}   \ottsym{:}   \ottnt{T_{{\mathrm{11}}}} ;  \theta ;  \delta  \, \mathbin{ \text{implies} } \,   \\  \,  &  \,  &  \,  \quad  \,  \quad  \, \ottnt{T_{{\mathrm{12}}}}  \simeq  \ottnt{T_{{\mathrm{22}}}}   \ottsym{:}   \ast ;  \theta ;   \delta    [  \,  (  \ottnt{v_{{\mathrm{1}}}}  ,  \ottnt{v_{{\mathrm{2}}}}  ) /  \mathit{x}  \,  ]   
      \\[1ex]

      \forall   \alpha  .  \ottnt{T_{{\mathrm{1}}}}   \simeq   \forall   \alpha  .  \ottnt{T_{{\mathrm{2}}}}    \ottsym{:}   \ast ;  \theta ;  \delta  & \iff &
      \forall    \ottnt{T'_{{\mathrm{1}}}}  ,  \ottnt{T'_{{\mathrm{2}}}}   ,  \ottnt{r}   .~  \langle  \ottnt{r}  \ottsym{,}  \ottnt{T'_{{\mathrm{1}}}}  \ottsym{,}  \ottnt{T'_{{\mathrm{2}}}}  \rangle \, \mathbin{ \text{implies} } \,   \\  \,  &  \,  &  \,  \quad  \, \ottnt{T_{{\mathrm{1}}}}  \simeq  \ottnt{T_{{\mathrm{2}}}}   \ottsym{:}   \ast ;  \theta \,  \{  \,  \alpha  \mapsto  \ottnt{r} , \ottnt{T'_{{\mathrm{1}}}} , \ottnt{T'_{{\mathrm{2}}}}  \,  \}  ;  \delta 
      \\[1ex]

      \{  \mathit{x}  \mathord{:}  \ottnt{T_{{\mathrm{1}}}}   \mathop{\mid}   \ottnt{e_{{\mathrm{1}}}}  \}   \simeq   \{  \mathit{x}  \mathord{:}  \ottnt{T_{{\mathrm{2}}}}   \mathop{\mid}   \ottnt{e_{{\mathrm{2}}}}  \}    \ottsym{:}   \ast ;  \theta ;  \delta  & \iff &
      \ottnt{T_{{\mathrm{1}}}}  \simeq  \ottnt{T_{{\mathrm{2}}}}   \ottsym{:}   \ast ;  \theta ;  \delta  \AND \\
     & & \quad
     \forall   \ottnt{v_{{\mathrm{1}}}}  ,  \ottnt{v_{{\mathrm{2}}}}   .~   \ottnt{v_{{\mathrm{1}}}}  \simeq_{\mathtt{v} }  \ottnt{v_{{\mathrm{2}}}}   \ottsym{:}   \ottnt{T_{{\mathrm{1}}}} ;  \theta ;  \delta  \, \mathbin{ \text{implies} } \,  \\  \,  &  \,  &  \,  \quad  \,  \quad  \,   \theta_{{\mathrm{1}}}  (   \delta_{{\mathrm{1}}}  (  \ottnt{e_{{\mathrm{1}}}} \, [  \ottnt{v_{{\mathrm{1}}}}  \ottsym{/}  \mathit{x}  ]  )   )   \simeq_{\mathtt{e} }   \theta_{{\mathrm{2}}}  (   \delta_{{\mathrm{2}}}  (  \ottnt{e_{{\mathrm{2}}}} \, [  \ottnt{v_{{\mathrm{2}}}}  \ottsym{/}  \mathit{x}  ]  )   )    \ottsym{:}    \mathsf{Bool}  ;  \theta ;  \delta 
  \end{array}\]

  \caption{Value, term, and type relations}
  \label{fig:logirel}
 \end{fhfigure*}

 \begin{defi}[Value, Term, and Type Relations]
  \label{def:act-type}
  We define the value relation $ \ottnt{v_{{\mathrm{1}}}}  \simeq_{\mathtt{v} }  \ottnt{v_{{\mathrm{2}}}}   \ottsym{:}   \ottnt{T} ;  \theta ;  \delta $,
  the term relation $ \ottnt{e_{{\mathrm{1}}}}  \simeq_{\mathtt{e} }  \ottnt{e_{{\mathrm{2}}}}   \ottsym{:}   \ottnt{T} ;  \theta ;  \delta $, and the type relation
  $ \ottnt{T_{{\mathrm{1}}}}  \simeq  \ottnt{T_{{\mathrm{2}}}}   \ottsym{:}   \ast ;  \theta ;  \delta $ by using the rules in \fig{logirel}.
  In these relations, values and terms (resp.\ types) on the left hand side are
  closed and well typed (resp.\ well formed) and those on the right hand side
  are closed:
  \begin{itemize}
   \item if $ \ottnt{v_{{\mathrm{1}}}}  \simeq_{\mathtt{v} }  \ottnt{v_{{\mathrm{2}}}}   \ottsym{:}   \ottnt{T} ;  \theta ;  \delta $, then
         $\emptyset  \vdash  \ottnt{v_{{\mathrm{1}}}}  \ottsym{:}   \theta_{{\mathrm{1}}}  (   \delta_{{\mathrm{1}}}  (  \ottnt{T}  )   ) $ and
         $ \mathit{FV}  (  \ottnt{v_{{\mathrm{2}}}}  )   \mathrel{\cup}   \mathit{FTV}  (  \ottnt{v_{{\mathrm{2}}}}  )   \ottsym{=}  \emptyset$;
   \item if $ \ottnt{e_{{\mathrm{1}}}}  \simeq_{\mathtt{e} }  \ottnt{e_{{\mathrm{2}}}}   \ottsym{:}   \ottnt{T} ;  \theta ;  \delta $, then
         $\emptyset  \vdash  \ottnt{e_{{\mathrm{1}}}}  \ottsym{:}   \theta_{{\mathrm{1}}}  (   \delta_{{\mathrm{1}}}  (  \ottnt{T}  )   ) $ and
         $ \mathit{FV}  (  \ottnt{e_{{\mathrm{2}}}}  )   \mathrel{\cup}   \mathit{FTV}  (  \ottnt{e_{{\mathrm{2}}}}  )   \ottsym{=}  \emptyset$; and
   \item if $ \ottnt{T_{{\mathrm{1}}}}  \simeq  \ottnt{T_{{\mathrm{2}}}}   \ottsym{:}   \ast ;  \theta ;  \delta $, then
         $\emptyset  \vdash   \theta_{{\mathrm{1}}}  (   \delta_{{\mathrm{1}}}  (  \ottnt{T_{{\mathrm{1}}}}  )   ) $ and
         $ \mathit{FV}  (  \ottnt{T_{{\mathrm{2}}}}  )   \mathrel{\cup}   \mathit{FTV}  (  \ottnt{T_{{\mathrm{2}}}}  )   \subseteq   \mathit{dom}  (  \theta  )   \mathrel{\cup}   \mathit{dom}  (  \delta  ) $.
  \end{itemize}
 \end{defi}

 The definitions of value, term, and type relations are quite similar to the
 prior
 work~\cite{Belo/Greenberg/Igarashi/Pierce_2011_ESOP,Sekiyama/Igarashi/Greenberg_2016_TOPLAS},
 but we explain them here briefly.
 Value relations on $\ottnt{B}$ and $\alpha$ are standard.
 Related values $\ottnt{v_{{\mathrm{1}}}}$ and $\ottnt{v_{{\mathrm{2}}}}$ at function type
 $ \mathit{x} \mathord{:} \ottnt{T_{{\mathrm{1}}}} \rightarrow \ottnt{T_{{\mathrm{2}}}} $ have to produce related values when applied to
 related arguments $\ottnt{v'_{{\mathrm{1}}}}$ and $\ottnt{v'_{{\mathrm{2}}}}$ at $\ottnt{T_{{\mathrm{1}}}}$.
 Since $\ottnt{T_{{\mathrm{2}}}}$ may depend on arguments, $\ottnt{v'_{{\mathrm{1}}}}$ and $\ottnt{v'_{{\mathrm{2}}}}$ are recorded in value
 assignment $\delta$ so that the arguments can be referred to by
 refinements in $\ottnt{T_{{\mathrm{2}}}}$.
 Values related at $ \forall   \alpha  .  \ottnt{T} $ produces related values, regardless of 
 the interpretation $(\ottnt{r}, \ottnt{T_{{\mathrm{1}}}}, \ottnt{T_{{\mathrm{2}}}})$ for $\alpha$.
 %
 %
 Values related at $ \{  \mathit{x}  \mathord{:}  \ottnt{T}   \mathop{\mid}   \ottnt{e}  \} $ have to be related at the underlying type
 $\ottnt{T}$ and satisfy refinement $\ottnt{e}$.
 What values and types should be substituted for free variables in $\ottnt{e}$ are
 found in $\delta$ and $\theta$; we evaluate the refinement obtained by applying $\theta$ and $\delta$.
 Term relations contain terms that raise blame with the same label or
 evaluate to related values.
 Type relations, intuitively, relate types with the same ``denotation.''
 Function types are related if both domain and codomain types are related.
 The codomain types may depend on values of the domain types, so we require them
 to be related under an extension of $\delta$ with any pair of values related
 at the \emph{well-formed} domain type---we choose the well-formed type, not the
 possibly ill-formed one, since the index type in a value relation has to be well
 formed (Definition~\ref{def:act-type}).
 Universal types $ \forall   \alpha  .  \ottnt{T_{{\mathrm{1}}}} $ and $ \forall   \alpha  .  \ottnt{T_{{\mathrm{2}}}} $ are related if $\ottnt{T_{{\mathrm{1}}}}$
 and $\ottnt{T_{{\mathrm{2}}}}$ are related under an extension of $\theta$ with any
 interpretation.
 Refinement types are related if both the underlying types and the refinements
 are related; we choose values for the bound variable from the value relation
 indexed by the underlying type $\ottnt{T_{{\mathrm{1}}}}$ on the left hand side because it is well formed.

 Now, we extend term relations for closed terms to open terms.
\begin{defi}
  The relation $\Gamma  \vdash  \theta  \ottsym{;}  \delta$ is defined by: $\Gamma  \vdash  \theta  \ottsym{;}  \delta$
  if and only if
  \begin{enumerate}
   \item for any $\alpha \, \in \, \Gamma$, $\alpha \, \in \,  \mathit{dom}  (  \theta  ) $ and
   \item for any $ \mathit{x}  \mathord{:}  \ottnt{T}   \in   \Gamma $,
         $ {}   \delta_{{\mathrm{1}}}  (  \mathit{x}  )   {}  \simeq_{\mathtt{v} }  {}   \delta_{{\mathrm{2}}}  (  \mathit{x}  )   {}   \ottsym{:}   \ottnt{T} ;  \theta ;  \delta $.
  \end{enumerate}
 \end{defi}

  \begin{defi}[Logical Relation]
   Terms $\ottnt{e_{{\mathrm{1}}}}$ and $\ottnt{e_{{\mathrm{2}}}}$ are \emph{logically related at
   $\ottnt{T}$ under $\Gamma$}, written $\Gamma  \vdash  \ottnt{e_{{\mathrm{1}}}} \,  \mathrel{ \simeq }  \, \ottnt{e_{{\mathrm{2}}}}  \ottsym{:}  \ottnt{T}$, if and only if
   (1) $\Gamma  \vdash  \ottnt{e_{{\mathrm{1}}}}  \ottsym{:}  \ottnt{T}$, (2) $ \mathit{FV}  (  \ottnt{e_{{\mathrm{2}}}}  )   \mathrel{\cup}   \mathit{FTV}  (  \ottnt{e_{{\mathrm{2}}}}  )   \subseteq   \mathit{dom}  (  \Gamma  ) $, and
   (3) 
   $  \theta_{{\mathrm{1}}}  (   \delta_{{\mathrm{1}}}  (  \ottnt{e_{{\mathrm{1}}}}  )   )   \simeq_{\mathtt{e} }   \theta_{{\mathrm{2}}}  (   \delta_{{\mathrm{2}}}  (  \ottnt{e_{{\mathrm{2}}}}  )   )    \ottsym{:}   \ottnt{T} ;  \theta ;  \delta $,
   for any $\theta$ and $\delta$ such that $\Gamma  \vdash  \theta  \ottsym{;}  \delta$.
   Similarly, types $\ottnt{T_{{\mathrm{1}}}}$ and $\ottnt{T_{{\mathrm{2}}}}$ are \emph{logically related under
   $\Gamma$}, written $ \Gamma \vdash \ottnt{T_{{\mathrm{1}}}}  \mathrel{ \simeq }  \ottnt{T_{{\mathrm{2}}}}  : \ast $, if and only if
   (1) $\Gamma  \vdash  \ottnt{T_{{\mathrm{1}}}}$, (2) $ \mathit{FV}  (  \ottnt{T_{{\mathrm{2}}}}  )   \mathrel{\cup}   \mathit{FTV}  (  \ottnt{T_{{\mathrm{2}}}}  )   \subseteq   \mathit{dom}  (  \Gamma  ) $, and
   (3) $ \ottnt{T_{{\mathrm{1}}}}  \simeq  \ottnt{T_{{\mathrm{2}}}}   \ottsym{:}   \ast ;  \theta ;  \delta $,
   for any $\theta$ and $\delta$ such that $\Gamma  \vdash  \theta  \ottsym{;}  \delta$.
 \end{defi}

 \subsection{Soundness and Completeness}
 \label{sec:logical_relation-statement}

 We state the soundness and the completeness of the logical relation with
 respect to the semityped contextual equivalence; we prove them in \sect{proving}.
 \begin{restatable}[Soundness]{thm}{fhlrsound}
  \label{thm:fh-lr-sound}
  For any $\Gamma_{{\mathrm{1}}}, ..., \Gamma_{\ottmv{n}}$, $\ottnt{e_{{\mathrm{11}}}}, ...,  { \ottnt{e} }_{  \ottsym{1} \mathit{n}  } $,
  $\ottnt{e_{{\mathrm{21}}}}, ...,  { \ottnt{e} }_{  \ottsym{2} \mathit{n}  } $, and $\ottnt{T_{{\mathrm{1}}}}, ..., \ottnt{T_{\ottmv{n}}}$,
  if $\Gamma_{\ottmv{i}}  \vdash   { \ottnt{e} }_{  \ottsym{1} \ottmv{i}  }  \,  \mathrel{ \simeq }  \,  { \ottnt{e} }_{  \ottsym{2} \ottmv{i}  }   \ottsym{:}  \ottnt{T_{\ottmv{i}}}$ for $\ottmv{i} \, \in \,  \{  \, \ottsym{1}  ,\, ... \, ,  \mathit{n} \,  \} $,
  then $ \overline{ \Gamma_{\ottmv{i}}   \vdash    { \ottnt{e} }_{  \ottsym{1} \ottmv{i}  }    =_\mathsf{ctx}    { \ottnt{e} }_{  \ottsym{2} \ottmv{i}  }    \ottsym{:}   \ottnt{T_{\ottmv{i}}} }^{ \ottmv{i} } $.
 \end{restatable}
 \begin{restatable}[Completeness with respect to Typed Terms]{thm}{fhlrcompletetyped}
  \label{thm:fh-lr-complete-typed}
 If
 $ \overline{ \Gamma_{\ottmv{i}}  \vdash   { \ottnt{e} }_{  \ottsym{1} \ottmv{i}  }  \, =_\mathsf{ctx} \,  { \ottnt{e} }_{  \ottsym{2} \ottmv{i}  }   \ottsym{:}  \ottnt{T_{\ottmv{i}}} }^{ \ottmv{i} } $ and
 $\Gamma_{\ottmv{j}}  \vdash   { \ottnt{e} }_{  \ottsym{2} \ottmv{j}  }   \ottsym{:}  \ottnt{T_{\ottmv{j}}}$ for any $\ottmv{j}$,
 then
 $\Gamma_{\ottmv{j}}  \vdash   { \ottnt{e} }_{  \ottsym{1} \ottmv{j}  }  \,  \mathrel{ \simeq }  \,  { \ottnt{e} }_{  \ottsym{2} \ottmv{j}  }   \ottsym{:}  \ottnt{T_{\ottmv{j}}}$ for any $\ottmv{j}$.
 \end{restatable}

\section{Proving soundness and completeness}
\label{sec:proving}

This section gives proofs of the soundness and the completeness of the logical relation.
The readers who read this paper for the first time can skip this section.

\subsection{Soundness}
 We start with describing an overview of the proof and then detail it.

\subsubsection{Overview}
 \label{sec:proving-overview}
 Following the prior work on program reasoning with logical
 relations~\cite{Pitts_2005_ATTAPL,Ahmed_2006_ESOP,Dreyer/Ahmed/Birkedal_2011_LMCS,Ahmed/Jamner/Siek/Wadler_2017_ICFP},
 our proof of the soundness rests on so-called the \emph{fundamental property},
 which states that a logical relation is closed under term
 constructors.\footnote{In some
 work~\cite{Ahmed_2006_ESOP,Dreyer/Ahmed/Birkedal_2011_LMCS,Ahmed/Jamner/Siek/Wadler_2017_ICFP}
 the fundamental property means reflexivity of logical relations, but in this
 work it does compatibility of the logical relation as in
 Pitts~\cite{Pitts_2005_ATTAPL}.}
 If we have the fundamental property, it is easy to show the soundness.

 In manifest contracts, dependency of types on terms makes proving the
 fundamental property difficult.
 To see it, let us try to prove that the logical relation is closed under the
 term application constructor:
 \begin{quotation}
  if $\Gamma  \vdash  \ottnt{e_{{\mathrm{11}}}} \,  \mathrel{ \simeq }  \, \ottnt{e_{{\mathrm{21}}}}  \ottsym{:}  \ottsym{(}   \mathit{x} \mathord{:} \ottnt{T_{{\mathrm{1}}}} \rightarrow \ottnt{T_{{\mathrm{2}}}}   \ottsym{)}$ and $\Gamma  \vdash  \ottnt{e_{{\mathrm{12}}}} \,  \mathrel{ \simeq }  \, \ottnt{e_{{\mathrm{22}}}}  \ottsym{:}  \ottnt{T_{{\mathrm{1}}}}$,
  then $\Gamma  \vdash  \ottnt{e_{{\mathrm{11}}}} \, \ottnt{e_{{\mathrm{12}}}} \,  \mathrel{ \simeq }  \, \ottnt{e_{{\mathrm{21}}}} \, \ottnt{e_{{\mathrm{22}}}}  \ottsym{:}  \ottnt{T_{{\mathrm{2}}}} \, [  \ottnt{e_{{\mathrm{12}}}}  \ottsym{/}  \mathit{x}  ]$.
 \end{quotation}
 A problem occurs in the case that $\ottnt{T_{{\mathrm{2}}}}$ is a refinement type
 $ \{  \mathit{y}  \mathord{:}  \ottnt{T'_{{\mathrm{2}}}}   \mathop{\mid}   \ottnt{e'_{{\mathrm{2}}}}  \} $.
 In that case, we have to prove that
 \[
   \ottnt{e_{{\mathrm{11}}}} \, \ottnt{e_{{\mathrm{12}}}}  \simeq_{\mathtt{e} }  \ottnt{e_{{\mathrm{21}}}} \, \ottnt{e_{{\mathrm{22}}}}   \ottsym{:}    \{  \mathit{y}  \mathord{:}  \ottnt{T'_{{\mathrm{2}}}}   \mathop{\mid}   \ottnt{e'_{{\mathrm{2}}}}  \}  \, [  \ottnt{e_{{\mathrm{12}}}}  \ottsym{/}  \mathit{x}  ] ;  \theta ;  \delta 
 \]
 (we omit $\theta_{\ottmv{i}}$ and $\delta_{\ottmv{i}}$ in application terms for simplicity).
 Specifically, we have to show that the evaluation results of
 both $\ottnt{e_{{\mathrm{11}}}} \, \ottnt{e_{{\mathrm{12}}}}$ and $\ottnt{e_{{\mathrm{21}}}} \, \ottnt{e_{{\mathrm{22}}}}$ satisfy refinement $\ottnt{e'_{{\mathrm{2}}}} \, [  \ottnt{e_{{\mathrm{12}}}}  \ottsym{/}  \mathit{x}  ]$.
 On the one hand, it is trivial that $\ottnt{e_{{\mathrm{11}}}} \, \ottnt{e_{{\mathrm{12}}}}$ satisfies $\ottnt{e'_{{\mathrm{2}}}} \, [  \ottnt{e_{{\mathrm{12}}}}  \ottsym{/}  \mathit{x}  ]$
 because the type of $\ottnt{e_{{\mathrm{11}}}} \, \ottnt{e_{{\mathrm{12}}}}$ is $ \{  \mathit{y}  \mathord{:}  \ottnt{T'_{{\mathrm{2}}}}   \mathop{\mid}   \ottnt{e'_{{\mathrm{2}}}}  \}  \, [  \ottnt{e_{{\mathrm{12}}}}  \ottsym{/}  \mathit{x}  ]$ and well-typed
 terms satisfy all refinements in their types (\prop:ref{fh-val-satis-c}).
 On the other hand, while it is easy to show $\ottnt{e_{{\mathrm{21}}}} \, \ottnt{e_{{\mathrm{22}}}}$ satisfies
 $\ottnt{e'_{{\mathrm{2}}}} \, [  \ottnt{e_{{\mathrm{22}}}}  \ottsym{/}  \mathit{x}  ]$, proving that $\ottnt{e_{{\mathrm{21}}}} \, \ottnt{e_{{\mathrm{22}}}}$ satisfies $\ottnt{e'_{{\mathrm{2}}}} \, [  \ottnt{e_{{\mathrm{12}}}}  \ottsym{/}  \mathit{x}  ]$
 is nontrivial.
 Our key idea to addressing the nontrivial case is to \emph{assume} that refinement
 $\ottnt{e'_{{\mathrm{2}}}}$ is logically related to itself.
 This assumption allows us to show that $\ottnt{e'_{{\mathrm{2}}}} \, [  \ottnt{e_{{\mathrm{12}}}}  \ottsym{/}  \mathit{x}  ]$ and $\ottnt{e'_{{\mathrm{2}}}} \, [  \ottnt{e_{{\mathrm{22}}}}  \ottsym{/}  \mathit{x}  ]$
 are logically related since so are $\ottnt{e_{{\mathrm{12}}}}$ and $\ottnt{e_{{\mathrm{22}}}}$.
 Since logically related Boolean expressions evaluate to the same value (if
 any), we obtain that $\ottnt{e_{{\mathrm{21}}}} \, \ottnt{e_{{\mathrm{22}}}}$ satisfies $\ottnt{e'_{{\mathrm{2}}}} \, [  \ottnt{e_{{\mathrm{12}}}}  \ottsym{/}  \mathit{x}  ]$ if and only if
 it does $\ottnt{e'_{{\mathrm{2}}}} \, [  \ottnt{e_{{\mathrm{22}}}}  \ottsym{/}  \mathit{x}  ]$.
 Since the latter can be shown easily, we achieve the goal.
 %
 %
 %
 For a rigorous proof following this idea, we assume that $\Gamma$, $\ottnt{e_{{\mathrm{12}}}}$,
 $\ottnt{T'_{{\mathrm{2}}}}$, and $\ottnt{T_{{\mathrm{1}}}}$ are also logically related to themselves.
 \begin{defi}[Self-Related Typing Contexts]
  $\Gamma$ is \emph{self-related} if and only if
  $ \Gamma_{{\mathrm{1}}} \vdash \ottnt{T}  \mathrel{ \simeq }  \ottnt{T}  : \ast $ for any $\Gamma_{{\mathrm{1}}}$ and $\ottnt{T}$ such that $\Gamma  \ottsym{=}   \Gamma_{{\mathrm{1}}}  ,  \mathit{x}  \mathord{:}  \ottnt{T}   \ottsym{,}  \Gamma'_{{\mathrm{1}}}$.
 \end{defi}
 In summary, we show that:
 \begin{quotation}
   \emph{Suppose that $\Gamma$ is self-related,
   $\Gamma  \vdash  \ottnt{e_{{\mathrm{12}}}} \,  \mathrel{ \simeq }  \, \ottnt{e_{{\mathrm{12}}}}  \ottsym{:}  \ottnt{T_{{\mathrm{1}}}}$, and $ \Gamma \vdash  \mathit{x} \mathord{:} \ottnt{T_{{\mathrm{1}}}} \rightarrow \ottnt{T_{{\mathrm{2}}}}   \mathrel{ \simeq }   \mathit{x} \mathord{:} \ottnt{T_{{\mathrm{1}}}} \rightarrow \ottnt{T_{{\mathrm{2}}}}   : \ast $.}
   If $\Gamma  \vdash  \ottnt{e_{{\mathrm{11}}}} \,  \mathrel{ \simeq }  \, \ottnt{e_{{\mathrm{21}}}}  \ottsym{:}  \ottsym{(}   \mathit{x} \mathord{:} \ottnt{T_{{\mathrm{1}}}} \rightarrow \ottnt{T_{{\mathrm{2}}}}   \ottsym{)}$ and $\Gamma  \vdash  \ottnt{e_{{\mathrm{12}}}} \,  \mathrel{ \simeq }  \, \ottnt{e_{{\mathrm{22}}}}  \ottsym{:}  \ottnt{T_{{\mathrm{1}}}}$,
   then $\Gamma  \vdash  \ottnt{e_{{\mathrm{11}}}} \, \ottnt{e_{{\mathrm{12}}}} \,  \mathrel{ \simeq }  \, \ottnt{e_{{\mathrm{21}}}} \, \ottnt{e_{{\mathrm{22}}}}  \ottsym{:}  \ottnt{T_{{\mathrm{2}}}} \, \ottsym{[}  \ottnt{e_{{\mathrm{12}}}}  \ottsym{/}  \mathit{x}  \ottsym{]}$.
 \end{quotation}
 These additional assumptions, which we call \emph{self-relatedness}, are needed
 also in other term constructors such as type application.
 Self-relatedness assumptions are discharged once the parametricity, which amounts
 to reflexivity of the logical relation, is shown.

 We believe that the parametricity can be shown independently of the fundamental
 property, but their proofs are quite similar, so we organize a proof of
 the soundness as follows to avoid writing similar proofs and save the amount of
 work.
 \begin{enumerate}
  \item Prove that the logical relation is closed under each constructor under
        self-relatedness assumptions.

  \item Prove the parametricity with the lemmas shown in (1).

  \item Prove the soundness of the logical relation by discharging the self-relatedness
        assumptions from the lemmas shown in (1) with the parametricity.
 \end{enumerate}

 \subsubsection{Proof}
 The proof proceeds as follows.
 We start with showing weakening and strengthening of the logical relation
 (Lemmas~\ref{lem:fh-lr-ws-start}--\ref{lem:fh-lr-ws-end}), which are used
 broadly throughout the proof.
 We next prove the most challenging cases in the fundamental property: term
 application (Lemmas~\ref{lem:fh-lr-app-start}--\ref{lem:fh-lr-app-end}) and
 type application (Lemmas~\ref{lem:fh-lr-tapp-start}--\ref{lem:fh-lr-tapp-end}).
 After showing the remaining cases of the fundamental property
 (Lemmas~\ref{lem:fh-lr-other-start}--\ref{lem:fh-lr-other-end}), we prove
 the parametricity (\prop:ref{fh-lr-param}){}
 and then the soundness of the logical relation (\refthm{fh-lr-sound}).

 \paragraph{\bf Weakening and strengthening}
 \label{sec:proving-ws}
 %
 %
 \begin{prop}[name=Value Weakening/Strengthening]{fh-lr-val-ws}
  \label{lem:fh-lr-ws-start}
  Suppose that $\mathit{x}$ is a fresh variable.
  If $\ottnt{v_{{\mathrm{1}}}}$ is a closed well-typed value and $\ottnt{v_{{\mathrm{2}}}}$ is a closed (but not necessarily well-typed) value,
  then:
  \begin{statements}
   \item(trel) $ \ottnt{e_{{\mathrm{1}}}}  \simeq_{\mathtt{e} }  \ottnt{e_{{\mathrm{2}}}}   \ottsym{:}   \ottnt{T} ;  \theta ;  \delta $ iff
         $ \ottnt{e_{{\mathrm{1}}}}  \simeq_{\mathtt{e} }  \ottnt{e_{{\mathrm{2}}}}   \ottsym{:}   \ottnt{T} ;  \theta ;   \delta    [  \,  (  \ottnt{v_{{\mathrm{1}}}}  ,  \ottnt{v_{{\mathrm{2}}}}  ) /  \mathit{x}  \,  ]   $;
   \item(typrel) $ \ottnt{T_{{\mathrm{1}}}}  \simeq  \ottnt{T_{{\mathrm{2}}}}   \ottsym{:}   \ast ;  \theta ;  \delta $ iff
         $ \ottnt{T_{{\mathrm{1}}}}  \simeq  \ottnt{T_{{\mathrm{2}}}}   \ottsym{:}   \ast ;  \theta ;   \delta    [  \,  (  \ottnt{v_{{\mathrm{1}}}}  ,  \ottnt{v_{{\mathrm{2}}}}  ) /  \mathit{x}  \,  ]   $; and
   \item(tctx) $\Gamma  \ottsym{,}  \Gamma'  \vdash  \theta  \ottsym{;}  \delta$ and $ \ottnt{v_{{\mathrm{1}}}}  \simeq_{\mathtt{v} }  \ottnt{v_{{\mathrm{2}}}}   \ottsym{:}   \ottnt{T} ;  \theta ;  \delta $ iff
         $ \Gamma  ,  \mathit{x}  \mathord{:}  \ottnt{T}   \ottsym{,}  \Gamma'  \vdash  \theta  \ottsym{;}   \delta    [  \,  (  \ottnt{v_{{\mathrm{1}}}}  ,  \ottnt{v_{{\mathrm{2}}}}  ) /  \mathit{x}  \,  ]  $.
  \end{statements}
  Moreover, we have the following weakening lemmas:
  \begin{statements}
   \setcounter{enumi}{3}
   \item(log) If $\Gamma  \ottsym{,}  \Gamma'  \vdash  \ottnt{e_{{\mathrm{1}}}} \,  \mathrel{ \simeq }  \, \ottnt{e_{{\mathrm{2}}}}  \ottsym{:}  \ottnt{T'}$ and $\Gamma  \vdash  \ottnt{T}$,
         then $ \Gamma  ,  \mathit{x}  \mathord{:}  \ottnt{T}   \ottsym{,}  \Gamma'  \vdash  \ottnt{e_{{\mathrm{1}}}} \,  \mathrel{ \simeq }  \, \ottnt{e_{{\mathrm{2}}}}  \ottsym{:}  \ottnt{T'}$.
   \item(typlog) If $ \Gamma  \ottsym{,}  \Gamma' \vdash \ottnt{T_{{\mathrm{1}}}}  \mathrel{ \simeq }  \ottnt{T_{{\mathrm{2}}}}  : \ast $ and $\Gamma  \vdash  \ottnt{T}$,
         then $  \Gamma  ,  \mathit{x}  \mathord{:}  \ottnt{T}   \ottsym{,}  \Gamma' \vdash \ottnt{T_{{\mathrm{1}}}}  \mathrel{ \simeq }  \ottnt{T_{{\mathrm{2}}}}  : \ast $.
  \end{statements}

  \proof

  \noindent
  \begin{enmrt}
   \item By straightforward induction on $\ottnt{T}$.
   \item By straightforward induction on $\ottnt{T_{{\mathrm{1}}}}$.
   \item By definition and \prop:ref{(trel)}.
   \item By \prop:ref{(tctx)} and \prop:ref{(trel)}.
   \item By \prop:ref{(tctx)} and \prop:ref{(typrel)}. \qedhere
  \end{enmrt}
\end{prop}

\begin{prop}[name=Type Weakening/Strengthening]{fh-lr-typ-ws}
 \label{lem:fh-lr-ws-end}
 Suppose that $\alpha$ is a fresh type variable.
 \begin{statements}
  \item(trel) $ \ottnt{e_{{\mathrm{1}}}}  \simeq_{\mathtt{e} }  \ottnt{e_{{\mathrm{2}}}}   \ottsym{:}   \ottnt{T} ;  \theta ;  \delta $ and $\langle  \ottnt{r}  \ottsym{,}  \ottnt{T_{{\mathrm{1}}}}  \ottsym{,}  \ottnt{T_{{\mathrm{2}}}}  \rangle$ iff
        $ \ottnt{e_{{\mathrm{1}}}}  \simeq_{\mathtt{e} }  \ottnt{e_{{\mathrm{2}}}}   \ottsym{:}   \ottnt{T} ;  \theta \,  \{  \,  \alpha  \mapsto  \ottnt{r} , \ottnt{T_{{\mathrm{1}}}} , \ottnt{T_{{\mathrm{2}}}}  \,  \}  ;  \delta $;
  \item(typrel) $ \ottnt{T'_{{\mathrm{1}}}}  \simeq  \ottnt{T'_{{\mathrm{2}}}}   \ottsym{:}   \ast ;  \theta ;  \delta $ and $\langle  \ottnt{r}  \ottsym{,}  \ottnt{T_{{\mathrm{1}}}}  \ottsym{,}  \ottnt{T_{{\mathrm{2}}}}  \rangle$ iff
        $ \ottnt{T'_{{\mathrm{1}}}}  \simeq  \ottnt{T'_{{\mathrm{2}}}}   \ottsym{:}   \ast ;  \theta \,  \{  \,  \alpha  \mapsto  \ottnt{r} , \ottnt{T_{{\mathrm{1}}}} , \ottnt{T_{{\mathrm{2}}}}  \,  \}  ;  \delta $; and
  \item(tctx) $\Gamma  \ottsym{,}  \Gamma'  \vdash  \theta  \ottsym{;}  \delta$ and $\langle  \ottnt{r}  \ottsym{,}  \ottnt{T_{{\mathrm{1}}}}  \ottsym{,}  \ottnt{T_{{\mathrm{2}}}}  \rangle$
        iff $\Gamma  \ottsym{,}  \alpha  \ottsym{,}  \Gamma'  \vdash  \theta \,  \{  \,  \alpha  \mapsto  \ottnt{r} , \ottnt{T_{{\mathrm{1}}}} , \ottnt{T_{{\mathrm{2}}}}  \,  \}   \ottsym{;}  \delta$.
  \item(log) $\Gamma  \ottsym{,}  \Gamma'  \vdash  \ottnt{e_{{\mathrm{1}}}} \,  \mathrel{ \simeq }  \, \ottnt{e_{{\mathrm{2}}}}  \ottsym{:}  \ottnt{T}$ iff $\Gamma  \ottsym{,}  \alpha  \ottsym{,}  \Gamma'  \vdash  \ottnt{e_{{\mathrm{1}}}} \,  \mathrel{ \simeq }  \, \ottnt{e_{{\mathrm{2}}}}  \ottsym{:}  \ottnt{T}$.
  \item(typlog) $ \Gamma  \ottsym{,}  \Gamma' \vdash \ottnt{T_{{\mathrm{1}}}}  \mathrel{ \simeq }  \ottnt{T_{{\mathrm{2}}}}  : \ast $ iff $ \Gamma  \ottsym{,}  \alpha  \ottsym{,}  \Gamma' \vdash \ottnt{T_{{\mathrm{1}}}}  \mathrel{ \simeq }  \ottnt{T_{{\mathrm{2}}}}  : \ast $.
 \end{statements}

 \proof

 Similar to \prop:ref{fh-lr-val-ws}.
\end{prop}

\paragraph{\bf Fundamental property: term application}
\label{sec:proving-term-app}
To show that the logical relation is closed under term application,
we have to prove that, if
$ \ottnt{v_{{\mathrm{11}}}}  \simeq_{\mathtt{v} }  \ottnt{v_{{\mathrm{21}}}}   \ottsym{:}   \ottsym{(}   \mathit{x} \mathord{:} \ottnt{T_{{\mathrm{1}}}} \rightarrow \ottnt{T_{{\mathrm{2}}}}   \ottsym{)} ;  \theta ;  \delta $ and
$ \ottnt{v_{{\mathrm{12}}}}  \simeq_{\mathtt{v} }  \ottnt{v_{{\mathrm{22}}}}   \ottsym{:}   \ottnt{T_{{\mathrm{1}}}} ;  \theta ;  \delta $, then
$ \ottnt{v_{{\mathrm{11}}}} \, \ottnt{v_{{\mathrm{12}}}}  \simeq_{\mathtt{e} }  \ottnt{v_{{\mathrm{21}}}} \, \ottnt{v_{{\mathrm{22}}}}   \ottsym{:}   \ottnt{T_{{\mathrm{2}}}} \, [  \ottnt{v_{{\mathrm{12}}}}  \ottsym{/}  \mathit{x}  ] ;  \theta ;  \delta $.
However, the definition of the logical relation states only that
$ \ottnt{v_{{\mathrm{11}}}} \, \ottnt{v_{{\mathrm{12}}}}  \simeq_{\mathtt{e} }  \ottnt{v_{{\mathrm{21}}}} \, \ottnt{v_{{\mathrm{22}}}}   \ottsym{:}   \ottnt{T_{{\mathrm{2}}}} ;  \theta ;   \delta    [  \,  (  \ottnt{v_{{\mathrm{12}}}}  ,  \ottnt{v_{{\mathrm{22}}}}  ) /  \mathit{x}  \,  ]   $.
Thus, we have to show that the term relation indexed by $\ottnt{T_{{\mathrm{2}}}} \, [  \ottnt{v_{{\mathrm{12}}}}  \ottsym{/}  \mathit{x}  ]$ with
$\delta$ is equivalent to the one indexed by $\ottnt{T_{{\mathrm{2}}}}$ with
$ \delta    [  \,  (  \ottnt{v_{{\mathrm{12}}}}  ,  \ottnt{v_{{\mathrm{22}}}}  ) /  \mathit{x}  \,  ]  $.
This property is generalized to the so-called \emph{term
compositionality}~\cite{Belo/Greenberg/Igarashi/Pierce_2011_ESOP,Sekiyama/Igarashi/Greenberg_2016_TOPLAS}.
To prove the term compositionality, we first show that, for any $\ottnt{v}$, value
assignments $ \delta    [  \,  (  \ottnt{v}  ,  \ottnt{v_{{\mathrm{12}}}}  ) /  \mathit{x}  \,  ]  $ and $ \delta    [  \,  (  \ottnt{v}  ,  \ottnt{v_{{\mathrm{22}}}}  ) /  \mathit{x}  \,  ]  $ are not
distinguished by term relations.
The following lemma also shows that term relations cannot distinguish type
interpretations that refer to different, possibly ill-formed types; this is used
in the case of type application.
\begin{prop}{fh-lr-untyped-exchange-trel}
 \label{lem:fh-lr-app-start}
 Given $\theta$, $\theta'$, $\delta$, and $\delta'$, suppose
 that
 $\{ (\alpha,\ottnt{r},\ottnt{T_{{\mathrm{1}}}}) \mid \exists \, \ottnt{T_{{\mathrm{2}}}}  .~   \theta  (  \alpha  ) = (  \ottnt{r} ,  \ottnt{T_{{\mathrm{1}}}} ,  \ottnt{T_{{\mathrm{2}}}}  )  \} =
  \{ (\alpha,\ottnt{r},\ottnt{T_{{\mathrm{1}}}}) \mid \exists \, \ottnt{T_{{\mathrm{2}}}}  .~   \theta'  (  \alpha  ) = (  \ottnt{r} ,  \ottnt{T_{{\mathrm{1}}}} ,  \ottnt{T_{{\mathrm{2}}}}  )  \}$
 and
 $\delta_{{\mathrm{1}}}  \ottsym{=}  \delta'_{{\mathrm{1}}}$.
 %
 If $ \ottnt{T}  \simeq  \ottnt{T}   \ottsym{:}   \ast ;  \theta ;  \delta $ and $ \ottnt{T}  \simeq  \ottnt{T}   \ottsym{:}   \ast ;  \theta' ;  \delta' $,
 then
 $ \ottnt{e_{{\mathrm{1}}}}  \simeq_{\mathtt{e} }  \ottnt{e_{{\mathrm{2}}}}   \ottsym{:}   \ottnt{T} ;  \theta ;  \delta $ iff $ \ottnt{e_{{\mathrm{1}}}}  \simeq_{\mathtt{e} }  \ottnt{e_{{\mathrm{2}}}}   \ottsym{:}   \ottnt{T} ;  \theta' ;  \delta' $
 for any $\ottnt{e_{{\mathrm{1}}}}$ and $\ottnt{e_{{\mathrm{2}}}}$.

 \proof

 By induction on $\ottnt{T}$.
 The interesting case is that $\ottnt{T}  \ottsym{=}   \{  \mathit{x}  \mathord{:}  \ottnt{T'}   \mathop{\mid}   \ottnt{e'}  \} $.
 If $\ottnt{e_{{\mathrm{1}}}}$ and $\ottnt{e_{{\mathrm{2}}}}$ raise blame, the conclusion follows straightforwardly.
 Otherwise, $\ottnt{e_{{\mathrm{1}}}}  \longrightarrow^{\ast}  \ottnt{v_{{\mathrm{1}}}}$ and $\ottnt{e_{{\mathrm{2}}}}  \longrightarrow^{\ast}  \ottnt{v_{{\mathrm{2}}}}$ for some $\ottnt{v_{{\mathrm{1}}}}$ and $\ottnt{v_{{\mathrm{2}}}}$,
 and, by definition, we have to show that
 \[\begin{array}{c}
   \ottnt{v_{{\mathrm{1}}}}  \simeq_{\mathtt{v} }  \ottnt{v_{{\mathrm{2}}}}   \ottsym{:}   \ottnt{T'} ;  \theta ;  \delta  \text{ and }
    \theta_{{\mathrm{1}}}  (   \delta_{{\mathrm{1}}}  (  \ottnt{e'} \, [  \ottnt{v_{{\mathrm{1}}}}  \ottsym{/}  \mathit{x}  ]  )   )   \longrightarrow^{\ast}   \mathsf{true}  \text{ and }
    \theta_{{\mathrm{2}}}  (   \delta_{{\mathrm{2}}}  (  \ottnt{e'} \, [  \ottnt{v_{{\mathrm{2}}}}  \ottsym{/}  \mathit{x}  ]  )   )   \longrightarrow^{\ast}   \mathsf{true}  \\
  \text{iff} \\
   \ottnt{v_{{\mathrm{1}}}}  \simeq_{\mathtt{v} }  \ottnt{v_{{\mathrm{2}}}}   \ottsym{:}   \ottnt{T'} ;  \theta' ;  \delta'  \text{ and }
    \theta'_{{\mathrm{1}}}  (   \delta'_{{\mathrm{1}}}  (  \ottnt{e'} \, [  \ottnt{v_{{\mathrm{1}}}}  \ottsym{/}  \mathit{x}  ]  )   )   \longrightarrow^{\ast}   \mathsf{true}  \text{ and }
    \theta'_{{\mathrm{2}}}  (   \delta'_{{\mathrm{2}}}  (  \ottnt{e'} \, [  \ottnt{v_{{\mathrm{2}}}}  \ottsym{/}  \mathit{x}  ]  )   )   \longrightarrow^{\ast}   \mathsf{true} .
 \end{array}\]
 We show only the left-to-right direction, but the other is also shown in a similar
 way.
 Since $  \{  \mathit{x}  \mathord{:}  \ottnt{T'}   \mathop{\mid}   \ottnt{e'}  \}   \simeq   \{  \mathit{x}  \mathord{:}  \ottnt{T'}   \mathop{\mid}   \ottnt{e'}  \}    \ottsym{:}   \ast ;  \theta ;  \delta $ and
 $  \{  \mathit{x}  \mathord{:}  \ottnt{T'}   \mathop{\mid}   \ottnt{e'}  \}   \simeq   \{  \mathit{x}  \mathord{:}  \ottnt{T'}   \mathop{\mid}   \ottnt{e'}  \}    \ottsym{:}   \ast ;  \theta' ;  \delta' $,
 we have $ \ottnt{T'}  \simeq  \ottnt{T'}   \ottsym{:}   \ast ;  \theta ;  \delta $ and $ \ottnt{T'}  \simeq  \ottnt{T'}   \ottsym{:}   \ast ;  \theta' ;  \delta' $.
 Since $ \ottnt{v_{{\mathrm{1}}}}  \simeq_{\mathtt{v} }  \ottnt{v_{{\mathrm{2}}}}   \ottsym{:}   \ottnt{T'} ;  \theta ;  \delta $ by the assumption in the left-to-right
 direction, we have
 \[
   \ottnt{v_{{\mathrm{1}}}}  \simeq_{\mathtt{v} }  \ottnt{v_{{\mathrm{2}}}}   \ottsym{:}   \ottnt{T'} ;  \theta' ;  \delta' 
 \]
 by the IH.
 By the assumptions of this lemma,
 $ \theta_{{\mathrm{1}}}  (   \delta_{{\mathrm{1}}}  (  \ottnt{e'} \, [  \ottnt{v_{{\mathrm{1}}}}  \ottsym{/}  \mathit{x}  ]  )   )   \ottsym{=}   \theta'_{{\mathrm{1}}}  (   \delta'_{{\mathrm{1}}}  (  \ottnt{e'} \, [  \ottnt{v_{{\mathrm{1}}}}  \ottsym{/}  \mathit{x}  ]  )   ) $.
 Since in the left-to-right direction we assume that
 $ \theta_{{\mathrm{1}}}  (   \delta_{{\mathrm{1}}}  (  \ottnt{e'} \, [  \ottnt{v_{{\mathrm{1}}}}  \ottsym{/}  \mathit{x}  ]  )   )   \longrightarrow^{\ast}   \mathsf{true} $,
 we have
 \[
   \theta'_{{\mathrm{1}}}  (   \delta'_{{\mathrm{1}}}  (  \ottnt{e'} \, [  \ottnt{v_{{\mathrm{1}}}}  \ottsym{/}  \mathit{x}  ]  )   )   \longrightarrow^{\ast}   \mathsf{true} .
 \]
 Since $  \{  \mathit{x}  \mathord{:}  \ottnt{T'}   \mathop{\mid}   \ottnt{e'}  \}   \simeq   \{  \mathit{x}  \mathord{:}  \ottnt{T'}   \mathop{\mid}   \ottnt{e'}  \}    \ottsym{:}   \ast ;  \theta' ;  \delta' $ and
 $ \ottnt{v_{{\mathrm{1}}}}  \simeq_{\mathtt{v} }  \ottnt{v_{{\mathrm{2}}}}   \ottsym{:}   \ottnt{T'} ;  \theta' ;  \delta' $, we have
 \[
    \theta'_{{\mathrm{1}}}  (   \delta'_{{\mathrm{1}}}  (  \ottnt{e'} \, [  \ottnt{v_{{\mathrm{1}}}}  \ottsym{/}  \mathit{x}  ]  )   )   \simeq_{\mathtt{e} }   \theta'_{{\mathrm{2}}}  (   \delta'_{{\mathrm{2}}}  (  \ottnt{e'} \, [  \ottnt{v_{{\mathrm{2}}}}  \ottsym{/}  \mathit{x}  ]  )   )    \ottsym{:}    \mathsf{Bool}  ;  \theta' ;  \delta' 
 \]
 by definition.
 Since terms related at $ \mathsf{Bool} $ evaluate to the same value,
 we have
 \[
   \theta'_{{\mathrm{2}}}  (   \delta'_{{\mathrm{2}}}  (  \ottnt{e'} \, [  \ottnt{v_{{\mathrm{2}}}}  \ottsym{/}  \mathit{x}  ]  )   )   \longrightarrow^{\ast}   \mathsf{true} .
 \] \qedhere
\end{prop}
\begin{prop}{fh-lr-val-exchange-wf}
 Suppose that $ \Gamma  ,  \mathit{x}  \mathord{:}  \ottnt{T}   \ottsym{,}  \Gamma'$ is self-related.
 If 
 $ \Gamma  ,  \mathit{x}  \mathord{:}  \ottnt{T}   \ottsym{,}  \Gamma'  \vdash  \theta  \ottsym{;}   \delta    [  \,  (  \ottnt{v_{{\mathrm{1}}}}  ,  \ottnt{v_{{\mathrm{2}}}}  ) /  \mathit{x}  \,  ]  $ and
 $ \ottnt{v_{{\mathrm{1}}}}  \simeq_{\mathtt{v} }  \ottnt{v'_{{\mathrm{2}}}}   \ottsym{:}   \ottnt{T} ;  \theta ;   \delta    [  \,  (  \ottnt{v_{{\mathrm{1}}}}  ,  \ottnt{v_{{\mathrm{2}}}}  ) /  \mathit{x}  \,  ]   $,
 then
 $ \Gamma  ,  \mathit{x}  \mathord{:}  \ottnt{T}   \ottsym{,}  \Gamma'  \vdash  \theta  \ottsym{;}   \delta    [  \,  (  \ottnt{v_{{\mathrm{1}}}}  ,  \ottnt{v'_{{\mathrm{2}}}}  ) /  \mathit{x}  \,  ]  $.

 \proof

 By induction on $\Gamma'$.
 The case that $\Gamma'  \ottsym{=}   \Gamma''  ,  \mathit{y}  \mathord{:}  \ottnt{T'} $ is shown
 by \prop:ref{fh-lr-untyped-exchange-trel}.
\end{prop}

%
Now, we show the term compositionality.
In the statement, $\ottnt{v'_{{\mathrm{1}}}}$ and $\ottnt{v'_{{\mathrm{2}}}}$ correspond to $\ottnt{v_{{\mathrm{12}}}}$ and $\ottnt{v_{{\mathrm{22}}}}$, respectively, in the paragraph informally exlaining this property and $\ottnt{e'}$ to $\ottnt{e_{{\mathrm{12}}}}$
discussed in the second paragraph of \sect{proving-overview}.
\begin{prop}[name=Term Compositionality]{fh-lr-term-comp}
 Suppose that $ \Gamma  ,  \mathit{x}  \mathord{:}  \ottnt{T'}   \ottsym{,}  \Gamma'$ is self-related and $  \Gamma  ,  \mathit{x}  \mathord{:}  \ottnt{T'}   \ottsym{,}  \Gamma' \vdash \ottnt{T}  \mathrel{ \simeq }  \ottnt{T}  : \ast $.
 If
 $ \Gamma  ,  \mathit{x}  \mathord{:}  \ottnt{T'}   \ottsym{,}  \Gamma'  \vdash  \theta  \ottsym{;}   \delta    [  \,  (  \ottnt{v'_{{\mathrm{1}}}}  ,  \ottnt{v'_{{\mathrm{2}}}}  ) /  \mathit{x}  \,  ]  $ and
 $ \theta_{{\mathrm{1}}}  (   \delta_{{\mathrm{1}}}  (  \ottnt{e'}  )   )   \longrightarrow^{\ast}  \ottnt{v'_{{\mathrm{1}}}}$ and
 $ \theta_{{\mathrm{2}}}  (   \delta_{{\mathrm{2}}}  (  \ottnt{e'}  )   )   \longrightarrow^{\ast}  \ottnt{v''_{{\mathrm{2}}}}$ and
 $ \ottnt{v'_{{\mathrm{1}}}}  \simeq_{\mathtt{v} }  \ottnt{v''_{{\mathrm{2}}}}   \ottsym{:}   \ottnt{T'} ;  \theta ;  \delta $,
 then
 $ \ottnt{e_{{\mathrm{1}}}}  \simeq_{\mathtt{e} }  \ottnt{e_{{\mathrm{2}}}}   \ottsym{:}   \ottnt{T} ;  \theta ;   \delta    [  \,  (  \ottnt{v'_{{\mathrm{1}}}}  ,  \ottnt{v'_{{\mathrm{2}}}}  ) /  \mathit{x}  \,  ]   $ iff
 $ \ottnt{e_{{\mathrm{1}}}}  \simeq_{\mathtt{e} }  \ottnt{e_{{\mathrm{2}}}}   \ottsym{:}   \ottnt{T} \, [  \ottnt{e'}  \ottsym{/}  \mathit{x}  ] ;  \theta ;  \delta $
 for any $\ottnt{e_{{\mathrm{1}}}}$ and $\ottnt{e_{{\mathrm{2}}}}$.

 \proof

 By induction on $\ottnt{T}$.
 If $\ottnt{e_{{\mathrm{1}}}}$ and $\ottnt{e_{{\mathrm{2}}}}$ raise blame, then the conclusion is obvious.
 In what follows, suppose that
 $\ottnt{e_{{\mathrm{1}}}}  \longrightarrow^{\ast}  \ottnt{v_{{\mathrm{1}}}}$ and $\ottnt{e_{{\mathrm{2}}}}  \longrightarrow^{\ast}  \ottnt{v_{{\mathrm{2}}}}$.
 All cases except that $\ottnt{T}$ is a refinement type are straightforward
 by the IH(s).

 Let us consider the case that $\ottnt{T}  \ottsym{=}   \{  \mathit{y}  \mathord{:}  \ottnt{T''}   \mathop{\mid}   \ottnt{e''}  \} $.
 Without loss of generality, we can suppose that $\mathit{y} \, \notin \,  \mathit{dom}  (  \delta  ) $.
 We have to show:
 \[\begin{array}{lcl}
   \ottnt{v_{{\mathrm{1}}}}  \simeq_{\mathtt{v} }  \ottnt{v_{{\mathrm{2}}}}   \ottsym{:}   \ottnt{T''} ;  \theta ;   \delta    [  \,  (  \ottnt{v'_{{\mathrm{1}}}}  ,  \ottnt{v'_{{\mathrm{2}}}}  ) /  \mathit{x}  \,  ]    &&
    \ottnt{v_{{\mathrm{1}}}}  \simeq_{\mathtt{v} }  \ottnt{v_{{\mathrm{2}}}}   \ottsym{:}   \ottnt{T''} \, [  \ottnt{e'}  \ottsym{/}  \mathit{x}  ] ;  \theta ;  \delta  \\
   \theta_{{\mathrm{1}}}  (   \delta_{{\mathrm{1}}}  (  \ottnt{e''} \, [  \ottnt{v'_{{\mathrm{1}}}}  \ottsym{/}  \mathit{x}  ]  )   )  \, [  \ottnt{v_{{\mathrm{1}}}}  \ottsym{/}  \mathit{y}  ]  \longrightarrow^{\ast}   \mathsf{true}  & \text{iff} &
    \theta_{{\mathrm{1}}}  (   \delta_{{\mathrm{1}}}  (  \ottnt{e''} \, [  \ottnt{e'}  \ottsym{/}  \mathit{x}  ]  )   )  \, [  \ottnt{v_{{\mathrm{1}}}}  \ottsym{/}  \mathit{y}  ]  \longrightarrow^{\ast}   \mathsf{true}  \\
   \theta_{{\mathrm{2}}}  (   \delta_{{\mathrm{2}}}  (  \ottnt{e''} \, [  \ottnt{v'_{{\mathrm{2}}}}  \ottsym{/}  \mathit{x}  ]  )   )  \, [  \ottnt{v_{{\mathrm{2}}}}  \ottsym{/}  \mathit{y}  ]  \longrightarrow^{\ast}   \mathsf{true}  &&
    \theta_{{\mathrm{2}}}  (   \delta_{{\mathrm{2}}}  (  \ottnt{e''} \, [  \ottnt{e'}  \ottsym{/}  \mathit{x}  ]  )   )  \, [  \ottnt{v_{{\mathrm{2}}}}  \ottsym{/}  \mathit{y}  ]  \longrightarrow^{\ast}   \mathsf{true} 
   \end{array}\]
 %

 First, we show the left-to-right direction.
 Since $  \Gamma  ,  \mathit{x}  \mathord{:}  \ottnt{T'}   \ottsym{,}  \Gamma' \vdash  \{  \mathit{y}  \mathord{:}  \ottnt{T''}   \mathop{\mid}   \ottnt{e''}  \}   \mathrel{ \simeq }   \{  \mathit{y}  \mathord{:}  \ottnt{T''}   \mathop{\mid}   \ottnt{e''}  \}   : \ast $,
 it is easy to show that $  \Gamma  ,  \mathit{x}  \mathord{:}  \ottnt{T'}   \ottsym{,}  \Gamma' \vdash \ottnt{T''}  \mathrel{ \simeq }  \ottnt{T''}  : \ast $.
 Since $ \ottnt{v_{{\mathrm{1}}}}  \simeq_{\mathtt{v} }  \ottnt{v_{{\mathrm{2}}}}   \ottsym{:}   \ottnt{T''} ;  \theta ;   \delta    [  \,  (  \ottnt{v'_{{\mathrm{1}}}}  ,  \ottnt{v'_{{\mathrm{2}}}}  ) /  \mathit{x}  \,  ]   $
 by the assumption in the left-to-right direction, we have
 \[
   \ottnt{v_{{\mathrm{1}}}}  \simeq_{\mathtt{v} }  \ottnt{v_{{\mathrm{2}}}}   \ottsym{:}   \ottnt{T''} \, [  \ottnt{e'}  \ottsym{/}  \mathit{x}  ] ;  \theta ;  \delta 
 \]
 by the IH.
 Since $ \theta_{{\mathrm{1}}}  (   \delta_{{\mathrm{1}}}  (  \ottnt{e'}  )   )   \longrightarrow^{\ast}  \ottnt{v'_{{\mathrm{1}}}}$ and
 $ \theta_{{\mathrm{1}}}  (   \delta_{{\mathrm{1}}}  (  \ottnt{e''} \, [  \ottnt{v'_{{\mathrm{1}}}}  \ottsym{/}  \mathit{x}  ]  )   )  \, [  \ottnt{v_{{\mathrm{1}}}}  \ottsym{/}  \mathit{y}  ]  \longrightarrow^{\ast}   \mathsf{true} $ (the assumption in the
 left-to-right direction), we have
 \[
   \theta_{{\mathrm{1}}}  (   \delta_{{\mathrm{1}}}  (  \ottnt{e''} \, [  \ottnt{e'}  \ottsym{/}  \mathit{x}  ]  )   )  \, [  \ottnt{v_{{\mathrm{1}}}}  \ottsym{/}  \mathit{y}  ]  \longrightarrow^{\ast}   \mathsf{true} 
 \]
 by Cotermination (\prop:ref{fh-coterm-true}).
 The remaining obligation is
 \[
   \theta_{{\mathrm{2}}}  (   \delta_{{\mathrm{2}}}  (  \ottnt{e''} \, [  \ottnt{e'}  \ottsym{/}  \mathit{x}  ]  )   )  \, [  \ottnt{v_{{\mathrm{2}}}}  \ottsym{/}  \mathit{y}  ]  \longrightarrow^{\ast}   \mathsf{true} .
 \]
 Since $\ottnt{v''_{{\mathrm{2}}}}$ is the evaluation result of $ \theta_{{\mathrm{2}}}  (   \delta_{{\mathrm{2}}}  (  \ottnt{e'}  )   ) $, it
 suffices to show that, by Cotermination,
 \[
   \theta_{{\mathrm{2}}}  (   \delta_{{\mathrm{2}}}  (  \ottnt{e''} \, [  \ottnt{v''_{{\mathrm{2}}}}  \ottsym{/}  \mathit{x}  ]  )   )  \, [  \ottnt{v_{{\mathrm{2}}}}  \ottsym{/}  \mathit{y}  ]  \longrightarrow^{\ast}   \mathsf{true} .
 \]
 Since $ \Gamma  ,  \mathit{x}  \mathord{:}  \ottnt{T'}   \ottsym{,}  \Gamma'  \vdash  \theta  \ottsym{;}   \delta    [  \,  (  \ottnt{v'_{{\mathrm{1}}}}  ,  \ottnt{v'_{{\mathrm{2}}}}  ) /  \mathit{x}  \,  ]  $ and
 $ \ottnt{v_{{\mathrm{1}}}}  \simeq_{\mathtt{v} }  \ottnt{v_{{\mathrm{2}}}}   \ottsym{:}   \ottnt{T''} ;  \theta ;   \delta    [  \,  (  \ottnt{v'_{{\mathrm{1}}}}  ,  \ottnt{v'_{{\mathrm{2}}}}  ) /  \mathit{x}  \,  ]   $,
 we have
 \begin{equation}
    \Gamma  ,  \mathit{x}  \mathord{:}  \ottnt{T'}   \ottsym{,}  \Gamma'  ,  \mathit{y}  \mathord{:}  \ottnt{T''}   \vdash  \theta  \ottsym{;}    \delta    [  \,  (  \ottnt{v'_{{\mathrm{1}}}}  ,  \ottnt{v'_{{\mathrm{2}}}}  ) /  \mathit{x}  \,  ]      [  \,  (  \ottnt{v_{{\mathrm{1}}}}  ,  \ottnt{v_{{\mathrm{2}}}}  ) /  \mathit{y}  \,  ]  
   \label{eqn:fh-lr-term-comp-one}
 \end{equation}
 by the weakening (\prop:ref{fh-lr-val-ws}).
 Since $ \ottnt{v'_{{\mathrm{1}}}}  \simeq_{\mathtt{v} }  \ottnt{v''_{{\mathrm{2}}}}   \ottsym{:}   \ottnt{T'} ;  \theta ;  \delta $,
 we have
 \begin{equation}
   \ottnt{v'_{{\mathrm{1}}}}  \simeq_{\mathtt{v} }  \ottnt{v''_{{\mathrm{2}}}}   \ottsym{:}   \ottnt{T'} ;  \theta ;    \delta    [  \,  (  \ottnt{v'_{{\mathrm{1}}}}  ,  \ottnt{v'_{{\mathrm{2}}}}  ) /  \mathit{x}  \,  ]      [  \,  (  \ottnt{v_{{\mathrm{1}}}}  ,  \ottnt{v_{{\mathrm{2}}}}  ) /  \mathit{y}  \,  ]   
   \label{eqn:fh-lr-term-comp-two}
 \end{equation}
 by the weakening.
 By applying \prop:ref{fh-lr-val-exchange-wf} to
 (\ref{eqn:fh-lr-term-comp-one}) and (\ref{eqn:fh-lr-term-comp-two}),
 we have
 \[
    \Gamma  ,  \mathit{x}  \mathord{:}  \ottnt{T'}   \ottsym{,}  \Gamma'  ,  \mathit{y}  \mathord{:}  \ottnt{T''}   \vdash  \theta  \ottsym{;}    \delta    [  \,  (  \ottnt{v'_{{\mathrm{1}}}}  ,  \ottnt{v''_{{\mathrm{2}}}}  ) /  \mathit{x}  \,  ]      [  \,  (  \ottnt{v_{{\mathrm{1}}}}  ,  \ottnt{v_{{\mathrm{2}}}}  ) /  \mathit{y}  \,  ]  .
 \]
 Since $  \Gamma  ,  \mathit{x}  \mathord{:}  \ottnt{T'}   \ottsym{,}  \Gamma'  ,  \mathit{y}  \mathord{:}  \ottnt{T''}   \vdash  \ottnt{e''} \,  \mathrel{ \simeq }  \, \ottnt{e''}  \ottsym{:}   \mathsf{Bool} $
 from $  \Gamma  ,  \mathit{x}  \mathord{:}  \ottnt{T'}   \ottsym{,}  \Gamma' \vdash  \{  \mathit{y}  \mathord{:}  \ottnt{T''}   \mathop{\mid}   \ottnt{e''}  \}   \mathrel{ \simeq }   \{  \mathit{y}  \mathord{:}  \ottnt{T''}   \mathop{\mid}   \ottnt{e''}  \}   : \ast $,
 we have
 \[
    \theta_{{\mathrm{1}}}  (   \delta_{{\mathrm{1}}}  (  \ottnt{e''} \, [  \ottnt{v'_{{\mathrm{1}}}}  \ottsym{/}  \mathit{x}  ]  )   )  \, [  \ottnt{v_{{\mathrm{1}}}}  \ottsym{/}  \mathit{y}  ]  \simeq_{\mathtt{e} }   \theta_{{\mathrm{2}}}  (   \delta_{{\mathrm{2}}}  (  \ottnt{e''} \, [  \ottnt{v''_{{\mathrm{2}}}}  \ottsym{/}  \mathit{x}  ]  )   )  \, [  \ottnt{v_{{\mathrm{2}}}}  \ottsym{/}  \mathit{y}  ]   \ottsym{:}    \mathsf{Bool}  ;  \theta ;    \delta    [  \,  (  \ottnt{v'_{{\mathrm{1}}}}  ,  \ottnt{v''_{{\mathrm{2}}}}  ) /  \mathit{x}  \,  ]      [  \,  (  \ottnt{v_{{\mathrm{1}}}}  ,  \ottnt{v_{{\mathrm{2}}}}  ) /  \mathit{y}  \,  ]   .
 \]
 Since the term on the left-hand side evaluates to $ \mathsf{true} $ (the assumption in
 the left-to-right direction), the one on the right-hand side also evaluates to
 $ \mathsf{true} $ by definition.
 Hence, we finish.

 The other direction is shown in a similar way except the case of
 \[
   \theta_{{\mathrm{2}}}  (   \delta_{{\mathrm{2}}}  (  \ottnt{e''} \, [  \ottnt{v'_{{\mathrm{2}}}}  \ottsym{/}  \mathit{x}  ]  )   )  \, [  \ottnt{v_{{\mathrm{2}}}}  \ottsym{/}  \mathit{y}  ]  \longrightarrow^{\ast}   \mathsf{true} .
 \]
 This case is shown as follows.
 From (\ref{eqn:fh-lr-term-comp-one}), which can be shown also in the
 right-to-left direction with the IH, and
 $  \Gamma  ,  \mathit{x}  \mathord{:}  \ottnt{T'}   \ottsym{,}  \Gamma'  ,  \mathit{y}  \mathord{:}  \ottnt{T''}   \vdash  \ottnt{e''} \,  \mathrel{ \simeq }  \, \ottnt{e''}  \ottsym{:}   \mathsf{Bool} $, it is found that
 \[
    \theta_{{\mathrm{1}}}  (   \delta_{{\mathrm{1}}}  (  \ottnt{e''} \, [  \ottnt{v'_{{\mathrm{1}}}}  \ottsym{/}  \mathit{x}  ]  )   )  \, [  \ottnt{v_{{\mathrm{1}}}}  \ottsym{/}  \mathit{y}  ]  \simeq_{\mathtt{e} }   \theta_{{\mathrm{2}}}  (   \delta_{{\mathrm{2}}}  (  \ottnt{e''} \, [  \ottnt{v'_{{\mathrm{2}}}}  \ottsym{/}  \mathit{x}  ]  )   )  \, [  \ottnt{v_{{\mathrm{2}}}}  \ottsym{/}  \mathit{y}  ]   \ottsym{:}    \mathsf{Bool}  ;  \theta ;    \delta    [  \,  (  \ottnt{v'_{{\mathrm{1}}}}  ,  \ottnt{v'_{{\mathrm{2}}}}  ) /  \mathit{x}  \,  ]      [  \,  (  \ottnt{v_{{\mathrm{1}}}}  ,  \ottnt{v_{{\mathrm{2}}}}  ) /  \mathit{y}  \,  ]   .
 \]
 Since it is found that the term on the left-hand side evaluates to $ \mathsf{true} $
 by applying Cotermination to $ \theta_{{\mathrm{1}}}  (   \delta_{{\mathrm{1}}}  (  \ottnt{e''} \, [  \ottnt{e'}  \ottsym{/}  \mathit{x}  ]  )   )  \, [  \ottnt{v_{{\mathrm{1}}}}  \ottsym{/}  \mathit{y}  ]  \longrightarrow^{\ast}   \mathsf{true} $,
 so does the one on the right-hand side, which we want to show.
\end{prop}

\begin{prop}
 [name=Compatibility under Self-relatedness Assumption: Application]
 {fh-lr-comp-app-refl-assump}
 \label{lem:fh-lr-app-end}
 Suppose that $\Gamma$ is self-related and
 $\Gamma  \vdash  \ottnt{e_{{\mathrm{12}}}} \,  \mathrel{ \simeq }  \, \ottnt{e_{{\mathrm{12}}}}  \ottsym{:}  \ottnt{T_{{\mathrm{1}}}}$ and
 $ \Gamma \vdash  \mathit{x} \mathord{:} \ottnt{T_{{\mathrm{1}}}} \rightarrow \ottnt{T_{{\mathrm{2}}}}   \mathrel{ \simeq }   \mathit{x} \mathord{:} \ottnt{T_{{\mathrm{1}}}} \rightarrow \ottnt{T_{{\mathrm{2}}}}   : \ast $.
 If
 $\Gamma  \vdash  \ottnt{e_{{\mathrm{11}}}} \,  \mathrel{ \simeq }  \, \ottnt{e_{{\mathrm{21}}}}  \ottsym{:}  \ottsym{(}   \mathit{x} \mathord{:} \ottnt{T_{{\mathrm{1}}}} \rightarrow \ottnt{T_{{\mathrm{2}}}}   \ottsym{)}$ and
 $\Gamma  \vdash  \ottnt{e_{{\mathrm{12}}}} \,  \mathrel{ \simeq }  \, \ottnt{e_{{\mathrm{22}}}}  \ottsym{:}  \ottnt{T_{{\mathrm{1}}}}$,
 then
 $\Gamma  \vdash  \ottnt{e_{{\mathrm{11}}}} \, \ottnt{e_{{\mathrm{12}}}} \,  \mathrel{ \simeq }  \, \ottnt{e_{{\mathrm{21}}}} \, \ottnt{e_{{\mathrm{22}}}}  \ottsym{:}  \ottnt{T_{{\mathrm{2}}}} \, [  \ottnt{e_{{\mathrm{12}}}}  \ottsym{/}  \mathit{x}  ]$.

 \proof

 Suppose that $\Gamma  \vdash  \theta  \ottsym{;}  \delta$.
 Let
 $\ottnt{e'_{{\mathrm{11}}}}  \ottsym{=}   \theta_{{\mathrm{1}}}  (   \delta_{{\mathrm{1}}}  (  \ottnt{e_{{\mathrm{11}}}}  )   ) $,
 $\ottnt{e'_{{\mathrm{12}}}}  \ottsym{=}   \theta_{{\mathrm{1}}}  (   \delta_{{\mathrm{1}}}  (  \ottnt{e_{{\mathrm{12}}}}  )   ) $,
 $\ottnt{e'_{{\mathrm{21}}}}  \ottsym{=}   \theta_{{\mathrm{2}}}  (   \delta_{{\mathrm{2}}}  (  \ottnt{e_{{\mathrm{21}}}}  )   ) $, and
 $\ottnt{e'_{{\mathrm{22}}}}  \ottsym{=}   \theta_{{\mathrm{2}}}  (   \delta_{{\mathrm{2}}}  (  \ottnt{e_{{\mathrm{22}}}}  )   ) $.
 It suffices to show that
 \[  \ottnt{e'_{{\mathrm{11}}}} \, \ottnt{e'_{{\mathrm{12}}}}  \simeq_{\mathtt{e} }  \ottnt{e'_{{\mathrm{21}}}} \, \ottnt{e'_{{\mathrm{22}}}}   \ottsym{:}   \ottnt{T_{{\mathrm{2}}}} \, [  \ottnt{e_{{\mathrm{12}}}}  \ottsym{/}  \mathit{x}  ] ;  \theta ;  \delta . \]
 If both $\ottnt{e'_{{\mathrm{11}}}}$ and $\ottnt{e'_{{\mathrm{21}}}}$ or both $\ottnt{e'_{{\mathrm{12}}}}$ and $\ottnt{e'_{{\mathrm{22}}}}$ raise
 blame, the conclusion is obvious.
 Otherwise, we can suppose that
 $\ottnt{e'_{{\mathrm{11}}}}  \longrightarrow^{\ast}  \ottnt{v_{{\mathrm{11}}}}$ and
 $\ottnt{e'_{{\mathrm{12}}}}  \longrightarrow^{\ast}  \ottnt{v_{{\mathrm{12}}}}$ and
 $\ottnt{e'_{{\mathrm{21}}}}  \longrightarrow^{\ast}  \ottnt{v_{{\mathrm{21}}}}$ and
 $\ottnt{e'_{{\mathrm{22}}}}  \longrightarrow^{\ast}  \ottnt{v_{{\mathrm{22}}}}$ for some $\ottnt{v_{{\mathrm{11}}}}$, $\ottnt{v_{{\mathrm{12}}}}$, $\ottnt{v_{{\mathrm{21}}}}$, and $\ottnt{v_{{\mathrm{22}}}}$,
 and it suffices to show that
 \[  \ottnt{v_{{\mathrm{11}}}} \, \ottnt{v_{{\mathrm{12}}}}  \simeq_{\mathtt{e} }  \ottnt{v_{{\mathrm{21}}}} \, \ottnt{v_{{\mathrm{22}}}}   \ottsym{:}   \ottnt{T_{{\mathrm{2}}}} \, [  \ottnt{e_{{\mathrm{12}}}}  \ottsym{/}  \mathit{x}  ] ;  \theta ;  \delta . \]
 Since $ \ottnt{e'_{{\mathrm{11}}}}  \simeq_{\mathtt{e} }  \ottnt{e'_{{\mathrm{21}}}}   \ottsym{:}    \mathit{x} \mathord{:} \ottnt{T_{{\mathrm{1}}}} \rightarrow \ottnt{T_{{\mathrm{2}}}}  ;  \theta ;  \delta $
 and $ \ottnt{e'_{{\mathrm{12}}}}  \simeq_{\mathtt{e} }  \ottnt{e'_{{\mathrm{22}}}}   \ottsym{:}   \ottnt{T_{{\mathrm{1}}}} ;  \theta ;  \delta $,
 we have
 $ \ottnt{v_{{\mathrm{11}}}}  \simeq_{\mathtt{v} }  \ottnt{v_{{\mathrm{21}}}}   \ottsym{:}    \mathit{x} \mathord{:} \ottnt{T_{{\mathrm{1}}}} \rightarrow \ottnt{T_{{\mathrm{2}}}}  ;  \theta ;  \delta $ and
 $ \ottnt{v_{{\mathrm{12}}}}  \simeq_{\mathtt{v} }  \ottnt{v_{{\mathrm{22}}}}   \ottsym{:}   \ottnt{T_{{\mathrm{1}}}} ;  \theta ;  \delta $.
 Thus, $ \ottnt{v_{{\mathrm{11}}}} \, \ottnt{v_{{\mathrm{12}}}}  \simeq_{\mathtt{e} }  \ottnt{v_{{\mathrm{21}}}} \, \ottnt{v_{{\mathrm{22}}}}   \ottsym{:}   \ottnt{T_{{\mathrm{2}}}} ;  \theta ;   \delta    [  \,  (  \ottnt{v_{{\mathrm{12}}}}  ,  \ottnt{v_{{\mathrm{22}}}}  ) /  \mathit{x}  \,  ]   $
 by definition.
 Since $ \Gamma \vdash  \mathit{x} \mathord{:} \ottnt{T_{{\mathrm{1}}}} \rightarrow \ottnt{T_{{\mathrm{2}}}}   \mathrel{ \simeq }   \mathit{x} \mathord{:} \ottnt{T_{{\mathrm{1}}}} \rightarrow \ottnt{T_{{\mathrm{2}}}}   : \ast $,
 we have $ \Gamma \vdash \ottnt{T_{{\mathrm{1}}}}  \mathrel{ \simeq }  \ottnt{T_{{\mathrm{1}}}}  : \ast $ and $  \Gamma  ,  \mathit{x}  \mathord{:}  \ottnt{T_{{\mathrm{1}}}}  \vdash \ottnt{T_{{\mathrm{2}}}}  \mathrel{ \simeq }  \ottnt{T_{{\mathrm{2}}}}  : \ast $.
 Since $\Gamma$ is self-related, so is $ \Gamma  ,  \mathit{x}  \mathord{:}  \ottnt{T_{{\mathrm{1}}}} $.
 Since $\Gamma  \vdash  \theta  \ottsym{;}  \delta$,
 we have $ \Gamma  ,  \mathit{x}  \mathord{:}  \ottnt{T_{{\mathrm{1}}}}   \vdash  \theta  \ottsym{;}   \delta    [  \,  (  \ottnt{v_{{\mathrm{12}}}}  ,  \ottnt{v_{{\mathrm{22}}}}  ) /  \mathit{x}  \,  ]  $
 by the weakening (\prop:ref{fh-lr-val-ws}).
 We have $ \theta_{{\mathrm{1}}}  (   \delta_{{\mathrm{1}}}  (  \ottnt{e_{{\mathrm{12}}}}  )   )   \longrightarrow^{\ast}  \ottnt{v_{{\mathrm{12}}}}$.
 Since $ \theta_{{\mathrm{1}}}  (   \delta_{{\mathrm{1}}}  (  \ottnt{e_{{\mathrm{12}}}}  )   )  = \ottnt{e'_{{\mathrm{12}}}}  \longrightarrow^{\ast}  \ottnt{v_{{\mathrm{12}}}}$ and
 $\Gamma  \vdash  \ottnt{e_{{\mathrm{12}}}} \,  \mathrel{ \simeq }  \, \ottnt{e_{{\mathrm{12}}}}  \ottsym{:}  \ottnt{T_{{\mathrm{1}}}}$,
 we have $ \theta_{{\mathrm{2}}}  (   \delta_{{\mathrm{2}}}  (  \ottnt{e_{{\mathrm{12}}}}  )   )   \longrightarrow^{\ast}  \ottnt{v'_{{\mathrm{12}}}}$ and
 $ \ottnt{v_{{\mathrm{12}}}}  \simeq_{\mathtt{v} }  \ottnt{v'_{{\mathrm{12}}}}   \ottsym{:}   \ottnt{T_{{\mathrm{1}}}} ;  \theta ;  \delta $ for some $\ottnt{v'_{{\mathrm{12}}}}$.
 Thus, by the term compositionality (\prop:ref{fh-lr-term-comp}),
 we finish.
\end{prop}

\paragraph{\bf Fundamental property: type application}
\label{sec:proving-type-app}
We show that the logical relation is closed under type applications, that is,
if
$ \ottnt{v_{{\mathrm{1}}}}  \simeq_{\mathtt{v} }  \ottnt{v_{{\mathrm{2}}}}   \ottsym{:}    \forall   \alpha  .  \ottnt{T}  ;  \theta ;  \delta $ and
$ \ottnt{T_{{\mathrm{1}}}}  \simeq  \ottnt{T_{{\mathrm{2}}}}   \ottsym{:}   \ast ;  \theta ;  \delta $,
then $ \ottnt{v_{{\mathrm{1}}}} \, \ottnt{T_{{\mathrm{1}}}}  \simeq_{\mathtt{e} }  \ottnt{v_{{\mathrm{2}}}} \, \ottnt{T_{{\mathrm{2}}}}   \ottsym{:}   \ottnt{T} \, [  \ottnt{T_{{\mathrm{1}}}}  \ottsym{/}  \alpha  ] ;  \theta ;  \delta $.
To this end, for a reason similar to the case of term applications, we show the
\emph{type compositionality}, which states that the term relation indexed by
$\ottnt{T} \, [  \ottnt{T_{{\mathrm{1}}}}  \ottsym{/}  \alpha  ]$ with $\theta$ coincides with the one indexed by $\ottnt{T}$ with
$\theta \,  \{  \,  \alpha  \mapsto ( \ottnt{r} , \ottnt{T_{{\mathrm{1}}}} , \ottnt{T_{{\mathrm{2}}}} )  \,  \} $ for some $\ottnt{r}$.
Since $\ottnt{r}$ gives an interpretation of $\alpha$ and $\alpha$ is replaced with
$\ottnt{T_{{\mathrm{1}}}}$ in the former, it is natural to choose the term relation $ \ottnt{e_{{\mathrm{1}}}}  \simeq_{\mathtt{e} }  \ottnt{e_{{\mathrm{2}}}}   \ottsym{:}   \ottnt{T_{{\mathrm{1}}}} ;  \theta ;  \delta $ indexed by
$\ottnt{T_{{\mathrm{1}}}}$ as $\ottnt{r}$.
We first show that the term relation satisfies requirements to interpretations (Lemmas~\ref{lem:fh-lr-tapp-start}--\ref{lem:fh-lr-well-formed-interpret}) and
then the type compositionality (\prop:ref{fh-lr-typ-comp}).

The first requirement which we show that term relations satisfy is that,
if $ \theta  (  \alpha  ) = (  \ottnt{r} ,  \ottnt{T_{{\mathrm{1}}}} ,  \ottnt{T_{{\mathrm{2}}}}  ) $ and $\ottsym{(}  \ottnt{v_{{\mathrm{1}}}}  \ottsym{,}  \ottnt{v_{{\mathrm{2}}}}  \ottsym{)} \, \in \, \ottnt{r}$,
then there exists some $\ottnt{v'_{{\mathrm{1}}}}$ such that $\langle  \ottnt{T_{{\mathrm{1}}}}  \Rightarrow  \ottnt{T_{{\mathrm{1}}}}  \rangle   ^{ \ell }  \, \ottnt{v_{{\mathrm{1}}}}  \longrightarrow^{\ast}  \ottnt{v'_{{\mathrm{1}}}}$ and
$\ottsym{(}  \ottnt{v'_{{\mathrm{1}}}}  \ottsym{,}  \ottnt{v_{{\mathrm{2}}}}  \ottsym{)} \, \in \, \ottnt{r}$.
This is generalized to elimination of reflexive casts.
{}
\begin{prop}[name=Elimination of Reflexive Casts on Left]{fh-lr-elim-refl-cast}
 {\iffull\else\label{lem:fh-lr-tapp-start}\fi}
 If
 $ \ottnt{T_{{\mathrm{1}}}}  \simeq  \ottnt{T_{{\mathrm{1}}}}   \ottsym{:}   \ast ;  \theta ;  \delta $ and
 $ \ottnt{T_{{\mathrm{2}}}}  \simeq  \ottnt{T_{{\mathrm{2}}}}   \ottsym{:}   \ast ;  \theta ;  \delta $ and
 $ \ottnt{T_{{\mathrm{1}}}}  \simeq  \ottnt{T_{{\mathrm{2}}}}   \ottsym{:}   \ast ;  \theta ;  \delta $ and
 $ \ottnt{T_{{\mathrm{2}}}}  \simeq  \ottnt{T_{{\mathrm{1}}}}   \ottsym{:}   \ast ;  \theta ;  \delta $,
 then
 $  \theta_{{\mathrm{1}}}  (   \delta_{{\mathrm{1}}}  (  \langle  \ottnt{T_{{\mathrm{1}}}}  \Rightarrow  \ottnt{T_{{\mathrm{2}}}}  \rangle   ^{ \ell }   )   )   \simeq_{\mathtt{v} }   \theta_{{\mathrm{2}}}  (   \delta_{{\mathrm{2}}}  (    \lambda    \mathit{x}  \mathord{:}  \ottnt{T_{{\mathrm{1}}}}  .  \mathit{x}   )   )    \ottsym{:}   \ottnt{T_{{\mathrm{1}}}}  \rightarrow  \ottnt{T_{{\mathrm{2}}}} ;  \theta ;  \delta $.

 \proof

 By course-of-values induction on the sum of sizes of $\ottnt{T_{{\mathrm{1}}}}$ and $\ottnt{T_{{\mathrm{2}}}}$.
 By definition, it suffices to show that, for any $\ottnt{v_{{\mathrm{1}}}}$ and $\ottnt{v_{{\mathrm{2}}}}$ such that
 $ \ottnt{v_{{\mathrm{1}}}}  \simeq_{\mathtt{v} }  \ottnt{v_{{\mathrm{2}}}}   \ottsym{:}   \ottnt{T_{{\mathrm{1}}}} ;  \theta ;  \delta $,
 \[
    \theta_{{\mathrm{1}}}  (   \delta_{{\mathrm{1}}}  (  \langle  \ottnt{T_{{\mathrm{1}}}}  \Rightarrow  \ottnt{T_{{\mathrm{2}}}}  \rangle   ^{ \ell }   )   )  \, \ottnt{v_{{\mathrm{1}}}}  \simeq_{\mathtt{e} }  \ottnt{v_{{\mathrm{2}}}}   \ottsym{:}   \ottnt{T_{{\mathrm{2}}}} ;  \theta ;  \delta .
 \]
 By case analysis on the derivation of $ \ottnt{T_{{\mathrm{1}}}}  \simeq  \ottnt{T_{{\mathrm{2}}}}   \ottsym{:}   \ast ;  \theta ;  \delta $.
 \begin{itemize}
  \case $ \alpha  \simeq  \alpha   \ottsym{:}   \ast ;  \theta ;  \delta $:
   Since $ \ottnt{v_{{\mathrm{1}}}}  \simeq_{\mathtt{v} }  \ottnt{v_{{\mathrm{2}}}}   \ottsym{:}   \alpha ;  \theta ;  \delta $,
   there exists some $\ottnt{r'}$, $\ottnt{T'_{{\mathrm{1}}}}$, and $\ottnt{T'_{{\mathrm{2}}}}$ such that
   $ \theta  (  \alpha  ) = (  \ottnt{r'} ,  \ottnt{T'_{{\mathrm{1}}}} ,  \ottnt{T'_{{\mathrm{2}}}}  ) $ and $\langle  \ottnt{r'}  \ottsym{,}  \ottnt{T'_{{\mathrm{1}}}}  \ottsym{,}  \ottnt{T'_{{\mathrm{2}}}}  \rangle$ and $\ottsym{(}  \ottnt{v_{{\mathrm{1}}}}  \ottsym{,}  \ottnt{v_{{\mathrm{2}}}}  \ottsym{)} \, \in \, \ottnt{r'}$.
   Since $ \theta_{{\mathrm{1}}}  (   \delta_{{\mathrm{1}}}  (  \langle  \ottnt{T_{{\mathrm{1}}}}  \Rightarrow  \ottnt{T_{{\mathrm{2}}}}  \rangle   ^{ \ell }   )   )   \ottsym{=}  \langle  \ottnt{T'_{{\mathrm{1}}}}  \Rightarrow  \ottnt{T'_{{\mathrm{1}}}}  \rangle   ^{ \ell } $ and
   $\ottnt{r} \, \in \,  \mathsf{VRel}  (  \ottnt{T'_{{\mathrm{1}}}} ,  \ottnt{T'_{{\mathrm{2}}}}  ) $,
   there exists some $\ottnt{v'_{{\mathrm{1}}}}$ such that
   $ \theta_{{\mathrm{1}}}  (   \delta_{{\mathrm{1}}}  (  \langle  \ottnt{T_{{\mathrm{1}}}}  \Rightarrow  \ottnt{T_{{\mathrm{2}}}}  \rangle   ^{ \ell }   )   )  \, \ottnt{v_{{\mathrm{1}}}}  \longrightarrow^{\ast}  \ottnt{v'_{{\mathrm{1}}}}$ and
   $\ottsym{(}  \ottnt{v'_{{\mathrm{1}}}}  \ottsym{,}  \ottnt{v_{{\mathrm{2}}}}  \ottsym{)} \, \in \, \ottnt{r'}$.
   We have $ \ottnt{v'_{{\mathrm{1}}}}  \simeq_{\mathtt{v} }  \ottnt{v_{{\mathrm{2}}}}   \ottsym{:}   \alpha ;  \theta ;  \delta $, and so we finish.

  \case $ \ottnt{B}  \simeq  \ottnt{B}   \ottsym{:}   \ast ;  \theta ;  \delta $: Obvious since $\ottnt{T_{{\mathrm{1}}}} = \ottnt{T_{{\mathrm{2}}}} = \ottnt{B}$.

  \case $  \mathit{x} \mathord{:} \ottnt{T_{{\mathrm{11}}}} \rightarrow \ottnt{T_{{\mathrm{12}}}}   \simeq   \mathit{x} \mathord{:} \ottnt{T_{{\mathrm{21}}}} \rightarrow \ottnt{T_{{\mathrm{22}}}}    \ottsym{:}   \ast ;  \theta ;  \delta $:
   Without loss of generality, we can suppose that $\mathit{x} \, \notin \,  \mathit{dom}  (  \delta  ) $.
   By \E{Red}/\R{Fun},
   \[\begin{array}{l}
     \theta_{{\mathrm{1}}}  (   \delta_{{\mathrm{1}}}  (  \langle  \ottnt{T_{{\mathrm{1}}}}  \Rightarrow  \ottnt{T_{{\mathrm{2}}}}  \rangle   ^{ \ell }   )   )  \, \ottnt{v_{{\mathrm{1}}}}  \longrightarrow  \\
      \qquad  \theta_{{\mathrm{1}}}  (   \delta_{{\mathrm{1}}}  (    \lambda    \mathit{x}  \mathord{:}  \ottnt{T_{{\mathrm{21}}}}  .   \mathsf{let}  ~  \mathit{y}  \mathord{:}  \ottnt{T_{{\mathrm{11}}}}  \equal  \langle  \ottnt{T_{{\mathrm{21}}}}  \Rightarrow  \ottnt{T_{{\mathrm{11}}}}  \rangle   ^{ \ell }  \, \mathit{x}  ~ \ottliteralin ~  \langle  \ottnt{T_{{\mathrm{12}}}} \, [  \mathit{y}  \ottsym{/}  \mathit{x}  ]  \Rightarrow  \ottnt{T_{{\mathrm{22}}}}  \rangle   ^{ \ell }   \, \ottsym{(}  \ottnt{v_{{\mathrm{1}}}} \, \mathit{y}  \ottsym{)}   )   ) 
     \end{array}\]
   for some fresh variable $\mathit{y}$.
   It thus suffices to show that
   \[\begin{array}{ll}
    &   \theta_{{\mathrm{1}}}  (   \delta_{{\mathrm{1}}}  (    \lambda    \mathit{x}  \mathord{:}  \ottnt{T_{{\mathrm{21}}}}  .   \mathsf{let}  ~  \mathit{y}  \mathord{:}  \ottnt{T_{{\mathrm{11}}}}  \equal  \langle  \ottnt{T_{{\mathrm{21}}}}  \Rightarrow  \ottnt{T_{{\mathrm{11}}}}  \rangle   ^{ \ell }  \, \mathit{x}  ~ \ottliteralin ~  \langle  \ottnt{T_{{\mathrm{12}}}} \, [  \mathit{y}  \ottsym{/}  \mathit{x}  ]  \Rightarrow  \ottnt{T_{{\mathrm{22}}}}  \rangle   ^{ \ell }   \, \ottsym{(}  \ottnt{v_{{\mathrm{1}}}} \, \mathit{y}  \ottsym{)}   )   )     \\   \simeq_{\mathtt{v} }   &    \ottnt{v_{{\mathrm{2}}}}   \ottsym{:}    \mathit{x} \mathord{:} \ottnt{T_{{\mathrm{21}}}} \rightarrow \ottnt{T_{{\mathrm{22}}}}  ;  \theta ;  \delta .
   \end{array}\]
   By definition, for any $\ottnt{v'_{{\mathrm{1}}}}$ and $\ottnt{v'_{{\mathrm{2}}}}$ such that
   $ \ottnt{v'_{{\mathrm{1}}}}  \simeq_{\mathtt{v} }  \ottnt{v'_{{\mathrm{2}}}}   \ottsym{:}   \ottnt{T_{{\mathrm{21}}}} ;  \theta ;  \delta $, we have to show that
   \[\begin{array}{ll}
    &   \theta_{{\mathrm{1}}}  (   \delta_{{\mathrm{1}}}  (   \mathsf{let}  ~  \mathit{y}  \mathord{:}  \ottnt{T_{{\mathrm{11}}}}  \equal  \langle  \ottnt{T_{{\mathrm{21}}}}  \Rightarrow  \ottnt{T_{{\mathrm{11}}}}  \rangle   ^{ \ell }  \, \ottnt{v'_{{\mathrm{1}}}}  ~ \ottliteralin ~  \langle  \ottnt{T_{{\mathrm{12}}}} \, [  \mathit{y}  \ottsym{/}  \mathit{x}  ]  \Rightarrow  \ottnt{T_{{\mathrm{22}}}} \, [  \ottnt{v'_{{\mathrm{1}}}}  \ottsym{/}  \mathit{x}  ]  \rangle   ^{ \ell }  \, \ottsym{(}  \ottnt{v_{{\mathrm{1}}}} \, \mathit{y}  \ottsym{)}   )   )     \\   \simeq_{\mathtt{e} }   &    \ottnt{v_{{\mathrm{2}}}} \, \ottnt{v'_{{\mathrm{2}}}}   \ottsym{:}   \ottnt{T_{{\mathrm{22}}}} ;  \theta ;   \delta    [  \,  (  \ottnt{v'_{{\mathrm{1}}}}  ,  \ottnt{v'_{{\mathrm{2}}}}  ) /  \mathit{x}  \,  ]   .
     \end{array}\]
   By the IH,
   $  \theta_{{\mathrm{1}}}  (   \delta_{{\mathrm{1}}}  (  \langle  \ottnt{T_{{\mathrm{21}}}}  \Rightarrow  \ottnt{T_{{\mathrm{11}}}}  \rangle   ^{ \ell }   )   )   \simeq_{\mathtt{v} }   \theta_{{\mathrm{2}}}  (   \delta_{{\mathrm{2}}}  (    \lambda    \mathit{x}  \mathord{:}  \ottnt{T_{{\mathrm{21}}}}  .  \mathit{x}   )   )    \ottsym{:}   \ottnt{T_{{\mathrm{21}}}}  \rightarrow  \ottnt{T_{{\mathrm{11}}}} ;  \theta ;  \delta $.
   Since $ \ottnt{v'_{{\mathrm{1}}}}  \simeq_{\mathtt{v} }  \ottnt{v'_{{\mathrm{2}}}}   \ottsym{:}   \ottnt{T_{{\mathrm{21}}}} ;  \theta ;  \delta $,
   we have $  \theta_{{\mathrm{1}}}  (   \delta_{{\mathrm{1}}}  (  \langle  \ottnt{T_{{\mathrm{21}}}}  \Rightarrow  \ottnt{T_{{\mathrm{11}}}}  \rangle   ^{ \ell }   )   )  \, \ottnt{v'_{{\mathrm{1}}}}  \simeq_{\mathtt{e} }  \ottnt{v'_{{\mathrm{2}}}}   \ottsym{:}   \ottnt{T_{{\mathrm{11}}}} ;  \theta ;  \delta $.
   Thus, there exists some $\ottnt{v''_{{\mathrm{1}}}}$ such that
   $ \theta_{{\mathrm{1}}}  (   \delta_{{\mathrm{1}}}  (  \langle  \ottnt{T_{{\mathrm{21}}}}  \Rightarrow  \ottnt{T_{{\mathrm{11}}}}  \rangle   ^{ \ell }   )   )  \, \ottnt{v'_{{\mathrm{1}}}}  \longrightarrow^{\ast}  \ottnt{v''_{{\mathrm{1}}}}$ and
   $ \ottnt{v''_{{\mathrm{1}}}}  \simeq_{\mathtt{v} }  \ottnt{v'_{{\mathrm{2}}}}   \ottsym{:}   \ottnt{T_{{\mathrm{11}}}} ;  \theta ;  \delta $.
   Hence, it suffices to show that
   \[
      \theta_{{\mathrm{1}}}  (   \delta_{{\mathrm{1}}}  (  \langle  \ottnt{T_{{\mathrm{12}}}} \, [  \ottnt{v''_{{\mathrm{1}}}}  \ottsym{/}  \mathit{x}  ]  \Rightarrow  \ottnt{T_{{\mathrm{22}}}} \, [  \ottnt{v'_{{\mathrm{1}}}}  \ottsym{/}  \mathit{x}  ]  \rangle   ^{ \ell }   )   )  \, \ottsym{(}  \ottnt{v_{{\mathrm{1}}}} \, \ottnt{v''_{{\mathrm{1}}}}  \ottsym{)}  \simeq_{\mathtt{e} }  \ottnt{v_{{\mathrm{2}}}} \, \ottnt{v'_{{\mathrm{2}}}}   \ottsym{:}   \ottnt{T_{{\mathrm{22}}}} ;  \theta ;   \delta    [  \,  (  \ottnt{v'_{{\mathrm{1}}}}  ,  \ottnt{v'_{{\mathrm{2}}}}  ) /  \mathit{x}  \,  ]   .
   \]
   Since $ \ottnt{v_{{\mathrm{1}}}}  \simeq_{\mathtt{v} }  \ottnt{v_{{\mathrm{2}}}}   \ottsym{:}    \mathit{x} \mathord{:} \ottnt{T_{{\mathrm{11}}}} \rightarrow \ottnt{T_{{\mathrm{12}}}}  ;  \theta ;  \delta $,
   we have $ \ottnt{v_{{\mathrm{1}}}} \, \ottnt{v''_{{\mathrm{1}}}}  \simeq_{\mathtt{e} }  \ottnt{v_{{\mathrm{2}}}} \, \ottnt{v'_{{\mathrm{2}}}}   \ottsym{:}   \ottnt{T_{{\mathrm{12}}}} ;  \theta ;   \delta    [  \,  (  \ottnt{v''_{{\mathrm{1}}}}  ,  \ottnt{v'_{{\mathrm{2}}}}  ) /  \mathit{x}  \,  ]   $.
   If $\ottnt{v_{{\mathrm{1}}}} \, \ottnt{v''_{{\mathrm{1}}}}$ and $\ottnt{v_{{\mathrm{2}}}} \, \ottnt{v'_{{\mathrm{2}}}}$ raise blame, we finish.
   Otherwise, $\ottnt{v_{{\mathrm{1}}}} \, \ottnt{v''_{{\mathrm{1}}}}  \longrightarrow^{\ast}  \ottnt{v'''_{{\mathrm{1}}}}$ and
   $\ottnt{v_{{\mathrm{2}}}} \, \ottnt{v'_{{\mathrm{2}}}}  \longrightarrow^{\ast}  \ottnt{v'''_{{\mathrm{2}}}}$ for some $\ottnt{v'''_{{\mathrm{1}}}}$ and $\ottnt{v'''_{{\mathrm{2}}}}$, and
   it suffices to show that
   \[
      \theta_{{\mathrm{1}}}  (   \delta_{{\mathrm{1}}}  (  \langle  \ottnt{T_{{\mathrm{12}}}} \, [  \ottnt{v''_{{\mathrm{1}}}}  \ottsym{/}  \mathit{x}  ]  \Rightarrow  \ottnt{T_{{\mathrm{22}}}} \, [  \ottnt{v'_{{\mathrm{1}}}}  \ottsym{/}  \mathit{x}  ]  \rangle   ^{ \ell }   )   )  \, \ottnt{v'''_{{\mathrm{1}}}}  \simeq_{\mathtt{e} }  \ottnt{v'''_{{\mathrm{2}}}}   \ottsym{:}   \ottnt{T_{{\mathrm{22}}}} ;  \theta ;   \delta    [  \,  (  \ottnt{v'_{{\mathrm{1}}}}  ,  \ottnt{v'_{{\mathrm{2}}}}  ) /  \mathit{x}  \,  ]   .
   \]
   We have $ \ottnt{v'''_{{\mathrm{1}}}}  \simeq_{\mathtt{v} }  \ottnt{v'''_{{\mathrm{2}}}}   \ottsym{:}   \ottnt{T_{{\mathrm{12}}}} ;  \theta ;   \delta    [  \,  (  \ottnt{v''_{{\mathrm{1}}}}  ,  \ottnt{v'_{{\mathrm{2}}}}  ) /  \mathit{x}  \,  ]   $.
   From the assumptions,
   we have:
   \begin{itemize}
    \item $ \ottnt{T_{{\mathrm{12}}}}  \simeq  \ottnt{T_{{\mathrm{12}}}}   \ottsym{:}   \ast ;  \theta ;   \delta    [  \,  (  \ottnt{v''_{{\mathrm{1}}}}  ,  \ottnt{v'_{{\mathrm{2}}}}  ) /  \mathit{x}  \,  ]   $
    \item $ \ottnt{T_{{\mathrm{22}}}}  \simeq  \ottnt{T_{{\mathrm{22}}}}   \ottsym{:}   \ast ;  \theta ;   \delta    [  \,  (  \ottnt{v'_{{\mathrm{1}}}}  ,  \ottnt{v'_{{\mathrm{2}}}}  ) /  \mathit{x}  \,  ]   $
    \item $ \ottnt{T_{{\mathrm{12}}}}  \simeq  \ottnt{T_{{\mathrm{22}}}}   \ottsym{:}   \ast ;  \theta ;   \delta    [  \,  (  \ottnt{v''_{{\mathrm{1}}}}  ,  \ottnt{v'_{{\mathrm{2}}}}  ) /  \mathit{x}  \,  ]   $
    \item $ \ottnt{T_{{\mathrm{22}}}}  \simeq  \ottnt{T_{{\mathrm{12}}}}   \ottsym{:}   \ast ;  \theta ;   \delta    [  \,  (  \ottnt{v'_{{\mathrm{1}}}}  ,  \ottnt{v'_{{\mathrm{2}}}}  ) /  \mathit{x}  \,  ]   $
   \end{itemize}
   Let $\delta'  \ottsym{=}    \delta    [  \,  (  \ottnt{v'_{{\mathrm{1}}}}  ,  \ottnt{v'_{{\mathrm{2}}}}  ) /  \mathit{x}  \,  ]      [  \,  (  \ottnt{v''_{{\mathrm{1}}}}  ,  \ottnt{v'_{{\mathrm{2}}}}  ) /  \mathit{y}  \,  ]  $.
   Since type relations are closed under $\alpha$-renaming,
   we have
   \begin{itemize}
    \item $ \ottnt{T_{{\mathrm{12}}}} \, [  \mathit{y}  \ottsym{/}  \mathit{x}  ]  \simeq  \ottnt{T_{{\mathrm{12}}}} \, [  \mathit{y}  \ottsym{/}  \mathit{x}  ]   \ottsym{:}   \ast ;  \theta ;  \delta' $
    \item $ \ottnt{T_{{\mathrm{22}}}}  \simeq  \ottnt{T_{{\mathrm{22}}}}   \ottsym{:}   \ast ;  \theta ;  \delta' $
    \item $ \ottnt{T_{{\mathrm{12}}}} \, [  \mathit{y}  \ottsym{/}  \mathit{x}  ]  \simeq  \ottnt{T_{{\mathrm{22}}}} \, [  \mathit{y}  \ottsym{/}  \mathit{x}  ]   \ottsym{:}   \ast ;  \theta ;  \delta' $
    \item $ \ottnt{T_{{\mathrm{22}}}}  \simeq  \ottnt{T_{{\mathrm{12}}}}   \ottsym{:}   \ast ;  \theta ;  \delta' $
   \end{itemize}
   by the weakening (\prop:ref{fh-lr-val-ws}).
   %
   Furthermore, we can show
   \begin{itemize}
    \item $ \ottnt{T_{{\mathrm{12}}}} \, [  \mathit{y}  \ottsym{/}  \mathit{x}  ]  \simeq  \ottnt{T_{{\mathrm{22}}}}   \ottsym{:}   \ast ;  \theta ;  \delta' $
          from $ \ottnt{T_{{\mathrm{12}}}} \, [  \mathit{y}  \ottsym{/}  \mathit{x}  ]  \simeq  \ottnt{T_{{\mathrm{22}}}} \, [  \mathit{y}  \ottsym{/}  \mathit{x}  ]   \ottsym{:}   \ast ;  \theta ;  \delta' $ and
    \item $ \ottnt{T_{{\mathrm{22}}}}  \simeq  \ottnt{T_{{\mathrm{12}}}} \, [  \mathit{y}  \ottsym{/}  \mathit{x}  ]   \ottsym{:}   \ast ;  \theta ;  \delta' $
          from $ \ottnt{T_{{\mathrm{22}}}}  \simeq  \ottnt{T_{{\mathrm{12}}}}   \ottsym{:}   \ast ;  \theta ;  \delta' $
   \end{itemize}
   because $\mathit{x}$ and $\mathit{y}$ have the same denotation in $\delta'_{{\mathrm{2}}}$, that
   is, $ \delta'_{{\mathrm{2}}}  (  \mathit{x}  )   \ottsym{=}   \delta'_{{\mathrm{2}}}  (  \mathit{y}  ) $.
   Thus, by the IH,
   \[\begin{array}{ll}
    &   \theta_{{\mathrm{1}}}  (   \delta_{{\mathrm{1}}}  (  \langle  \ottnt{T_{{\mathrm{12}}}} \, [  \mathit{y}  \ottsym{/}  \mathit{x}  ]  \Rightarrow  \ottnt{T_{{\mathrm{22}}}}  \rangle   ^{ \ell }   )  \, [  \ottnt{v'_{{\mathrm{1}}}}  \ottsym{/}  \mathit{x}  \ottsym{,}  \ottnt{v''_{{\mathrm{1}}}}  \ottsym{/}  \mathit{y}  ]  )     \\   \simeq_{\mathtt{v} }   &     \theta_{{\mathrm{2}}}  (   \delta_{{\mathrm{2}}}  (    \lambda    \mathit{x}  \mathord{:}  \ottnt{T_{{\mathrm{12}}}} \, [  \mathit{y}  \ottsym{/}  \mathit{x}  ]  .  \mathit{x}   )  \, [  \ottnt{v'_{{\mathrm{2}}}}  \ottsym{/}  \mathit{x}  \ottsym{,}  \ottnt{v'_{{\mathrm{2}}}}  \ottsym{/}  \mathit{y}  ]  )    \ottsym{:}   \ottnt{T_{{\mathrm{12}}}} \, [  \mathit{y}  \ottsym{/}  \mathit{x}  ]  \rightarrow  \ottnt{T_{{\mathrm{22}}}} ;  \theta ;  \delta' .
      \end{array}\]
   Since $ \ottnt{v'''_{{\mathrm{1}}}}  \simeq_{\mathtt{v} }  \ottnt{v'''_{{\mathrm{2}}}}   \ottsym{:}   \ottnt{T_{{\mathrm{12}}}} ;  \theta ;   \delta    [  \,  (  \ottnt{v''_{{\mathrm{1}}}}  ,  \ottnt{v'_{{\mathrm{2}}}}  ) /  \mathit{x}  \,  ]   $,
   we have $ \ottnt{v'''_{{\mathrm{1}}}}  \simeq_{\mathtt{v} }  \ottnt{v'''_{{\mathrm{2}}}}   \ottsym{:}   \ottnt{T_{{\mathrm{12}}}} \, [  \mathit{y}  \ottsym{/}  \mathit{x}  ] ;  \theta ;  \delta' $
   {\iffull
\else
    (term relations are closed under $\alpha$-renaming).
   \fi}
   Thus,
   \[
      \theta_{{\mathrm{1}}}  (   \delta_{{\mathrm{1}}}  (  \langle  \ottnt{T_{{\mathrm{12}}}} \, [  \ottnt{v''_{{\mathrm{1}}}}  \ottsym{/}  \mathit{x}  ]  \Rightarrow  \ottnt{T_{{\mathrm{22}}}} \, [  \ottnt{v'_{{\mathrm{1}}}}  \ottsym{/}  \mathit{x}  ]  \rangle   ^{ \ell }   )   )  \, \ottnt{v'''_{{\mathrm{1}}}}  \simeq_{\mathtt{e} }  \ottnt{v'''_{{\mathrm{2}}}}   \ottsym{:}   \ottnt{T_{{\mathrm{22}}}} ;  \theta ;   \delta    [  \,  (  \ottnt{v'_{{\mathrm{1}}}}  ,  \ottnt{v'_{{\mathrm{2}}}}  ) /  \mathit{x}  \,  ]   
   \]
   with the weakening.
   This is what we want to show.

  \case $  \forall   \alpha  .  \ottnt{T'_{{\mathrm{1}}}}   \simeq   \forall   \alpha  .  \ottnt{T'_{{\mathrm{2}}}}    \ottsym{:}   \ast ;  \theta ;  \delta $:
   Straightforward by the IH.

  \case $  \{  \mathit{x}  \mathord{:}  \ottnt{T'_{{\mathrm{1}}}}   \mathop{\mid}   \ottnt{e'_{{\mathrm{1}}}}  \}   \simeq   \{  \mathit{x}  \mathord{:}  \ottnt{T'_{{\mathrm{2}}}}   \mathop{\mid}   \ottnt{e'_{{\mathrm{2}}}}  \}    \ottsym{:}   \ast ;  \theta ;  \delta $:
   Without loss of generality, we can suppose that $\mathit{x} \, \notin \,  \mathit{dom}  (  \delta  ) $.
   By \E{Red}/\R{Forget},
   \[
     \theta_{{\mathrm{1}}}  (   \delta_{{\mathrm{1}}}  (  \langle  \ottnt{T_{{\mathrm{1}}}}  \Rightarrow  \ottnt{T_{{\mathrm{2}}}}  \rangle   ^{ \ell }   )   )  \, \ottnt{v_{{\mathrm{1}}}}  \longrightarrow   \theta_{{\mathrm{1}}}  (   \delta_{{\mathrm{1}}}  (  \langle  \ottnt{T'_{{\mathrm{1}}}}  \Rightarrow   \{  \mathit{x}  \mathord{:}  \ottnt{T'_{{\mathrm{2}}}}   \mathop{\mid}   \ottnt{e'_{{\mathrm{2}}}}  \}   \rangle   ^{ \ell }   )   )  \, \ottnt{v_{{\mathrm{1}}}}.
   \]
   Thus, it suffices to show that
   \[
      \theta_{{\mathrm{1}}}  (   \delta_{{\mathrm{1}}}  (  \langle  \ottnt{T'_{{\mathrm{1}}}}  \Rightarrow   \{  \mathit{x}  \mathord{:}  \ottnt{T'_{{\mathrm{2}}}}   \mathop{\mid}   \ottnt{e'_{{\mathrm{2}}}}  \}   \rangle   ^{ \ell }   )   )  \, \ottnt{v_{{\mathrm{1}}}}  \simeq_{\mathtt{e} }  \ottnt{v_{{\mathrm{2}}}}   \ottsym{:}    \{  \mathit{x}  \mathord{:}  \ottnt{T'_{{\mathrm{2}}}}   \mathop{\mid}   \ottnt{e'_{{\mathrm{2}}}}  \}  ;  \theta ;  \delta .
   \]
   By the IH,
   \[
      \theta_{{\mathrm{1}}}  (   \delta_{{\mathrm{1}}}  (  \langle  \ottnt{T'_{{\mathrm{1}}}}  \Rightarrow  \ottnt{T'_{{\mathrm{2}}}}  \rangle   ^{ \ell }   )   )   \simeq_{\mathtt{e} }   \theta_{{\mathrm{2}}}  (   \delta_{{\mathrm{2}}}  (    \lambda    \mathit{y}  \mathord{:}  \ottnt{T'_{{\mathrm{1}}}}  .  \mathit{y}   )   )    \ottsym{:}   \ottnt{T'_{{\mathrm{1}}}}  \rightarrow  \ottnt{T'_{{\mathrm{2}}}} ;  \theta ;  \delta .
   \]
  Since $ \ottnt{v_{{\mathrm{1}}}}  \simeq_{\mathtt{v} }  \ottnt{v_{{\mathrm{2}}}}   \ottsym{:}    \{  \mathit{x}  \mathord{:}  \ottnt{T'_{{\mathrm{1}}}}   \mathop{\mid}   \ottnt{e'_{{\mathrm{1}}}}  \}  ;  \theta ;  \delta $,
  we have $ \ottnt{v_{{\mathrm{1}}}}  \simeq_{\mathtt{v} }  \ottnt{v_{{\mathrm{2}}}}   \ottsym{:}   \ottnt{T'_{{\mathrm{1}}}} ;  \theta ;  \delta $ by definition.
  Thus,
  \[
     \theta_{{\mathrm{1}}}  (   \delta_{{\mathrm{1}}}  (  \langle  \ottnt{T'_{{\mathrm{1}}}}  \Rightarrow  \ottnt{T'_{{\mathrm{2}}}}  \rangle   ^{ \ell }   )   )  \, \ottnt{v_{{\mathrm{1}}}}  \simeq_{\mathtt{e} }  \ottnt{v_{{\mathrm{2}}}}   \ottsym{:}   \ottnt{T'_{{\mathrm{2}}}} ;  \theta ;  \delta .
  \]
  By definition, there exists some $\ottnt{v'_{{\mathrm{1}}}}$ such that
  $ \theta_{{\mathrm{1}}}  (   \delta_{{\mathrm{1}}}  (  \langle  \ottnt{T'_{{\mathrm{1}}}}  \Rightarrow  \ottnt{T'_{{\mathrm{2}}}}  \rangle   ^{ \ell }   )   )  \, \ottnt{v_{{\mathrm{1}}}}  \longrightarrow^{\ast}  \ottnt{v'_{{\mathrm{1}}}}$ and
  $ \ottnt{v'_{{\mathrm{1}}}}  \simeq_{\mathtt{v} }  \ottnt{v_{{\mathrm{2}}}}   \ottsym{:}   \ottnt{T'_{{\mathrm{2}}}} ;  \theta ;  \delta $.
  By \R{Forget} and \R{PreCheck},
  \[
    \theta_{{\mathrm{1}}}  (   \delta_{{\mathrm{1}}}  (  \langle  \ottnt{T'_{{\mathrm{1}}}}  \Rightarrow   \{  \mathit{x}  \mathord{:}  \ottnt{T'_{{\mathrm{2}}}}   \mathop{\mid}   \ottnt{e'_{{\mathrm{2}}}}  \}   \rangle   ^{ \ell }   )   )  \, \ottnt{v_{{\mathrm{1}}}}  \longrightarrow^{\ast}   \theta_{{\mathrm{1}}}  (   \delta_{{\mathrm{1}}}  (  \langle   \{  \mathit{x}  \mathord{:}  \ottnt{T'_{{\mathrm{2}}}}   \mathop{\mid}   \ottnt{e'_{{\mathrm{2}}}}  \}   \ottsym{,}  \ottnt{e'_{{\mathrm{2}}}} \, [  \ottnt{v'_{{\mathrm{1}}}}  \ottsym{/}  \mathit{x}  ]  \ottsym{,}  \ottnt{v'_{{\mathrm{1}}}}  \rangle   ^{ \ell }   )   ) .
  \]
  Thus, it suffices to show that
  \[
      \theta_{{\mathrm{1}}}  (   \delta_{{\mathrm{1}}}  (  \langle   \{  \mathit{x}  \mathord{:}  \ottnt{T'_{{\mathrm{2}}}}   \mathop{\mid}   \ottnt{e'_{{\mathrm{2}}}}  \}   \ottsym{,}  \ottnt{e'_{{\mathrm{2}}}} \, [  \ottnt{v'_{{\mathrm{1}}}}  \ottsym{/}  \mathit{x}  ]  \ottsym{,}  \ottnt{v'_{{\mathrm{1}}}}  \rangle   ^{ \ell }   )   )   \simeq_{\mathtt{e} }  \ottnt{v_{{\mathrm{2}}}}   \ottsym{:}    \{  \mathit{x}  \mathord{:}  \ottnt{T'_{{\mathrm{2}}}}   \mathop{\mid}   \ottnt{e'_{{\mathrm{2}}}}  \}  ;  \theta ;  \delta .
  \]

  We show
  \[
    \theta_{{\mathrm{1}}}  (   \delta_{{\mathrm{1}}}  (  \ottnt{e'_{{\mathrm{2}}}} \, [  \ottnt{v'_{{\mathrm{1}}}}  \ottsym{/}  \mathit{x}  ]  )   )   \longrightarrow^{\ast}   \mathsf{true} .
  \]
  Since $ \ottnt{v_{{\mathrm{1}}}}  \simeq_{\mathtt{v} }  \ottnt{v_{{\mathrm{2}}}}   \ottsym{:}    \{  \mathit{x}  \mathord{:}  \ottnt{T'_{{\mathrm{1}}}}   \mathop{\mid}   \ottnt{e'_{{\mathrm{1}}}}  \}  ;  \theta ;  \delta $,
  we have
  $ \ottnt{v_{{\mathrm{1}}}}  \simeq_{\mathtt{v} }  \ottnt{v_{{\mathrm{2}}}}   \ottsym{:}   \ottnt{T'_{{\mathrm{1}}}} ;  \theta ;  \delta $ and
  $ \theta_{{\mathrm{1}}}  (   \delta_{{\mathrm{1}}}  (  \ottnt{e'_{{\mathrm{1}}}} \, [  \ottnt{v_{{\mathrm{1}}}}  \ottsym{/}  \mathit{x}  ]  )   )   \longrightarrow^{\ast}   \mathsf{true} $.
  Since
  $  \{  \mathit{x}  \mathord{:}  \ottnt{T'_{{\mathrm{1}}}}   \mathop{\mid}   \ottnt{e'_{{\mathrm{1}}}}  \}   \simeq   \{  \mathit{x}  \mathord{:}  \ottnt{T'_{{\mathrm{2}}}}   \mathop{\mid}   \ottnt{e'_{{\mathrm{2}}}}  \}    \ottsym{:}   \ast ;  \theta ;  \delta $,
  we have
  \[
     \theta_{{\mathrm{1}}}  (   \delta_{{\mathrm{1}}}  (  \ottnt{e'_{{\mathrm{1}}}} \, [  \ottnt{v_{{\mathrm{1}}}}  \ottsym{/}  \mathit{x}  ]  )   )   \simeq_{\mathtt{e} }   \theta_{{\mathrm{2}}}  (   \delta_{{\mathrm{2}}}  (  \ottnt{e'_{{\mathrm{2}}}} \, [  \ottnt{v_{{\mathrm{2}}}}  \ottsym{/}  \mathit{x}  ]  )   )    \ottsym{:}    \mathsf{Bool}  ;  \theta ;  \delta .
  \]
  Since the term on the left-hand side evaluates to $ \mathsf{true} $,
  we have
  $ \theta_{{\mathrm{2}}}  (   \delta_{{\mathrm{2}}}  (  \ottnt{e'_{{\mathrm{2}}}} \, [  \ottnt{v_{{\mathrm{2}}}}  \ottsym{/}  \mathit{x}  ]  )   )   \longrightarrow^{\ast}   \mathsf{true} $.
  Since
  $  \{  \mathit{x}  \mathord{:}  \ottnt{T'_{{\mathrm{2}}}}   \mathop{\mid}   \ottnt{e'_{{\mathrm{2}}}}  \}   \simeq   \{  \mathit{x}  \mathord{:}  \ottnt{T'_{{\mathrm{2}}}}   \mathop{\mid}   \ottnt{e'_{{\mathrm{2}}}}  \}    \ottsym{:}   \ast ;  \theta ;  \delta $ and
  $ \ottnt{v'_{{\mathrm{1}}}}  \simeq_{\mathtt{v} }  \ottnt{v_{{\mathrm{2}}}}   \ottsym{:}   \ottnt{T'_{{\mathrm{2}}}} ;  \theta ;  \delta $,
  we have
  \[
     \theta_{{\mathrm{1}}}  (   \delta_{{\mathrm{1}}}  (  \ottnt{e'_{{\mathrm{2}}}} \, [  \ottnt{v'_{{\mathrm{1}}}}  \ottsym{/}  \mathit{x}  ]  )   )   \simeq_{\mathtt{e} }   \theta_{{\mathrm{2}}}  (   \delta_{{\mathrm{2}}}  (  \ottnt{e'_{{\mathrm{2}}}} \, [  \ottnt{v_{{\mathrm{2}}}}  \ottsym{/}  \mathit{x}  ]  )   )    \ottsym{:}    \mathsf{Bool}  ;  \theta ;  \delta .
  \]
  Since the term on the right-hand term evaluates to $ \mathsf{true} $,
  we have $ \theta_{{\mathrm{1}}}  (   \delta_{{\mathrm{1}}}  (  \ottnt{e'_{{\mathrm{2}}}} \, [  \ottnt{v'_{{\mathrm{1}}}}  \ottsym{/}  \mathit{x}  ]  )   )   \longrightarrow^{\ast}   \mathsf{true} $.

  Thus, $ \theta_{{\mathrm{1}}}  (   \delta_{{\mathrm{1}}}  (  \langle   \{  \mathit{x}  \mathord{:}  \ottnt{T'_{{\mathrm{2}}}}   \mathop{\mid}   \ottnt{e'_{{\mathrm{2}}}}  \}   \ottsym{,}  \ottnt{e'_{{\mathrm{2}}}} \, [  \ottnt{v'_{{\mathrm{1}}}}  \ottsym{/}  \mathit{x}  ]  \ottsym{,}  \ottnt{v'_{{\mathrm{1}}}}  \rangle   ^{ \ell }   )   )   \longrightarrow^{\ast}  \ottnt{v'_{{\mathrm{1}}}}$,
  and so it suffices to show that
  \[
    \ottnt{v'_{{\mathrm{1}}}}  \simeq_{\mathtt{v} }  \ottnt{v_{{\mathrm{2}}}}   \ottsym{:}    \{  \mathit{x}  \mathord{:}  \ottnt{T'_{{\mathrm{2}}}}   \mathop{\mid}   \ottnt{e'_{{\mathrm{2}}}}  \}  ;  \theta ;  \delta ,
  \]
  which follows by the facts that
  $ \ottnt{v'_{{\mathrm{1}}}}  \simeq_{\mathtt{v} }  \ottnt{v_{{\mathrm{2}}}}   \ottsym{:}   \ottnt{T'_{{\mathrm{2}}}} ;  \theta ;  \delta $ and
  $ \theta_{{\mathrm{1}}}  (   \delta_{{\mathrm{1}}}  (  \ottnt{e'_{{\mathrm{2}}}} \, [  \ottnt{v'_{{\mathrm{1}}}}  \ottsym{/}  \mathit{x}  ]  )   )   \longrightarrow^{\ast}   \mathsf{true} $ and
  $ \theta_{{\mathrm{2}}}  (   \delta_{{\mathrm{2}}}  (  \ottnt{e'_{{\mathrm{2}}}} \, [  \ottnt{v_{{\mathrm{2}}}}  \ottsym{/}  \mathit{x}  ]  )   )   \longrightarrow^{\ast}   \mathsf{true} $.  \qedhere
 \end{itemize}
\end{prop}

The other requirement about reflexive casts is shown similarly.
\begin{prop}[name=Elimination of Reflexive Casts on Right]{fh-lr-elim-refl-cast-right}
 If
 $ \ottnt{T}  \simeq  \ottnt{T}   \ottsym{:}   \ast ;  \theta ;  \delta $ and
 $ \ottnt{T}  \simeq  \ottnt{T_{{\mathrm{1}}}}   \ottsym{:}   \ast ;  \theta ;  \delta $ and
 $ \ottnt{T}  \simeq  \ottnt{T_{{\mathrm{2}}}}   \ottsym{:}   \ast ;  \theta ;  \delta $,
 then
 $  \theta_{{\mathrm{1}}}  (   \delta_{{\mathrm{1}}}  (    \lambda    \mathit{x}  \mathord{:}  \ottnt{T}  .  \mathit{x}   )   )   \simeq_{\mathtt{v} }   \theta_{{\mathrm{2}}}  (   \delta_{{\mathrm{2}}}  (  \langle  \ottnt{T_{{\mathrm{1}}}}  \Rightarrow  \ottnt{T_{{\mathrm{2}}}}  \rangle   ^{ \ell }   )   )    \ottsym{:}   \ottnt{T}  \rightarrow  \ottnt{T} ;  \theta ;  \delta $.

 \proof

 By course-of-values induction on the sum of sizes of $\ottnt{T_{{\mathrm{1}}}}$ and $\ottnt{T_{{\mathrm{2}}}}$.
\end{prop}

The final requirement is about CIU equivalence---if $\emptyset  \vdash  \ottnt{v} \, =_\mathsf{ciu} \, \ottnt{v_{{\mathrm{1}}}}  \ottsym{:}  \ottnt{T_{{\mathrm{1}}}}$
and $\ottsym{(}  \ottnt{v_{{\mathrm{1}}}}  \ottsym{,}  \ottnt{v_{{\mathrm{2}}}}  \ottsym{)} \, \in \, \ottnt{r}$, then $\ottsym{(}  \ottnt{v}  \ottsym{,}  \ottnt{v_{{\mathrm{2}}}}  \ottsym{)} \, \in \, \ottnt{r}$.
We show that term relations satisfy it by using the (restricted)
equivalence-respecting property~\cite{Pitts_2005_ATTAPL}.
\begin{prop}{fh-lr-comp-sectx-hole-red}
 If
 $\emptyset  \vdash  \ottnt{e_{{\mathrm{1}}}}  \ottsym{:}  \ottnt{T}$ and
 $\ottnt{e_{{\mathrm{1}}}}  \longrightarrow^{\ast}  \ottnt{e_{{\mathrm{2}}}}$ and
 $\emptyset  \vdash  \ottnt{E}^\ottnt{S}  \ottsym{:}  \ottsym{(}  \emptyset  \vdash  \ottnt{e_{{\mathrm{2}}}}  \ottsym{:}  \ottnt{T}  \ottsym{)}  \mathrel{\circ\hspace{-.4em}\rightarrow}  \ottnt{T'}$,
 then
 $\emptyset  \vdash  \ottnt{E}^\ottnt{S}  \ottsym{:}  \ottsym{(}  \emptyset  \vdash  \ottnt{e_{{\mathrm{1}}}}  \ottsym{:}  \ottnt{T}  \ottsym{)}  \mathrel{\circ\hspace{-.4em}\rightarrow}  \ottnt{T'}$.

 \proof

 Straightforward by induction on the derivation of
 $\emptyset  \vdash  \ottnt{E}^\ottnt{S}  \ottsym{:}  \ottsym{(}  \emptyset  \vdash  \ottnt{e_{{\mathrm{2}}}}  \ottsym{:}  \ottnt{T}  \ottsym{)}  \mathrel{\circ\hspace{-.4em}\rightarrow}  \ottnt{T'}$.
 \end{prop}
 \begin{prop}{fh-lr-comp-sectx-ctx-composed}
 If
 $\Gamma  \vdash  \ottnt{e}  \ottsym{:}  \ottnt{T}$ and
 $\emptyset  \vdash  \ottnt{C}  \ottsym{:}  \ottsym{(}  \Gamma  \vdash  \ottnt{e}  \ottsym{:}  \ottnt{T}  \ottsym{)}  \mathrel{\circ\hspace{-.4em}\rightarrow}  \ottnt{T'}$ and
 $\emptyset  \vdash  \ottnt{E}^\ottnt{S}  \ottsym{:}  \ottsym{(}  \emptyset  \vdash  \ottnt{C}  [  \ottnt{e}  ]  \ottsym{:}  \ottnt{T'}  \ottsym{)}  \mathrel{\circ\hspace{-.4em}\rightarrow}  \ottnt{T''}$,
 then $\emptyset  \vdash  \ottnt{E}^\ottnt{S}  [  \ottnt{C}  ]  \ottsym{:}  \ottsym{(}  \Gamma  \vdash  \ottnt{e}  \ottsym{:}  \ottnt{T}  \ottsym{)}  \mathrel{\circ\hspace{-.4em}\rightarrow}  \ottnt{T''}$.

 \proof

 Straightforward by induction on the derivation of
 $\emptyset  \vdash  \ottnt{E}^\ottnt{S}  \ottsym{:}  \ottsym{(}  \emptyset  \vdash  \ottnt{C}  [  \ottnt{e}  ]  \ottsym{:}  \ottnt{T'}  \ottsym{)}  \mathrel{\circ\hspace{-.4em}\rightarrow}  \ottnt{T''}$.
 \end{prop}
\begin{prop}[name=Equivalence-Respecting]{fh-lr-comp-equiv-res}
 If
 $\emptyset  \vdash  \ottnt{e_{{\mathrm{1}}}} \, =_\mathsf{ciu} \, \ottnt{e_{{\mathrm{2}}}}  \ottsym{:}   \theta_{{\mathrm{1}}}  (   \delta_{{\mathrm{1}}}  (  \ottnt{T}  )   ) $ and
 $ \ottnt{e_{{\mathrm{2}}}}  \simeq_{\mathtt{e} }  \ottnt{e_{{\mathrm{3}}}}   \ottsym{:}   \ottnt{T} ;  \theta ;  \delta $,
 then
 $ \ottnt{e_{{\mathrm{1}}}}  \simeq_{\mathtt{e} }  \ottnt{e_{{\mathrm{3}}}}   \ottsym{:}   \ottnt{T} ;  \theta ;  \delta $.

 \proof

 By induction on $\ottnt{T}$.
 If $\ottnt{e_{{\mathrm{1}}}}$, $\ottnt{e_{{\mathrm{2}}}}$, and $\ottnt{e_{{\mathrm{3}}}}$ raise blame, then we finish.
 Otherwise, $\ottnt{e_{{\mathrm{1}}}}  \longrightarrow^{\ast}  \ottnt{v_{{\mathrm{1}}}}$,
 $\ottnt{e_{{\mathrm{2}}}}  \longrightarrow^{\ast}  \ottnt{v_{{\mathrm{2}}}}$, and $\ottnt{e_{{\mathrm{3}}}}  \longrightarrow^{\ast}  \ottnt{v_{{\mathrm{3}}}}$
 for some $\ottnt{v_{{\mathrm{1}}}}$, $\ottnt{v_{{\mathrm{2}}}}$, and $\ottnt{v_{{\mathrm{3}}}}$.
 We have $ \ottnt{v_{{\mathrm{2}}}}  \simeq_{\mathtt{v} }  \ottnt{v_{{\mathrm{3}}}}   \ottsym{:}   \ottnt{T} ;  \theta ;  \delta $.
 By definition, it suffices to show that
 $ \ottnt{v_{{\mathrm{1}}}}  \simeq_{\mathtt{v} }  \ottnt{v_{{\mathrm{3}}}}   \ottsym{:}   \ottnt{T} ;  \theta ;  \delta $.
 By case analysis on $\ottnt{T}$.
 \begin{itemize}
  \case $\ottnt{T}  \ottsym{=}  \ottnt{B}$:
   Since $ \ottnt{v_{{\mathrm{2}}}}  \simeq_{\mathtt{v} }  \ottnt{v_{{\mathrm{3}}}}   \ottsym{:}   \ottnt{B} ;  \theta ;  \delta $,
   we have $\ottnt{v_{{\mathrm{2}}}}  \ottsym{=}  \ottnt{v_{{\mathrm{3}}}} = \ottnt{k} \, \in \,  {\cal K}_{ \ottnt{B} } $ for some $\ottnt{k}$.
   Let $\ottnt{E}^\ottnt{S}  \ottsym{=}  \langle   \mathsf{Bool}   \Rightarrow   \{  \mathit{x}  \mathord{:}   \mathsf{Bool}    \mathop{\mid}   \mathit{x}  \}   \rangle   ^{ \ell }  \, \ottsym{(}   \left[ \, \right]  \mathrel{=}_{ \ottnt{B} }  \ottnt{k}   \ottsym{)}$.
   Since $\emptyset  \vdash  \ottnt{E}^\ottnt{S}  \ottsym{:}  \ottsym{(}  \emptyset  \vdash  \ottnt{e_{{\mathrm{1}}}}  \ottsym{:}  \ottnt{B}  \ottsym{)}  \mathrel{\circ\hspace{-.4em}\rightarrow}   \{  \mathit{x}  \mathord{:}   \mathsf{Bool}    \mathop{\mid}   \mathit{x}  \} $ and
   $\emptyset  \vdash  \ottnt{e_{{\mathrm{1}}}} \, =_\mathsf{ciu} \, \ottnt{e_{{\mathrm{2}}}}  \ottsym{:}  \ottnt{B}$,
   we have $\ottnt{E}^\ottnt{S}  [  \ottnt{e_{{\mathrm{1}}}}  ]  \Downarrow  \ottnt{E}^\ottnt{S}  [  \ottnt{e_{{\mathrm{2}}}}  ]$.
   Since $\ottnt{e_{{\mathrm{2}}}}  \longrightarrow^{\ast}  \ottnt{k}$, we have $\ottnt{E}^\ottnt{S}  [  \ottnt{e_{{\mathrm{2}}}}  ]  \longrightarrow^{\ast}   \mathsf{true} $.
   If $\ottnt{v_{{\mathrm{1}}}}  \mathrel{\neq}  \ottnt{k}$, then $\ottnt{E}^\ottnt{S}  [  \ottnt{e_{{\mathrm{1}}}}  ]$ does not terminate at values, which is
   contradictory to $\ottnt{E}^\ottnt{S}  [  \ottnt{e_{{\mathrm{1}}}}  ]  \Downarrow  \ottnt{E}^\ottnt{S}  [  \ottnt{e_{{\mathrm{2}}}}  ]$.
   Thus, $\ottnt{v_{{\mathrm{1}}}}  \ottsym{=}  \ottnt{k}$.
   Since $\ottnt{v_{{\mathrm{3}}}}  \ottsym{=}  \ottnt{k}$, we have $ \ottnt{v_{{\mathrm{1}}}}  \simeq_{\mathtt{v} }  \ottnt{v_{{\mathrm{3}}}}   \ottsym{:}   \ottnt{B} ;  \theta ;  \delta $.

  \case $\ottnt{T}  \ottsym{=}  \alpha$:
   Since $ \ottnt{v_{{\mathrm{2}}}}  \simeq_{\mathtt{v} }  \ottnt{v_{{\mathrm{3}}}}   \ottsym{:}   \alpha ;  \theta ;  \delta $,
   there exists some $\ottnt{r}$, $\ottnt{T_{{\mathrm{1}}}}$, and $\ottnt{T_{{\mathrm{2}}}}$ such that
   $ \theta  (  \alpha  ) = (  \ottnt{r} ,  \ottnt{T_{{\mathrm{1}}}} ,  \ottnt{T_{{\mathrm{2}}}}  ) $ and
   $\ottsym{(}  \ottnt{v_{{\mathrm{2}}}}  \ottsym{,}  \ottnt{v_{{\mathrm{3}}}}  \ottsym{)} \, \in \, \ottnt{r}$.
   Since $\ottnt{r} \, \in \,  \mathsf{VRel}  (  \ottnt{T_{{\mathrm{1}}}} ,  \ottnt{T_{{\mathrm{2}}}}  ) $, it suffices to show that
   $\emptyset  \vdash  \ottnt{v_{{\mathrm{1}}}} \, =_\mathsf{ciu} \, \ottnt{v_{{\mathrm{2}}}}  \ottsym{:}  \ottnt{T_{{\mathrm{1}}}}$, that is,
   for any $\ottnt{E}^\ottnt{S}$ and $\ottnt{T'}$ such that
   $\emptyset  \vdash  \ottnt{E}^\ottnt{S}  \ottsym{:}  \ottsym{(}  \emptyset  \vdash  \ottnt{v_{{\mathrm{1}}}}  \ottsym{:}  \ottnt{T_{{\mathrm{1}}}}  \ottsym{)}  \mathrel{\circ\hspace{-.4em}\rightarrow}  \ottnt{T'}$,
   $\ottnt{E}^\ottnt{S}  [  \ottnt{v_{{\mathrm{1}}}}  ]  \Downarrow  \ottnt{E}^\ottnt{S}  [  \ottnt{v_{{\mathrm{2}}}}  ]$.
   Since $\emptyset  \vdash  \ottnt{e_{{\mathrm{1}}}}  \ottsym{:}  \ottnt{T_{{\mathrm{1}}}}$ and $\ottnt{e_{{\mathrm{1}}}}  \longrightarrow^{\ast}  \ottnt{v_{{\mathrm{1}}}}$,
   we have $\emptyset  \vdash  \ottnt{E}^\ottnt{S}  \ottsym{:}  \ottsym{(}  \emptyset  \vdash  \ottnt{e_{{\mathrm{1}}}}  \ottsym{:}  \ottnt{T_{{\mathrm{1}}}}  \ottsym{)}  \mathrel{\circ\hspace{-.4em}\rightarrow}  \ottnt{T'}$
   by \prop:ref{fh-lr-comp-sectx-hole-red}.
   Since $\emptyset  \vdash  \ottnt{e_{{\mathrm{1}}}} \, =_\mathsf{ciu} \, \ottnt{e_{{\mathrm{2}}}}  \ottsym{:}  \ottnt{T_{{\mathrm{1}}}}$,
   we have $\ottnt{E}^\ottnt{S}  [  \ottnt{e_{{\mathrm{1}}}}  ]  \Downarrow  \ottnt{E}^\ottnt{S}  [  \ottnt{e_{{\mathrm{2}}}}  ]$.
   Since $\ottnt{e_{{\mathrm{1}}}}  \longrightarrow^{\ast}  \ottnt{v_{{\mathrm{1}}}}$ and $\ottnt{e_{{\mathrm{2}}}}  \longrightarrow^{\ast}  \ottnt{v_{{\mathrm{2}}}}$,
   we have $\ottnt{E}^\ottnt{S}  [  \ottnt{v_{{\mathrm{1}}}}  ]  \Downarrow  \ottnt{E}^\ottnt{S}  [  \ottnt{v_{{\mathrm{2}}}}  ]$.

  \case $\ottnt{T}  \ottsym{=}   \mathit{x} \mathord{:} \ottnt{T_{{\mathrm{1}}}} \rightarrow \ottnt{T_{{\mathrm{2}}}} $:
   Without loss of generality, we can suppose that
   $\mathit{x} \, \notin \,  \mathit{dom}  (  \delta  ) $.
   By definition, it suffices to show that,
   for any $\ottnt{v'_{{\mathrm{1}}}}$ and $\ottnt{v'_{{\mathrm{3}}}}$ such that
   $ \ottnt{v'_{{\mathrm{1}}}}  \simeq_{\mathtt{v} }  \ottnt{v'_{{\mathrm{3}}}}   \ottsym{:}   \ottnt{T_{{\mathrm{1}}}} ;  \theta ;  \delta $,
   \[
     \ottnt{v_{{\mathrm{1}}}} \, \ottnt{v'_{{\mathrm{1}}}}  \simeq_{\mathtt{e} }  \ottnt{v_{{\mathrm{3}}}} \, \ottnt{v'_{{\mathrm{3}}}}   \ottsym{:}   \ottnt{T_{{\mathrm{2}}}} ;  \theta ;   \delta    [  \,  (  \ottnt{v'_{{\mathrm{1}}}}  ,  \ottnt{v'_{{\mathrm{3}}}}  ) /  \mathit{x}  \,  ]   .
   \]
   By the IH, it suffices to show that
   \begin{itemize}
    \item $\emptyset  \vdash  \ottnt{v_{{\mathrm{1}}}} \, \ottnt{v'_{{\mathrm{1}}}} \, =_\mathsf{ciu} \, \ottnt{v_{{\mathrm{2}}}} \, \ottnt{v'_{{\mathrm{1}}}}  \ottsym{:}   \theta_{{\mathrm{1}}}  (   \delta_{{\mathrm{1}}}  (  \ottnt{T_{{\mathrm{2}}}} \, [  \ottnt{v'_{{\mathrm{1}}}}  \ottsym{/}  \mathit{x}  ]  )   ) $ and
    \item $ \ottnt{v_{{\mathrm{2}}}} \, \ottnt{v'_{{\mathrm{1}}}}  \simeq_{\mathtt{e} }  \ottnt{v_{{\mathrm{3}}}} \, \ottnt{v'_{{\mathrm{3}}}}   \ottsym{:}   \ottnt{T_{{\mathrm{2}}}} ;  \theta ;   \delta    [  \,  (  \ottnt{v'_{{\mathrm{1}}}}  ,  \ottnt{v'_{{\mathrm{3}}}}  ) /  \mathit{x}  \,  ]   $.
   \end{itemize}
   The second is shown by
   $ \ottnt{v_{{\mathrm{2}}}}  \simeq_{\mathtt{v} }  \ottnt{v_{{\mathrm{3}}}}   \ottsym{:}    \mathit{x} \mathord{:} \ottnt{T_{{\mathrm{1}}}} \rightarrow \ottnt{T_{{\mathrm{2}}}}  ;  \theta ;  \delta $ and
   $ \ottnt{v'_{{\mathrm{1}}}}  \simeq_{\mathtt{v} }  \ottnt{v'_{{\mathrm{3}}}}   \ottsym{:}   \ottnt{T_{{\mathrm{1}}}} ;  \theta ;  \delta $.

   As for the first, it suffices to show that,
   for any $\ottnt{E}^\ottnt{S}$ and $\ottnt{T'}$ such that
   $\emptyset  \vdash  \ottnt{E}^\ottnt{S}  \ottsym{:}  \ottsym{(}  \emptyset  \vdash  \ottnt{v_{{\mathrm{1}}}} \, \ottnt{v'_{{\mathrm{1}}}}  \ottsym{:}   \theta_{{\mathrm{1}}}  (   \delta_{{\mathrm{1}}}  (  \ottnt{T_{{\mathrm{2}}}} \, [  \ottnt{v'_{{\mathrm{1}}}}  \ottsym{/}  \mathit{x}  ]  )   )   \ottsym{)}  \mathrel{\circ\hspace{-.4em}\rightarrow}  \ottnt{T'}$,
   \[
    \ottnt{E}^\ottnt{S}  [  \ottnt{v_{{\mathrm{1}}}} \, \ottnt{v'_{{\mathrm{1}}}}  ]  \Downarrow  \ottnt{E}^\ottnt{S}  [  \ottnt{v_{{\mathrm{2}}}} \, \ottnt{v'_{{\mathrm{1}}}}  ].
   \]
   Since
   \begin{itemize}
    \item $\emptyset  \vdash  \ottnt{v_{{\mathrm{1}}}}  \ottsym{:}   \theta_{{\mathrm{1}}}  (   \delta_{{\mathrm{1}}}  (   \mathit{x} \mathord{:} \ottnt{T_{{\mathrm{1}}}} \rightarrow \ottnt{T_{{\mathrm{2}}}}   )   ) $,
    \item $\emptyset  \vdash  \ottsym{(}  \left[ \, \right] \, \ottnt{v'_{{\mathrm{1}}}}  \ottsym{)}  \ottsym{:}  \ottsym{(}  \emptyset  \vdash  \ottnt{v_{{\mathrm{1}}}}  \ottsym{:}   \theta_{{\mathrm{1}}}  (   \delta_{{\mathrm{1}}}  (   \mathit{x} \mathord{:} \ottnt{T_{{\mathrm{1}}}} \rightarrow \ottnt{T_{{\mathrm{2}}}}   )   )   \ottsym{)}  \mathrel{\circ\hspace{-.4em}\rightarrow}   \theta_{{\mathrm{1}}}  (   \delta_{{\mathrm{1}}}  (  \ottnt{T_{{\mathrm{2}}}} \, [  \ottnt{v'_{{\mathrm{1}}}}  \ottsym{/}  \mathit{x}  ]  )   ) $, and
    \item $\emptyset  \vdash  \ottnt{E}^\ottnt{S}  \ottsym{:}  \ottsym{(}  \emptyset  \vdash  \ottnt{v_{{\mathrm{1}}}} \, \ottnt{v'_{{\mathrm{1}}}}  \ottsym{:}   \theta_{{\mathrm{1}}}  (   \delta_{{\mathrm{1}}}  (  \ottnt{T_{{\mathrm{2}}}} \, [  \ottnt{v'_{{\mathrm{1}}}}  \ottsym{/}  \mathit{x}  ]  )   )   \ottsym{)}  \mathrel{\circ\hspace{-.4em}\rightarrow}  \ottnt{T'}$,
   \end{itemize}
   we have
   \[
    \emptyset  \vdash   \ottnt{E}^\ottnt{S}   [  ~  \left[ \, \right] \, \ottnt{v'_{{\mathrm{1}}}}   ~  ]   \ottsym{:}  \ottsym{(}  \emptyset  \vdash  \ottnt{v_{{\mathrm{1}}}}  \ottsym{:}   \theta_{{\mathrm{1}}}  (   \delta_{{\mathrm{1}}}  (   \mathit{x} \mathord{:} \ottnt{T_{{\mathrm{1}}}} \rightarrow \ottnt{T_{{\mathrm{2}}}}   )   )   \ottsym{)}  \mathrel{\circ\hspace{-.4em}\rightarrow}  \ottnt{T'}
   \]
   by \prop:ref{fh-lr-comp-sectx-ctx-composed}.
   Since $\emptyset  \vdash  \ottnt{e_{{\mathrm{1}}}}  \ottsym{:}   \theta_{{\mathrm{1}}}  (   \delta_{{\mathrm{1}}}  (   \mathit{x} \mathord{:} \ottnt{T_{{\mathrm{1}}}} \rightarrow \ottnt{T_{{\mathrm{2}}}}   )   ) $ and
   $\ottnt{e_{{\mathrm{1}}}}  \longrightarrow^{\ast}  \ottnt{v_{{\mathrm{1}}}}$,
   we have
   \[
    \emptyset  \vdash   \ottnt{E}^\ottnt{S}   [  ~  \left[ \, \right] \, \ottnt{v'_{{\mathrm{1}}}}   ~  ]   \ottsym{:}  \ottsym{(}  \emptyset  \vdash  \ottnt{e_{{\mathrm{1}}}}  \ottsym{:}   \theta_{{\mathrm{1}}}  (   \delta_{{\mathrm{1}}}  (   \mathit{x} \mathord{:} \ottnt{T_{{\mathrm{1}}}} \rightarrow \ottnt{T_{{\mathrm{2}}}}   )   )   \ottsym{)}  \mathrel{\circ\hspace{-.4em}\rightarrow}  \ottnt{T'}
   \]
   by \prop:ref{fh-lr-comp-sectx-hole-red}.
   Since $\emptyset  \vdash  \ottnt{e_{{\mathrm{1}}}} \, =_\mathsf{ciu} \, \ottnt{e_{{\mathrm{2}}}}  \ottsym{:}   \theta_{{\mathrm{1}}}  (   \delta_{{\mathrm{1}}}  (   \mathit{x} \mathord{:} \ottnt{T_{{\mathrm{1}}}} \rightarrow \ottnt{T_{{\mathrm{2}}}}   )   ) $,
   we have $\ottnt{E}^\ottnt{S}  [  \ottnt{e_{{\mathrm{1}}}} \, \ottnt{v'_{{\mathrm{1}}}}  ]  \Downarrow  \ottnt{E}^\ottnt{S}  [  \ottnt{e_{{\mathrm{2}}}} \, \ottnt{v'_{{\mathrm{1}}}}  ]$.
   Since $\ottnt{e_{{\mathrm{1}}}}  \longrightarrow^{\ast}  \ottnt{v_{{\mathrm{1}}}}$ and $\ottnt{e_{{\mathrm{2}}}}  \longrightarrow^{\ast}  \ottnt{v_{{\mathrm{2}}}}$,
   we have $\ottnt{E}^\ottnt{S}  [  \ottnt{e_{{\mathrm{1}}}} \, \ottnt{v'_{{\mathrm{1}}}}  ]  \longrightarrow^{\ast}  \ottnt{E}^\ottnt{S}  [  \ottnt{v_{{\mathrm{1}}}} \, \ottnt{v'_{{\mathrm{1}}}}  ]$ and
   $\ottnt{E}^\ottnt{S}  [  \ottnt{e_{{\mathrm{2}}}} \, \ottnt{v'_{{\mathrm{1}}}}  ]  \longrightarrow^{\ast}  \ottnt{E}^\ottnt{S}  [  \ottnt{v_{{\mathrm{2}}}} \, \ottnt{v'_{{\mathrm{1}}}}  ]$.
   Thus, $\ottnt{E}^\ottnt{S}  [  \ottnt{v_{{\mathrm{1}}}} \, \ottnt{v'_{{\mathrm{1}}}}  ]  \Downarrow  \ottnt{E}^\ottnt{S}  [  \ottnt{v_{{\mathrm{2}}}} \, \ottnt{v'_{{\mathrm{1}}}}  ]$.

  \case $\ottnt{T}  \ottsym{=}   \forall   \alpha  .  \ottnt{T'} $:
   Similar to the case of $\ottnt{T}  \ottsym{=}   \mathit{x} \mathord{:} \ottnt{T_{{\mathrm{1}}}} \rightarrow \ottnt{T_{{\mathrm{2}}}} $.

  \case $\ottnt{T}  \ottsym{=}   \{  \mathit{x}  \mathord{:}  \ottnt{T'}   \mathop{\mid}   \ottnt{e'}  \} $:
   Without loss of generality, we can suppose that
   $\mathit{x} \, \notin \,  \mathit{dom}  (  \delta  ) $.
   By definition, it suffices to show that
   \begin{enumerate}
    \item \label{req:fh-lr-comp-equiv-res-one}
          $ \theta_{{\mathrm{1}}}  (   \delta_{{\mathrm{1}}}  (  \ottnt{e'} \, [  \ottnt{v_{{\mathrm{1}}}}  \ottsym{/}  \mathit{x}  ]  )   )   \longrightarrow^{\ast}   \mathsf{true} $,
    \item \label{req:fh-lr-comp-equiv-res-two}
          $ \theta_{{\mathrm{2}}}  (   \delta_{{\mathrm{2}}}  (  \ottnt{e'} \, [  \ottnt{v_{{\mathrm{3}}}}  \ottsym{/}  \mathit{x}  ]  )   )   \longrightarrow^{\ast}   \mathsf{true} $, and
    \item \label{req:fh-lr-comp-equiv-res-three}
          $ \ottnt{v_{{\mathrm{1}}}}  \simeq_{\mathtt{v} }  \ottnt{v_{{\mathrm{3}}}}   \ottsym{:}   \ottnt{T'} ;  \theta ;  \delta $.
   \end{enumerate}

   Since $\emptyset  \vdash  \ottnt{v_{{\mathrm{1}}}}  \ottsym{:}   \theta_{{\mathrm{1}}}  (   \delta_{{\mathrm{1}}}  (   \{  \mathit{x}  \mathord{:}  \ottnt{T'}   \mathop{\mid}   \ottnt{e'}  \}   )   ) $, 
   we have (\ref{req:fh-lr-comp-equiv-res-one})
   by the value inversion (\prop:ref{fh-val-satis-c}).
   Since $ \ottnt{v_{{\mathrm{2}}}}  \simeq_{\mathtt{v} }  \ottnt{v_{{\mathrm{3}}}}   \ottsym{:}    \{  \mathit{x}  \mathord{:}  \ottnt{T'}   \mathop{\mid}   \ottnt{e'}  \}  ;  \theta ;  \delta $,
   we have (\ref{req:fh-lr-comp-equiv-res-two}).

   As for (\ref{req:fh-lr-comp-equiv-res-three}),
   by the IH, it suffices to show that
   \begin{itemize}
    \item $\emptyset  \vdash  \ottnt{v_{{\mathrm{1}}}} \, =_\mathsf{ciu} \, \ottnt{v_{{\mathrm{2}}}}  \ottsym{:}   \theta_{{\mathrm{1}}}  (   \delta_{{\mathrm{1}}}  (  \ottnt{T'}  )   ) $ and
    \item $ \ottnt{v_{{\mathrm{2}}}}  \simeq_{\mathtt{v} }  \ottnt{v_{{\mathrm{3}}}}   \ottsym{:}   \ottnt{T'} ;  \theta ;  \delta $.
   \end{itemize}
   Since $ \ottnt{v_{{\mathrm{2}}}}  \simeq_{\mathtt{v} }  \ottnt{v_{{\mathrm{3}}}}   \ottsym{:}    \{  \mathit{x}  \mathord{:}  \ottnt{T'}   \mathop{\mid}   \ottnt{e'}  \}  ;  \theta ;  \delta $,
   we have the second by definition.
   We can show the first
   in a way similar to the case of $\ottnt{T}  \ottsym{=}   \mathit{x} \mathord{:} \ottnt{T_{{\mathrm{1}}}} \rightarrow \ottnt{T_{{\mathrm{2}}}} $.
   \qedhere
 \end{itemize}
\end{prop}

Now, we show that term relations are interpretations and then prove the \emph{type
compositionality}, which states that term relations indexed by $\ottnt{T} \, [  \ottnt{T_{{\mathrm{1}}}}  \ottsym{/}  \alpha  ]$
with $\theta$ and by $\ottnt{T}$ with $\theta \,  \{  \,  \alpha  \mapsto  \ottnt{r} , \ottnt{T_{{\mathrm{1}}}} , \ottnt{T_{{\mathrm{2}}}}  \,  \} $ are the same,
provided that $\ottnt{r}$ is a term relation indexed by $\ottnt{T_{{\mathrm{1}}}}$.
\begin{prop}[name=Term Relation as Interpretation]{fh-lr-well-formed-interpret}
  \label{lem:fh-lr-well-formed-interpret}
 Let $\ottnt{r} = \{ \ottsym{(}  \ottnt{v_{{\mathrm{1}}}}  \ottsym{,}  \ottnt{v_{{\mathrm{2}}}}  \ottsym{)} \mid  \ottnt{v_{{\mathrm{1}}}}  \simeq_{\mathtt{v} }  \ottnt{v_{{\mathrm{2}}}}   \ottsym{:}   \ottnt{T_{{\mathrm{1}}}} ;  \theta ;  \delta  \}$.
 If $ \ottnt{T_{{\mathrm{1}}}}  \simeq  \ottnt{T_{{\mathrm{1}}}}   \ottsym{:}   \ast ;  \theta ;  \delta $ and $ \ottnt{T_{{\mathrm{1}}}}  \simeq  \ottnt{T_{{\mathrm{2}}}}   \ottsym{:}   \ast ;  \theta ;  \delta $,
 then $\langle  \ottnt{r}  \ottsym{,}   \theta_{{\mathrm{1}}}  (   \delta_{{\mathrm{1}}}  (  \ottnt{T_{{\mathrm{1}}}}  )   )   \ottsym{,}   \theta_{{\mathrm{2}}}  (   \delta_{{\mathrm{2}}}  (  \ottnt{T_{{\mathrm{2}}}}  )   )   \rangle$.

 \proof

 {\iffull
\else
 $\ottnt{r} \, \in \,  \mathsf{VRel}  (   \theta_{{\mathrm{1}}}  (   \delta_{{\mathrm{1}}}  (  \ottnt{T_{{\mathrm{1}}}}  )   )  ,   \theta_{{\mathrm{2}}}  (   \delta_{{\mathrm{2}}}  (  \ottnt{T_{{\mathrm{2}}}}  )   )   ) $
 by \prop:ref{fh-lr-elim-refl-cast,fh-lr-elim-refl-cast-right,fh-lr-comp-equiv-res}.
 \fi}
\end{prop}
\begin{prop}[name=Type Compositionality]{fh-lr-typ-comp}
 Suppose that
 $ \Gamma  \ottsym{,}  \alpha  \ottsym{,}  \Gamma' \vdash \ottnt{T}  \mathrel{ \simeq }  \ottnt{T}  : \ast $ and $ \ottnt{T_{{\mathrm{1}}}}  \simeq  \ottnt{T_{{\mathrm{1}}}}   \ottsym{:}   \ast ;  \theta ;  \delta $.
 Also, assume that $\Gamma  \ottsym{,}  \alpha  \ottsym{,}  \Gamma'$ is self-related.
 Let
 $\ottnt{r} = \{ \ottsym{(}  \ottnt{v_{{\mathrm{1}}}}  \ottsym{,}  \ottnt{v_{{\mathrm{2}}}}  \ottsym{)} \mid  \ottnt{v_{{\mathrm{1}}}}  \simeq_{\mathtt{v} }  \ottnt{v_{{\mathrm{2}}}}   \ottsym{:}   \ottnt{T_{{\mathrm{1}}}} ;  \theta ;  \delta  \}$.
 If
 $\Gamma  \ottsym{,}  \alpha  \ottsym{,}  \Gamma'  \vdash  \theta \,  \{  \,  \alpha  \mapsto  \ottnt{r} ,  \theta_{{\mathrm{1}}}  (   \delta_{{\mathrm{1}}}  (  \ottnt{T_{{\mathrm{1}}}}  )   )  , \ottnt{T_{{\mathrm{2}}}}  \,  \}   \ottsym{;}  \delta$,
 then
 $ \ottnt{e_{{\mathrm{1}}}}  \simeq_{\mathtt{e} }  \ottnt{e_{{\mathrm{2}}}}   \ottsym{:}   \ottnt{T} ;  \theta \,  \{  \,  \alpha  \mapsto  \ottnt{r} ,  \theta_{{\mathrm{1}}}  (   \delta_{{\mathrm{1}}}  (  \ottnt{T_{{\mathrm{1}}}}  )   )  , \ottnt{T_{{\mathrm{2}}}}  \,  \}  ;  \delta $
 iff $ \ottnt{e_{{\mathrm{1}}}}  \simeq_{\mathtt{e} }  \ottnt{e_{{\mathrm{2}}}}   \ottsym{:}   \ottnt{T} \, [  \ottnt{T_{{\mathrm{1}}}}  \ottsym{/}  \alpha  ] ;  \theta ;  \delta $.

 \proof

 By induction on $\ottnt{T}$.
 If $\ottnt{e_{{\mathrm{1}}}}$ and $\ottnt{e_{{\mathrm{2}}}}$ raise blame, the conclusion follows straightforwardly.
 Otherwise,
 $\ottnt{e_{{\mathrm{1}}}}  \longrightarrow^{\ast}  \ottnt{v_{{\mathrm{1}}}}$ and $\ottnt{e_{{\mathrm{2}}}}  \longrightarrow^{\ast}  \ottnt{v_{{\mathrm{2}}}}$ for some $\ottnt{v_{{\mathrm{1}}}}$ and $\ottnt{v_{{\mathrm{2}}}}$, and
 it suffices to show that
 \[
   \ottnt{v_{{\mathrm{1}}}}  \simeq_{\mathtt{v} }  \ottnt{v_{{\mathrm{2}}}}   \ottsym{:}   \ottnt{T} ;  \theta' ;  \delta  \text{ iff }
   \ottnt{v_{{\mathrm{1}}}}  \simeq_{\mathtt{v} }  \ottnt{v_{{\mathrm{2}}}}   \ottsym{:}   \ottnt{T} \, [  \ottnt{T_{{\mathrm{1}}}}  \ottsym{/}  \alpha  ] ;  \theta ;  \delta 
 \]
 where $\theta'  \ottsym{=}  \theta \,  \{  \,  \alpha  \mapsto  \ottnt{r} ,  \theta_{{\mathrm{1}}}  (   \delta_{{\mathrm{1}}}  (  \ottnt{T_{{\mathrm{1}}}}  )   )  , \ottnt{T_{{\mathrm{2}}}}  \,  \} $.
 By case analysis on $\ottnt{T}$.
 \begin{itemize}
  \case $\ottnt{T}  \ottsym{=}  \beta$:
   Suppose that $ \ottnt{v_{{\mathrm{1}}}}  \simeq_{\mathtt{v} }  \ottnt{v_{{\mathrm{2}}}}   \ottsym{:}   \beta ;  \theta' ;  \delta $.
   We show that
   $ \ottnt{v_{{\mathrm{1}}}}  \simeq_{\mathtt{v} }  \ottnt{v_{{\mathrm{2}}}}   \ottsym{:}   \beta \, [  \ottnt{T_{{\mathrm{1}}}}  \ottsym{/}  \alpha  ] ;  \theta ;  \delta $.
   If $\alpha  \ottsym{=}  \beta$, then $\ottsym{(}  \ottnt{v_{{\mathrm{1}}}}  \ottsym{,}  \ottnt{v_{{\mathrm{2}}}}  \ottsym{)} \, \in \, \ottnt{r}$, that is,
   $ \ottnt{v_{{\mathrm{1}}}}  \simeq_{\mathtt{v} }  \ottnt{v_{{\mathrm{2}}}}   \ottsym{:}   \ottnt{T_{{\mathrm{1}}}} ;  \theta ;  \delta $.
   Since $\beta \, [  \ottnt{T_{{\mathrm{1}}}}  \ottsym{/}  \alpha  ]  \ottsym{=}  \ottnt{T_{{\mathrm{1}}}}$, we finish.
   Otherwise, if $\alpha  \mathrel{\neq}  \beta$, then obvious since $\beta \, [  \ottnt{T_{{\mathrm{1}}}}  \ottsym{/}  \alpha  ]  \ottsym{=}  \beta$ and
   $\theta'$ is an extension of $\theta$ with $\alpha$.

   Conversely, we suppose that $ \ottnt{v_{{\mathrm{1}}}}  \simeq_{\mathtt{v} }  \ottnt{v_{{\mathrm{2}}}}   \ottsym{:}   \beta \, [  \ottnt{T_{{\mathrm{1}}}}  \ottsym{/}  \alpha  ] ;  \theta ;  \delta $.
   We show that $ \ottnt{v_{{\mathrm{1}}}}  \simeq_{\mathtt{v} }  \ottnt{v_{{\mathrm{2}}}}   \ottsym{:}   \beta ;  \theta' ;  \delta $.
   If $\beta  \ottsym{=}  \alpha$, we have $ \ottnt{v_{{\mathrm{1}}}}  \simeq_{\mathtt{v} }  \ottnt{v_{{\mathrm{2}}}}   \ottsym{:}   \ottnt{T_{{\mathrm{1}}}} ;  \theta ;  \delta $, so
   $\ottsym{(}  \ottnt{v_{{\mathrm{1}}}}  \ottsym{,}  \ottnt{v_{{\mathrm{2}}}}  \ottsym{)} \, \in \, \ottnt{r}$ and $ \ottnt{v_{{\mathrm{1}}}}  \simeq_{\mathtt{v} }  \ottnt{v_{{\mathrm{2}}}}   \ottsym{:}   \alpha ;  \theta' ;  \delta $.
   Otherwise, if $\beta  \mathrel{\neq}  \alpha$, then obvious.

  \case $\ottnt{T}  \ottsym{=}  \ottnt{B}$: Obvious.
  \case $\ottnt{T}  \ottsym{=}   \mathit{x} \mathord{:} \ottnt{T'_{{\mathrm{1}}}} \rightarrow \ottnt{T'_{{\mathrm{2}}}} $: By the IHs.
  \case $\ottnt{T}  \ottsym{=}   \forall   \beta  .  \ottnt{T'} $: By the IH.
  \case $ \{  \mathit{x}  \mathord{:}  \ottnt{T'}   \mathop{\mid}   \ottnt{e'}  \} $:
   Without loss of generality, we can suppose that
   $\mathit{x} \, \notin \,  \mathit{dom}  (  \delta  ) $.
   We show:
   \[\begin{array}{lcl}
     \ottnt{v_{{\mathrm{1}}}}  \simeq_{\mathtt{v} }  \ottnt{v_{{\mathrm{2}}}}   \ottsym{:}   \ottnt{T'} ;  \theta' ;  \delta  &&
      \ottnt{v_{{\mathrm{1}}}}  \simeq_{\mathtt{v} }  \ottnt{v_{{\mathrm{2}}}}   \ottsym{:}   \ottnt{T'} \, [  \ottnt{T_{{\mathrm{1}}}}  \ottsym{/}  \alpha  ] ;  \theta ;  \delta  \\
     \theta_{{\mathrm{1}}}  (   \delta_{{\mathrm{1}}}  (  \ottnt{e'} \, [  \ottnt{T_{{\mathrm{1}}}}  \ottsym{/}  \alpha  ] \, [  \ottnt{v_{{\mathrm{1}}}}  \ottsym{/}  \mathit{x}  ]  )   )   \longrightarrow^{\ast}   \mathsf{true}  & \text{iff} &
      \theta_{{\mathrm{1}}}  (   \delta_{{\mathrm{1}}}  (  \ottnt{e'} \, [  \ottnt{T_{{\mathrm{1}}}}  \ottsym{/}  \alpha  ] \, [  \ottnt{v_{{\mathrm{1}}}}  \ottsym{/}  \mathit{x}  ]  )   )   \longrightarrow^{\ast}   \mathsf{true}  \\
     \theta_{{\mathrm{2}}}  (   \delta_{{\mathrm{2}}}  (  \ottnt{e'} \, [  \ottnt{T_{{\mathrm{2}}}}  \ottsym{/}  \alpha  ] \, [  \ottnt{v_{{\mathrm{2}}}}  \ottsym{/}  \mathit{x}  ]  )   )   \longrightarrow^{\ast}   \mathsf{true}  &&
      \theta_{{\mathrm{2}}}  (   \delta_{{\mathrm{2}}}  (  \ottnt{e'} \, [  \ottnt{T_{{\mathrm{1}}}}  \ottsym{/}  \alpha  ] \, [  \ottnt{v_{{\mathrm{2}}}}  \ottsym{/}  \mathit{x}  ]  )   )   \longrightarrow^{\ast}   \mathsf{true} 
     \end{array}\]
   We show only the left-to-right direction; the other is shown similarly.
   Since $ \Gamma  \ottsym{,}  \alpha  \ottsym{,}  \Gamma' \vdash  \{  \mathit{x}  \mathord{:}  \ottnt{T'}   \mathop{\mid}   \ottnt{e'}  \}   \mathrel{ \simeq }   \{  \mathit{x}  \mathord{:}  \ottnt{T'}   \mathop{\mid}   \ottnt{e'}  \}   : \ast $,
   it is easy to show that $ \Gamma  \ottsym{,}  \alpha  \ottsym{,}  \Gamma' \vdash \ottnt{T'}  \mathrel{ \simeq }  \ottnt{T'}  : \ast $.
   Since $ \ottnt{v_{{\mathrm{1}}}}  \simeq_{\mathtt{v} }  \ottnt{v_{{\mathrm{2}}}}   \ottsym{:}   \ottnt{T'} ;  \theta' ;  \delta $, we have
   \[
     \ottnt{v_{{\mathrm{1}}}}  \simeq_{\mathtt{v} }  \ottnt{v_{{\mathrm{2}}}}   \ottsym{:}   \ottnt{T'} \, [  \ottnt{T_{{\mathrm{1}}}}  \ottsym{/}  \alpha  ] ;  \theta ;  \delta 
   \]
   by the IH.
   We have the second case by the assumption of the left-to-right direction.
   The remaining case to be shown is:
   \[
     \theta_{{\mathrm{2}}}  (   \delta_{{\mathrm{2}}}  (  \ottnt{e'} \, [  \ottnt{T_{{\mathrm{1}}}}  \ottsym{/}  \alpha  ] \, [  \ottnt{v_{{\mathrm{2}}}}  \ottsym{/}  \mathit{x}  ]  )   )   \longrightarrow^{\ast}   \mathsf{true} 
   \]
   Since $\Gamma  \ottsym{,}  \alpha  \ottsym{,}  \Gamma'  \vdash  \theta'  \ottsym{;}  \delta$ by the assumption of this lemma,
   we have
   \begin{equation}
     \Gamma  \ottsym{,}  \alpha  \ottsym{,}  \Gamma'  ,  \mathit{x}  \mathord{:}  \ottnt{T'}   \vdash  \theta'  \ottsym{;}   \delta    [  \,  (  \ottnt{v_{{\mathrm{1}}}}  ,  \ottnt{v_{{\mathrm{2}}}}  ) /  \mathit{x}  \,  ]  
    \label{eqn:fh-lr-typ-comp-one}
   \end{equation}
   by the weakening (\prop:ref{fh-lr-val-ws}).
   Since $ \ottnt{T_{{\mathrm{1}}}}  \simeq  \ottnt{T_{{\mathrm{1}}}}   \ottsym{:}   \ast ;  \theta ;  \delta $,
   we have
   \begin{equation}
    \langle  \ottnt{r}  \ottsym{,}   \theta_{{\mathrm{1}}}  (   \delta_{{\mathrm{1}}}  (  \ottnt{T_{{\mathrm{1}}}}  )   )   \ottsym{,}   \theta_{{\mathrm{2}}}  (   \delta_{{\mathrm{2}}}  (  \ottnt{T_{{\mathrm{1}}}}  )   )   \rangle
     \label{eqn:fh-lr-typ-comp-two}
   \end{equation}
   by \prop:ref{fh-lr-well-formed-interpret}.
   {\iffull
\else 
   Then, we can show that
   \[
     \Gamma  \ottsym{,}  \alpha  \ottsym{,}  \Gamma'  ,  \mathit{x}  \mathord{:}  \ottnt{T'}   \vdash  \theta \,  \{  \,  \alpha  \mapsto  \ottnt{r} ,  \theta_{{\mathrm{1}}}  (   \delta_{{\mathrm{1}}}  (  \ottnt{T_{{\mathrm{1}}}}  )   )  ,  \theta_{{\mathrm{2}}}  (   \delta_{{\mathrm{2}}}  (  \ottnt{T_{{\mathrm{1}}}}  )   )   \,  \}   \ottsym{;}   \delta    [  \,  (  \ottnt{v_{{\mathrm{1}}}}  ,  \ottnt{v_{{\mathrm{2}}}}  ) /  \mathit{x}  \,  ]  
   \]
   by induction on $\Gamma'$ of (\ref{eqn:fh-lr-typ-comp-one}) with (\ref{eqn:fh-lr-typ-comp-two})
   and \prop:ref{fh-lr-untyped-exchange-trel}.
   \fi} 
   Since $ \Gamma  \ottsym{,}  \alpha  \ottsym{,}  \Gamma'  ,  \mathit{x}  \mathord{:}  \ottnt{T'}   \vdash  \ottnt{e'} \,  \mathrel{ \simeq }  \, \ottnt{e'}  \ottsym{:}   \mathsf{Bool} $,
   we have
   \[
      \theta_{{\mathrm{1}}}  (   \delta_{{\mathrm{1}}}  (  \ottnt{e'} \, [  \ottnt{T_{{\mathrm{1}}}}  \ottsym{/}  \alpha  ] \, [  \ottnt{v_{{\mathrm{1}}}}  \ottsym{/}  \mathit{x}  ]  )   )   \simeq_{\mathtt{e} }   \theta_{{\mathrm{2}}}  (   \delta_{{\mathrm{2}}}  (  \ottnt{e'} \, [  \ottnt{T_{{\mathrm{1}}}}  \ottsym{/}  \alpha  ] \, [  \ottnt{v_{{\mathrm{2}}}}  \ottsym{/}  \mathit{x}  ]  )   )    \ottsym{:}    \mathsf{Bool}  ;  \theta ;  \delta .
   \]
   Since the term on the left-hand side evaluates to $ \mathsf{true} $,
   so does the one on the right-hand side, which we want to show.
   \qedhere
 \end{itemize}
\end{prop}

\begin{prop}
 [name=Compatibility under Self-relatedness Assumption: Type Application]
 {fh-lr-comp-tapp-refl-assump}
 \label{lem:fh-lr-tapp-end}
 Suppose that
 $ \Gamma \vdash  \forall   \alpha  .  \ottnt{T}   \mathrel{ \simeq }   \forall   \alpha  .  \ottnt{T}   : \ast $ and
 $ \Gamma \vdash \ottnt{T_{{\mathrm{1}}}}  \mathrel{ \simeq }  \ottnt{T_{{\mathrm{1}}}}  : \ast $ and
 that $\Gamma$ is self-related.
 If
 $\Gamma  \vdash  \ottnt{e_{{\mathrm{1}}}} \,  \mathrel{ \simeq }  \, \ottnt{e_{{\mathrm{2}}}}  \ottsym{:}   \forall   \alpha  .  \ottnt{T} $ and
 $ \Gamma \vdash \ottnt{T_{{\mathrm{1}}}}  \mathrel{ \simeq }  \ottnt{T_{{\mathrm{2}}}}  : \ast $,
 then
 $\Gamma  \vdash  \ottnt{e_{{\mathrm{1}}}} \, \ottnt{T_{{\mathrm{1}}}} \,  \mathrel{ \simeq }  \, \ottnt{e_{{\mathrm{2}}}} \, \ottnt{T_{{\mathrm{2}}}}  \ottsym{:}  \ottnt{T} \, [  \ottnt{T_{{\mathrm{1}}}}  \ottsym{/}  \alpha  ]$.

 \proof

 Suppose that $\Gamma  \vdash  \theta  \ottsym{;}  \delta$.
 Also, let
 $\ottnt{e'_{{\mathrm{1}}}}  \ottsym{=}   \theta_{{\mathrm{1}}}  (   \delta_{{\mathrm{1}}}  (  \ottnt{e_{{\mathrm{1}}}}  )   ) $,
 $\ottnt{e'_{{\mathrm{2}}}}  \ottsym{=}   \theta_{{\mathrm{2}}}  (   \delta_{{\mathrm{2}}}  (  \ottnt{e_{{\mathrm{2}}}}  )   ) $,
 $\ottnt{T'_{{\mathrm{1}}}}  \ottsym{=}   \theta_{{\mathrm{1}}}  (   \delta_{{\mathrm{1}}}  (  \ottnt{T_{{\mathrm{1}}}}  )   ) $ and
 $\ottnt{T'_{{\mathrm{2}}}}  \ottsym{=}   \theta_{{\mathrm{2}}}  (   \delta_{{\mathrm{2}}}  (  \ottnt{T_{{\mathrm{2}}}}  )   ) $.
 It suffices to show that
 \[
   \ottnt{e'_{{\mathrm{1}}}} \, \ottnt{T'_{{\mathrm{1}}}}  \simeq_{\mathtt{e} }  \ottnt{e'_{{\mathrm{2}}}} \, \ottnt{T'_{{\mathrm{2}}}}   \ottsym{:}   \ottnt{T} \, [  \ottnt{T_{{\mathrm{1}}}}  \ottsym{/}  \alpha  ] ;  \theta ;  \delta .
 \]
 Since $\Gamma  \vdash  \ottnt{e_{{\mathrm{1}}}} \,  \mathrel{ \simeq }  \, \ottnt{e_{{\mathrm{2}}}}  \ottsym{:}   \forall   \alpha  .  \ottnt{T} $, we have
 $ \ottnt{e'_{{\mathrm{1}}}}  \simeq_{\mathtt{e} }  \ottnt{e'_{{\mathrm{2}}}}   \ottsym{:}    \forall   \alpha  .  \ottnt{T}  ;  \theta ;  \delta $.
 If $\ottnt{e'_{{\mathrm{1}}}}$ and $\ottnt{e'_{{\mathrm{2}}}}$ raise blame, we finish.
 Otherwise,
 $\ottnt{e'_{{\mathrm{1}}}}  \longrightarrow^{\ast}  \ottnt{v_{{\mathrm{1}}}}$ and $\ottnt{e'_{{\mathrm{2}}}}  \longrightarrow^{\ast}  \ottnt{v_{{\mathrm{2}}}}$ for some $\ottnt{v_{{\mathrm{1}}}}$ and $\ottnt{v_{{\mathrm{2}}}}$, and
 it suffices to show that
 \[
   \ottnt{v_{{\mathrm{1}}}} \, \ottnt{T'_{{\mathrm{1}}}}  \simeq_{\mathtt{e} }  \ottnt{v_{{\mathrm{2}}}} \, \ottnt{T'_{{\mathrm{2}}}}   \ottsym{:}   \ottnt{T} \, [  \ottnt{T_{{\mathrm{1}}}}  \ottsym{/}  \alpha  ] ;  \theta ;  \delta .
 \]
 We also have $ \ottnt{v_{{\mathrm{1}}}}  \simeq_{\mathtt{v} }  \ottnt{v_{{\mathrm{2}}}}   \ottsym{:}    \forall   \alpha  .  \ottnt{T}  ;  \theta ;  \delta $.

 Let $\ottnt{r} = \{ \ottsym{(}  \ottnt{v'_{{\mathrm{1}}}}  \ottsym{,}  \ottnt{v'_{{\mathrm{2}}}}  \ottsym{)} \mid  \ottnt{v'_{{\mathrm{1}}}}  \simeq_{\mathtt{v} }  \ottnt{v'_{{\mathrm{2}}}}   \ottsym{:}   \ottnt{T_{{\mathrm{1}}}} ;  \theta ;  \delta  \}$.
 Since $ \ottnt{T_{{\mathrm{1}}}}  \simeq  \ottnt{T_{{\mathrm{1}}}}   \ottsym{:}   \ast ;  \theta ;  \delta $ and $ \ottnt{T_{{\mathrm{1}}}}  \simeq  \ottnt{T_{{\mathrm{2}}}}   \ottsym{:}   \ast ;  \theta ;  \delta $,
 we have $\langle  \ottnt{r}  \ottsym{,}  \ottnt{T'_{{\mathrm{1}}}}  \ottsym{,}  \ottnt{T'_{{\mathrm{2}}}}  \rangle$ by \prop:ref{fh-lr-well-formed-interpret}.
 Since $ \ottnt{v_{{\mathrm{1}}}}  \simeq_{\mathtt{v} }  \ottnt{v_{{\mathrm{2}}}}   \ottsym{:}    \forall   \alpha  .  \ottnt{T}  ;  \theta ;  \delta $,
 we have
 $ \ottnt{v_{{\mathrm{1}}}} \, \ottnt{T'_{{\mathrm{1}}}}  \simeq_{\mathtt{e} }  \ottnt{v_{{\mathrm{2}}}} \, \ottnt{T'_{{\mathrm{2}}}}   \ottsym{:}   \ottnt{T} ;  \theta \,  \{  \,  \alpha  \mapsto  \ottnt{r} , \ottnt{T'_{{\mathrm{1}}}} , \ottnt{T'_{{\mathrm{2}}}}  \,  \}  ;  \delta $.
 Since
 $ \Gamma  \ottsym{,}  \alpha \vdash \ottnt{T}  \mathrel{ \simeq }  \ottnt{T}  : \ast $ and
 $\Gamma  \ottsym{,}  \alpha  \vdash  \theta \,  \{  \,  \alpha  \mapsto  \ottnt{r} , \ottnt{T'_{{\mathrm{1}}}} , \ottnt{T'_{{\mathrm{2}}}}  \,  \}   \ottsym{;}  \delta$ by the weakening
 (\prop:ref{fh-lr-typ-ws}), and
 $ \ottnt{T_{{\mathrm{1}}}}  \simeq  \ottnt{T_{{\mathrm{1}}}}   \ottsym{:}   \ast ;  \theta ;  \delta $ and
 $\Gamma  \ottsym{,}  \alpha$ is self-related,
 we have
 $ \ottnt{v_{{\mathrm{1}}}} \, \ottnt{T'_{{\mathrm{1}}}}  \simeq_{\mathtt{e} }  \ottnt{v_{{\mathrm{2}}}} \, \ottnt{T'_{{\mathrm{2}}}}   \ottsym{:}   \ottnt{T} \, [  \ottnt{T_{{\mathrm{1}}}}  \ottsym{/}  \alpha  ] ;  \theta ;  \delta $
 by the type compositionality (\prop:ref{fh-lr-typ-comp}).
\end{prop}

\paragraph{\bf Fundamental property: other constructors}
\label{sec:proving-other-case}

We show remaining cases of the fundamental property.
\begin{prop}[name=Compatibility: Variable]{fh-lr-comp-var}
 \label{lem:fh-lr-other-start}
 If $ \mathord{ \vdash } ~  \Gamma $ and $ \mathit{x}  \mathord{:}  \ottnt{T}   \in   \Gamma $, then $\Gamma  \vdash  \mathit{x} \,  \mathrel{ \simeq }  \, \mathit{x}  \ottsym{:}  \ottnt{T}$.

 \proof

 Straightforward by definition.
\end{prop}

\begin{prop}[name=Compatibility: Constant]{fh-lr-comp-const}
 If $ \mathord{ \vdash } ~  \Gamma $, then $\Gamma  \vdash  \ottnt{k} \,  \mathrel{ \simeq }  \, \ottnt{k}  \ottsym{:}   \mathsf{ty}  (  \ottnt{k}  ) $.

 \proof

 Let $\Gamma  \vdash  \theta  \ottsym{;}  \delta$.
 It suffices to show that
 $ \ottnt{k}  \simeq_{\mathtt{v} }  \ottnt{k}   \ottsym{:}    \mathsf{ty}  (  \ottnt{k}  )  ;  \theta ;  \delta $.
 By the assumptions that $ \mathit{unref}  (   \mathsf{ty}  (  \ottnt{k}  )   )   \ottsym{=}  \ottnt{B}$ for some $\ottnt{B}$ and
 that $\ottnt{k} \, \in \,  {\cal K}_{ \ottnt{B} } $,
 we have $ \ottnt{k}  \simeq_{\mathtt{v} }  \ottnt{k}   \ottsym{:}    \mathit{unref}  (   \mathsf{ty}  (  \ottnt{k}  )   )  ;  \theta ;  \delta $.
 Since constants satisfy contracts on their types and
 $ \mathsf{ty}  (  \ottnt{k}  ) $ is closed,
 $ \ottnt{k}  \simeq_{\mathtt{v} }  \ottnt{k}   \ottsym{:}    \mathsf{ty}  (  \ottnt{k}  )  ;  \theta ;  \delta $ by definition.
\end{prop}

\begin{prop}[name=Compatibility under Self-relatedness Assumption: Op]
 {fh-lr-comp-op-refl-assump}

 Suppose that $\Gamma$ is self-related and that
 $ \mathsf{ty}  (  {\tt op}  )   \ottsym{=}  {}  \mathit{x_{{\mathrm{1}}}}  \ottsym{:}  \ottnt{T_{{\mathrm{1}}}}  \rightarrow \, ... \, \rightarrow  \mathit{x_{\ottmv{n}}}  \ottsym{:}  \ottnt{T_{\ottmv{n}}}  {}  \rightarrow  \ottnt{T}$.
 Moreover, assume that, for any $\ottmv{i} \, \in \,  \{  \, \ottsym{1}  ,\, ... \, ,  \mathit{n} \,  \} $,
 $\Gamma  \vdash   { \ottnt{e} }_{  \ottsym{1} \ottmv{i}  }  \,  \mathrel{ \simeq }  \,  { \ottnt{e} }_{  \ottsym{1} \ottmv{i}  }   \ottsym{:}  \ottnt{T_{\ottmv{i}}} \, [  \ottnt{e_{{\mathrm{11}}}}  \ottsym{/}  \mathit{x_{{\mathrm{1}}}}  ,\, ... \, ,   { \ottnt{e} }_{  \ottsym{1} \ottmv{i}   \ottsym{-}  \ottsym{1} }   \ottsym{/}   \mathit{x} _{ \ottmv{i}  \ottsym{-}  \ottsym{1} }   ]$ and
 $ \Gamma \vdash \ottnt{T_{\ottmv{i}}} \, [  \ottnt{e_{{\mathrm{11}}}}  \ottsym{/}  \mathit{x_{{\mathrm{1}}}}  ,\, ... \, ,   { \ottnt{e} }_{  \ottsym{1} \ottmv{i}   \ottsym{-}  \ottsym{1} }   \ottsym{/}   \mathit{x} _{ \ottmv{i}  \ottsym{-}  \ottsym{1} }   ]  \mathrel{ \simeq }  \ottnt{T_{\ottmv{i}}} \, [  \ottnt{e_{{\mathrm{11}}}}  \ottsym{/}  \mathit{x_{{\mathrm{1}}}}  ,\, ... \, ,   { \ottnt{e} }_{  \ottsym{1} \ottmv{i}   \ottsym{-}  \ottsym{1} }   \ottsym{/}   \mathit{x} _{ \ottmv{i}  \ottsym{-}  \ottsym{1} }   ]  : \ast $.
 If $\Gamma  \vdash   { \ottnt{e} }_{  \ottsym{1} \ottmv{i}  }  \,  \mathrel{ \simeq }  \,  { \ottnt{e} }_{  \ottsym{2} \ottmv{i}  }   \ottsym{:}  \ottnt{T_{\ottmv{i}}} \, [  \ottnt{e_{{\mathrm{11}}}}  \ottsym{/}  \mathit{x_{{\mathrm{1}}}}  ,\, ... \, ,   { \ottnt{e} }_{  \ottsym{1} \ottmv{i}   \ottsym{-}  \ottsym{1} }   \ottsym{/}   \mathit{x} _{ \ottmv{i}  \ottsym{-}  \ottsym{1} }   ]$ for any $\ottmv{i} \, \in \,  \{  \, \ottsym{1}  ,\, ... \, ,  \mathit{n} \,  \} $,
 then $\Gamma  \vdash  {\tt op} \, \ottsym{(}  \ottnt{e_{{\mathrm{11}}}}  ,\, ... \, ,   { \ottnt{e} }_{  \ottsym{1} \mathit{n}  }   \ottsym{)} \,  \mathrel{ \simeq }  \, {\tt op} \, \ottsym{(}  \ottnt{e_{{\mathrm{21}}}}  ,\, ... \, ,   { \ottnt{e} }_{  \ottsym{2} \mathit{n}  }   \ottsym{)}  \ottsym{:}  \ottnt{T} \, [  \ottnt{e_{{\mathrm{11}}}}  \ottsym{/}  \mathit{x_{{\mathrm{1}}}}  ,\, ... \, ,   { \ottnt{e} }_{  \ottsym{1} \mathit{n}  }   \ottsym{/}  \mathit{x_{\ottmv{n}}}  ]$.

 \proof

 Similar to the case of term application.
\end{prop}

\begin{prop}[name=Compatibility: Abstraction]{fh-lr-comp-abs}
 If
 $ \Gamma  ,  \mathit{x}  \mathord{:}  \ottnt{T_{{\mathrm{11}}}}   \vdash  \ottnt{e_{{\mathrm{1}}}} \,  \mathrel{ \simeq }  \, \ottnt{e_{{\mathrm{2}}}}  \ottsym{:}  \ottnt{T_{{\mathrm{12}}}}$ and
 $ \mathit{FV}  (  \ottnt{T_{{\mathrm{21}}}}  )   \mathrel{\cup}   \mathit{FTV}  (  \ottnt{T_{{\mathrm{21}}}}  )   \subseteq   \mathit{dom}  (  \Gamma  ) $,
 then
 $\Gamma  \vdash    \lambda    \mathit{x}  \mathord{:}  \ottnt{T_{{\mathrm{11}}}}  .  \ottnt{e_{{\mathrm{1}}}}  \,  \mathrel{ \simeq }  \,   \lambda    \mathit{x}  \mathord{:}  \ottnt{T_{{\mathrm{21}}}}  .  \ottnt{e_{{\mathrm{2}}}}   \ottsym{:}  \ottsym{(}   \mathit{x} \mathord{:} \ottnt{T_{{\mathrm{11}}}} \rightarrow \ottnt{T_{{\mathrm{12}}}}   \ottsym{)}$.

 \proof

 Let $\Gamma  \vdash  \theta  \ottsym{;}  \delta$.
 By definition, it suffices to show that,
 for any $\ottnt{v_{{\mathrm{1}}}}$ and $\ottnt{v_{{\mathrm{2}}}}$ such that $ \ottnt{v_{{\mathrm{1}}}}  \simeq_{\mathtt{v} }  \ottnt{v_{{\mathrm{2}}}}   \ottsym{:}   \ottnt{T_{{\mathrm{11}}}} ;  \theta ;  \delta $,
 $  \theta_{{\mathrm{1}}}  (   \delta_{{\mathrm{1}}}  (    \lambda    \mathit{x}  \mathord{:}  \ottnt{T_{{\mathrm{11}}}}  .  \ottnt{e_{{\mathrm{1}}}}   )   )  \, \ottnt{v_{{\mathrm{1}}}}  \simeq_{\mathtt{e} }   \theta_{{\mathrm{2}}}  (   \delta_{{\mathrm{2}}}  (    \lambda    \mathit{x}  \mathord{:}  \ottnt{T_{{\mathrm{21}}}}  .  \ottnt{e_{{\mathrm{2}}}}   )   )  \, \ottnt{v_{{\mathrm{2}}}}   \ottsym{:}   \ottnt{T_{{\mathrm{12}}}} ;  \theta ;   \delta    [  \,  (  \ottnt{v_{{\mathrm{1}}}}  ,  \ottnt{v_{{\mathrm{2}}}}  ) /  \mathit{x}  \,  ]   $.
 Since $ \Gamma  ,  \mathit{x}  \mathord{:}  \ottnt{T_{{\mathrm{11}}}}   \vdash  \theta  \ottsym{;}   \delta    [  \,  (  \ottnt{v_{{\mathrm{1}}}}  ,  \ottnt{v_{{\mathrm{2}}}}  ) /  \mathit{x}  \,  ]  $
 by the weakening (\prop:ref{fh-lr-val-ws}), and
 $ \Gamma  ,  \mathit{x}  \mathord{:}  \ottnt{T_{{\mathrm{11}}}}   \vdash  \ottnt{e_{{\mathrm{1}}}} \,  \mathrel{ \simeq }  \, \ottnt{e_{{\mathrm{2}}}}  \ottsym{:}  \ottnt{T_{{\mathrm{12}}}}$, we finish.
\end{prop}

{}

\begin{prop}[name=Compatibility under Self-relatedness Assumption: Cast]
 {fh-lr-comp-cast-refl-assump}
 Suppose that $ \Gamma \vdash \ottnt{T_{{\mathrm{11}}}}  \mathrel{ \simeq }  \ottnt{T_{{\mathrm{11}}}}  : \ast $ and $ \Gamma \vdash \ottnt{T_{{\mathrm{12}}}}  \mathrel{ \simeq }  \ottnt{T_{{\mathrm{12}}}}  : \ast $.
 If $ \Gamma \vdash \ottnt{T_{{\mathrm{11}}}}  \mathrel{ \simeq }  \ottnt{T_{{\mathrm{21}}}}  : \ast $ and $ \Gamma \vdash \ottnt{T_{{\mathrm{12}}}}  \mathrel{ \simeq }  \ottnt{T_{{\mathrm{22}}}}  : \ast $ and $\ottnt{T_{{\mathrm{11}}}}  \mathrel{\parallel}  \ottnt{T_{{\mathrm{12}}}}$,
 then $\Gamma  \vdash  \langle  \ottnt{T_{{\mathrm{11}}}}  \Rightarrow  \ottnt{T_{{\mathrm{12}}}}  \rangle   ^{ \ell }  \,  \mathrel{ \simeq }  \, \langle  \ottnt{T_{{\mathrm{21}}}}  \Rightarrow  \ottnt{T_{{\mathrm{22}}}}  \rangle   ^{ \ell }   \ottsym{:}  \ottnt{T_{{\mathrm{11}}}}  \rightarrow  \ottnt{T_{{\mathrm{12}}}}$.

 \proof

 {\iffull
\else
 It suffices to show:
 \begin{quotation}
  If
  $\ottnt{T_{{\mathrm{11}}}}  \mathrel{\parallel}  \ottnt{T_{{\mathrm{12}}}}$,
  $ \ottnt{T_{{\mathrm{11}}}}  \simeq  \ottnt{T_{{\mathrm{11}}}}   \ottsym{:}   \ast ;  \theta ;  \delta $,
  $ \ottnt{T_{{\mathrm{12}}}}  \simeq  \ottnt{T_{{\mathrm{12}}}}   \ottsym{:}   \ast ;  \theta ;  \delta $,
  $ \ottnt{T_{{\mathrm{11}}}}  \simeq  \ottnt{T_{{\mathrm{21}}}}   \ottsym{:}   \ast ;  \theta ;  \delta $, and
  $ \ottnt{T_{{\mathrm{12}}}}  \simeq  \ottnt{T_{{\mathrm{22}}}}   \ottsym{:}   \ast ;  \theta ;  \delta $,
  then
  \[
     \theta_{{\mathrm{1}}}  (   \delta_{{\mathrm{1}}}  (  \langle  \ottnt{T_{{\mathrm{11}}}}  \Rightarrow  \ottnt{T_{{\mathrm{12}}}}  \rangle   ^{ \ell }   )   )   \simeq_{\mathtt{v} }   \theta_{{\mathrm{2}}}  (   \delta_{{\mathrm{2}}}  (  \langle  \ottnt{T_{{\mathrm{21}}}}  \Rightarrow  \ottnt{T_{{\mathrm{22}}}}  \rangle   ^{ \ell }   )   )    \ottsym{:}   \ottnt{T_{{\mathrm{11}}}}  \rightarrow  \ottnt{T_{{\mathrm{12}}}} ;  \theta ;  \delta .
  \]
 \end{quotation}
 We prove this by strong induction on the sum of sizes of $\ottnt{T_{{\mathrm{11}}}}$ and
 $\ottnt{T_{{\mathrm{12}}}}$ as elimination of reflexive casts (\prop:ref{fh-lr-elim-refl-cast});
 the details are omitted.
 \fi}
\end{prop}

\begin{prop}[name=Compatibility: Type Abstraction]{fh-lr-comp-tabs}
 If $\Gamma  \ottsym{,}  \alpha  \vdash  \ottnt{e_{{\mathrm{1}}}} \,  \mathrel{ \simeq }  \, \ottnt{e_{{\mathrm{2}}}}  \ottsym{:}  \ottnt{T}$,
 then $\Gamma  \vdash   \Lambda\!  \, \alpha  .~  \ottnt{e_{{\mathrm{1}}}} \,  \mathrel{ \simeq }  \,  \Lambda\!  \, \alpha  .~  \ottnt{e_{{\mathrm{2}}}}  \ottsym{:}   \forall   \alpha  .  \ottnt{T} $.

 \proof

 Let $\Gamma  \vdash  \theta  \ottsym{;}  \delta$.
 By definition, it suffices to show that,
 for any $\ottnt{r}$, $\ottnt{T_{{\mathrm{1}}}}$, and $\ottnt{T_{{\mathrm{2}}}}$ such that $\langle  \ottnt{r}  \ottsym{,}  \ottnt{T_{{\mathrm{1}}}}  \ottsym{,}  \ottnt{T_{{\mathrm{2}}}}  \rangle$,
 $  \theta_{{\mathrm{1}}}  (   \delta_{{\mathrm{1}}}  (   \Lambda\!  \, \alpha  .~  \ottnt{e_{{\mathrm{1}}}}  )   )  \, \ottnt{T_{{\mathrm{1}}}}  \simeq_{\mathtt{e} }   \theta_{{\mathrm{2}}}  (   \delta_{{\mathrm{2}}}  (   \Lambda\!  \, \alpha  .~  \ottnt{e_{{\mathrm{2}}}}  )   )  \, \ottnt{T_{{\mathrm{2}}}}   \ottsym{:}   \ottnt{T} ;  \theta \,  \{  \,  \alpha  \mapsto  \ottnt{r} , \ottnt{T_{{\mathrm{1}}}} , \ottnt{T_{{\mathrm{2}}}}  \,  \}  ;  \delta $.
 Since $\Gamma  \ottsym{,}  \alpha  \vdash  \theta \,  \{  \,  \alpha  \mapsto  \ottnt{r} , \ottnt{T_{{\mathrm{1}}}} , \ottnt{T_{{\mathrm{2}}}}  \,  \}   \ottsym{;}  \delta$
 by the weakening ({\prop:ref{fh-lr-typ-ws}}), and
 $\Gamma  \ottsym{,}  \alpha  \vdash  \ottnt{e_{{\mathrm{1}}}} \,  \mathrel{ \simeq }  \, \ottnt{e_{{\mathrm{2}}}}  \ottsym{:}  \ottnt{T}$, we finish.
\end{prop}

\begin{prop}[name=Compatibility: Type Conversion]{fh-lr-comp-conv}
 If $ \mathord{ \vdash } ~  \Gamma $ and $\emptyset  \vdash  \ottnt{e_{{\mathrm{1}}}} \,  \mathrel{ \simeq }  \, \ottnt{e_{{\mathrm{2}}}}  \ottsym{:}  \ottnt{T_{{\mathrm{1}}}}$ and
 $\emptyset  \vdash  \ottnt{T_{{\mathrm{2}}}}$ and $\ottnt{T_{{\mathrm{1}}}}  \equiv  \ottnt{T_{{\mathrm{2}}}}$,
 then $\Gamma  \vdash  \ottnt{e_{{\mathrm{1}}}} \,  \mathrel{ \simeq }  \, \ottnt{e_{{\mathrm{2}}}}  \ottsym{:}  \ottnt{T_{{\mathrm{2}}}}$.

 \proof

 It suffices to show that, if $\ottnt{T_{{\mathrm{1}}}}  \equiv  \ottnt{T_{{\mathrm{2}}}}$, then
 $ \ottnt{e_{{\mathrm{1}}}}  \simeq_{\mathtt{e} }  \ottnt{e_{{\mathrm{2}}}}   \ottsym{:}   \ottnt{T_{{\mathrm{1}}}} ;  \theta ;  \delta $ \text{ iff } $ \ottnt{e_{{\mathrm{1}}}}  \simeq_{\mathtt{e} }  \ottnt{e_{{\mathrm{2}}}}   \ottsym{:}   \ottnt{T_{{\mathrm{2}}}} ;  \theta ;  \delta $.
 We consider the case of $\ottnt{T_{{\mathrm{1}}}}  \Rrightarrow  \ottnt{T_{{\mathrm{2}}}}$
 (other cases are shown straightforwardly).
 There exist $\ottnt{T}$, $\ottnt{e'_{{\mathrm{1}}}}$, $\ottnt{e'_{{\mathrm{2}}}}$, and
 $\mathit{x}$ such that $\ottnt{T_{{\mathrm{1}}}}  \ottsym{=}  \ottnt{T} \, [  \ottnt{e'_{{\mathrm{1}}}}  \ottsym{/}  \mathit{x}  ]$ and $\ottnt{T_{{\mathrm{2}}}}  \ottsym{=}  \ottnt{T} \, [  \ottnt{e'_{{\mathrm{2}}}}  \ottsym{/}  \mathit{x}  ]$ and
 $\ottnt{e'_{{\mathrm{1}}}}  \longrightarrow  \ottnt{e'_{{\mathrm{2}}}}$.
 If $\ottnt{e_{{\mathrm{1}}}}$ and $\ottnt{e_{{\mathrm{2}}}}$ raise blame, then obvious.
 Otherwise,
 $\ottnt{e_{{\mathrm{1}}}}  \longrightarrow^{\ast}  \ottnt{v_{{\mathrm{1}}}}$ and $\ottnt{e_{{\mathrm{2}}}}  \longrightarrow^{\ast}  \ottnt{v_{{\mathrm{2}}}}$ for some $\ottnt{v_{{\mathrm{1}}}}$ and $\ottnt{v_{{\mathrm{2}}}}$, and
 it suffices to show that
 $ \ottnt{v_{{\mathrm{1}}}}  \simeq_{\mathtt{v} }  \ottnt{v_{{\mathrm{2}}}}   \ottsym{:}   \ottnt{T} \, [  \ottnt{e'_{{\mathrm{1}}}}  \ottsym{/}  \mathit{x}  ] ;  \theta ;  \delta $ iff 
 $ \ottnt{v_{{\mathrm{1}}}}  \simeq_{\mathtt{v} }  \ottnt{v_{{\mathrm{2}}}}   \ottsym{:}   \ottnt{T} \, [  \ottnt{e'_{{\mathrm{2}}}}  \ottsym{/}  \mathit{x}  ] ;  \theta ;  \delta $.
 Straightforward by induction on $\ottnt{T}$.
 The case that $\ottnt{T}$ is a refinement type is shown with
 Cotermination (\prop:ref{fh-coterm-true}).
\end{prop}

\begin{prop}[name=Compatibility under Self-relatedness Assumption: Active Check]
 {fh-lr-comp-acheck-refl-assump}
 Suppose that $ \mathord{ \vdash } ~  \Gamma $ and $ \emptyset \vdash  \{  \mathit{x}  \mathord{:}  \ottnt{T_{{\mathrm{1}}}}   \mathop{\mid}   \ottnt{e_{{\mathrm{1}}}}  \}   \mathrel{ \simeq }   \{  \mathit{x}  \mathord{:}  \ottnt{T_{{\mathrm{1}}}}   \mathop{\mid}   \ottnt{e_{{\mathrm{1}}}}  \}   : \ast $.
 If
 $\emptyset  \vdash  \ottnt{e'_{{\mathrm{1}}}} \,  \mathrel{ \simeq }  \, \ottnt{e'_{{\mathrm{2}}}}  \ottsym{:}   \mathsf{Bool} $ and
 $\emptyset  \vdash  \ottnt{v_{{\mathrm{1}}}} \,  \mathrel{ \simeq }  \, \ottnt{v_{{\mathrm{2}}}}  \ottsym{:}  \ottnt{T_{{\mathrm{1}}}}$ and
 $\ottnt{e_{{\mathrm{1}}}} \, [  \ottnt{v_{{\mathrm{1}}}}  \ottsym{/}  \mathit{x}  ]  \longrightarrow^{\ast}  \ottnt{e'_{{\mathrm{1}}}}$,
 then
 $\Gamma  \vdash  \langle   \{  \mathit{x}  \mathord{:}  \ottnt{T_{{\mathrm{1}}}}   \mathop{\mid}   \ottnt{e_{{\mathrm{1}}}}  \}   \ottsym{,}  \ottnt{e'_{{\mathrm{1}}}}  \ottsym{,}  \ottnt{v_{{\mathrm{1}}}}  \rangle   ^{ \ell }  \,  \mathrel{ \simeq }  \, \langle   \{  \mathit{x}  \mathord{:}  \ottnt{T_{{\mathrm{2}}}}   \mathop{\mid}   \ottnt{e_{{\mathrm{2}}}}  \}   \ottsym{,}  \ottnt{e'_{{\mathrm{2}}}}  \ottsym{,}  \ottnt{v_{{\mathrm{2}}}}  \rangle   ^{ \ell }   \ottsym{:}   \{  \mathit{x}  \mathord{:}  \ottnt{T_{{\mathrm{1}}}}   \mathop{\mid}   \ottnt{e_{{\mathrm{1}}}}  \} $.

 \proof

 It suffices to show that
 $ \langle   \{  \mathit{x}  \mathord{:}  \ottnt{T_{{\mathrm{1}}}}   \mathop{\mid}   \ottnt{e_{{\mathrm{1}}}}  \}   \ottsym{,}  \ottnt{e'_{{\mathrm{1}}}}  \ottsym{,}  \ottnt{v}  \rangle   ^{ \ell }   \simeq_{\mathtt{e} }  \langle   \{  \mathit{x}  \mathord{:}  \ottnt{T_{{\mathrm{2}}}}   \mathop{\mid}   \ottnt{e_{{\mathrm{2}}}}  \}   \ottsym{,}  \ottnt{e'_{{\mathrm{2}}}}  \ottsym{,}  \ottnt{v_{{\mathrm{2}}}}  \rangle   ^{ \ell }    \ottsym{:}    \{  \mathit{x}  \mathord{:}  \ottnt{T_{{\mathrm{1}}}}   \mathop{\mid}   \ottnt{e_{{\mathrm{1}}}}  \}  ;  \emptyset ;  \emptyset $.
 If $\ottnt{e'_{{\mathrm{1}}}}$ and $\ottnt{e'_{{\mathrm{2}}}}$ raise blame, then obvious.
 Otherwise,
 $\ottnt{e'_{{\mathrm{1}}}}  \longrightarrow^{\ast}  \ottnt{v'_{{\mathrm{1}}}}$ and $\ottnt{e'_{{\mathrm{2}}}}  \longrightarrow^{\ast}  \ottnt{v'_{{\mathrm{2}}}}$ for some $\ottnt{v'_{{\mathrm{1}}}}$ and $\ottnt{v'_{{\mathrm{2}}}}$, and
 it suffices to show that
 \[  \langle   \{  \mathit{x}  \mathord{:}  \ottnt{T_{{\mathrm{1}}}}   \mathop{\mid}   \ottnt{e_{{\mathrm{1}}}}  \}   \ottsym{,}  \ottnt{v'_{{\mathrm{1}}}}  \ottsym{,}  \ottnt{v_{{\mathrm{1}}}}  \rangle   ^{ \ell }   \simeq_{\mathtt{e} }  \langle   \{  \mathit{x}  \mathord{:}  \ottnt{T_{{\mathrm{2}}}}   \mathop{\mid}   \ottnt{e_{{\mathrm{2}}}}  \}   \ottsym{,}  \ottnt{v'_{{\mathrm{2}}}}  \ottsym{,}  \ottnt{v_{{\mathrm{2}}}}  \rangle   ^{ \ell }    \ottsym{:}    \{  \mathit{x}  \mathord{:}  \ottnt{T_{{\mathrm{1}}}}   \mathop{\mid}   \ottnt{e_{{\mathrm{1}}}}  \}  ;  \emptyset ;  \emptyset . \]
 Since $ \ottnt{v'_{{\mathrm{1}}}}  \simeq_{\mathtt{v} }  \ottnt{v'_{{\mathrm{2}}}}   \ottsym{:}    \mathsf{Bool}  ;  \emptyset ;  \emptyset $,
 there are two cases we have to consider.
 If $\ottnt{v'_{{\mathrm{1}}}}  \ottsym{=}  \ottnt{v'_{{\mathrm{2}}}}  \ottsym{=}   \mathsf{false} $, then
 $\langle   \{  \mathit{x}  \mathord{:}  \ottnt{T_{{\mathrm{1}}}}   \mathop{\mid}   \ottnt{e_{{\mathrm{1}}}}  \}   \ottsym{,}  \ottnt{v'_{{\mathrm{1}}}}  \ottsym{,}  \ottnt{v_{{\mathrm{1}}}}  \rangle   ^{ \ell }   \longrightarrow   \mathord{\Uparrow}  \ell $ and
 $\langle   \{  \mathit{x}  \mathord{:}  \ottnt{T_{{\mathrm{2}}}}   \mathop{\mid}   \ottnt{e_{{\mathrm{2}}}}  \}   \ottsym{,}  \ottnt{v'_{{\mathrm{2}}}}  \ottsym{,}  \ottnt{v_{{\mathrm{2}}}}  \rangle   ^{ \ell }   \longrightarrow   \mathord{\Uparrow}  \ell $, and so
 we finish.
 Otherwise, if $\ottnt{v'_{{\mathrm{1}}}}  \ottsym{=}  \ottnt{v'_{{\mathrm{2}}}}  \ottsym{=}   \mathsf{true} $, then
 $\langle   \{  \mathit{x}  \mathord{:}  \ottnt{T_{{\mathrm{1}}}}   \mathop{\mid}   \ottnt{e_{{\mathrm{1}}}}  \}   \ottsym{,}  \ottnt{v'_{{\mathrm{1}}}}  \ottsym{,}  \ottnt{v_{{\mathrm{1}}}}  \rangle   ^{ \ell }   \longrightarrow  \ottnt{v_{{\mathrm{1}}}}$ and
 $\langle   \{  \mathit{x}  \mathord{:}  \ottnt{T_{{\mathrm{2}}}}   \mathop{\mid}   \ottnt{e_{{\mathrm{2}}}}  \}   \ottsym{,}  \ottnt{v'_{{\mathrm{2}}}}  \ottsym{,}  \ottnt{v_{{\mathrm{2}}}}  \rangle   ^{ \ell }   \longrightarrow  \ottnt{v_{{\mathrm{2}}}}$.
 Thus, it suffices to show that $ \ottnt{v_{{\mathrm{1}}}}  \simeq_{\mathtt{v} }  \ottnt{v_{{\mathrm{2}}}}   \ottsym{:}    \{  \mathit{x}  \mathord{:}  \ottnt{T_{{\mathrm{1}}}}   \mathop{\mid}   \ottnt{e_{{\mathrm{1}}}}  \}  ;  \emptyset ;  \emptyset $,
 that is,
 (1) $ \ottnt{v_{{\mathrm{1}}}}  \simeq_{\mathtt{v} }  \ottnt{v_{{\mathrm{2}}}}   \ottsym{:}   \ottnt{T_{{\mathrm{1}}}} ;  \emptyset ;  \emptyset $,
 (2) $\ottnt{e_{{\mathrm{1}}}} \, [  \ottnt{v_{{\mathrm{1}}}}  \ottsym{/}  \mathit{x}  ]  \longrightarrow^{\ast}   \mathsf{true} $, and
 (3) $\ottnt{e_{{\mathrm{1}}}} \, [  \ottnt{v_{{\mathrm{2}}}}  \ottsym{/}  \mathit{x}  ]  \longrightarrow^{\ast}   \mathsf{true} $.
 We have
 $ \ottnt{v_{{\mathrm{1}}}}  \simeq_{\mathtt{v} }  \ottnt{v_{{\mathrm{2}}}}   \ottsym{:}   \ottnt{T_{{\mathrm{1}}}} ;  \emptyset ;  \emptyset $ by the assumption, and
 $\ottnt{e_{{\mathrm{1}}}} \, [  \ottnt{v_{{\mathrm{1}}}}  \ottsym{/}  \mathit{x}  ]  \longrightarrow^{\ast}  \ottnt{e'_{{\mathrm{1}}}}  \longrightarrow^{\ast}   \mathsf{true} $.
 Since $  \{  \mathit{x}  \mathord{:}  \ottnt{T_{{\mathrm{1}}}}   \mathop{\mid}   \ottnt{e_{{\mathrm{1}}}}  \}   \simeq   \{  \mathit{x}  \mathord{:}  \ottnt{T_{{\mathrm{1}}}}   \mathop{\mid}   \ottnt{e_{{\mathrm{1}}}}  \}    \ottsym{:}   \ast ;  \emptyset ;  \emptyset $,
 we have
 $ \ottnt{e_{{\mathrm{1}}}} \, [  \ottnt{v_{{\mathrm{1}}}}  \ottsym{/}  \mathit{x}  ]  \simeq_{\mathtt{e} }  \ottnt{e_{{\mathrm{1}}}} \, [  \ottnt{v_{{\mathrm{2}}}}  \ottsym{/}  \mathit{x}  ]   \ottsym{:}    \mathsf{Bool}  ;  \emptyset ;  \emptyset $.
 Thus, $\ottnt{e_{{\mathrm{1}}}} \, [  \ottnt{v_{{\mathrm{2}}}}  \ottsym{/}  \mathit{x}  ]  \longrightarrow^{\ast}   \mathsf{true} $.
\end{prop}

\begin{prop}[name=Compatibility under Self-relatedness Assumption: Waiting Check]
 {fh-lr-comp-wcheck-refl-assump}
 Suppose that $ \Gamma \vdash  \{  \mathit{x}  \mathord{:}  \ottnt{T_{{\mathrm{1}}}}   \mathop{\mid}   \ottnt{e_{{\mathrm{1}}}}  \}   \mathrel{ \simeq }   \{  \mathit{x}  \mathord{:}  \ottnt{T_{{\mathrm{1}}}}   \mathop{\mid}   \ottnt{e_{{\mathrm{1}}}}  \}   : \ast $.
 If
 $ \Gamma \vdash  \{  \mathit{x}  \mathord{:}  \ottnt{T_{{\mathrm{1}}}}   \mathop{\mid}   \ottnt{e_{{\mathrm{1}}}}  \}   \mathrel{ \simeq }   \{  \mathit{x}  \mathord{:}  \ottnt{T_{{\mathrm{2}}}}   \mathop{\mid}   \ottnt{e_{{\mathrm{2}}}}  \}   : \ast $ and
 $\Gamma  \vdash  \ottnt{e'_{{\mathrm{1}}}} \,  \mathrel{ \simeq }  \, \ottnt{e'_{{\mathrm{2}}}}  \ottsym{:}  \ottnt{T_{{\mathrm{1}}}}$,
 then
 $\Gamma  \vdash   \langle\!\langle  \,  \{  \mathit{x}  \mathord{:}  \ottnt{T_{{\mathrm{1}}}}   \mathop{\mid}   \ottnt{e_{{\mathrm{1}}}}  \}   \ottsym{,}  \ottnt{e'_{{\mathrm{1}}}} \,  \rangle\!\rangle  \,  ^{ \ell }  \,  \mathrel{ \simeq }  \,  \langle\!\langle  \,  \{  \mathit{x}  \mathord{:}  \ottnt{T_{{\mathrm{2}}}}   \mathop{\mid}   \ottnt{e_{{\mathrm{2}}}}  \}   \ottsym{,}  \ottnt{e'_{{\mathrm{2}}}} \,  \rangle\!\rangle  \,  ^{ \ell }   \ottsym{:}   \{  \mathit{x}  \mathord{:}  \ottnt{T_{{\mathrm{1}}}}   \mathop{\mid}   \ottnt{e_{{\mathrm{1}}}}  \} $.

 \proof

 It suffices to show that, if
 $  \{  \mathit{x}  \mathord{:}  \ottnt{T_{{\mathrm{1}}}}   \mathop{\mid}   \ottnt{e_{{\mathrm{1}}}}  \}   \simeq   \{  \mathit{x}  \mathord{:}  \ottnt{T_{{\mathrm{1}}}}   \mathop{\mid}   \ottnt{e_{{\mathrm{1}}}}  \}    \ottsym{:}   \ast ;  \theta ;  \delta $ and
 $  \{  \mathit{x}  \mathord{:}  \ottnt{T_{{\mathrm{1}}}}   \mathop{\mid}   \ottnt{e_{{\mathrm{1}}}}  \}   \simeq   \{  \mathit{x}  \mathord{:}  \ottnt{T_{{\mathrm{2}}}}   \mathop{\mid}   \ottnt{e_{{\mathrm{2}}}}  \}    \ottsym{:}   \ast ;  \theta ;  \delta $ and
 $ \ottnt{e'_{{\mathrm{1}}}}  \simeq_{\mathtt{e} }  \ottnt{e'_{{\mathrm{2}}}}   \ottsym{:}   \ottnt{T_{{\mathrm{1}}}} ;  \theta ;  \delta $,
 then
 $  \langle\!\langle  \,  \theta_{{\mathrm{1}}}  (   \delta_{{\mathrm{1}}}  (   \{  \mathit{x}  \mathord{:}  \ottnt{T_{{\mathrm{1}}}}   \mathop{\mid}   \ottnt{e_{{\mathrm{1}}}}  \}   )   )   \ottsym{,}  \ottnt{e'_{{\mathrm{1}}}} \,  \rangle\!\rangle  \,  ^{ \ell }   \simeq_{\mathtt{e} }   \langle\!\langle  \,  \theta_{{\mathrm{2}}}  (   \delta_{{\mathrm{2}}}  (   \{  \mathit{x}  \mathord{:}  \ottnt{T_{{\mathrm{2}}}}   \mathop{\mid}   \ottnt{e_{{\mathrm{2}}}}  \}   )   )   \ottsym{,}  \ottnt{e'_{{\mathrm{2}}}} \,  \rangle\!\rangle  \,  ^{ \ell }    \ottsym{:}    \{  \mathit{x}  \mathord{:}  \ottnt{T_{{\mathrm{1}}}}   \mathop{\mid}   \ottnt{e_{{\mathrm{1}}}}  \}  ;  \theta ;  \delta $.
 If $\ottnt{e'_{{\mathrm{1}}}}$ and $\ottnt{e'_{{\mathrm{2}}}}$ raise blame, then obvious.
 Otherwise, $\ottnt{e'_{{\mathrm{1}}}}  \longrightarrow^{\ast}  \ottnt{v'_{{\mathrm{1}}}}$ and $\ottnt{e'_{{\mathrm{2}}}}  \longrightarrow^{\ast}  \ottnt{v'_{{\mathrm{2}}}}$
 for some $\ottnt{v'_{{\mathrm{1}}}}$ and $\ottnt{v'_{{\mathrm{2}}}}$, and it suffices to show that
 \[
    \theta_{{\mathrm{1}}}  (   \delta_{{\mathrm{1}}}  (  \langle   \{  \mathit{x}  \mathord{:}  \ottnt{T_{{\mathrm{1}}}}   \mathop{\mid}   \ottnt{e_{{\mathrm{1}}}}  \}   \ottsym{,}  \ottnt{e_{{\mathrm{1}}}} \, [  \ottnt{v'_{{\mathrm{1}}}}  \ottsym{/}  \mathit{x}  ]  \ottsym{,}  \ottnt{v'_{{\mathrm{1}}}}  \rangle   ^{ \ell }   )   )   \simeq_{\mathtt{e} }   \theta_{{\mathrm{2}}}  (   \delta_{{\mathrm{2}}}  (  \langle   \{  \mathit{x}  \mathord{:}  \ottnt{T_{{\mathrm{2}}}}   \mathop{\mid}   \ottnt{e_{{\mathrm{2}}}}  \}   \ottsym{,}  \ottnt{e_{{\mathrm{2}}}} \, [  \ottnt{v'_{{\mathrm{2}}}}  \ottsym{/}  \mathit{x}  ]  \ottsym{,}  \ottnt{v'_{{\mathrm{2}}}}  \rangle   ^{ \ell }   )   )    \ottsym{:}    \{  \mathit{x}  \mathord{:}  \ottnt{T_{{\mathrm{1}}}}   \mathop{\mid}   \ottnt{e_{{\mathrm{1}}}}  \}  ;  \theta ;  \delta .
 \]
 We have $ \ottnt{v'_{{\mathrm{1}}}}  \simeq_{\mathtt{v} }  \ottnt{v'_{{\mathrm{2}}}}   \ottsym{:}   \ottnt{T_{{\mathrm{1}}}} ;  \theta ;  \delta $.
 Since $  \{  \mathit{x}  \mathord{:}  \ottnt{T_{{\mathrm{1}}}}   \mathop{\mid}   \ottnt{e_{{\mathrm{1}}}}  \}   \simeq   \{  \mathit{x}  \mathord{:}  \ottnt{T_{{\mathrm{2}}}}   \mathop{\mid}   \ottnt{e_{{\mathrm{2}}}}  \}    \ottsym{:}   \ast ;  \theta ;  \delta $,
 we have $  \theta_{{\mathrm{1}}}  (   \delta_{{\mathrm{1}}}  (  \ottnt{e_{{\mathrm{1}}}} \, [  \ottnt{v'_{{\mathrm{1}}}}  \ottsym{/}  \mathit{x}  ]  )   )   \simeq_{\mathtt{e} }   \theta_{{\mathrm{2}}}  (   \delta_{{\mathrm{2}}}  (  \ottnt{e_{{\mathrm{2}}}} \, [  \ottnt{v'_{{\mathrm{2}}}}  \ottsym{/}  \mathit{x}  ]  )   )    \ottsym{:}    \mathsf{Bool}  ;  \theta ;  \delta $.
 The remaining proceeds as in active check
 (\prop:ref{fh-lr-comp-acheck-refl-assump}).
\end{prop}

\begin{prop}[name=Compatibility: Exact]{fh-lr-comp-exact-refl-assump}
 Suppose that $ \mathord{ \vdash } ~  \Gamma $ and $ \emptyset \vdash  \{  \mathit{x}  \mathord{:}  \ottnt{T}   \mathop{\mid}   \ottnt{e}  \}   \mathrel{ \simeq }   \{  \mathit{x}  \mathord{:}  \ottnt{T}   \mathop{\mid}   \ottnt{e}  \}   : \ast $.
 If
 $\emptyset  \vdash  \ottnt{v_{{\mathrm{1}}}} \,  \mathrel{ \simeq }  \, \ottnt{v_{{\mathrm{2}}}}  \ottsym{:}  \ottnt{T}$ and
 $\ottnt{e} \, [  \ottnt{v_{{\mathrm{1}}}}  \ottsym{/}  \mathit{x}  ]  \longrightarrow^{\ast}   \mathsf{true} $,
 then
 $\Gamma  \vdash  \ottnt{v_{{\mathrm{1}}}} \,  \mathrel{ \simeq }  \, \ottnt{v_{{\mathrm{2}}}}  \ottsym{:}   \{  \mathit{x}  \mathord{:}  \ottnt{T}   \mathop{\mid}   \ottnt{e}  \} $.

 \proof

 By the weakening (\prop:ref{fh-lr-val-ws,fh-lr-typ-ws}),
 it suffices to show that
 $ \ottnt{v_{{\mathrm{1}}}}  \simeq_{\mathtt{v} }  \ottnt{v_{{\mathrm{2}}}}   \ottsym{:}    \{  \mathit{x}  \mathord{:}  \ottnt{T}   \mathop{\mid}   \ottnt{e}  \}  ;  \emptyset ;  \emptyset $.
 Since $ \ottnt{v_{{\mathrm{1}}}}  \simeq_{\mathtt{v} }  \ottnt{v_{{\mathrm{2}}}}   \ottsym{:}   \ottnt{T} ;  \emptyset ;  \emptyset $ and
 $\ottnt{e} \, [  \ottnt{v_{{\mathrm{1}}}}  \ottsym{/}  \mathit{x}  ]  \longrightarrow^{\ast}   \mathsf{true} $,
 it suffices to show that $\ottnt{e} \, [  \ottnt{v_{{\mathrm{2}}}}  \ottsym{/}  \mathit{x}  ]  \longrightarrow^{\ast}   \mathsf{true} $.
 Since $ \mathit{x}  \mathord{:}  \ottnt{T}   \vdash  \emptyset  \ottsym{;}   [  \,  (  \ottnt{v_{{\mathrm{1}}}}  ,  \ottnt{v_{{\mathrm{2}}}}  ) /  \mathit{x}  \,  ] $ and
 $ \mathit{x}  \mathord{:}  \ottnt{T}   \vdash  \ottnt{e} \,  \mathrel{ \simeq }  \, \ottnt{e}  \ottsym{:}   \mathsf{Bool} $,
 we have $ \ottnt{e} \, [  \ottnt{v_{{\mathrm{1}}}}  \ottsym{/}  \mathit{x}  ]  \simeq_{\mathtt{e} }  \ottnt{e} \, [  \ottnt{v_{{\mathrm{2}}}}  \ottsym{/}  \mathit{x}  ]   \ottsym{:}    \mathsf{Bool}  ;  \emptyset ;   [  \,  (  \ottnt{v_{{\mathrm{1}}}}  ,  \ottnt{v_{{\mathrm{2}}}}  ) /  \mathit{x}  \,  ]  $.
 Since $\ottnt{e} \, [  \ottnt{v_{{\mathrm{1}}}}  \ottsym{/}  \mathit{x}  ]  \longrightarrow^{\ast}   \mathsf{true} $,
 we have $\ottnt{e} \, [  \ottnt{v_{{\mathrm{2}}}}  \ottsym{/}  \mathit{x}  ]  \longrightarrow^{\ast}   \mathsf{true} $.
\end{prop}

\begin{prop}[name=Compatibility: Forget]{fh-lr-comp-forget}
 \label{lem:fh-lr-other-end}
 If $ \mathord{ \vdash } ~  \Gamma $ and $\emptyset  \vdash  \ottnt{v_{{\mathrm{1}}}} \,  \mathrel{ \simeq }  \, \ottnt{v_{{\mathrm{2}}}}  \ottsym{:}   \{  \mathit{x}  \mathord{:}  \ottnt{T}   \mathop{\mid}   \ottnt{e}  \} $,
 then $\Gamma  \vdash  \ottnt{v_{{\mathrm{1}}}} \,  \mathrel{ \simeq }  \, \ottnt{v_{{\mathrm{2}}}}  \ottsym{:}  \ottnt{T}$.

 \proof

 Straightforward by definition.
\end{prop}

\paragraph{\bf Parametricity}
\label{sec:proving-param}
Before showing the parametricity, we prove that the logical relation for open types
is closed under type substitution, which is needed to show that $\Gamma  \vdash  \ottnt{e}  \ottsym{:}  \ottnt{T}$
implies $ \Gamma \vdash \ottnt{T}  \mathrel{ \simeq }  \ottnt{T}  : \ast $.
{}
\begin{prop}
 [name=Type Substitutivity in Type under Self-relatedness Assumption]
 {fh-lr-typ-subst-typ-refl-assump}
 \label{lem:fh-lr-typ-subst-typ-refl-assump}
 Suppose that $\Gamma  \ottsym{,}  \alpha  \ottsym{,}  \Gamma'$ is self-related and that
 $ \Gamma  \ottsym{,}  \alpha  \ottsym{,}  \Gamma' \vdash \ottnt{T_{{\mathrm{11}}}}  \mathrel{ \simeq }  \ottnt{T_{{\mathrm{11}}}}  : \ast $ and $ \Gamma \vdash \ottnt{T_{{\mathrm{12}}}}  \mathrel{ \simeq }  \ottnt{T_{{\mathrm{12}}}}  : \ast $.
 If $ \Gamma  \ottsym{,}  \alpha  \ottsym{,}  \Gamma' \vdash \ottnt{T_{{\mathrm{11}}}}  \mathrel{ \simeq }  \ottnt{T_{{\mathrm{21}}}}  : \ast $ and $ \Gamma \vdash \ottnt{T_{{\mathrm{12}}}}  \mathrel{ \simeq }  \ottnt{T_{{\mathrm{22}}}}  : \ast $.,
 then $ \Gamma  \ottsym{,}  \Gamma'  [  \ottnt{T_{{\mathrm{12}}}}  \ottsym{/}  \alpha  ] \vdash \ottnt{T_{{\mathrm{11}}}} \, [  \ottnt{T_{{\mathrm{12}}}}  \ottsym{/}  \alpha  ]  \mathrel{ \simeq }  \ottnt{T_{{\mathrm{21}}}} \, [  \ottnt{T_{{\mathrm{22}}}}  \ottsym{/}  \alpha  ]  : \ast $.

 \proof

 By induction on $\ottnt{T_{{\mathrm{11}}}}$.
 Let $\Gamma  \ottsym{,}  \Gamma'  [  \ottnt{T_{{\mathrm{12}}}}  \ottsym{/}  \alpha  ]  \vdash  \theta  \ottsym{;}  \delta$.
 We show that
 \[
   \ottnt{T_{{\mathrm{11}}}} \, [  \ottnt{T_{{\mathrm{12}}}}  \ottsym{/}  \alpha  ]  \simeq  \ottnt{T_{{\mathrm{21}}}} \, [  \ottnt{T_{{\mathrm{22}}}}  \ottsym{/}  \alpha  ]   \ottsym{:}   \ast ;  \theta ;  \delta .
 \]
 Let
 $\ottnt{r} = \{ \ottsym{(}  \ottnt{v_{{\mathrm{1}}}}  \ottsym{,}  \ottnt{v_{{\mathrm{2}}}}  \ottsym{)} \mid  \ottnt{v_{{\mathrm{1}}}}  \simeq_{\mathtt{v} }  \ottnt{v_{{\mathrm{2}}}}   \ottsym{:}   \ottnt{T_{{\mathrm{12}}}} ;  \theta ;  \delta  \}$ and
 $\theta'  \ottsym{=}  \theta \,  \{  \,  \alpha  \mapsto  \ottnt{r} ,  \theta_{{\mathrm{1}}}  (   \delta_{{\mathrm{1}}}  (  \ottnt{T_{{\mathrm{12}}}}  )   )  ,  \theta_{{\mathrm{2}}}  (   \delta_{{\mathrm{2}}}  (  \ottnt{T_{{\mathrm{22}}}}  )   )   \,  \} $.
 {\iffull
\else
 Since $\Gamma  \ottsym{,}  \Gamma'  [  \ottnt{T_{{\mathrm{12}}}}  \ottsym{/}  \alpha  ]  \vdash  \theta  \ottsym{;}  \delta$,
 we have $\Gamma  \ottsym{,}  \alpha  \ottsym{,}  \Gamma'  \vdash  \theta'  \ottsym{;}  \delta$; it is shown by
 induction on $\Gamma'$.
 \fi}
 Thus, $ \ottnt{T_{{\mathrm{11}}}}  \simeq  \ottnt{T_{{\mathrm{21}}}}   \ottsym{:}   \ast ;  \theta' ;  \delta $.
 The remaining is straightforward by case analysis on the derivation of
 $ \ottnt{T_{{\mathrm{11}}}}  \simeq  \ottnt{T_{{\mathrm{21}}}}   \ottsym{:}   \ast ;  \theta' ;  \delta $;
 we need the type compositionality (\prop:ref{fh-lr-typ-comp})
 in the case that both $\ottnt{T_{{\mathrm{11}}}}$ and $\ottnt{T_{{\mathrm{21}}}}$ are refinement types.
\end{prop}

\begin{prop}[type=thm,name=Parametricity]{fh-lr-param}
 \noindent
 \begin{statements}
  \item(term) If $\Gamma  \vdash  \ottnt{e}  \ottsym{:}  \ottnt{T}$,
    then $\Gamma  \vdash  \ottnt{e} \,  \mathrel{ \simeq }  \, \ottnt{e}  \ottsym{:}  \ottnt{T}$ and $ \Gamma \vdash \ottnt{T}  \mathrel{ \simeq }  \ottnt{T}  : \ast $ and $\Gamma$ is self-related.
  \item(type) If $\Gamma  \vdash  \ottnt{T}$,
    then $ \Gamma \vdash \ottnt{T}  \mathrel{ \simeq }  \ottnt{T}  : \ast $ and $\Gamma$ is self-related.
  \item(tctx) If $ \mathord{ \vdash } ~  \Gamma $,
    then $\Gamma$ is self-related.
 \end{statements}

 \proof

 The three statements are simultaneously proved by induction on the derivations of the judgments with the compatibility lemmas shown above.
 The case of \T{TApp} uses \prop:ref{fh-lr-typ-subst-typ-refl-assump}.
 In the case of \T{App}, we can show $ \Gamma \vdash \ottnt{T}  \mathrel{ \simeq }  \ottnt{T}  : \ast $ by the IH since
 \T{App} has premise $\Gamma  \vdash  \ottnt{T}$.
\end{prop}

{}

\paragraph{\bf Soundness}
\label{sec:proving-sound}

By the parametricity, we can discharge self-relatedness assumptions from the compatibility
lemmas, which leads to the fundamental property, and so we are ready
to show the soundness of the logical relation.
\begin{prop}[name=Adequacy]{fh-lr-behav-rel-term}
 If $ \ottnt{e_{{\mathrm{1}}}}  \simeq_{\mathtt{e} }  \ottnt{e_{{\mathrm{2}}}}   \ottsym{:}   \ottnt{T} ;  \theta ;  \delta $,
 then $\ottnt{e_{{\mathrm{1}}}}  \Downarrow  \ottnt{e_{{\mathrm{2}}}}$.

 \proof

 Obvious.
\end{prop}

\fhlrsound*
\begin{proof}
 {\iffull
\else
 We can show that, for any $\ottnt{C}$ and $\ottnt{T}$ such that
 $\emptyset  \vdash  \ottnt{C}  \ottsym{:}   \overline{ \Gamma_{\ottmv{i}}  \vdash   { \ottnt{e} }_{  \ottsym{1} \ottmv{i}  }   \ottsym{:}  \ottnt{T_{\ottmv{i}}} }^{ \ottmv{i} }   \mathrel{\circ\hspace{-.4em}\rightarrow}  \ottnt{T}$,
 \[
  \emptyset  \vdash  \ottnt{C}  [  \ottnt{e_{{\mathrm{11}}}}  ,\, ... \, ,   { \ottnt{e} }_{  \ottsym{1} \mathit{n}  }   ] \,  \mathrel{ \simeq }  \, \ottnt{C}  [  \ottnt{e_{{\mathrm{21}}}}  ,\, ... \, ,   { \ottnt{e} }_{  \ottsym{2} \mathit{n}  }   ]  \ottsym{:}  \ottnt{T}
 \]
 using the compatibility lemmas with the parametricity (\prop:ref{fh-lr-param}).
 Then,
 \[
  \ottnt{C}  [  \ottnt{e_{{\mathrm{11}}}}  ,\, ... \, ,   { \ottnt{e} }_{  \ottsym{1} \mathit{n}  }   ]  \Downarrow  \ottnt{C}  [  \ottnt{e_{{\mathrm{21}}}}  ,\, ... \, ,   { \ottnt{e} }_{  \ottsym{2} \mathit{n}  }   ]
 \]
 by the adequacy (\prop:ref{fh-lr-behav-rel-term}).
 \fi}
\end{proof}

 \subsection{Completeness}
 \label{sec:logical_relation-complete}
 We also show the completeness of the logical relation with respect to \emph{typed}
 contextual equivalence, that is, contextually equivalent terms are logically
 related if they are both well typed.
 The completeness proof is via CIU-equivalence: we show that (1) contextually
 equivalent terms are CIU-equivalent (\prop:ref{fh-lr-ciu-complete}) and (2)
 \emph{well-typed}, CIU-equivalent terms are logically related
 (\prop:ref{fh-lr-ciu-sound-typed}).
 Using these lemmas, we can show that well-typed, contextually equivalent terms
 are logically related (\refthm{fh-lr-complete-typed}).
 The completeness enables us to show (restricted) transitivity of semityped
 contextual equivalence (\prop:ref{fh-lr-ctx-trans}).

 To prove CIU-equivalence of contextually equivalent terms, we start with
 defining functions to close open terms according to typing contexts and
 closing substitutions.
 These functions are used to construct program contexts in semityped contextual
 equivalence.
 \begin{defi}
 For $\ottnt{e}$ and $\Gamma$, $ \mathit{Abs}  (  \Gamma ;  \ottnt{e}  ) $ denotes a term
 that takes term and type variables bound in $\Gamma$ as arguments:
 \[\begin{array}{lcl}
   \mathit{Abs}  (  \emptyset ;  \ottnt{e}  )  &=& \ottnt{e} \\
   \mathit{Abs}  (   \Gamma  ,  \mathit{x}  \mathord{:}  \ottnt{T}  ;  \ottnt{e}  )  &=&  \mathit{Abs}  (  \Gamma ;    \lambda    \mathit{x}  \mathord{:}  \ottnt{T}  .  \ottnt{e}   )  \\
   \mathit{Abs}  (  \Gamma  \ottsym{,}  \alpha ;  \ottnt{e}  )    &=&  \mathit{Abs}  (  \Gamma ;   \Lambda\!  \, \alpha  .~  \ottnt{e}  ) 
 \end{array}\]
 For $\sigma$ and $\Gamma$, $ \mathit{App}  (  \Gamma ;  \sigma ;  \ottnt{e}  ) $
 denotes a term that is applied to values and types to which $\sigma$ maps:
 \[\begin{array}{lcl}
   \mathit{App}  (  \emptyset ;  \sigma ;  \ottnt{e}  )  &=& \ottnt{e} \\
   \mathit{App}  (   \Gamma  ,  \mathit{x}  \mathord{:}  \ottnt{T}  ;  \sigma ;  \ottnt{e}  )  &=&  \mathit{App}  (  \Gamma ;  \sigma ;  \ottnt{e}  )  \, \ottsym{(}   \sigma  (  \mathit{x}  )   \ottsym{)} \\
   \mathit{App}  (  \Gamma  \ottsym{,}  \alpha ;  \sigma ;  \ottnt{e}  )    &=&  \mathit{App}  (  \Gamma ;  \sigma ;  \ottnt{e}  )  \, \ottsym{(}   \sigma  (  \alpha  )   \ottsym{)}
 \end{array}\]
 For $\ottnt{T}$ and $\Gamma$, $ \mathit{AbsType}  (  \Gamma ;  \ottnt{T}  ) $
 denotes a type that abstracts variables bound in $\Gamma$.
 \[\begin{array}{lcl}
   \mathit{AbsType}  (  \emptyset ;  \ottnt{T}  )  &=&  \emptyset  \\
   \mathit{AbsType}  (   \Gamma  ,  \mathit{x}  \mathord{:}  \ottnt{T'}  ;  \ottnt{T}  )  &=&  \mathit{AbsType}  (  \Gamma ;   \mathit{x} \mathord{:} \ottnt{T'} \rightarrow \ottnt{T}   )  \\
   \mathit{AbsType}  (  \Gamma  \ottsym{,}  \alpha ;  \ottnt{T}  )    &=&  \mathit{AbsType}  (  \Gamma ;   \forall   \alpha  .  \ottnt{T}   ) 
 \end{array}\]
\end{defi}
\begin{prop}{fh-lr-comp-app-abs-typed}
 \noindent
 \begin{statements}
  \item(abs) If $\Gamma  \vdash  \ottnt{e}  \ottsym{:}  \ottnt{T}$, then $\emptyset  \vdash   \mathit{Abs}  (  \Gamma ;  \ottnt{e}  )   \ottsym{:}   \mathit{AbsType}  (  \Gamma ;  \ottnt{T}  ) $.
  \item(app) If $\emptyset  \vdash  \ottnt{e}  \ottsym{:}   \mathit{AbsType}  (  \Gamma ;  \ottnt{T}  ) $ and $\Gamma  \vdash  \sigma$,
   then $\emptyset  \vdash   \mathit{App}  (  \Gamma ;  \sigma ;  \ottnt{e}  )   \ottsym{:}   \sigma  (  \ottnt{T}  ) $.
  \item(abs-app) If $\Gamma  \vdash  \ottnt{e}  \ottsym{:}  \ottnt{T}$ and $\Gamma  \vdash  \sigma$,
   then $\emptyset  \vdash   \mathit{App}  (  \Gamma ;  \sigma ;   \mathit{Abs}  (  \Gamma ;  \ottnt{e}  )   )   \ottsym{:}   \sigma  (  \ottnt{T}  ) $.
 \end{statements}

 \proof

 The first and second cases are shown by induction on $\Gamma$ straightforwardly.
 The third case is a corollary of the combination of the first and second cases.
\end{prop}

\begin{prop}[name={$ =_\mathsf{ctx}  \mathbin{\subseteq}  =_\mathsf{ciu} $}]{fh-lr-ciu-complete}
 If $ \overline{ \Gamma_{\ottmv{i}}  \vdash   { \ottnt{e} }_{  \ottsym{1} \ottmv{i}  }  \, =_\mathsf{ctx} \,  { \ottnt{e} }_{  \ottsym{2} \ottmv{i}  }   \ottsym{:}  \ottnt{T_{\ottmv{i}}} }^{ \ottmv{i} } $,
 then $\Gamma_{\ottmv{j}}  \vdash   { \ottnt{e} }_{  \ottsym{1} \ottmv{j}  }  \, =_\mathsf{ciu} \,  { \ottnt{e} }_{  \ottsym{2} \ottmv{j}  }   \ottsym{:}  \ottnt{T_{\ottmv{j}}}$ for any $\ottmv{j}$.

 \proof

 We show that
 if $\Gamma  \vdash  \ottnt{e_{{\mathrm{1}}}} \, =_\mathsf{ctx} \, \ottnt{e_{{\mathrm{2}}}}  \ottsym{:}  \ottnt{T}$, then $\Gamma  \vdash  \ottnt{e_{{\mathrm{1}}}} \, =_\mathsf{ciu} \, \ottnt{e_{{\mathrm{2}}}}  \ottsym{:}  \ottnt{T}$.
 By definition, it suffices to show that,
 for any $\sigma$ and $\ottnt{E}^\ottnt{S}$ such that
 $\Gamma  \vdash  \sigma$ and $\emptyset  \vdash  \ottnt{E}^\ottnt{S}  \ottsym{:}  \ottsym{(}  \emptyset  \vdash   \sigma  (  \ottnt{e_{{\mathrm{1}}}}  )   \ottsym{:}   \sigma  (  \ottnt{T}  )   \ottsym{)}  \mathrel{\circ\hspace{-.4em}\rightarrow}  \ottnt{T'}$,
 \[ \ottnt{E}^\ottnt{S}  [   \sigma  (  \ottnt{e_{{\mathrm{1}}}}  )   ]  \Downarrow  \ottnt{E}^\ottnt{S}  [   \sigma  (  \ottnt{e_{{\mathrm{2}}}}  )   ]. \]

 Here,
 $ \mathit{App}  (  \Gamma ;  \sigma ;   \mathit{Abs}  (  \Gamma ;  \ottnt{e_{{\mathrm{1}}}}  )   )   \longrightarrow^{\ast}   \sigma  (  \ottnt{e_{{\mathrm{1}}}}  ) $ and
 $ \mathit{App}  (  \Gamma ;  \sigma ;   \mathit{Abs}  (  \Gamma ;  \ottnt{e_{{\mathrm{2}}}}  )   )   \longrightarrow^{\ast}   \sigma  (  \ottnt{e_{{\mathrm{2}}}}  ) $.
 Since $\emptyset  \vdash  \ottnt{E}^\ottnt{S}  \ottsym{:}  \ottsym{(}  \emptyset  \vdash   \sigma  (  \ottnt{e_{{\mathrm{1}}}}  )   \ottsym{:}   \sigma  (  \ottnt{T}  )   \ottsym{)}  \mathrel{\circ\hspace{-.4em}\rightarrow}  \ottnt{T'}$,
 {\iffull\else we can show \fi}
 \begin{equation}
  \emptyset  \vdash  \ottnt{E}^\ottnt{S}  \ottsym{:}  \ottsym{(}  \emptyset  \vdash   \mathit{App}  (  \Gamma ;  \sigma ;   \mathit{Abs}  (  \Gamma ;  \ottnt{e_{{\mathrm{1}}}}  )   )   \ottsym{:}   \sigma  (  \ottnt{T}  )   \ottsym{)}  \mathrel{\circ\hspace{-.4em}\rightarrow}  \ottnt{T'}
   \label{eqn:fh-lr-ciu-complete-one}
 \end{equation}
   by \prop:ref{fh-lr-comp-sectx-hole-red}.
 %
 By context typing rules,
 $\emptyset  \vdash   \mathit{App}  (  \Gamma ;  \sigma ;   \mathit{Abs}  (  \Gamma ;  \left[ \, \right]  )   )   \ottsym{:}  \ottsym{(}  \Gamma  \vdash  \ottnt{e_{{\mathrm{1}}}}  \ottsym{:}  \ottnt{T}  \ottsym{)}  \mathrel{\circ\hspace{-.4em}\rightarrow}   \sigma  (  \ottnt{T}  ) $.
 %
    Thus, by \prop:ref{fh-lr-comp-sectx-ctx-composed} with
    (\ref{eqn:fh-lr-ciu-complete-one}):
 \[
  \emptyset  \vdash  \ottnt{E}^\ottnt{S}  [   \mathit{App}  (  \Gamma ;  \sigma ;   \mathit{Abs}  (  \Gamma ;  \left[ \, \right]  )   )   ]  \ottsym{:}  \ottsym{(}  \Gamma  \vdash  \ottnt{e_{{\mathrm{1}}}}  \ottsym{:}  \ottnt{T}  \ottsym{)}  \mathrel{\circ\hspace{-.4em}\rightarrow}  \ottnt{T'}
 \]
 {\iffull\else
   by induction on (\ref{eqn:fh-lr-ciu-complete-one}).
 \fi}
 Since $\Gamma  \vdash  \ottnt{e_{{\mathrm{1}}}} \, =_\mathsf{ctx} \, \ottnt{e_{{\mathrm{2}}}}  \ottsym{:}  \ottnt{T}$,
 we have
 \[
  \ottnt{E}^\ottnt{S}  [   \mathit{App}  (  \Gamma ;  \sigma ;   \mathit{Abs}  (  \Gamma ;  \ottnt{e_{{\mathrm{1}}}}  )   )   ]  \Downarrow  \ottnt{E}^\ottnt{S}  [   \mathit{App}  (  \Gamma ;  \sigma ;   \mathit{Abs}  (  \Gamma ;  \ottnt{e_{{\mathrm{2}}}}  )   )   ]
 \]
 by definition.
 Since $\ottnt{E}^\ottnt{S}  [   \mathit{App}  (  \Gamma ;  \sigma ;   \mathit{Abs}  (  \Gamma ;  \ottnt{e_{\ottmv{i}}}  )   )   ]  \Downarrow  \ottnt{E}^\ottnt{S}  [   \sigma  (  \ottnt{e_{\ottmv{i}}}  )   ]$ for $\ottmv{i} \, \in \,  \{   \ottsym{1}  \ottsym{,}  \ottsym{2}   \} $,
 we finish.
\end{prop}
It is shown by the equivalence-respecting property that CIU-equivalent terms are logically
related.
We write substitution $ \theta_{{\mathrm{1}}} \circ \delta_{{\mathrm{1}}} $ for the concatenation of $\theta_{{\mathrm{1}}}$
and $\delta_{{\mathrm{1}}}$.
Note that $\Gamma  \vdash   \theta_{{\mathrm{1}}} \circ \delta_{{\mathrm{1}}} $ if $\Gamma  \vdash  \theta  \ottsym{;}  \delta$.
\begin{prop}[name={$ =_\mathsf{ciu}  \mathbin{\subseteq}  \simeq $ with respect to Typed Terms}]
 {fh-lr-ciu-sound-typed}
 If $\Gamma  \vdash  \ottnt{e_{{\mathrm{1}}}} \, =_\mathsf{ciu} \, \ottnt{e_{{\mathrm{2}}}}  \ottsym{:}  \ottnt{T}$ and $\Gamma  \vdash  \ottnt{e_{{\mathrm{2}}}}  \ottsym{:}  \ottnt{T}$,
 then $\Gamma  \vdash  \ottnt{e_{{\mathrm{1}}}} \,  \mathrel{ \simeq }  \, \ottnt{e_{{\mathrm{2}}}}  \ottsym{:}  \ottnt{T}$.

 \proof

 Let $\Gamma  \vdash  \theta  \ottsym{;}  \delta$.
 It suffices to show that
 $  \theta_{{\mathrm{1}}}  (   \delta_{{\mathrm{1}}}  (  \ottnt{e_{{\mathrm{1}}}}  )   )   \simeq_{\mathtt{e} }   \theta_{{\mathrm{2}}}  (   \delta_{{\mathrm{2}}}  (  \ottnt{e_{{\mathrm{2}}}}  )   )    \ottsym{:}   \ottnt{T} ;  \theta ;  \delta $.
 Since $\Gamma  \vdash  \ottnt{e_{{\mathrm{2}}}}  \ottsym{:}  \ottnt{T}$, we have $\Gamma  \vdash  \ottnt{e_{{\mathrm{2}}}} \,  \mathrel{ \simeq }  \, \ottnt{e_{{\mathrm{2}}}}  \ottsym{:}  \ottnt{T}$ by the parametricity
 (\prop:ref{fh-lr-param}), and so
 $  \theta_{{\mathrm{1}}}  (   \delta_{{\mathrm{1}}}  (  \ottnt{e_{{\mathrm{2}}}}  )   )   \simeq_{\mathtt{e} }   \theta_{{\mathrm{2}}}  (   \delta_{{\mathrm{2}}}  (  \ottnt{e_{{\mathrm{2}}}}  )   )    \ottsym{:}   \ottnt{T} ;  \theta ;  \delta $.
 Since $\Gamma  \vdash  \ottnt{e_{{\mathrm{1}}}} \, =_\mathsf{ciu} \, \ottnt{e_{{\mathrm{2}}}}  \ottsym{:}  \ottnt{T}$ and $\Gamma  \vdash   \theta_{{\mathrm{1}}} \circ \delta_{{\mathrm{1}}} $,
 we have $\emptyset  \vdash   \theta_{{\mathrm{1}}}  (   \delta_{{\mathrm{1}}}  (  \ottnt{e_{{\mathrm{1}}}}  )   )  \, =_\mathsf{ciu} \,  \theta_{{\mathrm{1}}}  (   \delta_{{\mathrm{1}}}  (  \ottnt{e_{{\mathrm{2}}}}  )   )   \ottsym{:}   \theta_{{\mathrm{1}}}  (   \delta_{{\mathrm{1}}}  (  \ottnt{T}  )   ) $.
 Thus, we finish by the equivalence-respecting property (\prop:ref{fh-lr-comp-equiv-res}).
\end{prop}

\fhlrcompletetyped*
\begin{proof}
 By \prop:ref{fh-lr-ciu-complete,fh-lr-ciu-sound-typed}.
\end{proof}

 We can show transitivity of semityped contextual equivalence for well-typed terms
 via the completeness.
 \begin{prop}[name=Transitivity of the Logical Relation]{fh-lr-trans}
 If $\Gamma  \vdash  \ottnt{e_{{\mathrm{1}}}} \,  \mathrel{ \simeq }  \, \ottnt{e_{{\mathrm{2}}}}  \ottsym{:}  \ottnt{T}$ and $\Gamma  \vdash  \ottnt{e_{{\mathrm{2}}}} \,  \mathrel{ \simeq }  \, \ottnt{e_{{\mathrm{3}}}}  \ottsym{:}  \ottnt{T}$,
 then $\Gamma  \vdash  \ottnt{e_{{\mathrm{1}}}} \,  \mathrel{ \simeq }  \, \ottnt{e_{{\mathrm{3}}}}  \ottsym{:}  \ottnt{T}$.

 \proof

 Let $\Gamma  \vdash  \theta  \ottsym{;}  \delta$.
 We show that
 $  \theta_{{\mathrm{1}}}  (   \delta_{{\mathrm{1}}}  (  \ottnt{e_{{\mathrm{1}}}}  )   )   \simeq_{\mathtt{e} }   \theta_{{\mathrm{2}}}  (   \delta_{{\mathrm{2}}}  (  \ottnt{e_{{\mathrm{3}}}}  )   )    \ottsym{:}   \ottnt{T} ;  \theta ;  \delta $.
 Since $\Gamma  \vdash  \ottnt{e_{{\mathrm{1}}}} \,  \mathrel{ \simeq }  \, \ottnt{e_{{\mathrm{2}}}}  \ottsym{:}  \ottnt{T}$, we have $\Gamma  \vdash  \ottnt{e_{{\mathrm{1}}}} \, =_\mathsf{ciu} \, \ottnt{e_{{\mathrm{2}}}}  \ottsym{:}  \ottnt{T}$ by
 \refthm{fh-lr-sound} and \prop:ref{fh-lr-ciu-complete} (note that $\ottnt{e_{{\mathrm{2}}}}$ is
 well typed).
 Since $\Gamma  \vdash   \theta_{{\mathrm{1}}} \circ \delta_{{\mathrm{1}}} $,
 we have
 $\emptyset  \vdash   \theta_{{\mathrm{1}}}  (   \delta_{{\mathrm{1}}}  (  \ottnt{e_{{\mathrm{1}}}}  )   )  \, =_\mathsf{ciu} \,  \theta_{{\mathrm{1}}}  (   \delta_{{\mathrm{1}}}  (  \ottnt{e_{{\mathrm{2}}}}  )   )   \ottsym{:}   \theta_{{\mathrm{1}}}  (   \delta_{{\mathrm{1}}}  (  \ottnt{T}  )   ) $.
 Since $\Gamma  \vdash  \ottnt{e_{{\mathrm{2}}}} \,  \mathrel{ \simeq }  \, \ottnt{e_{{\mathrm{3}}}}  \ottsym{:}  \ottnt{T}$,
 we have $  \theta_{{\mathrm{1}}}  (   \delta_{{\mathrm{1}}}  (  \ottnt{e_{{\mathrm{2}}}}  )   )   \simeq_{\mathtt{e} }   \theta_{{\mathrm{2}}}  (   \delta_{{\mathrm{2}}}  (  \ottnt{e_{{\mathrm{3}}}}  )   )    \ottsym{:}   \ottnt{T} ;  \theta ;  \delta $.
 By the equivalence-respecting property (\prop:ref{fh-lr-comp-equiv-res}), we finish.
\end{prop}

\begin{prop}[type=cor,name=Transitivity of Semityped Contextual Equivalence]
 {fh-lr-ctx-trans}
 If $\Gamma  \vdash  \ottnt{e_{{\mathrm{1}}}} \, =_\mathsf{ctx} \, \ottnt{e_{{\mathrm{2}}}}  \ottsym{:}  \ottnt{T}$ and $\Gamma  \vdash  \ottnt{e_{{\mathrm{2}}}} \, =_\mathsf{ctx} \, \ottnt{e_{{\mathrm{3}}}}  \ottsym{:}  \ottnt{T}$ and
 $\Gamma  \vdash  \ottnt{e_{{\mathrm{3}}}}  \ottsym{:}  \ottnt{T}$, then $\Gamma  \vdash  \ottnt{e_{{\mathrm{1}}}} \, =_\mathsf{ctx} \, \ottnt{e_{{\mathrm{3}}}}  \ottsym{:}  \ottnt{T}$.
\end{prop}

\section{Reasoning about Casts}
\label{sec:reasoning}

This section shows correctness of three cast reasoning techniques---the upcast
elimination, the selfification, and the cast decomposition---using the logical
relation developed in \sect{logical_relation}.

\subsection{Upcast Elimination}
\label{sec:reasoning-upcast-elim}

We first introduce subtyping for {\fhfix} and then show that an upcast and an
identity function are logically related.
Thanks to the soundness of the logical relation with respect to semityped
contextual equivalence (\refthm{fh-lr-sound}), it implies that they are
contextually equivalent.

\begin{fhfigure*}[t!]
 \begin{flushleft}
  \framebox{$\Gamma  \vdash  \ottnt{T_{{\mathrm{1}}}}  \ottsym{<:}  \ottnt{T_{{\mathrm{2}}}}$} \quad {\bf{Subtyping Rules}}
 \end{flushleft}

 \begin{center}
  $\ottdruleSXXBase{}$ \hfil
  $\ottdruleSXXTVar{}$ \hfil
  $\ottdruleSXXForall{}$ \\[1.5ex]
  $\ottdruleSXXFun{}$ \\[1.5ex]
  $\ottdruleSXXRefineL{}$
  $\ottdruleSXXRefineR{}$
 \end{center}

 \vspace{1em}
 \begin{flushleft}
  \framebox{$\Gamma  \models  \ottnt{e}$} \quad {\bf{Satisfaction Rule}}
 \end{flushleft}

 \begin{center}
  $\ottdruleSatis{}$
 \end{center}
 \caption{Subtyping rules.}
 \label{fig:subtype}
\end{fhfigure*}

\fig{subtype} shows subtyping rules, which are similar to Belo et
al.~\cite{Belo/Greenberg/Igarashi/Pierce_2011_ESOP} except that we decompose
the subtyping rule for refinement types into two simple rules.
Subtyping judgment $\Gamma  \vdash  \ottnt{T_{{\mathrm{1}}}}  \ottsym{<:}  \ottnt{T_{{\mathrm{2}}}}$ takes typing context $\Gamma$ for
checking refinements in $\ottnt{T_{{\mathrm{2}}}}$.
Base types and type variables can be subtypes of only themselves
(\Sub{Base} and \Sub{TVar}).
\Sub{Forall} checks that body types of universal types are in subtyping.
As for function types, subtyping is contravariant on the domain types
and covariant on the codomain types \Sub{Fun}.
The subtyping judgment on codomain types assumes that the type of argument
$\mathit{x}$ is $\ottnt{T_{{\mathrm{21}}}}$, a subtype of the other domain type $\ottnt{T_{{\mathrm{11}}}}$, but
codomain type $\ottnt{T_{{\mathrm{12}}}}$ refers to $\mathit{x}$ as $\ottnt{T_{{\mathrm{11}}}}$.
Since the type system of {\fhfix} does not allow subsumption for subtyping
(unlike Knowles and Flanagan~\cite{Knowles/Flanagan_2010_TOPLAS}), we
force $\mathit{x}$ to be of $\ottnt{T_{{\mathrm{11}}}}$ by inserting an upcast, which can be
eliminated after showing the upcast elimination.
We can refine a subtype furthermore \Sub{RefineL}.
By contrast, a supertype can be refined if we can prove that any value of the
subtype satisfies the refinement \Sub{RefineR}.
Term $\ottnt{e}$ is satisfied under $\Gamma$ ($\Gamma  \models  \ottnt{e}$) if,
for any closing substitution $\sigma$ that respects $\Gamma$,
$ \sigma  (  \ottnt{e}  ) $ evaluates to $ \mathsf{true} $.
\Sub{RefineR} also inserts an upcast since
satisfaction assumes that the type of $\mathit{x}$ is subtype $\ottnt{T_{{\mathrm{1}}}}$ but
$\ottnt{e}$ refers to it as supertype $\ottnt{T_{{\mathrm{2}}}}$.

We show that an upcast and an identity function are contextually equivalent via
the logical relation.

\begin{prop}{fh-elim-upcast-aux}
 If
 $\Gamma  \vdash  \ottnt{T_{{\mathrm{1}}}}$ and
 $\Gamma  \vdash  \ottnt{T_{{\mathrm{2}}}}$ and
 $\Gamma  \vdash  \ottnt{T_{{\mathrm{1}}}}  \ottsym{<:}  \ottnt{T_{{\mathrm{2}}}}$ and
 $\Gamma  \vdash  \theta  \ottsym{;}  \delta$,
 then
 $  \theta_{{\mathrm{1}}}  (   \delta_{{\mathrm{1}}}  (  \langle  \ottnt{T_{{\mathrm{1}}}}  \Rightarrow  \ottnt{T_{{\mathrm{2}}}}  \rangle   ^{ \ell }   )   )   \simeq_{\mathtt{v} }   \theta_{{\mathrm{2}}}  (   \delta_{{\mathrm{2}}}  (    \lambda    \mathit{x}  \mathord{:}  \ottnt{T_{{\mathrm{1}}}}  .  \mathit{x}   )   )    \ottsym{:}   \ottnt{T_{{\mathrm{1}}}}  \rightarrow  \ottnt{T_{{\mathrm{2}}}} ;  \theta ;  \delta $.

 \proof

 By induction on $\Gamma  \vdash  \ottnt{T_{{\mathrm{1}}}}  \ottsym{<:}  \ottnt{T_{{\mathrm{2}}}}$.
 It suffices to show that,
 for any $\ottnt{v_{{\mathrm{1}}}}$ and $\ottnt{v_{{\mathrm{2}}}}$ such that $ \ottnt{v_{{\mathrm{1}}}}  \simeq_{\mathtt{v} }  \ottnt{v_{{\mathrm{2}}}}   \ottsym{:}   \ottnt{T_{{\mathrm{1}}}} ;  \theta ;  \delta $,
 \[
    \theta_{{\mathrm{1}}}  (   \delta_{{\mathrm{1}}}  (  \langle  \ottnt{T_{{\mathrm{1}}}}  \Rightarrow  \ottnt{T_{{\mathrm{2}}}}  \rangle   ^{ \ell }   )   )  \, \ottnt{v_{{\mathrm{1}}}}  \simeq_{\mathtt{e} }  \ottnt{v_{{\mathrm{2}}}}   \ottsym{:}   \ottnt{T_{{\mathrm{2}}}} ;  \theta ;  \delta .
 \]
 We proceed by case analysis on the rule applied last to derive
 $\Gamma  \vdash  \ottnt{T_{{\mathrm{1}}}}  \ottsym{<:}  \ottnt{T_{{\mathrm{2}}}}$.
 \begin{itemize}
  \case \Sub{Base}: Obvious
   since $\langle  \ottnt{B}  \Rightarrow  \ottnt{B}  \rangle   ^{ \ell }  \, \ottnt{v_{{\mathrm{1}}}}  \longrightarrow  \ottnt{v_{{\mathrm{1}}}}$ and $ \ottnt{v_{{\mathrm{1}}}}  \simeq_{\mathtt{v} }  \ottnt{v_{{\mathrm{2}}}}   \ottsym{:}   \ottnt{B} ;  \theta ;  \delta $.

  \case \Sub{TVar}:
   We are given $\Gamma  \vdash  \alpha  \ottsym{<:}  \alpha$.
   Since $\Gamma  \vdash  \alpha$ and $\Gamma  \vdash  \theta  \ottsym{;}  \delta$,
   there exists some $\ottnt{r'}$, $\ottnt{T'_{{\mathrm{1}}}}$, and $\ottnt{T'_{{\mathrm{2}}}}$ such that
   $ \theta  (  \alpha  ) = (  \ottnt{r'} ,  \ottnt{T'_{{\mathrm{1}}}} ,  \ottnt{T'_{{\mathrm{2}}}}  ) $.
   Since $ \ottnt{v_{{\mathrm{1}}}}  \simeq_{\mathtt{v} }  \ottnt{v_{{\mathrm{2}}}}   \ottsym{:}   \alpha ;  \theta ;  \delta $, we have $\ottsym{(}  \ottnt{v_{{\mathrm{1}}}}  \ottsym{,}  \ottnt{v_{{\mathrm{2}}}}  \ottsym{)} \, \in \, \ottnt{r'}$.
   Since $\ottnt{r'} \, \in \,  \mathsf{VRel}  (  \ottnt{T'_{{\mathrm{1}}}} ,  \ottnt{T'_{{\mathrm{2}}}}  ) $,
   there exists some $\ottnt{v'_{{\mathrm{1}}}}$ such that
   $\langle  \ottnt{T'_{{\mathrm{1}}}}  \Rightarrow  \ottnt{T'_{{\mathrm{1}}}}  \rangle   ^{ \ell }  \, \ottnt{v_{{\mathrm{1}}}}  \longrightarrow^{\ast}  \ottnt{v'_{{\mathrm{1}}}}$ and $\ottsym{(}  \ottnt{v'_{{\mathrm{1}}}}  \ottsym{,}  \ottnt{v_{{\mathrm{2}}}}  \ottsym{)} \, \in \, \ottnt{r'}$.
   Thus, $ \ottnt{v'_{{\mathrm{1}}}}  \simeq_{\mathtt{v} }  \ottnt{v_{{\mathrm{2}}}}   \ottsym{:}   \alpha ;  \theta ;  \delta $.

  \case \Sub{Fun}:
   We are given $\Gamma  \vdash   \mathit{x} \mathord{:} \ottnt{T_{{\mathrm{11}}}} \rightarrow \ottnt{T_{{\mathrm{12}}}}   \ottsym{<:}   \mathit{x} \mathord{:} \ottnt{T_{{\mathrm{21}}}} \rightarrow \ottnt{T_{{\mathrm{22}}}} $.
   By inversion, we have
   $\Gamma  \vdash  \ottnt{T_{{\mathrm{21}}}}  \ottsym{<:}  \ottnt{T_{{\mathrm{11}}}}$ and $ \Gamma  ,  \mathit{x}  \mathord{:}  \ottnt{T_{{\mathrm{21}}}}   \vdash  \ottnt{T_{{\mathrm{12}}}} \, [  \langle  \ottnt{T_{{\mathrm{21}}}}  \Rightarrow  \ottnt{T_{{\mathrm{11}}}}  \rangle   ^{ \ell }  \, \mathit{x}  \ottsym{/}  \mathit{x}  ]  \ottsym{<:}  \ottnt{T_{{\mathrm{22}}}}$.
   Without loss of generality, we can suppose that $\mathit{x} \, \notin \,  \mathit{dom}  (  \delta  ) $.
   By \E{Red}/\R{Fun},
   \[\begin{array}{l}
     \theta_{{\mathrm{1}}}  (   \delta_{{\mathrm{1}}}  (  \langle   \mathit{x} \mathord{:} \ottnt{T_{{\mathrm{11}}}} \rightarrow \ottnt{T_{{\mathrm{12}}}}   \Rightarrow   \mathit{x} \mathord{:} \ottnt{T_{{\mathrm{21}}}} \rightarrow \ottnt{T_{{\mathrm{22}}}}   \rangle   ^{ \ell }   )   )  \, \ottnt{v_{{\mathrm{1}}}}  \longrightarrow  \\ \quad
      \theta_{{\mathrm{1}}}  (   \delta_{{\mathrm{1}}}  (    \lambda    \mathit{x}  \mathord{:}  \ottnt{T_{{\mathrm{21}}}}  .   \mathsf{let}  ~  \mathit{y}  \mathord{:}  \ottnt{T_{{\mathrm{11}}}}  \equal  \langle  \ottnt{T_{{\mathrm{21}}}}  \Rightarrow  \ottnt{T_{{\mathrm{11}}}}  \rangle   ^{ \ell }  \, \mathit{x}  ~ \ottliteralin ~  \langle  \ottnt{T_{{\mathrm{12}}}} \, [  \mathit{y}  \ottsym{/}  \mathit{x}  ]  \Rightarrow  \ottnt{T_{{\mathrm{22}}}}  \rangle   ^{ \ell }   \, \ottsym{(}  \ottnt{v_{{\mathrm{1}}}} \, \mathit{y}  \ottsym{)}   )   ) 
     \end{array}\]
   for a fresh variable $\mathit{y}$.
   By definition, it suffices to show that,
   for any $\ottnt{v'_{{\mathrm{1}}}}$ and $\ottnt{v'_{{\mathrm{2}}}}$ such that $ \ottnt{v'_{{\mathrm{1}}}}  \simeq_{\mathtt{v} }  \ottnt{v'_{{\mathrm{2}}}}   \ottsym{:}   \ottnt{T_{{\mathrm{21}}}} ;  \theta ;  \delta $,
   \[\begin{array}{ll}
    &   \theta_{{\mathrm{1}}}  (   \delta_{{\mathrm{1}}}  (    \lambda    \mathit{x}  \mathord{:}  \ottnt{T_{{\mathrm{21}}}}  .   \mathsf{let}  ~  \mathit{y}  \mathord{:}  \ottnt{T_{{\mathrm{11}}}}  \equal  \langle  \ottnt{T_{{\mathrm{21}}}}  \Rightarrow  \ottnt{T_{{\mathrm{11}}}}  \rangle   ^{ \ell }  \, \mathit{x}  ~ \ottliteralin ~  \langle  \ottnt{T_{{\mathrm{12}}}} \, [  \mathit{y}  \ottsym{/}  \mathit{x}  ]  \Rightarrow  \ottnt{T_{{\mathrm{22}}}}  \rangle   ^{ \ell }   \, \ottsym{(}  \ottnt{v_{{\mathrm{1}}}} \, \mathit{y}  \ottsym{)}   )   )  \, \ottnt{v'_{{\mathrm{1}}}}    \\   \simeq_{\mathtt{e} }   &    \ottnt{v_{{\mathrm{2}}}} \, \ottnt{v'_{{\mathrm{2}}}}   \ottsym{:}   \ottnt{T_{{\mathrm{22}}}} ;  \theta ;   \delta    [  \,  (  \ottnt{v'_{{\mathrm{1}}}}  ,  \ottnt{v'_{{\mathrm{2}}}}  ) /  \mathit{x}  \,  ]   .
     \end{array}\]

   Since $\Gamma  \vdash   \mathit{x} \mathord{:} \ottnt{T_{{\mathrm{11}}}} \rightarrow \ottnt{T_{{\mathrm{12}}}} $ and $\Gamma  \vdash   \mathit{x} \mathord{:} \ottnt{T_{{\mathrm{21}}}} \rightarrow \ottnt{T_{{\mathrm{22}}}} $,
   we have $\Gamma  \vdash  \ottnt{T_{{\mathrm{11}}}}$ and $\Gamma  \vdash  \ottnt{T_{{\mathrm{21}}}}$
   by their inversion.
   Since $\Gamma  \vdash  \ottnt{T_{{\mathrm{21}}}}  \ottsym{<:}  \ottnt{T_{{\mathrm{11}}}}$ and $\Gamma  \vdash  \theta  \ottsym{;}  \delta$,
   we have
   \[
      \theta_{{\mathrm{1}}}  (   \delta_{{\mathrm{1}}}  (  \langle  \ottnt{T_{{\mathrm{21}}}}  \Rightarrow  \ottnt{T_{{\mathrm{11}}}}  \rangle   ^{ \ell }   )   )   \simeq_{\mathtt{v} }   \theta_{{\mathrm{2}}}  (   \delta_{{\mathrm{2}}}  (    \lambda    \mathit{x}  \mathord{:}  \ottnt{T_{{\mathrm{21}}}}  .  \mathit{x}   )   )    \ottsym{:}   \ottnt{T_{{\mathrm{21}}}}  \rightarrow  \ottnt{T_{{\mathrm{11}}}} ;  \theta ;  \delta 
   \]
   by the IH.
   Since $ \ottnt{v'_{{\mathrm{1}}}}  \simeq_{\mathtt{v} }  \ottnt{v'_{{\mathrm{2}}}}   \ottsym{:}   \ottnt{T_{{\mathrm{21}}}} ;  \theta ;  \delta $,
   we have
   \[
      \theta_{{\mathrm{1}}}  (   \delta_{{\mathrm{1}}}  (  \langle  \ottnt{T_{{\mathrm{21}}}}  \Rightarrow  \ottnt{T_{{\mathrm{11}}}}  \rangle   ^{ \ell }   )   )  \, \ottnt{v'_{{\mathrm{1}}}}  \simeq_{\mathtt{e} }  \ottnt{v'_{{\mathrm{2}}}}   \ottsym{:}   \ottnt{T_{{\mathrm{11}}}} ;  \theta ;  \delta .
   \]
   By definition,
   there exists some $\ottnt{v''_{{\mathrm{1}}}}$ such that
   $ \theta_{{\mathrm{1}}}  (   \delta_{{\mathrm{1}}}  (  \langle  \ottnt{T_{{\mathrm{21}}}}  \Rightarrow  \ottnt{T_{{\mathrm{11}}}}  \rangle   ^{ \ell }   )   )  \, \ottnt{v'_{{\mathrm{1}}}}  \longrightarrow^{\ast}  \ottnt{v''_{{\mathrm{1}}}}$ and
   $ \ottnt{v''_{{\mathrm{1}}}}  \simeq_{\mathtt{v} }  \ottnt{v'_{{\mathrm{2}}}}   \ottsym{:}   \ottnt{T_{{\mathrm{11}}}} ;  \theta ;  \delta $.
   Thus, it suffices to show that
   \[
      \theta_{{\mathrm{1}}}  (   \delta_{{\mathrm{1}}}  (  \langle  \ottnt{T_{{\mathrm{12}}}} \, [  \ottnt{v''_{{\mathrm{1}}}}  \ottsym{/}  \mathit{x}  ]  \Rightarrow  \ottnt{T_{{\mathrm{22}}}} \, [  \ottnt{v'_{{\mathrm{1}}}}  \ottsym{/}  \mathit{x}  ]  \rangle   ^{ \ell }   )   )  \, \ottsym{(}  \ottnt{v_{{\mathrm{1}}}} \, \ottnt{v''_{{\mathrm{1}}}}  \ottsym{)}  \simeq_{\mathtt{e} }  \ottnt{v_{{\mathrm{2}}}} \, \ottnt{v'_{{\mathrm{2}}}}   \ottsym{:}   \ottnt{T_{{\mathrm{22}}}} ;  \theta ;   \delta    [  \,  (  \ottnt{v'_{{\mathrm{1}}}}  ,  \ottnt{v'_{{\mathrm{2}}}}  ) /  \mathit{x}  \,  ]   .
   \]
   Since $ \ottnt{v_{{\mathrm{1}}}}  \simeq_{\mathtt{v} }  \ottnt{v_{{\mathrm{2}}}}   \ottsym{:}    \mathit{x} \mathord{:} \ottnt{T_{{\mathrm{11}}}} \rightarrow \ottnt{T_{{\mathrm{12}}}}  ;  \theta ;  \delta $ and
   $ \ottnt{v''_{{\mathrm{1}}}}  \simeq_{\mathtt{v} }  \ottnt{v'_{{\mathrm{2}}}}   \ottsym{:}   \ottnt{T_{{\mathrm{11}}}} ;  \theta ;  \delta $,
   we have
   \[
     \ottnt{v_{{\mathrm{1}}}} \, \ottnt{v''_{{\mathrm{1}}}}  \simeq_{\mathtt{e} }  \ottnt{v_{{\mathrm{2}}}} \, \ottnt{v'_{{\mathrm{2}}}}   \ottsym{:}   \ottnt{T_{{\mathrm{12}}}} ;  \theta ;   \delta    [  \,  (  \ottnt{v''_{{\mathrm{1}}}}  ,  \ottnt{v'_{{\mathrm{2}}}}  ) /  \mathit{x}  \,  ]   .
   \]
   If $\ottnt{v_{{\mathrm{1}}}} \, \ottnt{v''_{{\mathrm{1}}}}$ and $\ottnt{v_{{\mathrm{2}}}} \, \ottnt{v'_{{\mathrm{2}}}}$ raise blame, we finish.
   Otherwise,
   $\ottnt{v_{{\mathrm{1}}}} \, \ottnt{v''_{{\mathrm{1}}}}  \longrightarrow^{\ast}  \ottnt{v'''_{{\mathrm{1}}}}$ and $\ottnt{v_{{\mathrm{2}}}} \, \ottnt{v'_{{\mathrm{2}}}}  \longrightarrow^{\ast}  \ottnt{v'''_{{\mathrm{2}}}}$
   for some $\ottnt{v'''_{{\mathrm{1}}}}$ and $\ottnt{v'''_{{\mathrm{2}}}}$, and it suffices to show that
   \[
      \theta_{{\mathrm{1}}}  (   \delta_{{\mathrm{1}}}  (  \langle  \ottnt{T_{{\mathrm{12}}}} \, [  \ottnt{v''_{{\mathrm{1}}}}  \ottsym{/}  \mathit{x}  ]  \Rightarrow  \ottnt{T_{{\mathrm{22}}}} \, [  \ottnt{v'_{{\mathrm{1}}}}  \ottsym{/}  \mathit{x}  ]  \rangle   ^{ \ell }   )   )  \, \ottnt{v'''_{{\mathrm{1}}}}  \simeq_{\mathtt{e} }  \ottnt{v'''_{{\mathrm{2}}}}   \ottsym{:}   \ottnt{T_{{\mathrm{22}}}} ;  \theta ;   \delta    [  \,  (  \ottnt{v'_{{\mathrm{1}}}}  ,  \ottnt{v'_{{\mathrm{2}}}}  ) /  \mathit{x}  \,  ]   .
   \]
   We also have $ \ottnt{v'''_{{\mathrm{1}}}}  \simeq_{\mathtt{v} }  \ottnt{v'''_{{\mathrm{2}}}}   \ottsym{:}   \ottnt{T_{{\mathrm{12}}}} ;  \theta ;   \delta    [  \,  (  \ottnt{v''_{{\mathrm{1}}}}  ,  \ottnt{v'_{{\mathrm{2}}}}  ) /  \mathit{x}  \,  ]   $.
   By $\alpha$-renaming $\mathit{x}$ in $\ottnt{T_{{\mathrm{12}}}}$ to $\mathit{y}$ {} and
   the weakening (\prop:ref{fh-lr-val-ws(trel)}),
   \[
     \ottnt{v'''_{{\mathrm{1}}}}  \simeq_{\mathtt{v} }  \ottnt{v'''_{{\mathrm{2}}}}   \ottsym{:}   \ottnt{T_{{\mathrm{12}}}} \, [  \mathit{y}  \ottsym{/}  \mathit{x}  ] ;  \theta ;    \delta    [  \,  (  \ottnt{v'_{{\mathrm{1}}}}  ,  \ottnt{v'_{{\mathrm{2}}}}  ) /  \mathit{x}  \,  ]      [  \,  (  \ottnt{v''_{{\mathrm{1}}}}  ,  \ottnt{v'_{{\mathrm{2}}}}  ) /  \mathit{y}  \,  ]   .
   \]
   Thus, it suffices to show that
   \begin{eqnarray}
    \begin{array}{ll}
    &   \theta_{{\mathrm{1}}}  (   \delta_{{\mathrm{1}}}  (  \langle  \ottnt{T_{{\mathrm{12}}}} \, [  \ottnt{v''_{{\mathrm{1}}}}  \ottsym{/}  \mathit{x}  ]  \Rightarrow  \ottnt{T_{{\mathrm{22}}}} \, [  \ottnt{v'_{{\mathrm{1}}}}  \ottsym{/}  \mathit{x}  ]  \rangle   ^{ \ell }   )   )     \\   \simeq_{\mathtt{v} }   &     \theta_{{\mathrm{2}}}  (   \delta_{{\mathrm{2}}}  (    \lambda    \mathit{z}  \mathord{:}  \ottnt{T'_{{\mathrm{12}}}}  .  \mathit{z}   )  \, [  \ottnt{v'_{{\mathrm{2}}}}  \ottsym{/}  \mathit{x}  ]  )    \ottsym{:}   \ottnt{T_{{\mathrm{12}}}} \, [  \mathit{y}  \ottsym{/}  \mathit{x}  ]  \rightarrow  \ottnt{T_{{\mathrm{22}}}} ;  \theta ;    \delta    [  \,  (  \ottnt{v'_{{\mathrm{1}}}}  ,  \ottnt{v'_{{\mathrm{2}}}}  ) /  \mathit{x}  \,  ]      [  \,  (  \ottnt{v''_{{\mathrm{1}}}}  ,  \ottnt{v'_{{\mathrm{2}}}}  ) /  \mathit{y}  \,  ]   .
     \label{req:fh-elim-upcast-aux-one}
    \end{array}
   \end{eqnarray}
   where $\ottnt{T'_{{\mathrm{12}}}}  \ottsym{=}  \ottnt{T_{{\mathrm{12}}}} \, [  \langle  \ottnt{T_{{\mathrm{21}}}}  \Rightarrow  \ottnt{T_{{\mathrm{11}}}}  \rangle   ^{ \ell }  \, \mathit{x}  \ottsym{/}  \mathit{x}  ]$.

   We first show
   \begin{equation}
    \begin{array}{ll}
     &   \theta_{{\mathrm{1}}}  (   \delta_{{\mathrm{1}}}  (  \langle  \ottnt{T_{{\mathrm{12}}}} \, [  \ottnt{v''_{{\mathrm{1}}}}  \ottsym{/}  \mathit{x}  ]  \Rightarrow  \ottnt{T_{{\mathrm{22}}}} \, [  \ottnt{v'_{{\mathrm{1}}}}  \ottsym{/}  \mathit{x}  ]  \rangle   ^{ \ell }   )   )     \\   \simeq_{\mathtt{v} }   &     \theta_{{\mathrm{2}}}  (   \delta_{{\mathrm{2}}}  (    \lambda    \mathit{z}  \mathord{:}  \ottnt{T'_{{\mathrm{12}}}}  .  \mathit{z}   )  \, [  \ottnt{v'_{{\mathrm{2}}}}  \ottsym{/}  \mathit{x}  ]  )    \ottsym{:}   \ottnt{T'_{{\mathrm{12}}}}  \rightarrow  \ottnt{T_{{\mathrm{22}}}} ;  \theta ;   \delta    [  \,  (  \ottnt{v'_{{\mathrm{1}}}}  ,  \ottnt{v'_{{\mathrm{2}}}}  ) /  \mathit{x}  \,  ]   
    \end{array}
     \label{eqn:fh-elim-upcast-aux-three}
   \end{equation}
   by applying the equivalence-respecting property (\prop:ref{fh-lr-comp-equiv-res}).
   Since $\Gamma  \vdash   \mathit{x} \mathord{:} \ottnt{T_{{\mathrm{11}}}} \rightarrow \ottnt{T_{{\mathrm{12}}}} $ and $\Gamma  \vdash   \mathit{x} \mathord{:} \ottnt{T_{{\mathrm{21}}}} \rightarrow \ottnt{T_{{\mathrm{22}}}} $,
   we have $ \Gamma  ,  \mathit{x}  \mathord{:}  \ottnt{T_{{\mathrm{11}}}}   \vdash  \ottnt{T_{{\mathrm{12}}}}$ and $ \Gamma  ,  \mathit{x}  \mathord{:}  \ottnt{T_{{\mathrm{21}}}}   \vdash  \ottnt{T_{{\mathrm{22}}}}$.
   By the typing weakening (\prop:ref{fh-weak-term}) and
   the term substitution (\prop:ref{fh-subst-term}),
   $ \Gamma  ,  \mathit{x}  \mathord{:}  \ottnt{T_{{\mathrm{21}}}}   \vdash  \ottnt{T'_{{\mathrm{12}}}}$.
   Since $\Gamma  \vdash  \theta  \ottsym{;}  \delta$ and $ \ottnt{v'_{{\mathrm{1}}}}  \simeq_{\mathtt{v} }  \ottnt{v'_{{\mathrm{2}}}}   \ottsym{:}   \ottnt{T_{{\mathrm{21}}}} ;  \theta ;  \delta $,
   we have $ \Gamma  ,  \mathit{x}  \mathord{:}  \ottnt{T_{{\mathrm{21}}}}   \vdash  \theta  \ottsym{;}   \delta    [  \,  (  \ottnt{v'_{{\mathrm{1}}}}  ,  \ottnt{v'_{{\mathrm{2}}}}  ) /  \mathit{x}  \,  ]  $
   by the weakening of the logical relation (\prop:ref{fh-lr-val-ws(tctx)}).
   Since $ \Gamma  ,  \mathit{x}  \mathord{:}  \ottnt{T_{{\mathrm{21}}}}   \vdash  \ottnt{T_{{\mathrm{22}}}}$ and $ \Gamma  ,  \mathit{x}  \mathord{:}  \ottnt{T_{{\mathrm{21}}}}   \vdash  \ottnt{T'_{{\mathrm{12}}}}  \ottsym{<:}  \ottnt{T_{{\mathrm{22}}}}$,
   we have
   \begin{equation}
      \theta_{{\mathrm{1}}}  (   \delta_{{\mathrm{1}}}  (  \langle  \ottnt{T'_{{\mathrm{12}}}}  \Rightarrow  \ottnt{T_{{\mathrm{22}}}}  \rangle   ^{ \ell }  \, [  \ottnt{v'_{{\mathrm{1}}}}  \ottsym{/}  \mathit{x}  ]  )   )   \simeq_{\mathtt{v} }   \theta_{{\mathrm{2}}}  (   \delta_{{\mathrm{2}}}  (    \lambda    \mathit{z}  \mathord{:}  \ottnt{T'_{{\mathrm{12}}}}  .  \mathit{z}   )  \, [  \ottnt{v'_{{\mathrm{2}}}}  \ottsym{/}  \mathit{x}  ]  )    \ottsym{:}   \ottnt{T'_{{\mathrm{12}}}}  \rightarrow  \ottnt{T_{{\mathrm{22}}}} ;  \theta ;   \delta    [  \,  (  \ottnt{v'_{{\mathrm{1}}}}  ,  \ottnt{v'_{{\mathrm{2}}}}  ) /  \mathit{x}  \,  ]   
     \label{eqn:fh-elim-upcast-aux-one}
   \end{equation}
   by the IH.
   Furthermore, since $ \theta_{{\mathrm{1}}}  (   \delta_{{\mathrm{1}}}  (  \langle  \ottnt{T_{{\mathrm{21}}}}  \Rightarrow  \ottnt{T_{{\mathrm{11}}}}  \rangle   ^{ \ell }   )   )  \, \ottnt{v'_{{\mathrm{1}}}}  \longrightarrow^{\ast}  \ottnt{v''_{{\mathrm{1}}}}$,
   we can show
   \begin{equation}
    \begin{array}{ll}
     & \emptyset  \vdash   \theta_{{\mathrm{1}}}  (   \delta_{{\mathrm{1}}}  (  \langle  \ottnt{T_{{\mathrm{12}}}} \, [  \ottnt{v''_{{\mathrm{1}}}}  \ottsym{/}  \mathit{x}  ]  \Rightarrow  \ottnt{T_{{\mathrm{22}}}} \, [  \ottnt{v'_{{\mathrm{1}}}}  \ottsym{/}  \mathit{x}  ]  \rangle   ^{ \ell }   )   )  \,  \\  \, =_\mathsf{ciu} \,  &  \,  \theta_{{\mathrm{1}}}  (   \delta_{{\mathrm{1}}}  (  \langle  \ottnt{T'_{{\mathrm{12}}}}  \Rightarrow  \ottnt{T_{{\mathrm{22}}}}  \rangle   ^{ \ell }   )  \, [  \ottnt{v'_{{\mathrm{1}}}}  \ottsym{/}  \mathit{x}  ]  )   \ottsym{:}   \theta_{{\mathrm{1}}}  (   \delta_{{\mathrm{1}}}  (  \ottnt{T'_{{\mathrm{12}}}}  \rightarrow  \ottnt{T_{{\mathrm{22}}}}  )  \, [  \ottnt{v'_{{\mathrm{1}}}}  \ottsym{/}  \mathit{x}  ]  ) 
    \end{array}
    \label{eqn:fh-elim-upcast-aux-two}
   \end{equation}
   by using Cotermination.
   From (\ref{eqn:fh-elim-upcast-aux-one}) and (\ref{eqn:fh-elim-upcast-aux-two}),
   the equivalence-respecting property derives (\ref{eqn:fh-elim-upcast-aux-three}).

   We show (\ref{req:fh-elim-upcast-aux-one})
   by applying (\ref{eqn:fh-elim-upcast-aux-three}) to the term compositionality (\prop:ref{fh-lr-term-comp}).
   Since $\ottnt{T'_{{\mathrm{12}}}}  \ottsym{=}  \ottnt{T_{{\mathrm{12}}}} \, [  \langle  \ottnt{T_{{\mathrm{21}}}}  \Rightarrow  \ottnt{T_{{\mathrm{11}}}}  \rangle   ^{ \ell }  \, \mathit{x}  \ottsym{/}  \mathit{x}  ]  \ottsym{=}  \ottnt{T_{{\mathrm{12}}}} \, [  \mathit{y}  \ottsym{/}  \mathit{x}  ] \, [  \langle  \ottnt{T_{{\mathrm{21}}}}  \Rightarrow  \ottnt{T_{{\mathrm{11}}}}  \rangle   ^{ \ell }  \, \mathit{x}  \ottsym{/}  \mathit{y}  ]$,
   it suffices to show that
   \begin{enumerate}
    \item \label{req:fh-elim-upcast-aux-Gself}
          $  \Gamma  ,  \mathit{x}  \mathord{:}  \ottnt{T_{{\mathrm{21}}}}   ,  \mathit{y}  \mathord{:}  \ottnt{T_{{\mathrm{11}}}} $ is self-related,
    \item \label{req:fh-elim-upcast-aux-T12self}
          $   \Gamma  ,  \mathit{x}  \mathord{:}  \ottnt{T_{{\mathrm{21}}}}   ,  \mathit{y}  \mathord{:}  \ottnt{T_{{\mathrm{11}}}}  \vdash \ottnt{T_{{\mathrm{12}}}} \, [  \mathit{y}  \ottsym{/}  \mathit{x}  ]  \mathrel{ \simeq }  \ottnt{T_{{\mathrm{12}}}} \, [  \mathit{y}  \ottsym{/}  \mathit{x}  ]  : \ast $,
    \item \label{req:fh-elim-upcast-aux-subst}
          $  \Gamma  ,  \mathit{x}  \mathord{:}  \ottnt{T_{{\mathrm{21}}}}   ,  \mathit{y}  \mathord{:}  \ottnt{T_{{\mathrm{11}}}}   \vdash  \theta  \ottsym{;}    \delta    [  \,  (  \ottnt{v'_{{\mathrm{1}}}}  ,  \ottnt{v'_{{\mathrm{2}}}}  ) /  \mathit{x}  \,  ]      [  \,  (  \ottnt{v''_{{\mathrm{1}}}}  ,  \ottnt{v'_{{\mathrm{2}}}}  ) /  \mathit{y}  \,  ]  $,
    \item \label{req:fh-elim-upcast-aux-v2}
          $ \theta_{{\mathrm{2}}}  (   \delta_{{\mathrm{2}}}  (  \langle  \ottnt{T_{{\mathrm{21}}}}  \Rightarrow  \ottnt{T_{{\mathrm{11}}}}  \rangle   ^{ \ell }   )   )  \, \ottnt{v'_{{\mathrm{2}}}}  \longrightarrow^{\ast}  \ottnt{v''_{{\mathrm{2}}}}$ and
          $ \ottnt{v''_{{\mathrm{1}}}}  \simeq_{\mathtt{v} }  \ottnt{v''_{{\mathrm{2}}}}   \ottsym{:}   \ottnt{T_{{\mathrm{11}}}} ;  \theta ;   \delta    [  \,  (  \ottnt{v'_{{\mathrm{1}}}}  ,  \ottnt{v'_{{\mathrm{2}}}}  ) /  \mathit{x}  \,  ]   $
          for some $\ottnt{v''_{{\mathrm{2}}}}$.
   \end{enumerate}

   Since $ \Gamma  ,  \mathit{x}  \mathord{:}  \ottnt{T_{{\mathrm{11}}}}   \vdash  \ottnt{T_{{\mathrm{12}}}}$, we have
   $  \Gamma  ,  \mathit{x}  \mathord{:}  \ottnt{T_{{\mathrm{21}}}}   ,  \mathit{y}  \mathord{:}  \ottnt{T_{{\mathrm{11}}}}   \vdash  \ottnt{T_{{\mathrm{12}}}} \, [  \mathit{y}  \ottsym{/}  \mathit{x}  ]$.
   By the parametricity (\prop:ref{fh-lr-param}), we have
   (\ref{req:fh-elim-upcast-aux-Gself}) and
   (\ref{req:fh-elim-upcast-aux-T12self}).

   Since $\Gamma  \vdash  \theta  \ottsym{;}  \delta$ and
   $ \ottnt{v''_{{\mathrm{1}}}}  \simeq_{\mathtt{v} }  \ottnt{v'_{{\mathrm{2}}}}   \ottsym{:}   \ottnt{T_{{\mathrm{11}}}} ;  \theta ;  \delta $ and
   $ \ottnt{v'_{{\mathrm{1}}}}  \simeq_{\mathtt{v} }  \ottnt{v'_{{\mathrm{2}}}}   \ottsym{:}   \ottnt{T_{{\mathrm{21}}}} ;  \theta ;  \delta $,
   we have (\ref{req:fh-elim-upcast-aux-subst})
   by the weakening (\prop:ref{fh-lr-val-ws(tctx)}).

   Since $\Gamma  \vdash  \theta  \ottsym{;}  \delta$ and
   $ \ottnt{v'_{{\mathrm{1}}}}  \simeq_{\mathtt{v} }  \ottnt{v'_{{\mathrm{2}}}}   \ottsym{:}   \ottnt{T_{{\mathrm{21}}}} ;  \theta ;  \delta $,
   we have $ \Gamma  ,  \mathit{x}  \mathord{:}  \ottnt{T_{{\mathrm{21}}}}   \vdash  \theta  \ottsym{;}   \delta    [  \,  (  \ottnt{v'_{{\mathrm{1}}}}  ,  \ottnt{v'_{{\mathrm{2}}}}  ) /  \mathit{x}  \,  ]  $.
   Since $ \Gamma  ,  \mathit{x}  \mathord{:}  \ottnt{T_{{\mathrm{21}}}}   \vdash  \langle  \ottnt{T_{{\mathrm{21}}}}  \Rightarrow  \ottnt{T_{{\mathrm{11}}}}  \rangle   ^{ \ell }  \, \mathit{x}  \ottsym{:}  \ottnt{T_{{\mathrm{11}}}}$,
   we have $ \Gamma  ,  \mathit{x}  \mathord{:}  \ottnt{T_{{\mathrm{21}}}}   \vdash  \langle  \ottnt{T_{{\mathrm{21}}}}  \Rightarrow  \ottnt{T_{{\mathrm{11}}}}  \rangle   ^{ \ell }  \, \mathit{x} \,  \mathrel{ \simeq }  \, \langle  \ottnt{T_{{\mathrm{21}}}}  \Rightarrow  \ottnt{T_{{\mathrm{11}}}}  \rangle   ^{ \ell }  \, \mathit{x}  \ottsym{:}  \ottnt{T_{{\mathrm{11}}}}$
   by the parametricity (\prop:ref{fh-lr-param}).
   Thus, by definition,
   \[
      \theta_{{\mathrm{1}}}  (   \delta_{{\mathrm{1}}}  (  \langle  \ottnt{T_{{\mathrm{21}}}}  \Rightarrow  \ottnt{T_{{\mathrm{11}}}}  \rangle   ^{ \ell }   )   )  \, \ottnt{v'_{{\mathrm{1}}}}  \simeq_{\mathtt{e} }   \theta_{{\mathrm{2}}}  (   \delta_{{\mathrm{2}}}  (  \langle  \ottnt{T_{{\mathrm{21}}}}  \Rightarrow  \ottnt{T_{{\mathrm{11}}}}  \rangle   ^{ \ell }   )   )  \, \ottnt{v'_{{\mathrm{2}}}}   \ottsym{:}   \ottnt{T_{{\mathrm{11}}}} ;  \theta ;   \delta    [  \,  (  \ottnt{v'_{{\mathrm{1}}}}  ,  \ottnt{v'_{{\mathrm{2}}}}  ) /  \mathit{x}  \,  ]   .
   \]
   Since $ \theta_{{\mathrm{1}}}  (   \delta_{{\mathrm{1}}}  (  \langle  \ottnt{T_{{\mathrm{21}}}}  \Rightarrow  \ottnt{T_{{\mathrm{11}}}}  \rangle   ^{ \ell }   )   )  \, \ottnt{v'_{{\mathrm{1}}}}  \longrightarrow^{\ast}  \ottnt{v''_{{\mathrm{1}}}}$,
   we have (\ref{req:fh-elim-upcast-aux-v2}).

  \case \Sub{Forall}: By the IH.
  \case \Sub{RefineL}: By the IH.
   {}

  \case \Sub{RefineR}:
   We are given $\Gamma  \vdash  \ottnt{T_{{\mathrm{1}}}}  \ottsym{<:}   \{  \mathit{x}  \mathord{:}  \ottnt{T'_{{\mathrm{2}}}}   \mathop{\mid}   \ottnt{e'_{{\mathrm{2}}}}  \} $.
   By inversion,
   $\Gamma  \vdash  \ottnt{T_{{\mathrm{1}}}}  \ottsym{<:}  \ottnt{T'_{{\mathrm{2}}}}$ and $ \Gamma  ,  \mathit{x}  \mathord{:}  \ottnt{T_{{\mathrm{1}}}}   \models  \ottnt{e'_{{\mathrm{2}}}} \, [  \langle  \ottnt{T_{{\mathrm{1}}}}  \Rightarrow  \ottnt{T'_{{\mathrm{2}}}}  \rangle   ^{ \ell }  \, \mathit{x}  \ottsym{/}  \mathit{x}  ]$.
   By \E{Red}/\R{Forget},
   \[
     \theta_{{\mathrm{1}}}  (   \delta_{{\mathrm{1}}}  (  \langle  \ottnt{T_{{\mathrm{1}}}}  \Rightarrow   \{  \mathit{x}  \mathord{:}  \ottnt{T'_{{\mathrm{2}}}}   \mathop{\mid}   \ottnt{e'_{{\mathrm{2}}}}  \}   \rangle   ^{ \ell }   )   )  \, \ottnt{v_{{\mathrm{1}}}}  \longrightarrow^{\ast}  \langle   \mathit{unref}  (   \theta_{{\mathrm{1}}}  (   \delta_{{\mathrm{1}}}  (  \ottnt{T_{{\mathrm{1}}}}  )   )   )   \Rightarrow   \theta_{{\mathrm{1}}}  (   \delta_{{\mathrm{1}}}  (   \{  \mathit{x}  \mathord{:}  \ottnt{T'_{{\mathrm{2}}}}   \mathop{\mid}   \ottnt{e'_{{\mathrm{2}}}}  \}   )   )   \rangle   ^{ \ell }  \, \ottnt{v_{{\mathrm{1}}}}.
   \]
   By \E{Red}/\R{PreCheck},
   \[\begin{array}{ll}
    & \langle   \mathit{unref}  (   \theta_{{\mathrm{1}}}  (   \delta_{{\mathrm{1}}}  (  \ottnt{T_{{\mathrm{1}}}}  )   )   )   \Rightarrow   \theta_{{\mathrm{1}}}  (   \delta_{{\mathrm{1}}}  (   \{  \mathit{x}  \mathord{:}  \ottnt{T'_{{\mathrm{2}}}}   \mathop{\mid}   \ottnt{e'_{{\mathrm{2}}}}  \}   )   )   \rangle   ^{ \ell }  \, \ottnt{v_{{\mathrm{1}}}} \\
     \longrightarrow  &
       \langle\!\langle  \,  \theta_{{\mathrm{1}}}  (   \delta_{{\mathrm{1}}}  (   \{  \mathit{x}  \mathord{:}  \ottnt{T'_{{\mathrm{2}}}}   \mathop{\mid}   \ottnt{e'_{{\mathrm{2}}}}  \}   )   )   \ottsym{,}  \langle   \mathit{unref}  (   \theta_{{\mathrm{1}}}  (   \delta_{{\mathrm{1}}}  (  \ottnt{T_{{\mathrm{1}}}}  )   )   )   \Rightarrow   \theta_{{\mathrm{1}}}  (   \delta_{{\mathrm{1}}}  (  \ottnt{T'_{{\mathrm{2}}}}  )   )   \rangle   ^{ \ell }  \, \ottnt{v_{{\mathrm{1}}}} \,  \rangle\!\rangle  \,  ^{ \ell } .
     \end{array}\]
   Thus, it suffices to show that
   \[
      \langle\!\langle  \,  \theta_{{\mathrm{1}}}  (   \delta_{{\mathrm{1}}}  (   \{  \mathit{x}  \mathord{:}  \ottnt{T'_{{\mathrm{2}}}}   \mathop{\mid}   \ottnt{e'_{{\mathrm{2}}}}  \}   )   )   \ottsym{,}  \langle   \mathit{unref}  (   \theta_{{\mathrm{1}}}  (   \delta_{{\mathrm{1}}}  (  \ottnt{T_{{\mathrm{1}}}}  )   )   )   \Rightarrow   \theta_{{\mathrm{1}}}  (   \delta_{{\mathrm{1}}}  (  \ottnt{T'_{{\mathrm{2}}}}  )   )   \rangle   ^{ \ell }  \, \ottnt{v_{{\mathrm{1}}}} \,  \rangle\!\rangle  \,  ^{ \ell }   \simeq_{\mathtt{e} }  \ottnt{v_{{\mathrm{2}}}}   \ottsym{:}    \{  \mathit{x}  \mathord{:}  \ottnt{T'_{{\mathrm{2}}}}   \mathop{\mid}   \ottnt{e'_{{\mathrm{2}}}}  \}  ;  \theta ;  \delta .
   \]

   Since $\Gamma  \vdash   \{  \mathit{x}  \mathord{:}  \ottnt{T'_{{\mathrm{2}}}}   \mathop{\mid}   \ottnt{e'_{{\mathrm{2}}}}  \} $,
   we have $\Gamma  \vdash  \ottnt{T'_{{\mathrm{2}}}}$ by its inversion.
   Thus, by the IH,
   \[
      \theta_{{\mathrm{1}}}  (   \delta_{{\mathrm{1}}}  (  \langle  \ottnt{T_{{\mathrm{1}}}}  \Rightarrow  \ottnt{T'_{{\mathrm{2}}}}  \rangle   ^{ \ell }   )   )   \simeq_{\mathtt{v} }   \theta_{{\mathrm{2}}}  (   \delta_{{\mathrm{2}}}  (    \lambda    \mathit{x}  \mathord{:}  \ottnt{T_{{\mathrm{1}}}}  .  \mathit{x}   )   )    \ottsym{:}   \ottnt{T_{{\mathrm{1}}}}  \rightarrow  \ottnt{T'_{{\mathrm{2}}}} ;  \theta ;  \delta .
   \]
   Since $ \ottnt{v_{{\mathrm{1}}}}  \simeq_{\mathtt{v} }  \ottnt{v_{{\mathrm{2}}}}   \ottsym{:}   \ottnt{T_{{\mathrm{1}}}} ;  \theta ;  \delta $,
   we have
   \[
      \theta_{{\mathrm{1}}}  (   \delta_{{\mathrm{1}}}  (  \langle  \ottnt{T_{{\mathrm{1}}}}  \Rightarrow  \ottnt{T'_{{\mathrm{2}}}}  \rangle   ^{ \ell }   )   )  \, \ottnt{v_{{\mathrm{1}}}}  \simeq_{\mathtt{e} }  \ottnt{v_{{\mathrm{2}}}}   \ottsym{:}   \ottnt{T'_{{\mathrm{2}}}} ;  \theta ;  \delta .
   \]
   Since
   $ \theta_{{\mathrm{1}}}  (   \delta_{{\mathrm{1}}}  (  \langle  \ottnt{T_{{\mathrm{1}}}}  \Rightarrow  \ottnt{T'_{{\mathrm{2}}}}  \rangle   ^{ \ell }   )   )  \, \ottnt{v_{{\mathrm{1}}}}  \longrightarrow^{\ast}  \langle   \mathit{unref}  (   \theta_{{\mathrm{1}}}  (   \delta_{{\mathrm{1}}}  (  \ottnt{T_{{\mathrm{1}}}}  )   )   )   \Rightarrow   \theta_{{\mathrm{1}}}  (   \delta_{{\mathrm{1}}}  (  \ottnt{T'_{{\mathrm{2}}}}  )   )   \rangle   ^{ \ell }  \, \ottnt{v_{{\mathrm{1}}}}$,
   we have
   \[
     \langle   \mathit{unref}  (   \theta_{{\mathrm{1}}}  (   \delta_{{\mathrm{1}}}  (  \ottnt{T_{{\mathrm{1}}}}  )   )   )   \Rightarrow   \theta_{{\mathrm{1}}}  (   \delta_{{\mathrm{1}}}  (  \ottnt{T'_{{\mathrm{2}}}}  )   )   \rangle   ^{ \ell }  \, \ottnt{v_{{\mathrm{1}}}}  \simeq_{\mathtt{e} }  \ottnt{v_{{\mathrm{2}}}}   \ottsym{:}   \ottnt{T'_{{\mathrm{2}}}} ;  \theta ;  \delta .
   \]
   By definition,
   there exists some $\ottnt{v'_{{\mathrm{1}}}}$ such that
   $\langle   \mathit{unref}  (   \theta_{{\mathrm{1}}}  (   \delta_{{\mathrm{1}}}  (  \ottnt{T_{{\mathrm{1}}}}  )   )   )   \Rightarrow   \theta_{{\mathrm{1}}}  (   \delta_{{\mathrm{1}}}  (  \ottnt{T'_{{\mathrm{2}}}}  )   )   \rangle   ^{ \ell }  \, \ottnt{v_{{\mathrm{1}}}}  \longrightarrow^{\ast}  \ottnt{v'_{{\mathrm{1}}}}$
   and
   $ \ottnt{v'_{{\mathrm{1}}}}  \simeq_{\mathtt{v} }  \ottnt{v_{{\mathrm{2}}}}   \ottsym{:}   \ottnt{T'_{{\mathrm{2}}}} ;  \theta ;  \delta $.
   Thus, it suffices to show that
   \[
      \langle\!\langle  \,  \theta_{{\mathrm{1}}}  (   \delta_{{\mathrm{1}}}  (   \{  \mathit{x}  \mathord{:}  \ottnt{T'_{{\mathrm{2}}}}   \mathop{\mid}   \ottnt{e'_{{\mathrm{2}}}}  \}   )   )   \ottsym{,}  \ottnt{v'_{{\mathrm{1}}}} \,  \rangle\!\rangle  \,  ^{ \ell }   \simeq_{\mathtt{e} }  \ottnt{v_{{\mathrm{2}}}}   \ottsym{:}    \{  \mathit{x}  \mathord{:}  \ottnt{T'_{{\mathrm{2}}}}   \mathop{\mid}   \ottnt{e'_{{\mathrm{2}}}}  \}  ;  \theta ;  \delta .
   \]

   Since
   $\Gamma  \vdash  \theta  \ottsym{;}  \delta$ and $ \ottnt{v_{{\mathrm{1}}}}  \simeq_{\mathtt{v} }  \ottnt{v_{{\mathrm{2}}}}   \ottsym{:}   \ottnt{T_{{\mathrm{1}}}} ;  \theta ;  \delta $,
   we have $ \Gamma  ,  \mathit{x}  \mathord{:}  \ottnt{T_{{\mathrm{1}}}}   \vdash   \theta_{{\mathrm{1}}} \circ \delta_{{\mathrm{1}}}   [  \ottnt{v_{{\mathrm{1}}}}  \ottsym{/}  \mathit{x}  ]$.
   Since $ \Gamma  ,  \mathit{x}  \mathord{:}  \ottnt{T_{{\mathrm{1}}}}   \models  \ottnt{e'_{{\mathrm{2}}}} \, [  \langle  \ottnt{T_{{\mathrm{1}}}}  \Rightarrow  \ottnt{T'_{{\mathrm{2}}}}  \rangle   ^{ \ell }  \, \mathit{x}  \ottsym{/}  \mathit{x}  ]$,
   we have $ \theta_{{\mathrm{1}}}  (   \delta_{{\mathrm{1}}}  (  \ottnt{e'_{{\mathrm{2}}}} \, [  \langle  \ottnt{T_{{\mathrm{1}}}}  \Rightarrow  \ottnt{T'_{{\mathrm{2}}}}  \rangle   ^{ \ell }  \, \ottnt{v_{{\mathrm{1}}}}  \ottsym{/}  \mathit{x}  ]  )   )   \longrightarrow^{\ast}   \mathsf{true} $.
   Since $ \theta_{{\mathrm{1}}}  (   \delta_{{\mathrm{1}}}  (  \langle  \ottnt{T_{{\mathrm{1}}}}  \Rightarrow  \ottnt{T'_{{\mathrm{2}}}}  \rangle   ^{ \ell }   )   )  \, \ottnt{v_{{\mathrm{1}}}}  \longrightarrow^{\ast}  \ottnt{v'_{{\mathrm{1}}}}$,
   we have $ \theta_{{\mathrm{1}}}  (   \delta_{{\mathrm{1}}}  (  \ottnt{e'_{{\mathrm{2}}}} \, [  \ottnt{v'_{{\mathrm{1}}}}  \ottsym{/}  \mathit{x}  ]  )   )   \longrightarrow^{\ast}   \mathsf{true} $
   by Cotermination (\prop:ref{fh-coterm-true}).
   Thus,
   $ \langle\!\langle  \,  \theta_{{\mathrm{1}}}  (   \delta_{{\mathrm{1}}}  (   \{  \mathit{x}  \mathord{:}  \ottnt{T'_{{\mathrm{2}}}}   \mathop{\mid}   \ottnt{e'_{{\mathrm{2}}}}  \}   )   )   \ottsym{,}  \ottnt{v'_{{\mathrm{1}}}} \,  \rangle\!\rangle  \,  ^{ \ell }   \longrightarrow^{\ast}  \ottnt{v'_{{\mathrm{1}}}}$
   (by \R{Check} and \R{OK}), and so
   it suffices to show that
   \[
     \ottnt{v'_{{\mathrm{1}}}}  \simeq_{\mathtt{v} }  \ottnt{v_{{\mathrm{2}}}}   \ottsym{:}    \{  \mathit{x}  \mathord{:}  \ottnt{T'_{{\mathrm{2}}}}   \mathop{\mid}   \ottnt{e'_{{\mathrm{2}}}}  \}  ;  \theta ;  \delta .
   \]
   Since
   $ \ottnt{v'_{{\mathrm{1}}}}  \simeq_{\mathtt{v} }  \ottnt{v_{{\mathrm{2}}}}   \ottsym{:}   \ottnt{T'_{{\mathrm{2}}}} ;  \theta ;  \delta $ and
   $ \theta_{{\mathrm{1}}}  (   \delta_{{\mathrm{1}}}  (  \ottnt{e'_{{\mathrm{2}}}} \, [  \ottnt{v'_{{\mathrm{1}}}}  \ottsym{/}  \mathit{x}  ]  )   )   \longrightarrow^{\ast}   \mathsf{true} $,
   it suffices to show that
   \[
     \theta_{{\mathrm{1}}}  (   \delta_{{\mathrm{1}}}  (  \ottnt{e'_{{\mathrm{2}}}} \, [  \ottnt{v_{{\mathrm{2}}}}  \ottsym{/}  \mathit{x}  ]  )   )   \longrightarrow^{\ast}   \mathsf{true} .
   \]
   Since $\Gamma  \vdash   \{  \mathit{x}  \mathord{:}  \ottnt{T'_{{\mathrm{2}}}}   \mathop{\mid}   \ottnt{e'_{{\mathrm{2}}}}  \} $,
   we have $ \Gamma  ,  \mathit{x}  \mathord{:}  \ottnt{T'_{{\mathrm{2}}}}   \vdash  \ottnt{e'_{{\mathrm{2}}}}  \ottsym{:}   \mathsf{Bool} $ by its inversion.
   By the parametricity (\prop:ref{fh-lr-param}),
   $ \Gamma  ,  \mathit{x}  \mathord{:}  \ottnt{T'_{{\mathrm{2}}}}   \vdash  \ottnt{e'_{{\mathrm{2}}}} \,  \mathrel{ \simeq }  \, \ottnt{e'_{{\mathrm{2}}}}  \ottsym{:}   \mathsf{Bool} $.
   Since $\Gamma  \vdash  \theta  \ottsym{;}  \delta$ and $ \ottnt{v'_{{\mathrm{1}}}}  \simeq_{\mathtt{v} }  \ottnt{v_{{\mathrm{2}}}}   \ottsym{:}   \ottnt{T'_{{\mathrm{2}}}} ;  \theta ;  \delta $,
   we have $ \Gamma  ,  \mathit{x}  \mathord{:}  \ottnt{T'_{{\mathrm{2}}}}   \vdash  \theta  \ottsym{;}   \delta    [  \,  (  \ottnt{v'_{{\mathrm{1}}}}  ,  \ottnt{v_{{\mathrm{2}}}}  ) /  \mathit{x}  \,  ]  $.
   Thus,
   \[
      \theta_{{\mathrm{1}}}  (   \delta_{{\mathrm{1}}}  (  \ottnt{e'_{{\mathrm{2}}}} \, [  \ottnt{v'_{{\mathrm{1}}}}  \ottsym{/}  \mathit{x}  ]  )   )   \simeq_{\mathtt{e} }   \theta_{{\mathrm{2}}}  (   \delta_{{\mathrm{2}}}  (  \ottnt{e'_{{\mathrm{2}}}} \, [  \ottnt{v_{{\mathrm{2}}}}  \ottsym{/}  \mathit{x}  ]  )   )    \ottsym{:}    \mathsf{Bool}  ;  \theta ;   \delta    [  \,  (  \ottnt{v'_{{\mathrm{1}}}}  ,  \ottnt{v_{{\mathrm{2}}}}  ) /  \mathit{x}  \,  ]   .
   \]
   Since the term on the left-hand side evaluates to $ \mathsf{true} $,
   we have $ \theta_{{\mathrm{1}}}  (   \delta_{{\mathrm{1}}}  (  \ottnt{e'_{{\mathrm{2}}}} \, [  \ottnt{v_{{\mathrm{2}}}}  \ottsym{/}  \mathit{x}  ]  )   )   \longrightarrow^{\ast}   \mathsf{true} $ by definition.
   \qedhere
 \end{itemize}
\end{prop}

\begin{prop}[type=thm,name=Upcast Elimination]{elim-upcast}
 If $\Gamma  \vdash  \ottnt{T_{{\mathrm{1}}}}$ and $\Gamma  \vdash  \ottnt{T_{{\mathrm{2}}}}$ and $\Gamma  \vdash  \ottnt{T_{{\mathrm{1}}}}  \ottsym{<:}  \ottnt{T_{{\mathrm{2}}}}$,
 then $\Gamma  \vdash  \langle  \ottnt{T_{{\mathrm{1}}}}  \Rightarrow  \ottnt{T_{{\mathrm{2}}}}  \rangle   ^{ \ell }  \, =_\mathsf{ctx} \, \ottsym{(}    \lambda    \mathit{x}  \mathord{:}  \ottnt{T_{{\mathrm{1}}}}  .  \mathit{x}   \ottsym{)}  \ottsym{:}  \ottnt{T_{{\mathrm{1}}}}  \rightarrow  \ottnt{T_{{\mathrm{2}}}}$.

 \proof

 By \prop:ref{fh-elim-upcast-aux}, $\Gamma  \vdash  \langle  \ottnt{T_{{\mathrm{1}}}}  \Rightarrow  \ottnt{T_{{\mathrm{2}}}}  \rangle   ^{ \ell }  \,  \mathrel{ \simeq }  \, \ottsym{(}    \lambda    \mathit{x}  \mathord{:}  \ottnt{T_{{\mathrm{1}}}}  .  \mathit{x}   \ottsym{)}  \ottsym{:}  \ottnt{T_{{\mathrm{1}}}}  \rightarrow  \ottnt{T_{{\mathrm{2}}}}$.
 By the soundness of the logical relation (\refthm{fh-lr-sound}), we finish.
\end{prop}

\subsection{Selfification}
\label{sec:reasoning-self}

Selfification embeds information of a term into its type so that we can get the
singleton type that identifies the
term~\cite{Ou/Tan/Mandelbaum/Walker_2004_TCS}.
For example, selfification of $\mathit{x}$ of $ \mathsf{Int} $ produces $ \{  \mathit{y}  \mathord{:}   \mathsf{Int}    \mathop{\mid}    \mathit{y}  \mathrel{=}_{  \mathsf{Int}  }  \mathit{x}   \} $, which identifies $\mathit{x}$, and that of $\ottnt{e}$ of $ \mathsf{Bool}   \rightarrow   \mathsf{Int} $ does
$ \mathit{x} \mathord{:}  \mathsf{Bool}  \rightarrow  \{  \mathit{y}  \mathord{:}   \mathsf{Int}    \mathop{\mid}    \mathit{y}  \mathrel{=}_{  \mathsf{Int}  }  \ottnt{e}  \, \mathit{x}  \}  $, which means functions that return the same
value as the result of call to $\ottnt{e}$.

We expose the power of the selfification combined with the upcast elimination via an
example using stacks.
{\ifexample
For the sake of simplicity, this section considers stacks for integers.
\fi}
First of all, let us
{\ifexample
recall (a part of) the refined interface of stacks for integers:
\else
assume type $ \mathsf{Stack} $ (which can be implemented as an abstract datatype
in {\fhfix}) and the following functions:
\fi}
\[\begin{array}{ll}
  \mathsf{is\_empty}  :  \mathsf{Stack}   \rightarrow   \mathsf{Bool}  &
  \mathsf{push}  :  \mathsf{Int}   \rightarrow   \mathsf{Stack}   \rightarrow   \{  \mathit{x}  \mathord{:}   \mathsf{Stack}    \mathop{\mid}    \mathsf{not}  \, \ottsym{(}   \mathsf{is\_empty}  \, \mathit{x}  \ottsym{)}  \}  \\
  \mathsf{empty}  :  \{  \mathit{x}  \mathord{:}   \mathsf{Stack}    \mathop{\mid}    \mathsf{is\_empty}  \, \mathit{x}  \}  &
  \mathsf{pop}  :  \{  \mathit{x}  \mathord{:}   \mathsf{Stack}    \mathop{\mid}    \mathsf{not}  \, \ottsym{(}   \mathsf{is\_empty}  \, \mathit{x}  \ottsym{)}  \}   \rightarrow   \mathsf{Stack} 
  \end{array}\]
{\ifexample\else{where $ \mathsf{is\_empty} $ returns whether a given stack is empty,
$ \mathsf{empty} $ is the empty stack, $ \mathsf{push} $ produces a nonempty stack by adding an
element at the top of a stack, and $ \mathsf{pop} $ returns the stack without the
topmost element.\fi}}
Since the type signature of $ \mathsf{push} $ ensures that the result stack is never
empty, expression $ \mathsf{pop}  \, \ottsym{(}   \mathsf{push}  \, \ottsym{2} \,  \mathsf{empty}   \ottsym{)}$ would be accepted.\footnote{Trivial
cast $\langle   \{  \mathit{x}  \mathord{:}   \mathsf{Stack}    \mathop{\mid}    \mathsf{is\_empty}  \, \mathit{x}  \}   \Rightarrow   \mathsf{Stack}   \rangle   ^{ \ell } $ to $ \mathsf{empty} $ is omitted here.}
However, the type of $ \mathsf{pop} $ guarantees nothing about stacks that it returns.
Thus, expression $ \mathsf{pop}  \, \ottsym{(}   \mathsf{pop}  \, \ottsym{(}   \mathsf{push}  \, \ottsym{2} \, \ottsym{(}   \mathsf{push}  \, \ottsym{3} \,  \mathsf{empty}   \ottsym{)}  \ottsym{)}  \ottsym{)}$ would be rejected because
the outermost $ \mathsf{pop} $ takes a possibly empty stack ($ \mathsf{Stack} $), not nonempty
stacks ($ \{  \mathit{x}  \mathord{:}   \mathsf{Stack}    \mathop{\mid}    \mathsf{is\_empty}  \, \mathit{x}  \} $), even though it is actually called with a
nonempty one.
Insertion of cast $\langle   \mathsf{Stack}   \Rightarrow   \{  \mathit{x}  \mathord{:}   \mathsf{Stack}    \mathop{\mid}    \mathsf{not}  \, \ottsym{(}   \mathsf{is\_empty}  \, \mathit{x}  \ottsym{)}  \}   \rangle   ^{ \ell } $ makes
the program acceptable, but it incurs additional, redundant overhead.
(Note that the upcast elimination cannot be applied here because
$\langle   \mathsf{Stack}   \Rightarrow   \{  \mathit{x}  \mathord{:}   \mathsf{Stack}    \mathop{\mid}    \mathsf{not}  \, \ottsym{(}   \mathsf{is\_empty}  \, \mathit{x}  \ottsym{)}  \}   \rangle   ^{ \ell } $ is not an upcast.)

Combination of the selfification and the upcast elimination solves this
unfortunate situation.
Selfification can give subexpression $ \mathsf{pop}  \, \ottsym{(}   \mathsf{push}  \, \ottsym{2} \, \ottsym{(}   \mathsf{push}  \, \ottsym{3} \,  \mathsf{empty}   \ottsym{)}  \ottsym{)}$
type $\ottnt{T}  \ottsym{=}   \{  \mathit{x}  \mathord{:}   \mathsf{Stack}    \mathop{\mid}    \mathit{x}  \mathrel{=}_{  \mathsf{Stack}  }   \mathsf{pop}   \, \ottsym{(}   \mathsf{push}  \, \ottsym{2} \, \ottsym{(}   \mathsf{push}  \, \ottsym{3} \,  \mathsf{empty}   \ottsym{)}  \ottsym{)}  \} $, which identifies
the subexpresssion.
%
%
Since $\ottnt{T}$ denotes the singleton stack with only $\ottsym{3}$, we expect that
$\langle  \ottnt{T}  \Rightarrow   \{  \mathit{x}  \mathord{:}   \mathsf{Stack}    \mathop{\mid}    \mathsf{not}  \, \ottsym{(}   \mathsf{is\_empty}  \, \mathit{x}  \ottsym{)}  \}   \rangle   ^{ \ell } $ is proven to be an upcast.
If so, by the upcast elimination, acceptable program
\[
  \mathsf{pop}  \, \ottsym{(}  \langle  \ottnt{T}  \Rightarrow   \{  \mathit{x}  \mathord{:}   \mathsf{Stack}    \mathop{\mid}    \mathsf{not}  \, \ottsym{(}   \mathsf{is\_empty}  \, \mathit{x}  \ottsym{)}  \}   \rangle   ^{ \ell }  \, \ottsym{(}   \mathsf{pop}  \, \ottsym{(}   \mathsf{push}  \, \ottsym{2} \, \ottsym{(}   \mathsf{push}  \, \ottsym{3} \,  \mathsf{empty}   \ottsym{)}  \ottsym{)}  \ottsym{)}  \ottsym{)}
\]
should be contextual equivalent to $ \mathsf{pop}  \, \ottsym{(}   \mathsf{pop}  \, \ottsym{(}   \mathsf{push}  \, \ottsym{2} \, \ottsym{(}   \mathsf{push}  \, \ottsym{3} \,  \mathsf{empty}   \ottsym{)}  \ottsym{)}  \ottsym{)}$, and so
it would be \emph{proven} that it does not get stuck.

Selfification function $ \mathit{self}  (  \ottnt{T} ,  \ottnt{e}  ) $, which returns a type into which term
$\ottnt{e}$ of $\ottnt{T}$ is embedded, is defined as follows.
\begin{defi}[Selfification]
 \[\begin{array}{llll}
   \mathit{self}  (  \ottnt{B} ,  \ottnt{e}  )  &=&  \{  \mathit{x}  \mathord{:}  \ottnt{B}   \mathop{\mid}    \mathit{x}  \mathrel{=}_{ \ottnt{B} }  \ottnt{e}   \}  &
   \text{(if $\mathit{x} \, \notin \,  \mathit{FV}  (  \ottnt{e}  ) $)} \\
   \mathit{self}  (  \alpha ,  \ottnt{e}  )  &=& \alpha \\
   \mathit{self}  (   \mathit{x} \mathord{:} \ottnt{T_{{\mathrm{1}}}} \rightarrow \ottnt{T_{{\mathrm{2}}}}  ,  \ottnt{e}  )  &=&  \mathit{x} \mathord{:} \ottnt{T_{{\mathrm{1}}}} \rightarrow  \mathit{self}  (  \ottnt{T_{{\mathrm{2}}}} ,  \ottnt{e} \, \mathit{x}  )   &
   \text{(if $\mathit{x} \, \notin \,  \mathit{FV}  (  \ottnt{e}  ) $)} \\
   \mathit{self}  (   \forall   \alpha  .  \ottnt{T}  ,  \ottnt{e}  )  &=&  \forall   \alpha  .   \mathit{self}  (  \ottnt{T} ,  \ottnt{e} \, \alpha  )   &
   \text{(if $\alpha \, \notin \,  \mathit{FTV}  (  \ottnt{e}  ) $)} \\
   \mathit{self}  (   \{  \mathit{x}  \mathord{:}  \ottnt{T'}   \mathop{\mid}   \ottnt{e'}  \}  ,  \ottnt{e}  )  &=&  \{  \mathit{x}  \mathord{:}  \ottnt{T''}   \mathop{\mid}    \mathsf{let}  ~  \mathit{x}  \mathord{:}  \ottnt{T'}  \equal  \langle  \ottnt{T''}  \Rightarrow  \ottnt{T'}  \rangle   ^{ \ell }  \, \mathit{x}  ~ \ottliteralin ~  \ottnt{e'}   \}  &
   \text{(if $\mathit{x} \, \notin \,  \mathit{FV}  (  \ottnt{e}  ) $)} \\ &&
   \multicolumn{2}{r}{
    \text{where $\ottnt{T''}  \ottsym{=}   \mathit{self}  (  \ottnt{T'} ,  \langle   \{  \mathit{x}  \mathord{:}  \ottnt{T'}   \mathop{\mid}   \ottnt{e'}  \}   \Rightarrow  \ottnt{T'}  \rangle   ^{ \ell }  \, \ottnt{e}  ) $}
   }
   \end{array}\]
\end{defi}
Selfification $ \mathit{self}  (  \ottnt{B} ,  \ottnt{e}  ) $ produces the most precise type for $\ottnt{e}$ in
that it is the singleton type which identifies $\ottnt{e}$, and $ \mathit{self}  (  \alpha ,  \ottnt{e}  ) $
returns $\alpha$ as it is because we cannot make type variables more precise
without polymorphic equality.
Selfification of function types $ \mathit{x} \mathord{:} \ottnt{T_{{\mathrm{1}}}} \rightarrow \ottnt{T_{{\mathrm{2}}}} $ and universal types $ \forall   \alpha  .  \ottnt{T} $ is forwarded to $\ottnt{T_{{\mathrm{2}}}}$ and $\ottnt{T}$, respectively.
Term $\ottnt{e}$ is applied to variables so that selfified types can identify
what $\ottnt{e}$ produces.
Selfifying refinement types $ \{  \mathit{x}  \mathord{:}  \ottnt{T'}   \mathop{\mid}   \ottnt{e'}  \} $ appears slightly tricky: it
selfifies the underlying type $\ottnt{T'}$ with $\langle   \{  \mathit{x}  \mathord{:}  \ottnt{T'}   \mathop{\mid}   \ottnt{e'}  \}   \Rightarrow  \ottnt{T'}  \rangle   ^{ \ell }  \, \ottnt{e}$ (the cast
makes $\ottnt{e}$ a term of $\ottnt{T'}$) and refines the result with refinement
$\ottnt{e'}$, but, since $\ottnt{e'}$ refers to $\mathit{x}$ of $\ottnt{T'}$ whereas the
selfified underlying type is $\ottnt{T''}$, cast $\langle  \ottnt{T''}  \Rightarrow  \ottnt{T'}  \rangle   ^{ \ell } $ is inserted at the
beginning of the refinement.
Label $\ell$ can be any because, as shown later, the casts never fail.

The rest of this section shows that inserting casts to selfified types causes no run-time
errors, which leads to use of selfification to any expression for free.
More formally, we prove that, given term $\ottnt{e}$ of $\ottnt{T}$,
$\langle  \ottnt{T}  \Rightarrow   \mathit{self}  (  \ottnt{T} ,  \ottnt{e}  )   \rangle   ^{ \ell }  \, \ottnt{e}$ is contextually equivalent to $\ottnt{e}$.
We start with showing that casts to selfified types are well typed, which is
implied by two facts: (1) $\ottnt{T}$ is compatible with $ \mathit{self}  (  \ottnt{T} ,  \ottnt{e}  ) $ and
(2) $ \mathit{self}  (  \ottnt{T} ,  \ottnt{e}  ) $ is well formed if $\ottnt{e}$ is well typed.
\begin{prop}{fh-self-compat}
 $\ottnt{T}  \mathrel{\parallel}   \mathit{self}  (  \ottnt{T} ,  \ottnt{e}  ) $.

 \proof

 Straightforward by induction on $\ottnt{T}$.
\end{prop}
\begin{prop}{fh-self-wf}
 If $\Gamma  \vdash  \ottnt{e}  \ottsym{:}  \ottnt{T}$, then $\Gamma  \vdash   \mathit{self}  (  \ottnt{T} ,  \ottnt{e}  ) $.

 \proof

 By induction on $\ottnt{T}$ with the fact that $\Gamma  \vdash  \ottnt{T}$, which is obtained from
 $\Gamma  \vdash  \ottnt{e}  \ottsym{:}  \ottnt{T}$.
\end{prop}
The selfification of refinement types involves casts from selfified types to the
underlying types, so we need to show that such casts also do not raise blame.
\begin{prop}{fh-self-elim-self2type}
 If $\Gamma  \vdash  \ottnt{e}  \ottsym{:}  \ottnt{T}$,
 then $\Gamma  \vdash  \langle   \mathit{self}  (  \ottnt{T} ,  \ottnt{e}  )   \Rightarrow  \ottnt{T}  \rangle   ^{ \ell }  \,  \mathrel{ \simeq }  \,   \lambda    \mathit{x}  \mathord{:}  \ottnt{T}  .  \mathit{x}   \ottsym{:}   \mathit{self}  (  \ottnt{T} ,  \ottnt{e}  )   \rightarrow  \ottnt{T}$.

 \proof

 By induction on $\ottnt{T}$.
 Let $\Gamma  \vdash  \theta  \ottsym{;}  \delta$.
 It suffices to show that, for any $\ottnt{v_{{\mathrm{1}}}}$ and $\ottnt{v_{{\mathrm{2}}}}$
 such that $ \ottnt{v_{{\mathrm{1}}}}  \simeq_{\mathtt{v} }  \ottnt{v_{{\mathrm{2}}}}   \ottsym{:}    \mathit{self}  (  \ottnt{T} ,  \ottnt{e}  )  ;  \theta ;  \delta $,
 \[
    \theta_{{\mathrm{1}}}  (   \delta_{{\mathrm{1}}}  (  \langle   \mathit{self}  (  \ottnt{T} ,  \ottnt{e}  )   \Rightarrow  \ottnt{T}  \rangle   ^{ \ell }   )   )  \, \ottnt{v_{{\mathrm{1}}}}  \simeq_{\mathtt{e} }  \ottnt{v_{{\mathrm{2}}}}   \ottsym{:}   \ottnt{T} ;  \theta ;  \delta .
 \]
 By case analysis on $\ottnt{T}$.
 \begin{itemize}
  \case $\ottnt{T}  \ottsym{=}  \ottnt{B}$: Trivial.
  \case $\ottnt{T}  \ottsym{=}  \alpha$:
   We have $ \mathit{self}  (  \ottnt{T} ,  \ottnt{e}  )   \ottsym{=}  \alpha$.
   Since reflexive casts are logically related to identity functions
   (\prop:ref{fh-lr-elim-refl-cast}),
   we have
   $  \theta_{{\mathrm{1}}}  (   \delta_{{\mathrm{1}}}  (  \langle  \alpha  \Rightarrow  \alpha  \rangle   ^{ \ell }   )   )   \simeq_{\mathtt{v} }   \theta_{{\mathrm{2}}}  (   \delta_{{\mathrm{2}}}  (    \lambda    \mathit{x}  \mathord{:}  \alpha  .  \mathit{x}   )   )    \ottsym{:}   \alpha  \rightarrow  \alpha ;  \theta ;  \delta $.
   Since $ \ottnt{v_{{\mathrm{1}}}}  \simeq_{\mathtt{v} }  \ottnt{v_{{\mathrm{2}}}}   \ottsym{:}   \alpha ;  \theta ;  \delta $, we finish by definition.

  \case $\ottnt{T}  \ottsym{=}   \mathit{x} \mathord{:} \ottnt{T_{{\mathrm{1}}}} \rightarrow \ottnt{T_{{\mathrm{2}}}} $:
   Without loss of generality, we can suppose that
   $\mathit{x} \, \notin \,  \mathit{dom}  (  \delta  ) $.
   We have $ \mathit{self}  (  \ottnt{T} ,  \ottnt{e}  )   \ottsym{=}   \mathit{x} \mathord{:} \ottnt{T_{{\mathrm{1}}}} \rightarrow  \mathit{self}  (  \ottnt{T_{{\mathrm{2}}}} ,  \ottnt{e} \, \mathit{x}  )  $.
   By \E{Red}/\R{Fun},
   \[\begin{array}{ll}
    &  \theta_{{\mathrm{1}}}  (   \delta_{{\mathrm{1}}}  (  \langle   \mathit{self}  (  \ottnt{T} ,  \ottnt{e}  )   \Rightarrow  \ottnt{T}  \rangle   ^{ \ell }   )   )  \, \ottnt{v_{{\mathrm{1}}}} \\
     \longrightarrow  &  \theta_{{\mathrm{1}}}  (   \delta_{{\mathrm{1}}}  (    \lambda    \mathit{x}  \mathord{:}  \ottnt{T_{{\mathrm{1}}}}  .   \mathsf{let}  ~  \mathit{y}  \mathord{:}  \ottnt{T_{{\mathrm{1}}}}  \equal  \langle  \ottnt{T_{{\mathrm{1}}}}  \Rightarrow  \ottnt{T_{{\mathrm{1}}}}  \rangle   ^{ \ell }  \, \mathit{x}  ~ \ottliteralin ~  \langle   \mathit{self}  (  \ottnt{T_{{\mathrm{2}}}} ,  \ottnt{e} \, \mathit{x}  )  \, [  \mathit{y}  \ottsym{/}  \mathit{x}  ]  \Rightarrow  \ottnt{T_{{\mathrm{2}}}}  \rangle   ^{ \ell }   \, \ottsym{(}  \ottnt{v_{{\mathrm{1}}}} \, \mathit{y}  \ottsym{)}   )   ) .\\
     \end{array}\]
   for fresh variable $\mathit{y}$.
   Thus, it suffices to show that, for any $\ottnt{v'_{{\mathrm{1}}}}$ and $\ottnt{v'_{{\mathrm{2}}}}$
   such that $ \ottnt{v'_{{\mathrm{1}}}}  \simeq_{\mathtt{v} }  \ottnt{v'_{{\mathrm{2}}}}   \ottsym{:}   \ottnt{T_{{\mathrm{1}}}} ;  \theta ;  \delta $,
   \[\begin{array}{ll}
    &   \theta_{{\mathrm{1}}}  (   \delta_{{\mathrm{1}}}  (   \mathsf{let}  ~  \mathit{y}  \mathord{:}  \ottnt{T_{{\mathrm{1}}}}  \equal  \langle  \ottnt{T_{{\mathrm{1}}}}  \Rightarrow  \ottnt{T_{{\mathrm{1}}}}  \rangle   ^{ \ell }  \, \mathit{x}  ~ \ottliteralin ~  \langle   \mathit{self}  (  \ottnt{T_{{\mathrm{2}}}} ,  \ottnt{e} \, \mathit{x}  )  \, [  \mathit{y}  \ottsym{/}  \mathit{x}  ]  \Rightarrow  \ottnt{T_{{\mathrm{2}}}}  \rangle   ^{ \ell }  \, \ottsym{(}  \ottnt{v_{{\mathrm{1}}}} \, \mathit{y}  \ottsym{)}   )  \, [  \ottnt{v'_{{\mathrm{1}}}}  \ottsym{/}  \mathit{x}  ]  )     \\   \simeq_{\mathtt{e} }   &    \ottnt{v_{{\mathrm{2}}}} \, \ottnt{v'_{{\mathrm{2}}}}   \ottsym{:}   \ottnt{T_{{\mathrm{2}}}} ;  \theta ;   \delta    [  \,  (  \ottnt{v'_{{\mathrm{1}}}}  ,  \ottnt{v'_{{\mathrm{2}}}}  ) /  \mathit{x}  \,  ]   .
     \end{array}\]

   Let $\ottnt{e'}  \ottsym{=}  \langle   \mathit{self}  (  \ottnt{T_{{\mathrm{2}}}} ,  \ottnt{e} \, \mathit{x}  )  \, [  \mathit{y}  \ottsym{/}  \mathit{x}  ]  \Rightarrow  \ottnt{T_{{\mathrm{2}}}}  \rangle   ^{ \ell }  \, \ottsym{(}  \mathit{z} \, \mathit{y}  \ottsym{)}$
   for fresh variable $\mathit{z}$.
   Since reflexive casts are logically related to identity functions
   (\prop:ref{fh-lr-elim-refl-cast}),
   we have $\Gamma  \vdash  \langle  \ottnt{T_{{\mathrm{1}}}}  \Rightarrow  \ottnt{T_{{\mathrm{1}}}}  \rangle   ^{ \ell }  \,  \mathrel{ \simeq }  \,   \lambda    \mathit{x}  \mathord{:}  \ottnt{T_{{\mathrm{1}}}}  .  \mathit{x}   \ottsym{:}  \ottnt{T_{{\mathrm{1}}}}  \rightarrow  \ottnt{T_{{\mathrm{1}}}}$.
   Thus, we have
   \begin{equation}
    \begin{array}{r@{\;}l}
      \Gamma  ,  \mathit{z}  \mathord{:}  \ottsym{(}   \mathit{x} \mathord{:} \ottnt{T_{{\mathrm{1}}}} \rightarrow  \mathit{self}  (  \ottnt{T_{{\mathrm{2}}}} ,  \ottnt{e} \, \mathit{x}  )    \ottsym{)}   ,  \mathit{x}  \mathord{:}  \ottnt{T_{{\mathrm{1}}}}   \vdash   &  \,  \mathsf{let}  ~  \mathit{y}  \mathord{:}  \ottnt{T_{{\mathrm{1}}}}  \equal  \langle  \ottnt{T_{{\mathrm{1}}}}  \Rightarrow  \ottnt{T_{{\mathrm{1}}}}  \rangle   ^{ \ell }  \, \mathit{x}  ~ \ottliteralin ~  \ottnt{e'}  \,  \mathrel{ \simeq }  \,  \\  \,  &  \,  \mathsf{let}  ~  \mathit{y}  \mathord{:}  \ottnt{T_{{\mathrm{1}}}}  \equal  \ottsym{(}    \lambda    \mathit{x}  \mathord{:}  \ottnt{T_{{\mathrm{1}}}}  .  \mathit{x}   \ottsym{)} \, \mathit{x}  ~ \ottliteralin ~  \ottnt{e'}   \ottsym{:}  \ottnt{T_{{\mathrm{2}}}}
    \end{array}
    \label{eqn:fh-self-elim-self2type-fun-one}
   \end{equation}
   by the fundamental property.

   Since $ \ottnt{v_{{\mathrm{1}}}}  \simeq_{\mathtt{v} }  \ottnt{v_{{\mathrm{2}}}}   \ottsym{:}    \mathit{x} \mathord{:} \ottnt{T_{{\mathrm{1}}}} \rightarrow  \mathit{self}  (  \ottnt{T_{{\mathrm{2}}}} ,  \ottnt{e} \, \mathit{x}  )   ;  \theta ;  \delta $ and
   $ \ottnt{v'_{{\mathrm{1}}}}  \simeq_{\mathtt{v} }  \ottnt{v'_{{\mathrm{2}}}}   \ottsym{:}   \ottnt{T_{{\mathrm{1}}}} ;  \theta ;  \delta $,
   we have $  \Gamma  ,  \mathit{z}  \mathord{:}  \ottsym{(}   \mathit{x} \mathord{:} \ottnt{T_{{\mathrm{1}}}} \rightarrow  \mathit{self}  (  \ottnt{T_{{\mathrm{2}}}} ,  \ottnt{e} \, \mathit{x}  )    \ottsym{)}   ,  \mathit{x}  \mathord{:}  \ottnt{T_{{\mathrm{1}}}}   \vdash  \theta  \ottsym{;}    \delta    [  \,  (  \ottnt{v_{{\mathrm{1}}}}  ,  \ottnt{v_{{\mathrm{2}}}}  ) /  \mathit{z}  \,  ]      [  \,  (  \ottnt{v'_{{\mathrm{1}}}}  ,  \ottnt{v'_{{\mathrm{2}}}}  ) /  \mathit{x}  \,  ]  $
   by the weakening (\prop:ref{fh-lr-val-ws}).
   Since logically related terms are CIU-equivalent
   (\refthm{fh-lr-sound} and \prop:ref{fh-lr-ciu-complete}),
   we have
   \[\begin{array}{l}
    \emptyset  \vdash   \theta_{{\mathrm{1}}}  (   \delta_{{\mathrm{1}}}  (   \mathsf{let}  ~  \mathit{y}  \mathord{:}  \ottnt{T_{{\mathrm{1}}}}  \equal  \langle  \ottnt{T_{{\mathrm{1}}}}  \Rightarrow  \ottnt{T_{{\mathrm{1}}}}  \rangle   ^{ \ell }  \, \mathit{x}  ~ \ottliteralin ~  \ottnt{e'} \, [  \ottnt{v_{{\mathrm{1}}}}  \ottsym{/}  \mathit{z}  ]   )  \, [  \ottnt{v'_{{\mathrm{1}}}}  \ottsym{/}  \mathit{x}  ]  )  \,  \\  \,  \quad  \, =_\mathsf{ciu} \,  \theta_{{\mathrm{1}}}  (   \delta_{{\mathrm{1}}}  (   \mathsf{let}  ~  \mathit{y}  \mathord{:}  \ottnt{T_{{\mathrm{1}}}}  \equal  \ottsym{(}    \lambda    \mathit{x}  \mathord{:}  \ottnt{T_{{\mathrm{1}}}}  .  \mathit{x}   \ottsym{)} \, \mathit{x}  ~ \ottliteralin ~  \ottnt{e'} \, [  \ottnt{v_{{\mathrm{1}}}}  \ottsym{/}  \mathit{z}  ]   )  \, [  \ottnt{v'_{{\mathrm{1}}}}  \ottsym{/}  \mathit{x}  ]  )   \ottsym{:}   \theta_{{\mathrm{1}}}  (   \delta_{{\mathrm{1}}}  (  \ottnt{T_{{\mathrm{2}}}}  )  \, [  \ottnt{v'_{{\mathrm{1}}}}  \ottsym{/}  \mathit{x}  ]  ) .
     \end{array}\]
   from (\ref{eqn:fh-self-elim-self2type-fun-one}).
   Thus, by the equivalence-respecting property (\prop:ref{fh-lr-comp-equiv-res}),
   it suffices to show that
   \[\begin{array}{ll}
    &   \theta_{{\mathrm{1}}}  (   \delta_{{\mathrm{1}}}  (   \mathsf{let}  ~  \mathit{y}  \mathord{:}  \ottnt{T_{{\mathrm{1}}}}  \equal  \ottsym{(}    \lambda    \mathit{x}  \mathord{:}  \ottnt{T_{{\mathrm{1}}}}  .  \mathit{x}   \ottsym{)} \, \mathit{x}  ~ \ottliteralin ~  \langle   \mathit{self}  (  \ottnt{T_{{\mathrm{2}}}} ,  \ottnt{e} \, \mathit{x}  )  \, [  \mathit{y}  \ottsym{/}  \mathit{x}  ]  \Rightarrow  \ottnt{T_{{\mathrm{2}}}}  \rangle   ^{ \ell }  \, \ottsym{(}  \ottnt{v_{{\mathrm{1}}}} \, \mathit{y}  \ottsym{)}   )  \, [  \ottnt{v'_{{\mathrm{1}}}}  \ottsym{/}  \mathit{x}  ]  )     \\   \simeq_{\mathtt{e} }   &    \ottnt{v_{{\mathrm{2}}}} \, \ottnt{v'_{{\mathrm{2}}}}   \ottsym{:}   \ottnt{T_{{\mathrm{2}}}} ;  \theta ;   \delta    [  \,  (  \ottnt{v'_{{\mathrm{1}}}}  ,  \ottnt{v'_{{\mathrm{2}}}}  ) /  \mathit{x}  \,  ]   ,
     \end{array}\]
   that is,
   \[\begin{array}{ll}
    &   \theta_{{\mathrm{1}}}  (   \delta_{{\mathrm{1}}}  (  \langle   \mathit{self}  (  \ottnt{T_{{\mathrm{2}}}} ,  \ottnt{e} \, \mathit{x}  )   \Rightarrow  \ottnt{T_{{\mathrm{2}}}}  \rangle   ^{ \ell }   )  \, [  \ottnt{v'_{{\mathrm{1}}}}  \ottsym{/}  \mathit{x}  ]  )  \, \ottsym{(}  \ottnt{v_{{\mathrm{1}}}} \, \ottnt{v'_{{\mathrm{1}}}}  \ottsym{)}  \simeq_{\mathtt{e} }  \ottnt{v_{{\mathrm{2}}}} \, \ottnt{v'_{{\mathrm{2}}}}   \ottsym{:}   \ottnt{T_{{\mathrm{2}}}} ;  \theta ;   \delta    [  \,  (  \ottnt{v'_{{\mathrm{1}}}}  ,  \ottnt{v'_{{\mathrm{2}}}}  ) /  \mathit{x}  \,  ]   .
     \end{array}\]
   Since $ \ottnt{v_{{\mathrm{1}}}}  \simeq_{\mathtt{v} }  \ottnt{v_{{\mathrm{2}}}}   \ottsym{:}    \mathit{x} \mathord{:} \ottnt{T_{{\mathrm{1}}}} \rightarrow  \mathit{self}  (  \ottnt{T_{{\mathrm{2}}}} ,  \ottnt{e} \, \mathit{x}  )   ;  \theta ;  \delta $ and
   $ \ottnt{v'_{{\mathrm{1}}}}  \simeq_{\mathtt{v} }  \ottnt{v'_{{\mathrm{2}}}}   \ottsym{:}   \ottnt{T_{{\mathrm{1}}}} ;  \theta ;  \delta $,
   we have $ \ottnt{v_{{\mathrm{1}}}} \, \ottnt{v'_{{\mathrm{1}}}}  \simeq_{\mathtt{e} }  \ottnt{v_{{\mathrm{2}}}} \, \ottnt{v'_{{\mathrm{2}}}}   \ottsym{:}    \mathit{self}  (  \ottnt{T_{{\mathrm{2}}}} ,  \ottnt{e} \, \mathit{x}  )  ;  \theta ;   \delta    [  \,  (  \ottnt{v'_{{\mathrm{1}}}}  ,  \ottnt{v'_{{\mathrm{2}}}}  ) /  \mathit{x}  \,  ]   $.
   If $\ottnt{v_{{\mathrm{1}}}} \, \ottnt{v'_{{\mathrm{1}}}}$ and $\ottnt{v_{{\mathrm{2}}}} \, \ottnt{v'_{{\mathrm{2}}}}$ raise blame, then we finish.
   Otherwise,
   $\ottnt{v_{{\mathrm{1}}}} \, \ottnt{v'_{{\mathrm{1}}}}  \longrightarrow^{\ast}  \ottnt{v''_{{\mathrm{1}}}}$ and
   $\ottnt{v_{{\mathrm{2}}}} \, \ottnt{v'_{{\mathrm{2}}}}  \longrightarrow^{\ast}  \ottnt{v''_{{\mathrm{2}}}}$ for some $\ottnt{v''_{{\mathrm{1}}}}$ and $\ottnt{v''_{{\mathrm{2}}}}$, and
   it suffices to show that
   \begin{equation}
     \begin{array}{ll}
    &   \theta_{{\mathrm{1}}}  (   \delta_{{\mathrm{1}}}  (  \langle   \mathit{self}  (  \ottnt{T_{{\mathrm{2}}}} ,  \ottnt{e} \, \mathit{x}  )   \Rightarrow  \ottnt{T_{{\mathrm{2}}}}  \rangle   ^{ \ell }   )  \, [  \ottnt{v'_{{\mathrm{1}}}}  \ottsym{/}  \mathit{x}  ]  )  \, \ottnt{v''_{{\mathrm{1}}}}  \simeq_{\mathtt{e} }  \ottnt{v''_{{\mathrm{2}}}}   \ottsym{:}   \ottnt{T_{{\mathrm{2}}}} ;  \theta ;   \delta    [  \,  (  \ottnt{v'_{{\mathrm{1}}}}  ,  \ottnt{v'_{{\mathrm{2}}}}  ) /  \mathit{x}  \,  ]   .
     \end{array}
     \label{req:fh-self-elim-self2type-fun-two}
   \end{equation}
   We also have $ \ottnt{v''_{{\mathrm{1}}}}  \simeq_{\mathtt{v} }  \ottnt{v''_{{\mathrm{2}}}}   \ottsym{:}    \mathit{self}  (  \ottnt{T_{{\mathrm{2}}}} ,  \ottnt{e} \, \mathit{x}  )  ;  \theta ;   \delta    [  \,  (  \ottnt{v'_{{\mathrm{1}}}}  ,  \ottnt{v'_{{\mathrm{2}}}}  ) /  \mathit{x}  \,  ]   $.
   Since $ \Gamma  ,  \mathit{x}  \mathord{:}  \ottnt{T_{{\mathrm{1}}}}   \vdash  \ottnt{e} \, \mathit{x}  \ottsym{:}  \ottnt{T_{{\mathrm{2}}}}$,
   we have
   \[
     \Gamma  ,  \mathit{x}  \mathord{:}  \ottnt{T_{{\mathrm{1}}}}   \vdash  \langle   \mathit{self}  (  \ottnt{T_{{\mathrm{2}}}} ,  \ottnt{e} \, \mathit{x}  )   \Rightarrow  \ottnt{T_{{\mathrm{2}}}}  \rangle   ^{ \ell }  \,  \mathrel{ \simeq }  \,   \lambda    \mathit{x}  \mathord{:}  \ottnt{T_{{\mathrm{2}}}}  .  \mathit{x}   \ottsym{:}   \mathit{self}  (  \ottnt{T_{{\mathrm{2}}}} ,  \ottnt{e} \, \mathit{x}  )   \rightarrow  \ottnt{T_{{\mathrm{2}}}}
   \]
   by the IH.
   Since $ \Gamma  ,  \mathit{x}  \mathord{:}  \ottnt{T_{{\mathrm{1}}}}   \vdash  \theta  \ottsym{;}   \delta    [  \,  (  \ottnt{v'_{{\mathrm{1}}}}  ,  \ottnt{v'_{{\mathrm{2}}}}  ) /  \mathit{x}  \,  ]  $,
   we have
   \[\begin{array}{ll}
    &   \theta_{{\mathrm{1}}}  (   \delta_{{\mathrm{1}}}  (  \langle   \mathit{self}  (  \ottnt{T_{{\mathrm{2}}}} ,  \ottnt{e} \, \mathit{x}  )   \Rightarrow  \ottnt{T_{{\mathrm{2}}}}  \rangle   ^{ \ell }  \, [  \ottnt{v'_{{\mathrm{1}}}}  \ottsym{/}  \mathit{x}  ]  )   )     \\   \simeq_{\mathtt{v} }   &     \theta_{{\mathrm{2}}}  (   \delta_{{\mathrm{2}}}  (    \lambda    \mathit{x}  \mathord{:}  \ottnt{T_{{\mathrm{2}}}}  .  \mathit{x}   )  \, [  \ottnt{v'_{{\mathrm{2}}}}  \ottsym{/}  \mathit{x}  ]  )    \ottsym{:}    \mathit{self}  (  \ottnt{T_{{\mathrm{2}}}} ,  \ottnt{e} \, \mathit{x}  )   \rightarrow  \ottnt{T_{{\mathrm{2}}}} ;  \theta ;   \delta    [  \,  (  \ottnt{v'_{{\mathrm{1}}}}  ,  \ottnt{v'_{{\mathrm{2}}}}  ) /  \mathit{x}  \,  ]   .
     \end{array}\]
   Since $ \ottnt{v''_{{\mathrm{1}}}}  \simeq_{\mathtt{v} }  \ottnt{v''_{{\mathrm{2}}}}   \ottsym{:}    \mathit{self}  (  \ottnt{T_{{\mathrm{2}}}} ,  \ottnt{e} \, \mathit{x}  )  ;  \theta ;   \delta    [  \,  (  \ottnt{v'_{{\mathrm{1}}}}  ,  \ottnt{v'_{{\mathrm{2}}}}  ) /  \mathit{x}  \,  ]   $,
   we have (\ref{req:fh-self-elim-self2type-fun-two}).

  \case $\ottnt{T}  \ottsym{=}   \forall   \alpha  .  \ottnt{T'} $: Straightforward by the IH.

  \case $\ottnt{T}  \ottsym{=}   \{  \mathit{x}  \mathord{:}  \ottnt{T'}   \mathop{\mid}   \ottnt{e'}  \} $:
   Without loss of generality, we can suppose that
   $\mathit{x} \, \notin \,  \mathit{dom}  (  \delta  ) $.
   We have $ \mathit{self}  (  \ottnt{T} ,  \ottnt{e}  )   \ottsym{=}   \{  \mathit{x}  \mathord{:}  \ottnt{T''}   \mathop{\mid}    \mathsf{let}  ~  \mathit{x}  \mathord{:}  \ottnt{T'}  \equal  \langle  \ottnt{T''}  \Rightarrow  \ottnt{T'}  \rangle   ^{ \ell }  \, \mathit{x}  ~ \ottliteralin ~  \ottnt{e'}   \} $
   where $\ottnt{T''}  \ottsym{=}   \mathit{self}  (  \ottnt{T'} ,  \langle   \{  \mathit{x}  \mathord{:}  \ottnt{T'}   \mathop{\mid}   \ottnt{e'}  \}   \Rightarrow  \ottnt{T'}  \rangle   ^{ \ell }  \, \ottnt{e}  ) $.
   By \E{Red}/\R{Forget},
   \[
     \theta_{{\mathrm{1}}}  (   \delta_{{\mathrm{1}}}  (  \langle   \mathit{self}  (  \ottnt{T} ,  \ottnt{e}  )   \Rightarrow  \ottnt{T}  \rangle   ^{ \ell }   )   )  \, \ottnt{v_{{\mathrm{1}}}}  \longrightarrow   \theta_{{\mathrm{1}}}  (   \delta_{{\mathrm{1}}}  (  \langle  \ottnt{T''}  \Rightarrow   \{  \mathit{x}  \mathord{:}  \ottnt{T'}   \mathop{\mid}   \ottnt{e'}  \}   \rangle   ^{ \ell }   )   )  \, \ottnt{v_{{\mathrm{1}}}}.
   \]
   It suffices to show that
   \[
      \theta_{{\mathrm{1}}}  (   \delta_{{\mathrm{1}}}  (  \langle  \ottnt{T''}  \Rightarrow   \{  \mathit{x}  \mathord{:}  \ottnt{T'}   \mathop{\mid}   \ottnt{e'}  \}   \rangle   ^{ \ell }   )   )  \, \ottnt{v_{{\mathrm{1}}}}  \simeq_{\mathtt{e} }  \ottnt{v_{{\mathrm{2}}}}   \ottsym{:}    \{  \mathit{x}  \mathord{:}  \ottnt{T'}   \mathop{\mid}   \ottnt{e'}  \}  ;  \theta ;  \delta .
   \]
   By the IH,
   $\Gamma  \vdash  \langle  \ottnt{T''}  \Rightarrow  \ottnt{T'}  \rangle   ^{ \ell }  \,  \mathrel{ \simeq }  \,   \lambda    \mathit{x}  \mathord{:}  \ottnt{T'}  .  \mathit{x}   \ottsym{:}  \ottnt{T''}  \rightarrow  \ottnt{T'}$, and so
   $  \theta_{{\mathrm{1}}}  (   \delta_{{\mathrm{1}}}  (  \langle  \ottnt{T''}  \Rightarrow  \ottnt{T'}  \rangle   ^{ \ell }   )   )   \simeq_{\mathtt{v} }   \theta_{{\mathrm{2}}}  (   \delta_{{\mathrm{2}}}  (    \lambda    \mathit{x}  \mathord{:}  \ottnt{T'}  .  \mathit{x}   )   )    \ottsym{:}   \ottnt{T''}  \rightarrow  \ottnt{T'} ;  \theta ;  \delta $.
   Since $ \ottnt{v_{{\mathrm{1}}}}  \simeq_{\mathtt{v} }  \ottnt{v_{{\mathrm{2}}}}   \ottsym{:}    \mathit{self}  (  \ottnt{T} ,  \ottnt{e}  )  ;  \theta ;  \delta $,
   we have $ \ottnt{v_{{\mathrm{1}}}}  \simeq_{\mathtt{v} }  \ottnt{v_{{\mathrm{2}}}}   \ottsym{:}   \ottnt{T''} ;  \theta ;  \delta $.
   Thus,
   $  \theta_{{\mathrm{1}}}  (   \delta_{{\mathrm{1}}}  (  \langle  \ottnt{T''}  \Rightarrow  \ottnt{T'}  \rangle   ^{ \ell }   )   )  \, \ottnt{v_{{\mathrm{1}}}}  \simeq_{\mathtt{e} }  \ottnt{v_{{\mathrm{2}}}}   \ottsym{:}   \ottnt{T'} ;  \theta ;  \delta $.
   By definition, there exists some $\ottnt{v'_{{\mathrm{1}}}}$ such that
   $ \theta_{{\mathrm{1}}}  (   \delta_{{\mathrm{1}}}  (  \langle  \ottnt{T''}  \Rightarrow  \ottnt{T'}  \rangle   ^{ \ell }   )   )  \, \ottnt{v_{{\mathrm{1}}}}  \longrightarrow^{\ast}  \ottnt{v'_{{\mathrm{1}}}}$ and
   $ \ottnt{v'_{{\mathrm{1}}}}  \simeq_{\mathtt{v} }  \ottnt{v_{{\mathrm{2}}}}   \ottsym{:}   \ottnt{T'} ;  \theta ;  \delta $.
   By \R{Forget} and \R{PreCheck},
   $ \theta_{{\mathrm{1}}}  (   \delta_{{\mathrm{1}}}  (  \langle  \ottnt{T''}  \Rightarrow   \{  \mathit{x}  \mathord{:}  \ottnt{T'}   \mathop{\mid}   \ottnt{e'}  \}   \rangle   ^{ \ell }   )   )  \, \ottnt{v_{{\mathrm{1}}}}  \longrightarrow^{\ast}   \theta_{{\mathrm{1}}}  (   \delta_{{\mathrm{1}}}  (  \langle   \{  \mathit{x}  \mathord{:}  \ottnt{T'}   \mathop{\mid}   \ottnt{e'}  \}   \ottsym{,}  \ottnt{e'} \, [  \ottnt{v'_{{\mathrm{1}}}}  \ottsym{/}  \mathit{x}  ]  \ottsym{,}  \ottnt{v'_{{\mathrm{1}}}}  \rangle   ^{ \ell }   )   ) $.
   Thus, it suffices to show that
   \[
      \theta_{{\mathrm{1}}}  (   \delta_{{\mathrm{1}}}  (  \langle   \{  \mathit{x}  \mathord{:}  \ottnt{T'}   \mathop{\mid}   \ottnt{e'}  \}   \ottsym{,}  \ottnt{e'} \, [  \ottnt{v'_{{\mathrm{1}}}}  \ottsym{/}  \mathit{x}  ]  \ottsym{,}  \ottnt{v'_{{\mathrm{1}}}}  \rangle   ^{ \ell }   )   )   \simeq_{\mathtt{e} }  \ottnt{v_{{\mathrm{2}}}}   \ottsym{:}    \{  \mathit{x}  \mathord{:}  \ottnt{T'}   \mathop{\mid}   \ottnt{e'}  \}  ;  \theta ;  \delta .
   \]
   Since $ \ottnt{v_{{\mathrm{1}}}}  \simeq_{\mathtt{v} }  \ottnt{v_{{\mathrm{2}}}}   \ottsym{:}    \mathit{self}  (  \ottnt{T} ,  \ottnt{e}  )  ;  \theta ;  \delta $,
   we have $ \theta_{{\mathrm{1}}}  (   \delta_{{\mathrm{1}}}  (   \mathsf{let}  ~  \mathit{x}  \mathord{:}  \ottnt{T'}  \equal  \langle  \ottnt{T''}  \Rightarrow  \ottnt{T'}  \rangle   ^{ \ell }  \, \ottnt{v_{{\mathrm{1}}}}  ~ \ottliteralin ~  \ottnt{e'}   )   )   \longrightarrow^{\ast}   \mathsf{true} $.
   Since $ \theta_{{\mathrm{1}}}  (   \delta_{{\mathrm{1}}}  (  \langle  \ottnt{T''}  \Rightarrow  \ottnt{T'}  \rangle   ^{ \ell }   )   )  \, \ottnt{v_{{\mathrm{1}}}}  \longrightarrow^{\ast}  \ottnt{v'_{{\mathrm{1}}}}$,
   we have $ \theta_{{\mathrm{1}}}  (   \delta_{{\mathrm{1}}}  (  \ottnt{e'} \, [  \ottnt{v'_{{\mathrm{1}}}}  \ottsym{/}  \mathit{x}  ]  )   )   \longrightarrow^{\ast}   \mathsf{true} $.
   Thus, it suffices to show that
   \[
     \ottnt{v'_{{\mathrm{1}}}}  \simeq_{\mathtt{v} }  \ottnt{v_{{\mathrm{2}}}}   \ottsym{:}    \{  \mathit{x}  \mathord{:}  \ottnt{T'}   \mathop{\mid}   \ottnt{e'}  \}  ;  \theta ;  \delta .
   \]
   Since $ \ottnt{v'_{{\mathrm{1}}}}  \simeq_{\mathtt{v} }  \ottnt{v_{{\mathrm{2}}}}   \ottsym{:}   \ottnt{T'} ;  \theta ;  \delta $ and
   $ \theta_{{\mathrm{1}}}  (   \delta_{{\mathrm{1}}}  (  \ottnt{e'} \, [  \ottnt{v'_{{\mathrm{1}}}}  \ottsym{/}  \mathit{x}  ]  )   )   \longrightarrow^{\ast}   \mathsf{true} $,
   it suffices to show that
   \[
     \theta_{{\mathrm{2}}}  (   \delta_{{\mathrm{2}}}  (  \ottnt{e'} \, [  \ottnt{v_{{\mathrm{2}}}}  \ottsym{/}  \mathit{x}  ]  )   )   \longrightarrow^{\ast}   \mathsf{true} .
   \]
   Since $\Gamma  \vdash   \{  \mathit{x}  \mathord{:}  \ottnt{T'}   \mathop{\mid}   \ottnt{e'}  \} $, we have
   $ \Gamma  ,  \mathit{x}  \mathord{:}  \ottnt{T'}   \vdash  \ottnt{e'} \,  \mathrel{ \simeq }  \, \ottnt{e'}  \ottsym{:}   \mathsf{Bool} $ by the parametricity (\prop:ref{fh-lr-param}).
   Since $ \Gamma  ,  \mathit{x}  \mathord{:}  \ottnt{T'}   \vdash  \theta  \ottsym{;}   \delta    [  \,  (  \ottnt{v'_{{\mathrm{1}}}}  ,  \ottnt{v_{{\mathrm{2}}}}  ) /  \mathit{x}  \,  ]  $,
   we have \sloppy{$  \theta_{{\mathrm{1}}}  (   \delta_{{\mathrm{1}}}  (  \ottnt{e'} \, [  \ottnt{v'_{{\mathrm{1}}}}  \ottsym{/}  \mathit{x}  ]  )   )   \simeq_{\mathtt{e} }   \theta_{{\mathrm{2}}}  (   \delta_{{\mathrm{2}}}  (  \ottnt{e'} \, [  \ottnt{v_{{\mathrm{2}}}}  \ottsym{/}  \mathit{x}  ]  )   )    \ottsym{:}    \mathsf{Bool}  ;  \theta ;   \delta    [  \,  (  \ottnt{v'_{{\mathrm{1}}}}  ,  \ottnt{v_{{\mathrm{2}}}}  ) /  \mathit{x}  \,  ]   $}.
   Since the term on the left-hand side evaluates to $ \mathsf{true} $,
   we have $ \theta_{{\mathrm{2}}}  (   \delta_{{\mathrm{2}}}  (  \ottnt{e'} \, [  \ottnt{v_{{\mathrm{2}}}}  \ottsym{/}  \mathit{x}  ]  )   )   \longrightarrow^{\ast}   \mathsf{true} $.
   \qedhere
 \end{itemize}
\end{prop}

{}

Now, we prove that casts to selfified types are redundant at run time.

\begin{prop}{fh-self-elim-type2self}
 If $\Gamma  \vdash  \ottnt{e_{{\mathrm{1}}}} \,  \mathrel{ \simeq }  \, \ottnt{e_{{\mathrm{2}}}}  \ottsym{:}  \ottnt{T}$,
 then $\Gamma  \vdash  \langle  \ottnt{T}  \Rightarrow   \mathit{self}  (  \ottnt{T} ,  \ottnt{e_{{\mathrm{1}}}}  )   \rangle   ^{ \ell }  \, \ottnt{e_{{\mathrm{1}}}} \,  \mathrel{ \simeq }  \, \ottnt{e_{{\mathrm{2}}}}  \ottsym{:}   \mathit{self}  (  \ottnt{T} ,  \ottnt{e_{{\mathrm{1}}}}  ) $.

 \proof

 By induction on $\ottnt{T}$.
 Let $\Gamma  \vdash  \theta  \ottsym{;}  \delta$.
 We show that
 \[
    \theta_{{\mathrm{1}}}  (   \delta_{{\mathrm{1}}}  (  \langle  \ottnt{T}  \Rightarrow   \mathit{self}  (  \ottnt{T} ,  \ottnt{e_{{\mathrm{1}}}}  )   \rangle   ^{ \ell }  \, \ottnt{e_{{\mathrm{1}}}}  )   )   \simeq_{\mathtt{e} }   \theta_{{\mathrm{2}}}  (   \delta_{{\mathrm{2}}}  (  \ottnt{e_{{\mathrm{2}}}}  )   )    \ottsym{:}    \mathit{self}  (  \ottnt{T} ,  \ottnt{e_{{\mathrm{1}}}}  )  ;  \theta ;  \delta .
 \]
 Since $\Gamma  \vdash  \ottnt{e_{{\mathrm{1}}}} \,  \mathrel{ \simeq }  \, \ottnt{e_{{\mathrm{2}}}}  \ottsym{:}  \ottnt{T}$,
 we have $  \theta_{{\mathrm{1}}}  (   \delta_{{\mathrm{1}}}  (  \ottnt{e_{{\mathrm{1}}}}  )   )   \simeq_{\mathtt{e} }   \theta_{{\mathrm{2}}}  (   \delta_{{\mathrm{2}}}  (  \ottnt{e_{{\mathrm{2}}}}  )   )    \ottsym{:}   \ottnt{T} ;  \theta ;  \delta $.
 If $ \theta_{{\mathrm{1}}}  (   \delta_{{\mathrm{1}}}  (  \ottnt{e_{{\mathrm{1}}}}  )   ) $ and $ \theta_{{\mathrm{2}}}  (   \delta_{{\mathrm{2}}}  (  \ottnt{e_{{\mathrm{2}}}}  )   ) $ raise blame,
 we finish.
 Otherwise,
 $ \theta_{{\mathrm{1}}}  (   \delta_{{\mathrm{1}}}  (  \ottnt{e_{{\mathrm{1}}}}  )   )   \longrightarrow^{\ast}  \ottnt{v_{{\mathrm{1}}}}$ and
 $ \theta_{{\mathrm{2}}}  (   \delta_{{\mathrm{2}}}  (  \ottnt{e_{{\mathrm{2}}}}  )   )   \longrightarrow^{\ast}  \ottnt{v_{{\mathrm{2}}}}$ for some $\ottnt{v_{{\mathrm{1}}}}$ and $\ottnt{v_{{\mathrm{2}}}}$, and
 it suffices to show that
 \[
    \theta_{{\mathrm{1}}}  (   \delta_{{\mathrm{1}}}  (  \langle  \ottnt{T}  \Rightarrow   \mathit{self}  (  \ottnt{T} ,  \ottnt{e_{{\mathrm{1}}}}  )   \rangle   ^{ \ell }   )   )  \, \ottnt{v_{{\mathrm{1}}}}  \simeq_{\mathtt{e} }  \ottnt{v_{{\mathrm{2}}}}   \ottsym{:}    \mathit{self}  (  \ottnt{T} ,  \ottnt{e_{{\mathrm{1}}}}  )  ;  \theta ;  \delta .
 \]
 We also have $ \ottnt{v_{{\mathrm{1}}}}  \simeq_{\mathtt{v} }  \ottnt{v_{{\mathrm{2}}}}   \ottsym{:}   \ottnt{T} ;  \theta ;  \delta $.
 By case analysis on $\ottnt{T}$.
 \begin{itemize}
  \case $\ottnt{T}  \ottsym{=}  \ottnt{B}$:
   We have $ \mathit{self}  (  \ottnt{T} ,  \ottnt{e_{{\mathrm{1}}}}  )   \ottsym{=}   \{  \mathit{x}  \mathord{:}  \ottnt{B}   \mathop{\mid}    \mathit{x}  \mathrel{=}_{ \ottnt{B} }  \ottnt{e_{{\mathrm{1}}}}   \} $.
   Since $ \theta_{{\mathrm{1}}}  (   \delta_{{\mathrm{1}}}  (  \ottnt{e_{{\mathrm{1}}}}  )   )   \longrightarrow^{\ast}  \ottnt{v_{{\mathrm{1}}}}$,
   we have
   $ \theta_{{\mathrm{1}}}  (   \delta_{{\mathrm{1}}}  (  \langle  \ottnt{T}  \Rightarrow   \mathit{self}  (  \ottnt{T} ,  \ottnt{e_{{\mathrm{1}}}}  )   \rangle   ^{ \ell }   )   )  \, \ottnt{v_{{\mathrm{1}}}}  \longrightarrow^{\ast}  \ottnt{v_{{\mathrm{1}}}}$.
   Thus, it suffices to show that
   \[
     \ottnt{v_{{\mathrm{1}}}}  \simeq_{\mathtt{v} }  \ottnt{v_{{\mathrm{2}}}}   \ottsym{:}    \{  \mathit{x}  \mathord{:}  \ottnt{B}   \mathop{\mid}    \mathit{x}  \mathrel{=}_{ \ottnt{B} }  \ottnt{e_{{\mathrm{1}}}}   \}  ;  \theta ;  \delta .
   \]
   Since $ \ottnt{v_{{\mathrm{1}}}}  \simeq_{\mathtt{v} }  \ottnt{v_{{\mathrm{2}}}}   \ottsym{:}   \ottnt{B} ;  \theta ;  \delta $ and
   $ \theta_{{\mathrm{1}}}  (   \delta_{{\mathrm{1}}}  (   \mathit{x}  \mathrel{=}_{ \ottnt{B} }  \ottnt{e_{{\mathrm{1}}}}   )  \, [  \ottnt{v_{{\mathrm{1}}}}  \ottsym{/}  \mathit{x}  ]  )   \longrightarrow^{\ast}   \mathsf{true} $,
   it suffices to show that
   \[
     \theta_{{\mathrm{2}}}  (   \delta_{{\mathrm{2}}}  (   \mathit{x}  \mathrel{=}_{ \ottnt{B} }  \ottnt{e_{{\mathrm{1}}}}   )  \, [  \ottnt{v_{{\mathrm{2}}}}  \ottsym{/}  \mathit{x}  ]  )   \longrightarrow^{\ast}   \mathsf{true} .
   \]
   Since $\Gamma  \vdash  \ottnt{e_{{\mathrm{1}}}}  \ottsym{:}  \ottnt{B}$,
   we have $\Gamma  \vdash  \ottnt{e_{{\mathrm{1}}}} \,  \mathrel{ \simeq }  \, \ottnt{e_{{\mathrm{1}}}}  \ottsym{:}  \ottnt{B}$ by the parametricity.
   Thus, by definition, $ \theta_{{\mathrm{2}}}  (   \delta_{{\mathrm{2}}}  (  \ottnt{e_{{\mathrm{1}}}}  )   )   \longrightarrow^{\ast}  \ottnt{v_{{\mathrm{1}}}}$.
   Since $\ottnt{v_{{\mathrm{1}}}}  \ottsym{=}  \ottnt{v_{{\mathrm{2}}}}$ from $ \ottnt{v_{{\mathrm{1}}}}  \simeq_{\mathtt{v} }  \ottnt{v_{{\mathrm{2}}}}   \ottsym{:}   \ottnt{B} ;  \theta ;  \delta $,
   we finish.

  \case $\ottnt{T}  \ottsym{=}  \alpha$: Obvious since $ \mathit{self}  (  \ottnt{T} ,  \ottnt{e_{{\mathrm{1}}}}  )   \ottsym{=}  \alpha$ and
   a reflexive cast is logically related to an identify function
   (\prop:ref{fh-lr-elim-refl-cast}).

  \case $\ottnt{T}  \ottsym{=}   \mathit{x} \mathord{:} \ottnt{T_{{\mathrm{1}}}} \rightarrow \ottnt{T_{{\mathrm{2}}}} $:
   Similar to the case of function types in \prop:ref{fh-self-elim-self2type}.

  \case $\ottnt{T}  \ottsym{=}   \forall   \alpha  .  \ottnt{T'} $: Straightforward by the IH.

  \case $\ottnt{T}  \ottsym{=}   \{  \mathit{x}  \mathord{:}  \ottnt{T'}   \mathop{\mid}   \ottnt{e'}  \} $:
   Without loss of generality, we can suppose that
   $\mathit{x} \, \notin \,  \mathit{dom}  (  \delta  ) $.
   We have $ \mathit{self}  (  \ottnt{T} ,  \ottnt{e_{{\mathrm{1}}}}  )   \ottsym{=}   \{  \mathit{x}  \mathord{:}  \ottnt{T''}   \mathop{\mid}    \mathsf{let}  ~  \mathit{x}  \mathord{:}  \ottnt{T'}  \equal  \langle  \ottnt{T''}  \Rightarrow  \ottnt{T'}  \rangle   ^{ \ell }  \, \mathit{x}  ~ \ottliteralin ~  \ottnt{e'}   \} $
   where $\ottnt{T''}  \ottsym{=}   \mathit{self}  (  \ottnt{T'} ,  \langle   \{  \mathit{x}  \mathord{:}  \ottnt{T'}   \mathop{\mid}   \ottnt{e'}  \}   \Rightarrow  \ottnt{T'}  \rangle   ^{ \ell }  \, \ottnt{e_{{\mathrm{1}}}}  ) $.
   By \E{Red}/\R{Forget},
   \[
     \theta_{{\mathrm{1}}}  (   \delta_{{\mathrm{1}}}  (  \langle  \ottnt{T}  \Rightarrow   \mathit{self}  (  \ottnt{T} ,  \ottnt{e_{{\mathrm{1}}}}  )   \rangle   ^{ \ell }   )   )  \, \ottnt{v_{{\mathrm{1}}}}  \longrightarrow   \theta_{{\mathrm{1}}}  (   \delta_{{\mathrm{1}}}  (  \langle  \ottnt{T'}  \Rightarrow   \mathit{self}  (  \ottnt{T} ,  \ottnt{e_{{\mathrm{1}}}}  )   \rangle   ^{ \ell }   )   )  \, \ottnt{v_{{\mathrm{1}}}}.
   \]
   Thus, it suffices to show that
   \[
      \theta_{{\mathrm{1}}}  (   \delta_{{\mathrm{1}}}  (  \langle  \ottnt{T'}  \Rightarrow   \mathit{self}  (  \ottnt{T} ,  \ottnt{e_{{\mathrm{1}}}}  )   \rangle   ^{ \ell }   )   )  \, \ottnt{v_{{\mathrm{1}}}}  \simeq_{\mathtt{e} }  \ottnt{v_{{\mathrm{2}}}}   \ottsym{:}    \mathit{self}  (  \ottnt{T} ,  \ottnt{e_{{\mathrm{1}}}}  )  ;  \theta ;  \delta .
   \]

   We first show
   \begin{equation}
      \theta_{{\mathrm{1}}}  (   \delta_{{\mathrm{1}}}  (  \langle  \ottnt{T'}  \Rightarrow  \ottnt{T''}  \rangle   ^{ \ell }  \, \ottsym{(}  \ottsym{(}    \lambda    \mathit{x}  \mathord{:}  \ottnt{T'}  .  \mathit{x}   \ottsym{)} \, \ottnt{v_{{\mathrm{1}}}}  \ottsym{)}  )   )   \simeq_{\mathtt{e} }  \ottnt{v_{{\mathrm{2}}}}   \ottsym{:}   \ottnt{T''} ;  \theta ;  \delta .
     \label{req:fh-self-elim-type2self-refine-one}
   \end{equation}
   by using the equivalence-respecting property (\prop:ref{fh-lr-comp-equiv-res}).
   {\iffull
\else
   We can show $\Gamma  \vdash  \langle   \{  \mathit{x}  \mathord{:}  \ottnt{T'}   \mathop{\mid}   \ottnt{e'}  \}   \Rightarrow  \ottnt{T'}  \rangle   ^{ \ell }  \,  \mathrel{ \simeq }  \,   \lambda    \mathit{x}  \mathord{:}   \{  \mathit{x}  \mathord{:}  \ottnt{T'}   \mathop{\mid}   \ottnt{e'}  \}   .  \mathit{x}   \ottsym{:}   \{  \mathit{x}  \mathord{:}  \ottnt{T'}   \mathop{\mid}   \ottnt{e'}  \}   \rightarrow  \ottnt{T'}$
   easily from the fact that a reflexive cast and an identity function are
   logically related (\prop:ref{fh-lr-elim-refl-cast}).
   \fi}
   Since the logical relation is compatible, we have
   $\Gamma  \vdash  \langle   \{  \mathit{x}  \mathord{:}  \ottnt{T'}   \mathop{\mid}   \ottnt{e'}  \}   \Rightarrow  \ottnt{T'}  \rangle   ^{ \ell }  \, \ottnt{e_{{\mathrm{1}}}} \,  \mathrel{ \simeq }  \, \ottsym{(}    \lambda    \mathit{x}  \mathord{:}   \{  \mathit{x}  \mathord{:}  \ottnt{T'}   \mathop{\mid}   \ottnt{e'}  \}   .  \mathit{x}   \ottsym{)} \, \ottnt{e_{{\mathrm{2}}}}  \ottsym{:}  \ottnt{T'}$.
   Thus, by the IH,
   \[
    \Gamma  \vdash  \langle  \ottnt{T'}  \Rightarrow  \ottnt{T''}  \rangle   ^{ \ell }  \, \ottsym{(}  \langle   \{  \mathit{x}  \mathord{:}  \ottnt{T'}   \mathop{\mid}   \ottnt{e'}  \}   \Rightarrow  \ottnt{T'}  \rangle   ^{ \ell }  \, \ottnt{e_{{\mathrm{1}}}}  \ottsym{)} \,  \mathrel{ \simeq }  \, \ottsym{(}    \lambda    \mathit{x}  \mathord{:}   \{  \mathit{x}  \mathord{:}  \ottnt{T'}   \mathop{\mid}   \ottnt{e'}  \}   .  \mathit{x}   \ottsym{)} \, \ottnt{e_{{\mathrm{2}}}}  \ottsym{:}  \ottnt{T''}.
   \]
   Since
   $ \theta_{{\mathrm{1}}}  (   \delta_{{\mathrm{1}}}  (  \ottnt{e_{{\mathrm{1}}}}  )   )   \longrightarrow^{\ast}  \ottnt{v_{{\mathrm{1}}}}$ and
   $ \theta_{{\mathrm{2}}}  (   \delta_{{\mathrm{2}}}  (  \ottnt{e_{{\mathrm{2}}}}  )   )   \longrightarrow^{\ast}  \ottnt{v_{{\mathrm{2}}}}$ and
   $ \theta_{{\mathrm{1}}}  (   \delta_{{\mathrm{1}}}  (  \langle   \{  \mathit{x}  \mathord{:}  \ottnt{T'}   \mathop{\mid}   \ottnt{e'}  \}   \Rightarrow  \ottnt{T'}  \rangle   ^{ \ell }   )   )  \, \ottnt{v_{{\mathrm{1}}}}  \longrightarrow   \theta_{{\mathrm{1}}}  (   \delta_{{\mathrm{1}}}  (  \langle  \ottnt{T'}  \Rightarrow  \ottnt{T'}  \rangle   ^{ \ell }   )   )  \, \ottnt{v_{{\mathrm{1}}}}$
   by \E{Red}/\R{Forget},
   we have
   \begin{equation}
      \theta_{{\mathrm{1}}}  (   \delta_{{\mathrm{1}}}  (  \langle  \ottnt{T'}  \Rightarrow  \ottnt{T''}  \rangle   ^{ \ell }  \, \ottsym{(}  \langle  \ottnt{T'}  \Rightarrow  \ottnt{T'}  \rangle   ^{ \ell }  \, \ottnt{v_{{\mathrm{1}}}}  \ottsym{)}  )   )   \simeq_{\mathtt{e} }  \ottnt{v_{{\mathrm{2}}}}   \ottsym{:}   \ottnt{T''} ;  \theta ;  \delta .
     \label{eqn:fh-self-elim-type2self-refine-two}
   \end{equation}
   Since an identity function is logically related to a reflexive cast
   (\prop:ref{fh-lr-elim-refl-cast-right}),
   we have
   \[
    \Gamma  \vdash  \ottsym{(}    \lambda    \mathit{x}  \mathord{:}  \ottnt{T'}  .  \mathit{x}   \ottsym{)} \,  \mathrel{ \simeq }  \, \langle  \ottnt{T'}  \Rightarrow  \ottnt{T'}  \rangle   ^{ \ell }   \ottsym{:}  \ottnt{T'}  \rightarrow  \ottnt{T'}.
   \]
   Thus, by the fundamental property,
   \[
     \Gamma  ,  \mathit{x}  \mathord{:}  \ottnt{T'}   \vdash  \langle  \ottnt{T'}  \Rightarrow  \ottnt{T''}  \rangle   ^{ \ell }  \, \ottsym{(}  \ottsym{(}    \lambda    \mathit{x}  \mathord{:}  \ottnt{T'}  .  \mathit{x}   \ottsym{)} \, \mathit{x}  \ottsym{)} \,  \mathrel{ \simeq }  \, \langle  \ottnt{T'}  \Rightarrow  \ottnt{T''}  \rangle   ^{ \ell }  \, \ottsym{(}  \langle  \ottnt{T'}  \Rightarrow  \ottnt{T'}  \rangle   ^{ \ell }  \, \mathit{x}  \ottsym{)}  \ottsym{:}  \ottnt{T''}.
   \]
   Since $ \ottnt{v_{{\mathrm{1}}}}  \simeq_{\mathtt{v} }  \ottnt{v_{{\mathrm{2}}}}   \ottsym{:}    \{  \mathit{x}  \mathord{:}  \ottnt{T'}   \mathop{\mid}   \ottnt{e'}  \}  ;  \theta ;  \delta $,
   we have $ \ottnt{v_{{\mathrm{1}}}}  \simeq_{\mathtt{v} }  \ottnt{v_{{\mathrm{2}}}}   \ottsym{:}   \ottnt{T'} ;  \theta ;  \delta $, and so
   $ \Gamma  ,  \mathit{x}  \mathord{:}  \ottnt{T'}   \vdash  \theta  \ottsym{;}   \delta    [  \,  (  \ottnt{v_{{\mathrm{1}}}}  ,  \ottnt{v_{{\mathrm{2}}}}  ) /  \mathit{x}  \,  ]  $.
   Thus, from the fact that logically related terms are CIU-equivalent
   (\refthm{fh-lr-sound} and \prop:ref{fh-lr-ciu-complete}) and
   the definition of CIU-equivalence,
   \begin{equation}
    \begin{array}{l}
     \emptyset  \vdash   \theta_{{\mathrm{1}}}  (   \delta_{{\mathrm{1}}}  (  \langle  \ottnt{T'}  \Rightarrow  \ottnt{T''}  \rangle   ^{ \ell }  \, \ottsym{(}  \ottsym{(}    \lambda    \mathit{x}  \mathord{:}  \ottnt{T'}  .  \mathit{x}   \ottsym{)} \, \ottnt{v_{{\mathrm{1}}}}  \ottsym{)}  )   )  \, =_\mathsf{ciu} \,  \\  \,  \qquad  \,  \theta_{{\mathrm{1}}}  (   \delta_{{\mathrm{1}}}  (  \langle  \ottnt{T'}  \Rightarrow  \ottnt{T''}  \rangle   ^{ \ell }  \, \ottsym{(}  \langle  \ottnt{T'}  \Rightarrow  \ottnt{T'}  \rangle   ^{ \ell }  \, \ottnt{v_{{\mathrm{1}}}}  \ottsym{)}  )   )   \ottsym{:}   \theta_{{\mathrm{1}}}  (   \delta_{{\mathrm{1}}}  (  \ottnt{T''}  )   ) .
    \end{array}
     \label{eqn:fh-self-elim-type2self-refine-three}
   \end{equation}
   From (\ref{eqn:fh-self-elim-type2self-refine-two})
   and (\ref{eqn:fh-self-elim-type2self-refine-three}),
   the equivalence-respecting property derives (\ref{req:fh-self-elim-type2self-refine-one}).

   From (\ref{req:fh-self-elim-type2self-refine-one}),
   there exists some $\ottnt{v'_{{\mathrm{1}}}}$ such that
   $ \theta_{{\mathrm{1}}}  (   \delta_{{\mathrm{1}}}  (  \langle  \ottnt{T'}  \Rightarrow  \ottnt{T''}  \rangle   ^{ \ell }  \, \ottnt{v_{{\mathrm{1}}}}  )   )   \longrightarrow^{\ast}  \ottnt{v'_{{\mathrm{1}}}}$ and
   $ \ottnt{v'_{{\mathrm{1}}}}  \simeq_{\mathtt{v} }  \ottnt{v_{{\mathrm{2}}}}   \ottsym{:}   \ottnt{T''} ;  \theta ;  \delta $.
   Thus,
   \[\begin{array}{l}
     \theta_{{\mathrm{1}}}  (   \delta_{{\mathrm{1}}}  (  \langle  \ottnt{T'}  \Rightarrow   \mathit{self}  (  \ottnt{T} ,  \ottnt{e_{{\mathrm{1}}}}  )   \rangle   ^{ \ell }   )   )  \, \ottnt{v_{{\mathrm{1}}}}  \longrightarrow^{\ast}  \\
      \quad  \theta_{{\mathrm{1}}}  (   \delta_{{\mathrm{1}}}  (  \langle   \mathit{self}  (  \ottnt{T} ,  \ottnt{e_{{\mathrm{1}}}}  )   \ottsym{,}   \mathsf{let}  ~  \mathit{x}  \mathord{:}  \ottnt{T'}  \equal  \langle  \ottnt{T''}  \Rightarrow  \ottnt{T'}  \rangle   ^{ \ell }  \, \ottnt{v'_{{\mathrm{1}}}}  ~ \ottliteralin ~  \ottnt{e'}   \ottsym{,}  \ottnt{v'_{{\mathrm{1}}}}  \rangle   ^{ \ell }   )   ) 
     \end{array}\]
   by \R{Forget} and \R{PreCheck}, and so
   it suffices to show that
   \[
      \theta_{{\mathrm{1}}}  (   \delta_{{\mathrm{1}}}  (  \langle   \mathit{self}  (  \ottnt{T} ,  \ottnt{e_{{\mathrm{1}}}}  )   \ottsym{,}   \mathsf{let}  ~  \mathit{x}  \mathord{:}  \ottnt{T'}  \equal  \langle  \ottnt{T''}  \Rightarrow  \ottnt{T'}  \rangle   ^{ \ell }  \, \ottnt{v'_{{\mathrm{1}}}}  ~ \ottliteralin ~  \ottnt{e'}   \ottsym{,}  \ottnt{v'_{{\mathrm{1}}}}  \rangle   ^{ \ell }   )   )   \simeq_{\mathtt{e} }  \ottnt{v_{{\mathrm{2}}}}   \ottsym{:}    \mathit{self}  (  \ottnt{T} ,  \ottnt{e_{{\mathrm{1}}}}  )  ;  \theta ;  \delta .
   \]
   Since $\Gamma  \vdash  \langle   \{  \mathit{x}  \mathord{:}  \ottnt{T'}   \mathop{\mid}   \ottnt{e'}  \}   \Rightarrow  \ottnt{T'}  \rangle   ^{ \ell }  \, \ottnt{e_{{\mathrm{1}}}}  \ottsym{:}  \ottnt{T'}$,
   we have $\Gamma  \vdash  \langle  \ottnt{T''}  \Rightarrow  \ottnt{T'}  \rangle   ^{ \ell }  \,  \mathrel{ \simeq }  \,   \lambda    \mathit{x}  \mathord{:}  \ottnt{T'}  .  \mathit{x}   \ottsym{:}  \ottnt{T''}  \rightarrow  \ottnt{T'}$
   by \prop:ref{fh-self-elim-self2type}.
   Since $ \ottnt{v'_{{\mathrm{1}}}}  \simeq_{\mathtt{v} }  \ottnt{v_{{\mathrm{2}}}}   \ottsym{:}   \ottnt{T''} ;  \theta ;  \delta $,
   we have $  \theta_{{\mathrm{1}}}  (   \delta_{{\mathrm{1}}}  (  \langle  \ottnt{T''}  \Rightarrow  \ottnt{T'}  \rangle   ^{ \ell }   )   )  \, \ottnt{v'_{{\mathrm{1}}}}  \simeq_{\mathtt{e} }  \ottnt{v_{{\mathrm{2}}}}   \ottsym{:}   \ottnt{T'} ;  \theta ;  \delta $.
   By definition, $ \theta_{{\mathrm{1}}}  (   \delta_{{\mathrm{1}}}  (  \langle  \ottnt{T''}  \Rightarrow  \ottnt{T'}  \rangle   ^{ \ell }   )   )  \, \ottnt{v'_{{\mathrm{1}}}}  \longrightarrow^{\ast}  \ottnt{v''_{{\mathrm{1}}}}$
   for some $\ottnt{v''_{{\mathrm{1}}}}$, and it suffices to show that
   \[
      \theta_{{\mathrm{1}}}  (   \delta_{{\mathrm{1}}}  (  \langle   \mathit{self}  (  \ottnt{T} ,  \ottnt{e_{{\mathrm{1}}}}  )   \ottsym{,}  \ottnt{e'} \, [  \ottnt{v''_{{\mathrm{1}}}}  \ottsym{/}  \mathit{x}  ]  \ottsym{,}  \ottnt{v'_{{\mathrm{1}}}}  \rangle   ^{ \ell }   )   )   \simeq_{\mathtt{e} }  \ottnt{v_{{\mathrm{2}}}}   \ottsym{:}    \mathit{self}  (  \ottnt{T} ,  \ottnt{e_{{\mathrm{1}}}}  )  ;  \theta ;  \delta .
   \]
   We also have $ \ottnt{v''_{{\mathrm{1}}}}  \simeq_{\mathtt{v} }  \ottnt{v_{{\mathrm{2}}}}   \ottsym{:}   \ottnt{T'} ;  \theta ;  \delta $.

   Since $\Gamma  \vdash   \{  \mathit{x}  \mathord{:}  \ottnt{T'}   \mathop{\mid}   \ottnt{e'}  \} $,
   we have $ \Gamma  ,  \mathit{x}  \mathord{:}  \ottnt{T'}   \vdash  \ottnt{e'} \,  \mathrel{ \simeq }  \, \ottnt{e'}  \ottsym{:}   \mathsf{Bool} $
   by the parametricity (\prop:ref{fh-lr-param}).
   Since $ \Gamma  ,  \mathit{x}  \mathord{:}  \ottnt{T'}   \vdash  \theta  \ottsym{;}   \delta    [  \,  (  \ottnt{v''_{{\mathrm{1}}}}  ,  \ottnt{v_{{\mathrm{2}}}}  ) /  \mathit{x}  \,  ]  $,
   we have $  \theta_{{\mathrm{1}}}  (   \delta_{{\mathrm{1}}}  (  \ottnt{e'} \, [  \ottnt{v''_{{\mathrm{1}}}}  \ottsym{/}  \mathit{x}  ]  )   )   \simeq_{\mathtt{e} }   \theta_{{\mathrm{2}}}  (   \delta_{{\mathrm{2}}}  (  \ottnt{e'} \, [  \ottnt{v_{{\mathrm{2}}}}  \ottsym{/}  \mathit{x}  ]  )   )    \ottsym{:}    \mathsf{Bool}  ;  \theta ;   \delta    [  \,  (  \ottnt{v''_{{\mathrm{1}}}}  ,  \ottnt{v_{{\mathrm{2}}}}  ) /  \mathit{x}  \,  ]   $.
   Since $ \ottnt{v_{{\mathrm{1}}}}  \simeq_{\mathtt{v} }  \ottnt{v_{{\mathrm{2}}}}   \ottsym{:}    \{  \mathit{x}  \mathord{:}  \ottnt{T'}   \mathop{\mid}   \ottnt{e'}  \}  ;  \theta ;  \delta $,
   we have $ \theta_{{\mathrm{2}}}  (   \delta_{{\mathrm{2}}}  (  \ottnt{e'} \, [  \ottnt{v_{{\mathrm{2}}}}  \ottsym{/}  \mathit{x}  ]  )   )   \longrightarrow^{\ast}   \mathsf{true} $, so
   $ \theta_{{\mathrm{1}}}  (   \delta_{{\mathrm{1}}}  (  \ottnt{e'} \, [  \ottnt{v''_{{\mathrm{1}}}}  \ottsym{/}  \mathit{x}  ]  )   )   \longrightarrow^{\ast}   \mathsf{true} $.
   Thus, it suffices to show that
   \[
     \ottnt{v'_{{\mathrm{1}}}}  \simeq_{\mathtt{v} }  \ottnt{v_{{\mathrm{2}}}}   \ottsym{:}    \mathit{self}  (  \ottnt{T} ,  \ottnt{e_{{\mathrm{1}}}}  )  ;  \theta ;  \delta .
   \]
   We have it by the discussion above. \qedhere
 \end{itemize}
\end{prop}

\begin{prop}[type=cor,name=Selfification Cast Elimination]{fh-self-elim}
 If $\Gamma  \vdash  \ottnt{e}  \ottsym{:}  \ottnt{T}$,
 then $\Gamma  \vdash  \langle  \ottnt{T}  \Rightarrow   \mathit{self}  (  \ottnt{T} ,  \ottnt{e}  )   \rangle   ^{ \ell }  \, \ottnt{e} \,  \\  \, =_\mathsf{ctx} \, \ottnt{e}  \ottsym{:}   \mathit{self}  (  \ottnt{T} ,  \ottnt{e}  ) $.

 \proof

 By the parametricity (\prop:ref{fh-lr-param}) and
 \prop:ref{fh-self-elim-type2self}.
\end{prop}

\subsection{Cast Decomposition}
\label{sec:reasoning-cast-decomp}
The upcast elimination enables us to eliminate redundant casts, but there are cases
that nonredundant casts produce redundant ones.
For example, let us consider
$\langle   \{  \mathit{x}  \mathord{:}   \mathsf{Int}    \mathop{\mid}   \mathit{x}  \mathrel{\neq}  \ottsym{0}  \}   \rightarrow   \{  \mathit{x}  \mathord{:}   \mathsf{Int}    \mathop{\mid}    \mathsf{prime?}  \, \mathit{x}  \}   \Rightarrow   \{  \mathit{x}  \mathord{:}   \mathsf{Int}    \mathop{\mid}    \mathit{x}  \mathrel{\ge} \ottsym{0}   \}   \rightarrow   \{  \mathit{x}  \mathord{:}   \mathsf{Int}    \mathop{\mid}    \mathit{x}  \mathrel{>} \ottsym{0}   \}   \rangle   ^{ \ell } $,
which is not an upcast because the argument check may fail.
This will be decomposed into two casts at run time:
one for the domain type---$\langle   \{  \mathit{x}  \mathord{:}   \mathsf{Int}    \mathop{\mid}    \mathit{x}  \mathrel{\ge} \ottsym{0}   \}   \Rightarrow   \{  \mathit{x}  \mathord{:}   \mathsf{Int}    \mathop{\mid}   \mathit{x}  \mathrel{\neq}  \ottsym{0}  \}   \rangle   ^{ \ell } $---and
one for the codomain type---$\langle   \{  \mathit{x}  \mathord{:}   \mathsf{Int}    \mathop{\mid}    \mathsf{prime?}  \, \mathit{x}  \}   \Rightarrow   \{  \mathit{x}  \mathord{:}   \mathsf{Int}    \mathop{\mid}    \mathit{x}  \mathrel{>} \ottsym{0}   \}   \rangle   ^{ \ell } $.
As mentioned above, the cast for the domain type cannot be eliminated because it
would fail if applied to zero.
By contrast, the cast for the codomain type is an upcast and so can be
eliminated without changing the behavior of a program.

Static decomposition of casts makes it possible to eliminate as many redundant
casts as possible.
For example, it allows us to statically decompose casts for function types into
ones for domain types and codomain types and eliminate them if they are
upcasts.
In what follows, we show how casts can be decomposed.
\begin{prop}{fh-cc-fun}
 If $\Gamma  \vdash  \langle   \mathit{x} \mathord{:} \ottnt{T_{{\mathrm{11}}}} \rightarrow \ottnt{T_{{\mathrm{12}}}}   \Rightarrow   \mathit{x} \mathord{:} \ottnt{T_{{\mathrm{21}}}} \rightarrow \ottnt{T_{{\mathrm{22}}}}   \rangle   ^{ \ell }   \ottsym{:}  \ottsym{(}   \mathit{x} \mathord{:} \ottnt{T_{{\mathrm{11}}}} \rightarrow \ottnt{T_{{\mathrm{12}}}}   \ottsym{)}  \rightarrow  \ottsym{(}   \mathit{x} \mathord{:} \ottnt{T_{{\mathrm{21}}}} \rightarrow \ottnt{T_{{\mathrm{22}}}}   \ottsym{)}$,
 then $\Gamma  \vdash  \langle   \mathit{x} \mathord{:} \ottnt{T_{{\mathrm{11}}}} \rightarrow \ottnt{T_{{\mathrm{12}}}}   \Rightarrow   \mathit{x} \mathord{:} \ottnt{T_{{\mathrm{21}}}} \rightarrow \ottnt{T_{{\mathrm{22}}}}   \rangle   ^{ \ell }  \,  \mathrel{ \simeq }  \,   \lambda    \mathit{z}  \mathord{:}  \ottsym{(}   \mathit{x} \mathord{:} \ottnt{T_{{\mathrm{11}}}} \rightarrow \ottnt{T_{{\mathrm{12}}}}   \ottsym{)}  .    \lambda    \mathit{x}  \mathord{:}  \ottnt{T_{{\mathrm{21}}}}  .   \mathsf{let}  ~  \mathit{y}  \mathord{:}  \ottnt{T_{{\mathrm{11}}}}  \equal  \langle  \ottnt{T_{{\mathrm{21}}}}  \Rightarrow  \ottnt{T_{{\mathrm{11}}}}  \rangle   ^{ \ell }  \, \mathit{x}  ~ \ottliteralin ~  \langle  \ottnt{T_{{\mathrm{12}}}} \, [  \mathit{y}  \ottsym{/}  \mathit{x}  ]  \Rightarrow  \ottnt{T_{{\mathrm{22}}}}  \rangle   ^{ \ell }     \, \ottsym{(}  \mathit{z} \, \mathit{y}  \ottsym{)}  \ottsym{:}  \ottsym{(}   \mathit{x} \mathord{:} \ottnt{T_{{\mathrm{11}}}} \rightarrow \ottnt{T_{{\mathrm{12}}}}   \ottsym{)}  \rightarrow  \ottsym{(}   \mathit{x} \mathord{:} \ottnt{T_{{\mathrm{21}}}} \rightarrow \ottnt{T_{{\mathrm{22}}}}   \ottsym{)}$.

 \proof

 By following \E{Red}/\R{Fun} and the parametricity.
\end{prop}

\begin{prop}{fh-cc-forall}
 If $\Gamma  \vdash  \langle   \forall   \alpha  .  \ottnt{T_{{\mathrm{1}}}}   \Rightarrow   \forall   \alpha  .  \ottnt{T_{{\mathrm{2}}}}   \rangle   ^{ \ell }   \ottsym{:}  \ottsym{(}   \forall   \alpha  .  \ottnt{T_{{\mathrm{1}}}}   \ottsym{)}  \rightarrow  \ottsym{(}   \forall   \alpha  .  \ottnt{T_{{\mathrm{2}}}}   \ottsym{)}$,
 then $\Gamma  \vdash  \langle   \forall   \alpha  .  \ottnt{T_{{\mathrm{1}}}}   \Rightarrow   \forall   \alpha  .  \ottnt{T_{{\mathrm{2}}}}   \rangle   ^{ \ell }  \,  \mathrel{ \simeq }  \,   \lambda    \mathit{x}  \mathord{:}   \forall   \alpha  .  \ottnt{T_{{\mathrm{1}}}}   .   \Lambda\!  \, \alpha  .~  \langle  \ottnt{T_{{\mathrm{1}}}}  \Rightarrow  \ottnt{T_{{\mathrm{2}}}}  \rangle   ^{ \ell }   \, \ottsym{(}  \mathit{x} \, \alpha  \ottsym{)}  \ottsym{:}  \ottsym{(}   \forall   \alpha  .  \ottnt{T_{{\mathrm{1}}}}   \ottsym{)}  \rightarrow  \ottsym{(}   \forall   \alpha  .  \ottnt{T_{{\mathrm{2}}}}   \ottsym{)}$.

 \proof

 By following \E{Red}/\R{Forall} and the parametricity.
\end{prop}

\begin{prop}{fh-cc-precheck}
 If $\Gamma  \vdash  \langle  \ottnt{T_{{\mathrm{1}}}}  \Rightarrow   \{  \mathit{x}  \mathord{:}  \ottnt{T_{{\mathrm{2}}}}   \mathop{\mid}   \ottnt{e_{{\mathrm{2}}}}  \}   \rangle   ^{ \ell }   \ottsym{:}  \ottnt{T_{{\mathrm{1}}}}  \rightarrow   \{  \mathit{x}  \mathord{:}  \ottnt{T_{{\mathrm{2}}}}   \mathop{\mid}   \ottnt{e_{{\mathrm{2}}}}  \} $,
 then $\Gamma  \vdash  \langle  \ottnt{T_{{\mathrm{1}}}}  \Rightarrow   \{  \mathit{x}  \mathord{:}  \ottnt{T_{{\mathrm{2}}}}   \mathop{\mid}   \ottnt{e_{{\mathrm{2}}}}  \}   \rangle   ^{ \ell }  \,  \mathrel{ \simeq }  \,   \lambda    \mathit{y}  \mathord{:}  \ottnt{T_{{\mathrm{1}}}}  .   \langle\!\langle  \,  \{  \mathit{x}  \mathord{:}  \ottnt{T_{{\mathrm{2}}}}   \mathop{\mid}   \ottnt{e_{{\mathrm{2}}}}  \}   \ottsym{,}  \langle  \ottnt{T_{{\mathrm{1}}}}  \Rightarrow  \ottnt{T_{{\mathrm{2}}}}  \rangle   ^{ \ell }  \, \mathit{y} \,  \rangle\!\rangle  \,  ^{ \ell }    \ottsym{:}  \ottnt{T_{{\mathrm{1}}}}  \rightarrow   \{  \mathit{x}  \mathord{:}  \ottnt{T_{{\mathrm{2}}}}   \mathop{\mid}   \ottnt{e_{{\mathrm{2}}}}  \} $.

 \proof

 By following \E{Red}/\R{PreCheck}, the parametricity, and the fact that,
 if $\langle  \ottnt{T_{{\mathrm{1}}}}  \Rightarrow  \ottnt{T_{{\mathrm{2}}}}  \rangle   ^{ \ell }  \, \ottnt{v}  \longrightarrow^{\ast}  \ottnt{v'}$,
 then $\langle  \ottnt{T_{{\mathrm{1}}}}  \Rightarrow   \{  \mathit{x}  \mathord{:}  \ottnt{T_{{\mathrm{2}}}}   \mathop{\mid}   \ottnt{e}  \}   \rangle   ^{ \ell }  \, \ottnt{v}  \longrightarrow^{\ast}  \langle   \{  \mathit{x}  \mathord{:}  \ottnt{T_{{\mathrm{2}}}}   \mathop{\mid}   \ottnt{e}  \}   \ottsym{,}  \ottnt{e} \, [  \ottnt{v'}  \ottsym{/}  \mathit{x}  ]  \ottsym{,}  \ottnt{v'}  \rangle   ^{ \ell } $.
\end{prop}
Since {\fhfix} allows waiting checks to be open, this decomposition is
type-preserving.

We can show that $\langle   \{  \mathit{x}  \mathord{:}  \ottnt{T_{{\mathrm{1}}}}   \mathop{\mid}   \ottnt{e_{{\mathrm{1}}}}  \}   \Rightarrow  \ottnt{T_{{\mathrm{2}}}}  \rangle   ^{ \ell } $ is logically related to $\langle  \ottnt{T_{{\mathrm{1}}}}  \Rightarrow  \ottnt{T_{{\mathrm{2}}}}  \rangle   ^{ \ell } $,
but it does not preserves the type, which makes further optimization based on
contextual equivalence impossible; note that the transitivity of the logical
relation requires the index types to be the same (see \prop:ref{fh-lr-trans}).
Instead, we show that, if $\langle  \ottnt{T_{{\mathrm{1}}}}  \Rightarrow  \ottnt{T_{{\mathrm{2}}}}  \rangle   ^{ \ell } $ is logically related to term $\ottnt{e}$,
then $\langle   \{  \mathit{x}  \mathord{:}  \ottnt{T_{{\mathrm{1}}}}   \mathop{\mid}   \ottnt{e_{{\mathrm{1}}}}  \}   \Rightarrow  \ottnt{T_{{\mathrm{2}}}}  \rangle   ^{ \ell } $ is also logically related to $\ottnt{e}$.
In this formulation, we can relate $\langle   \{  \mathit{x}  \mathord{:}  \ottnt{T_{{\mathrm{1}}}}   \mathop{\mid}   \ottnt{e_{{\mathrm{1}}}}  \}   \Rightarrow  \ottnt{T_{{\mathrm{2}}}}  \rangle   ^{ \ell } $ to fully optimized
term $\ottnt{e}$.

\begin{prop}{fh-cc-forget}
 If
 $\Gamma  \vdash  \langle   \{  \mathit{x}  \mathord{:}  \ottnt{T_{{\mathrm{1}}}}   \mathop{\mid}   \ottnt{e_{{\mathrm{1}}}}  \}   \Rightarrow  \ottnt{T_{{\mathrm{2}}}}  \rangle   ^{ \ell }   \ottsym{:}   \{  \mathit{x}  \mathord{:}  \ottnt{T_{{\mathrm{1}}}}   \mathop{\mid}   \ottnt{e_{{\mathrm{1}}}}  \}   \rightarrow  \ottnt{T_{{\mathrm{2}}}}$ and
 $\Gamma  \vdash  \langle  \ottnt{T_{{\mathrm{1}}}}  \Rightarrow  \ottnt{T_{{\mathrm{2}}}}  \rangle   ^{ \ell }  \,  \mathrel{ \simeq }  \, \ottnt{e}  \ottsym{:}  \ottnt{T_{{\mathrm{1}}}}  \rightarrow  \ottnt{T_{{\mathrm{2}}}}$,
 then $\Gamma  \vdash  \langle   \{  \mathit{x}  \mathord{:}  \ottnt{T_{{\mathrm{1}}}}   \mathop{\mid}   \ottnt{e_{{\mathrm{1}}}}  \}   \Rightarrow  \ottnt{T_{{\mathrm{2}}}}  \rangle   ^{ \ell }  \,  \mathrel{ \simeq }  \, \ottnt{e}  \ottsym{:}   \{  \mathit{x}  \mathord{:}  \ottnt{T_{{\mathrm{1}}}}   \mathop{\mid}   \ottnt{e_{{\mathrm{1}}}}  \}   \rightarrow  \ottnt{T_{{\mathrm{2}}}}$.

 \proof

 Let $\Gamma  \vdash  \theta  \ottsym{;}  \delta$.
 It suffices to show that
 \[
    \theta_{{\mathrm{1}}}  (   \delta_{{\mathrm{1}}}  (  \langle   \{  \mathit{x}  \mathord{:}  \ottnt{T_{{\mathrm{1}}}}   \mathop{\mid}   \ottnt{e_{{\mathrm{1}}}}  \}   \Rightarrow  \ottnt{T_{{\mathrm{2}}}}  \rangle   ^{ \ell }   )   )   \simeq_{\mathtt{e} }   \theta_{{\mathrm{2}}}  (   \delta_{{\mathrm{2}}}  (  \ottnt{e}  )   )    \ottsym{:}    \{  \mathit{x}  \mathord{:}  \ottnt{T_{{\mathrm{1}}}}   \mathop{\mid}   \ottnt{e_{{\mathrm{1}}}}  \}   \rightarrow  \ottnt{T_{{\mathrm{2}}}} ;  \theta ;  \delta .
 \]
 Since $\Gamma  \vdash  \langle  \ottnt{T_{{\mathrm{1}}}}  \Rightarrow  \ottnt{T_{{\mathrm{2}}}}  \rangle   ^{ \ell }  \,  \mathrel{ \simeq }  \, \ottnt{e}  \ottsym{:}  \ottnt{T_{{\mathrm{1}}}}  \rightarrow  \ottnt{T_{{\mathrm{2}}}}$,
 there exists some $\ottnt{v}$ such that $ \theta_{{\mathrm{2}}}  (   \delta_{{\mathrm{2}}}  (  \ottnt{e}  )   )   \longrightarrow^{\ast}  \ottnt{v}$ and
 $  \theta_{{\mathrm{1}}}  (   \delta_{{\mathrm{1}}}  (  \langle  \ottnt{T_{{\mathrm{1}}}}  \Rightarrow  \ottnt{T_{{\mathrm{2}}}}  \rangle   ^{ \ell }   )   )   \simeq_{\mathtt{v} }  \ottnt{v}   \ottsym{:}   \ottnt{T_{{\mathrm{1}}}}  \rightarrow  \ottnt{T_{{\mathrm{2}}}} ;  \theta ;  \delta $.
 Thus, it suffices to show that, for any $\ottnt{v_{{\mathrm{1}}}}$ and $\ottnt{v_{{\mathrm{2}}}}$ such that
 $ \ottnt{v_{{\mathrm{1}}}}  \simeq_{\mathtt{v} }  \ottnt{v_{{\mathrm{2}}}}   \ottsym{:}    \{  \mathit{x}  \mathord{:}  \ottnt{T_{{\mathrm{1}}}}   \mathop{\mid}   \ottnt{e_{{\mathrm{1}}}}  \}  ;  \theta ;  \delta $,
 \[
    \theta_{{\mathrm{1}}}  (   \delta_{{\mathrm{1}}}  (  \langle   \{  \mathit{x}  \mathord{:}  \ottnt{T_{{\mathrm{1}}}}   \mathop{\mid}   \ottnt{e_{{\mathrm{1}}}}  \}   \Rightarrow  \ottnt{T_{{\mathrm{2}}}}  \rangle   ^{ \ell }   )   )  \, \ottnt{v_{{\mathrm{1}}}}  \simeq_{\mathtt{e} }  \ottnt{v} \, \ottnt{v_{{\mathrm{2}}}}   \ottsym{:}   \ottnt{T_{{\mathrm{2}}}} ;  \theta ;  \delta .
 \]
 Since $ \theta_{{\mathrm{1}}}  (   \delta_{{\mathrm{1}}}  (  \langle   \{  \mathit{x}  \mathord{:}  \ottnt{T_{{\mathrm{1}}}}   \mathop{\mid}   \ottnt{e_{{\mathrm{1}}}}  \}   \Rightarrow  \ottnt{T_{{\mathrm{2}}}}  \rangle   ^{ \ell }   )   )  \, \ottnt{v_{{\mathrm{1}}}}  \longrightarrow   \theta_{{\mathrm{1}}}  (   \delta_{{\mathrm{1}}}  (  \langle  \ottnt{T_{{\mathrm{1}}}}  \Rightarrow  \ottnt{T_{{\mathrm{2}}}}  \rangle   ^{ \ell }   )   )  \, \ottnt{v_{{\mathrm{1}}}}$
 by \E{Red}/\R{Forget}, it suffices to show that
 \[
    \theta_{{\mathrm{1}}}  (   \delta_{{\mathrm{1}}}  (  \langle  \ottnt{T_{{\mathrm{1}}}}  \Rightarrow  \ottnt{T_{{\mathrm{2}}}}  \rangle   ^{ \ell }   )   )  \, \ottnt{v_{{\mathrm{1}}}}  \simeq_{\mathtt{e} }  \ottnt{v} \, \ottnt{v_{{\mathrm{2}}}}   \ottsym{:}   \ottnt{T_{{\mathrm{2}}}} ;  \theta ;  \delta .
 \]
 Since $ \ottnt{v_{{\mathrm{1}}}}  \simeq_{\mathtt{v} }  \ottnt{v_{{\mathrm{2}}}}   \ottsym{:}    \{  \mathit{x}  \mathord{:}  \ottnt{T_{{\mathrm{1}}}}   \mathop{\mid}   \ottnt{e_{{\mathrm{1}}}}  \}  ;  \theta ;  \delta $,
 we have $ \ottnt{v_{{\mathrm{1}}}}  \simeq_{\mathtt{v} }  \ottnt{v_{{\mathrm{2}}}}   \ottsym{:}   \ottnt{T_{{\mathrm{1}}}} ;  \theta ;  \delta $.
 Since $  \theta_{{\mathrm{1}}}  (   \delta_{{\mathrm{1}}}  (  \langle  \ottnt{T_{{\mathrm{1}}}}  \Rightarrow  \ottnt{T_{{\mathrm{2}}}}  \rangle   ^{ \ell }   )   )   \simeq_{\mathtt{v} }  \ottnt{v}   \ottsym{:}   \ottnt{T_{{\mathrm{1}}}}  \rightarrow  \ottnt{T_{{\mathrm{2}}}} ;  \theta ;  \delta $,
 we finish by definition.
\end{prop}

Finally, we show that reflexive casts are redundant.\footnote{We believe that this is
derived from the upcast elimination, but showing that subtyping is reflexive is not
trivial due to substitution on the subtype side in \Sub{Fun}.}
\begin{prop}{fh-cc-refl}
 If $\Gamma  \vdash  \ottnt{T}$, then $\Gamma  \vdash  \langle  \ottnt{T}  \Rightarrow  \ottnt{T}  \rangle   ^{ \ell }  \,  \mathrel{ \simeq }  \,   \lambda    \mathit{x}  \mathord{:}  \ottnt{T}  .  \mathit{x}   \ottsym{:}  \ottnt{T}  \rightarrow  \ottnt{T}$.

 \proof

 By the parametricity (\prop:ref{fh-lr-param}) and \prop:ref{fh-lr-elim-refl-cast}.
\end{prop}

As a byproduct of the cast decomposition, it turns out that our fussy semantics can
simulate Belo et al.'s sloppy semantics.
The sloppy semantics, as shown at the end of \sect{lang-semantics}, eliminates
reflexive casts immediately and checks only the outermost refinement if others
have been ensured already.
It is found that the former is simulated from \prop:ref{fh-cc-refl} and
the second from combination of \prop:ref{fh-cc-precheck,fh-cc-refl}.
As a result, the type system of {\fhfix} turns out to be sound also for the
sloppy semantics despite that the cotermination (\prop:ref{fh-coterm-true}), a
key property for the type soundness, does not hold under the sloppy
semantics~\cite{Sekiyama/Igarashi/Greenberg_2016_TOPLAS}.

\section{Related Work}
\label{sec:relwork}



\subsection{Simply-typed Manifest Contracts}
\label{sec:relwork-simply}
Flanagan~\cite{Flanagan_2006_POPL} introduced a simply typed manifest contract
calculus {\lamh} equipped with a subsumption rule for subtyping.
While the subsumption rule allows us to eliminate upcasts, its naive
introduction results in an occurrence of well typedness at a negative position
in the definition of the type system, especially, in the implication judgment
for refinements; it is unclear whether the type system with such a negative
occurrence is well defined.

To avoid the negative occurrence problem due to the subsumption rule while
keeping that rule, Knowles and Flanagan~\cite{Knowles/Flanagan_2010_TOPLAS}
designed another simply typed manifest contract calculus where the implication
judgment refers to denotations of types instead of well-typed values.
They gave a denotation of each type as a set of terms in the simply typed lambda
calculus and defined a manifest contract calculus equipped with a well-defined
type system using the denotations.
Flanagan and Knowles~\cite{Flanagan_2006_POPL,Knowles/Flanagan_2010_TOPLAS} also
developed a compilation algorithm that transforms possibly ill-typed programs to
well-typed ones by inserting casts everywhere a required type is not a supertype
of an actual type.
The compilation result depends on an external prover that judges implication
between refinements: the more powerful the prover is, the less upcasts are
inserted.
Although how many upcasts are inserted depends on the prover, what prover is
used does not have an influence on the final results of programs because upcasts
should behave as identity functions.
To substantiate this idea, Knowles and
Flanagan~\cite{Knowles/Flanagan_2010_TOPLAS} proved that an upcast is
contextually equivalent to an identity function via a logical relation.

Apart from parametric polymorphism, a major difference between Knowles and
Flanagan~\cite{Knowles/Flanagan_2010_TOPLAS} and our work is the treatment of
the subsumption for subtyping, which has a great influence on the metatheory of
manifest contract calculi.
Knowles and Flanagan allow for the subsumption in the definition of their
calculus.
While their type system with the subsumption rule makes it possible that an
upcast and an identity function have the same type, they need some device to
ensure that the type system is well defined; in fact, their type system is
defined based on semantic typing and semantic subtyping.
By contrast, following Belo et
al.~\cite{Belo/Greenberg/Igarashi/Pierce_2011_ESOP}, we consider subtyping after
defining {\fhfix}.
Since a type system defined in this ``post facto'' approach does not refer to
the implication judgment, it is well defined naturally.
As a result, we can discuss the metatheory, such as the subject reduction, of our
calculus without semantic typing and semantic subtyping.
However, in such a type system, an upcast and an identity function may not have
the same type.
To relate two terms of different types, we introduce \emph{semityped} contextual
equivalence.
Another difference is that, while Knowles and
Flanagan~\cite{Knowles/Flanagan_2010_TOPLAS} show the upcast elimination only for
cases that upcasts are closed,\footnote{Corollary 13 in Knowles and
Flanagan~\cite{Knowles/Flanagan_2010_TOPLAS} states that logically related,
\emph{open} terms $\ottnt{e_{{\mathrm{1}}}}$ and $\ottnt{e_{{\mathrm{2}}}}$ are contextually equivalent, but their
proof shows that result terms of \emph{capture-avoiding substitution} of
$\ottnt{e_{{\mathrm{1}}}}$ and $\ottnt{e_{{\mathrm{2}}}}$ for a variable in any context are observationally equal;
this proof is valid only if $\ottnt{e_{{\mathrm{1}}}}$ and $\ottnt{e_{{\mathrm{2}}}}$ are closed.} we deal with open
upcasts as well.
%

Ou et al.~\cite{Ou/Tan/Mandelbaum/Walker_2004_TCS} studied interoperability of
certified, dependently-typed parts and uncertified, simply-typed ones.
As in manifest contracts, coercion of simply-typed values to dependently-typed
ones is achieved by run-time checking.
Their dependent type system supports refinement types where refinements have to
be pure (i.e., they consist of only variables, constants, and primitive
operations with pure arguments), a subsumption rule for subtyping, and a typing
rule for selfification, which inspires the contract reasoning in
\sect{reasoning-self}.
Unlike the other work on manifest
contracts~\cite{Flanagan_2006_POPL,Knowles/Flanagan_2010_TOPLAS,Belo/Greenberg/Igarashi/Pierce_2011_ESOP,Greenberg_2013_PhD,Sekiyama/Igarashi/Greenberg_2016_TOPLAS,Sekiyama/Igarashi_2017_POPL},
they did not address elimination of run-time coercion.

\subsection{Polymorphic Manifest Contracts}
Belo et al.~\cite{Belo/Greenberg/Igarashi/Pierce_2011_ESOP} studied parametric
polymorphism in manifest contracts.
In particular, they introduced a polymorphic manifest contract calculus,
developed a logical relation, and showed the parametricity and the upcast elimination;
details are described in Greenberg's dissertation~\cite{Greenberg_2013_PhD}.
The semantics of their calculus is sloppy in that refinements that have been
ensured already are not checked at run time.
For example, a reflexive cast returns a given argument immediately because the
argument should been typed at the source type of the cast and satisfy the
refinements in the target type (note that the source and target types of a
reflexive cast are the same).
This sloppiness is important in their proof of the parametricity, especially, to
show that polymorphic cast $\langle  \alpha  \Rightarrow  \alpha  \rangle   ^{ \ell } $ is logically related to itself{}.
However, it turns out that their sloppy semantics does not satisfy
the cotermination, a key property for both the type soundness and
the parametricity~\cite{Sekiyama/Igarashi/Greenberg_2016_TOPLAS}.

Sekiyama et al.~\cite{Sekiyama/Igarashi/Greenberg_2016_TOPLAS} resolved the
problem in the sloppy semantics by equipping casts with \emph{delayed
substitution}, which makes it possible to show the cotermination even under sloppy
semantics.
Furthermore, they also proved the type soundness and the parametricity in the cast
semantics with delayed substitution, while leaving proving the upcast elimination
open.
Although their delayed substitution works well in the sloppy semantics, it makes
the metatheory of a manifest contract calculus, especially, the definition of
substitution, complicated.

We define a polymorphic manifest contract calculus with fussy cast semantics,
where all refinements to be satisfied are checked even if they have been ensured
already.
The fussy cast semantics, which is adopted also by the simply-typed manifest
contract calculus~\cite{Flanagan_2006_POPL,Knowles/Flanagan_2010_TOPLAS} and a
manifest contract calculus for algebraic
data types~\cite{Sekiyama/Nishida/Igarashi_2015_POPL} and mutable
states~\cite{Sekiyama/Igarashi_2017_POPL}, uses usual substitution and
simplifies the metatheory of manifest contract calculi.
Our logical relation for the fussy cast semantics requires interpretations of
type variables to be closed under reduction of applications of reflexive casts
because in the fussy semantics reflexive casts may produce wrappers of given
arguments.
Fortunately, we can construct such an interpretation from any binary relation on
closed values easily, because (well-typed) reflexive casts always succeed.
We furthermore introduce semityped contextual equivalence, show the soundness and
the completeness of the logical relation with respect to it, and prove correctness
of reasoning techniques including the upcast elimination.

%
The \textsc{Sage} language~\cite{Gronski/Knowles/Tomb/Freund/Flanagan_2006_SFPW}
supports key features in polymorphic manifest contracts---general refinements
(i.e., refining refinement types), casts, subtyping, and parametric
\sloppy{polymorphism---as} well as recursive functions, recursive types, the
dynamic type, and the Type:Type discipline, but the parametricity and the upcast
elimination for \textsc{Sage} have not been investigated.
In particular, parametricity for languages equipped with both refinement types
and the dynamic type is left open.

\subsection{Gradual Typing}
Gradual typing~\cite{Siek/Taha_2006_SFPW} is a methodology to achieve a full
spectrum from dynamically typed programs to statically typed ones.
A gradually typed language is considered to be an extension of a static type
system with the dynamic type (or called the unknown type) and it deals with
values of the dynamic type as ones of any other type and vice versa.
Ahmed et al.~\cite{Ahmed/Findler/Siek/Wadler_2011_POPL} and, more recently,
Igarashi et al.~\cite{YIgarashi/Sekiyama/AIgarashi_2017_ICFP} study gradual
typing with parametric polymorphism.
To ensure parametricity, polymorphic gradual typing has to prevent that ones
investigate what type a type variable is instantiated with at run time.
Ahmed et al.\ and Igarashi et al.\ achieved it with help of type bindings, which
are similar to delayed substitution in Sekiyama et
al.~\cite{Sekiyama/Igarashi/Greenberg_2016_TOPLAS},\footnote{Precisely, delayed
substitution comes from type binding.} inspired by a parametric multi-language
system by Matthews and Ahmed~\cite{Matthews/Ahmed_2008_ESOP}.
Ahmed et al.~\cite{Ahmed/Jamner/Siek/Wadler_2017_ICFP} actually proved the
parametricity of the polymorphic gradual typing with type bindings.
Adding the dynamic type to polymorphic manifest contracts is an
interesting future direction.

Gradual typing allows checks of a part of types to be deferred to run-time.
Lehmann and Tanter~\cite{Lehmann/Tanter_2017_POPL} apply this idea to refinement
checking.
They extended refinements with the unknown refinement ``$?$'', which means that
values satisfying this refinement may have some additional information but it is
unknown statically.
In the spirit of gradual typing, their system defers refinement checking with
the unknown refinement to run-time, while checking without the unknown
refinement is performed completely statically.
In other words, the unknown refinement works as a marker that indicates
refinements to be possibly checked at run time.
%
%
Their gradual refinement type system is similar to (the simply typed) manifest
contracts, but in their system the dynamic semantics depends on the subtyping
whereas, conversely, in manifest contracts the subtyping refers to the dynamic
semantics.
In their work, casts just check that one type is a (gradual) subtype of the
other using the subtyping.
Hence, upcasts behave as identity functions naturally and upcast elimination is
less meaningful than in manifest contracts.
Instead, they showed that their calculus satisfies key properties in gradual
typing.

\subsection{Parametricity with Run-Time Analysis}
%
%
Neis et al.~\cite{Neis/Dreyer/Rossberg_2009_ICFP} proved that a language with
run-time type analysis can be parametric by generating fresh type names
dynamically.
Their language allows for run-time investigation of what types are substituted
for type variables.
By contrast, in {\fhfix} type variables are compatible with (possibly refined)
themselves and the run-time analysis on type variables is not allowed.

\subsection{Program Equivalence in Dependent Type Systems}
While type conversion in manifest contracts is performed explicitly by casts,
there are many dependent type systems where type conversion is implicit.
In such a system, term equivalence plays an important role to judge whether a
required type matches with an actual type.
To investigate an influence of term equivalence on dependent type checking, Jia
et al.~\cite{Jia/Zhao/Sjoberg/Weirich_2010_POPL} equipped a dependent type
system with various instances of equivalence.
In particular, they introduced \emph{untyped} contextual equivalence as an
instance.
Since the dependent type system rests on an instance of term equivalence, if
their contextual equivalence has been \emph{typed}, the same issue as in
Flanagan~\cite{Flanagan_2006_POPL} would happen, as discussed in
\sect{relwork-simply}.
Although we also use contextual equivalence for type conversion, our contextual
equivalence can refer to the type system without such an issue since it is given
after defining the calculus.
\section{Conclusion}
\label{sec:conclusion}

This paper has introduced semityped contextual equivalence, which relates a
well-typed term to a contextually equivalent, possibly ill-typed term, and
formulated the upcast elimination in a manifest contract calculus without subtyping.
We have also developed a logical relation for a polymorphic manifest contract
calculus with fussy cast semantics and show that it is sound with respect to
semityped contextual equivalence and complete for well-typed terms.
We have applied the logical relation to show the upcast elimination and correctness
of the selfification and the cast decomposition.
We are interested in extending the logical relation to step-indexed logical
relations~\cite{Ahmed_2006_ESOP}, which are used broadly for languages with
recursive types and mutable references, and studying bisimulation-based
reasoning for manifest contracts.

\bibliographystyle{plain}
\bibliography{main}

\end{document}